\documentclass{article}

\PassOptionsToPackage{numbers}{natbib} 

\usepackage[final]{neurips_2022}

\usepackage[utf8]{inputenc} 
\usepackage[T1]{fontenc}    
\usepackage{hyperref}       
\usepackage{url}            
\usepackage{booktabs}       
\usepackage{amsfonts}       
\usepackage{nicefrac}       
\usepackage{microtype}      
\usepackage{xcolor}         

\usepackage{amsmath} 
\usepackage{amsthm} 
\usepackage{graphicx} 
\usepackage{subcaption} 
\usepackage{tikz} 
\usetikzlibrary{positioning}
\tikzset{main node/.style={circle,fill=blue!20,draw,minimum size=0.5cm,inner sep=0pt},}
\usepackage[inline]{enumitem} 

\newtheorem{theorem}{Theorem}
\newtheorem*{theorem*}{Theorem}

\newtheorem*{lemma*}{Lemma}

\newcommand{\faggr}[0]{f_{\text{aggr}}} 
\newcommand{\fupdate}[0]{f_{\text{update}}} 
\newcommand{\fread}[0]{f_{\text{read}}} 

\newcommand{\changemarker}[1]{#1} 

\definecolor{bluecite}{HTML}{0875b7}
\hypersetup{
  colorlinks=true,
  citecolor=bluecite,
  linkcolor=magenta
}

\title{Universally Expressive Communication in Multi-Agent Reinforcement Learning}

\author{%
  Matthew Morris\\
  \changemarker{InstaDeep Ltd. \& University of Oxford}\\
  \texttt{\changemarker{matthew.morris@cs.ox.ac.uk}} \\
   \And
   Thomas D. Barrett\\
  InstaDeep Ltd.\\
  \texttt{t.barrett@instadeep.com} \\
  \And
   Arnu Pretorius \\
   InstaDeep Ltd. \\
   \texttt{a.pretorius@instadeep.com} \\
}

\begin{document}

\maketitle

\begin{abstract}
Allowing agents to share information through communication is crucial for solving complex tasks in multi-agent reinforcement learning.
In this work, we consider the question of whether a given communication protocol can express an arbitrary policy.
By observing that many existing protocols can be viewed as instances of graph neural networks (GNNs), we demonstrate the equivalence of joint action selection to node labelling.
With standard GNN approaches provably limited in their expressive capacity, we draw from existing GNN literature and consider augmenting agent observations with: (1) unique agent IDs and (2) random noise.
We provide a theoretical analysis as to how these approaches yield universally expressive communication, and also prove them capable of targeting arbitrary sets of actions for identical agents.
Empirically, these augmentations are found to improve performance on tasks where expressive communication is required, whilst, in general, the optimal communication protocol is found to be task-dependent.
\end{abstract}
\section{Introduction}

Communication lies at the heart of many multi-agent reinforcement learning (MARL) systems.  In MARL, multiple agents must account for each other's actions during both training and execution and, indeed, solving complex tasks in high-dimensional spaces often requires a cooperative joint policy that is difficult, or even impossible, to learn independently.  Therefore, allowing agents to share information is crucial and how best to achieve this has remained a keen area of research since the seminal proposals of learned communication by \citet{foerster2016learning} and \citet{sukhbaatar2016learning}. Whilst no single universally-adopted approach has emerged, considerations for MARL communication include inductive biases that aid learning. For example, an agent's policy should often not depend on the order in which messages are received at a given time step. i.e.\ be \emph{permutation invariant}.

In this context, graph neural networks (GNNs) provide a rich framework for MARL communication. It is natural to consider agents as nodes in a graph, with communication channels corresponding to edges between them. GNNs are specifically designed to respect this (typically non-Euclidian) structure \citep{bronstein2021geometric} and, indeed, many of the most successful MARL communication models fall within this paradigm, including CommNet \citep{sukhbaatar2016learning}, IC3Net \citep{singh2018learning}, GA-Comm \citep{liu2020multi}, MAGIC \citep{niu2021multi}, Agent-Entity Graph \citep{agarwal2019learning}, IP \citep{qu2020intention}, TARMAC \citep{das2019tarmac}, IMMAC \citep{sun2021intrinsic}, DGN \citep{jiang2018graph}, VBC \citep{zhang2019efficient}, MAGNet \citep{malysheva2018deep}, and TMC \citep{zhang2020succinct}. \changemarker{Other models such as ATOC \citep{jiang2018learning} and BiCNet \citep{peng2017multiagent} do not fall within the paradigm since they use LSTMs for combining messages, which are not permutation invariant, and models such as RIAL, DIAL \citep{foerster2016learning}, ETCNet \citep{hu2020event}, and SchedNet \citep{kim2019learning} do not since they used a fixed message-passing structure.} However, although traditional GNNs -- such as those used in MARL to date -- can readily provide permutation invariant communication, they are not universally expressive.

The expressivity of GNNs is often considered in the context of the 1-WL graph coloring algorithm~\citep{weisfeiler1968reduction}. In brief, 1-WL tests if two graphs are non-isomorphic by iteratively re-coloring the nodes and has been proven to \textit{not} be universally expressive (i.e.\ there exist non-isomorphic graphs that 1-WL can't distinguish). Moreover, \citet{morris2019weisfeiler} and \citet{xu2018powerful} proved that for any two non-isomorphic graphs indistinguishable by 1-WL, there is no GNN that can produce different outputs for those two graphs. An example of such graphs is given in Figure \ref{fig:1wl_graphs}. This direct correspondence between GNNs and 1-WL equivalently limits the expressivity of any MARL communication built on top of GNNs.  Whist higher-order GNN architectures which go beyond 1-WL expressivity have been proposed (see \citep{morris2021weisfeiler} for an overview), many of these models do not scale well and are computationally infeasible in practice.  However, recent works have shown that augmenting the node features can provide an alternative path to increased expressivity without computationally expensive architectural changes \citep{abboud2020surprising, dasoulas2019coloring}.  It is then natural to ask if, and how, these advancements can be brought into the MARL setting.

\begin{figure}
    \centering
    \caption{A pair of graphs indistinguishable by 1-WL}
    \label{fig:1wl_graphs}

    \begin{tikzpicture}
        \node[main node] (1) {};
        \node[main node] (2) [right = of 1]  {};
        \node[main node] (3) [right = of 2] {};
        \node[main node] (4) [right = of 3] {};
        \node[main node] (5) [below = of 4] {};
        \node[main node] (6) [left = of 5] {};
        \node[main node] (7) [left = of 6] {};
        \node[main node] (8) [left = of 7] {};
    
        \path[draw,thick]
        (1) edge node {} (2)
        (2) edge node {} (3)
        (3) edge node {} (4)
        (4) edge node {} (5)
        (5) edge node {} (6)
        (6) edge node {} (7)
        (7) edge node {} (8)
        (8) edge node {} (1)
        ;
        \begin{scope}[xshift=7cm]
        \node[main node, fill=orange] (1) {};
        \node[main node, fill=orange] (2) [right = of 1]  {};
        \node[main node, fill=orange] (3) [below = of 2] {};
        \node[main node, fill=orange] (4) [below = of 1] {};
    
        \path[draw,thick]
        (1) edge node {} (2)
        (2) edge node {} (3)
        (3) edge node {} (4)
        (4) edge node {} (1)
        ;
        \end{scope}
        
        \begin{scope}[xshift=10cm]
        \node[main node, fill=orange] (1) {};
        \node[main node, fill=orange] (2) [right = of 1]  {};
        \node[main node, fill=orange] (3) [below = of 2] {};
        \node[main node, fill=orange] (4) [below = of 1] {};
    
        \path[draw,thick]
        (1) edge node {} (2)
        (2) edge node {} (3)
        (3) edge node {} (4)
        (4) edge node {} (1)
        ;
        \end{scope}
    \end{tikzpicture}
\end{figure}
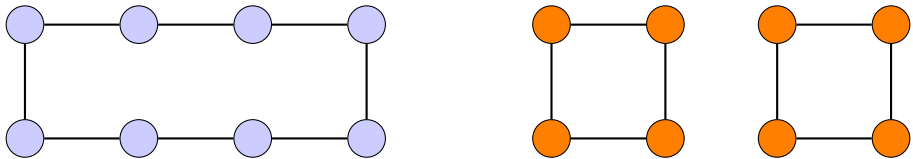

In this paper, we investigate the effectiveness of GNNs for universally expressive communication in MARL. We define \emph{Graph Decision Networks} (GDNs), a framework for MARL communication which captures many of the most successful methods. We highlight the correspondence between GDNs and the node labeling problem in GNNs, thus making concrete the limits of GDN expressivity. For moving beyond these limits, we consider two augmentations from the GNN literature -- random node initialization (RNI)~\citep{abboud2020surprising} and colored local iterative procedure (CLIP)~\citep{dasoulas2019coloring}, where random noise and unique labels are added to graph nodes, respectively. We provide a theoretical analysis as to how these algorithms yield universally expressive communication in MARL, and also prove their ability to solve coordination problems where the optimal policy requires arbitrary sets of actions from identical agents.
We then perform an empirical study where we augment several state-of-the-art MARL communication algorithms with RNI and unique labels.
By evaluating performance across both standard benchmarks and specifically designed tasks, we show that when complex non-local coordination or symmetry breaking is required, universally expressive communication can provide significant performance improvements.  However, in more moderate cases, augmented communication can reduce convergence speeds and result in suboptimal policies.
Therefore, whilst more expressive GNN architectures are required to improve performance on certain problems, a more complete picture relating expressivity to downstream performance remains an open question for future work.
\section{Background}
\paragraph{Multi-Agent Reinforcement Learning}
We consider the setting of Decentralized Partially Observable Markov Decision Processes \citep{oliehoek2012decentralized} augmented with communication between agents. At each timestep $t$ every agent $i \in \{ 1, ..., N \}$ gets a local observation $o^t_i$, takes an action $a^t_i$, and receives a reward $r^t_i$. We consider two agent paradigms: value-based \citep{tampuu2017multiagent} and actor-critic \citep{foerster2018counterfactual, lowe2017multi}. For brevity, we collectively refer to the policy network in actor-critic methods and the Q-network in value-based methods as the \emph{actor network}. In this paper, we consider by default parameter sharing between agent's networks, which is often used to yield faster and more stable training in MARL \citep{foerster2016learning, gupta2017cooperative, rashid2018qmix, yang2018mean}.

\paragraph{Graph Neural Networks}
GNNs can refer to a large variety of models; in this paper, we define the term to correspond to the definition of Message Passing Neural Networks (MPNNs) by \citet{gilmer2017neural}, which is the most common GNN architecture. Notable instances of this architecture include Graph Convolutional Networks (GCNs) \citep{duvenaud2015convolutional}, GraphSAGE \citep{hamilton2017inductive}, and Graph Attention Networks (GATs) \citep{velivckovic2017graph}. A GNN consists of multiple message-passing layers, each of which updates the node attributes / labels (terms used interchangeably). For layer $m$ and node $i$ with current attribute $v_i^m$, the new attribute $v_i^{m+1}$ is computed as

$$ v_i^{m+1} := \fupdate^{\theta_m}(v_i^{m},~ \faggr^{\theta_m'}(\{ v_j^m ~|~ j \in N(i) \}),~ \fread^{\theta_m''}(\{ v_j^m ~|~ j \in V(G) \}) ) $$

where $N(i)$ is all nodes with edges connecting to $i$ and $\theta_m, \theta_m', \theta_m''$ are the (possibly trainable) parameters of the update, aggregation, and readout functions for layer $m$. Parameters may be shared between layers, e.g. $\theta_0 = \theta_1$. The functions $\faggr^{\theta_m'}, \fread^{\theta_m''}$ are permutation invariant. Importantly, GNNs are invariant / equivariant graph functions.
\section{Expressivity of Multi-Agent Communication} \label{sec:theory}

\subsection{Graph Decision Networks} \label{sec:gdn}
Many of the most successful MARL communication methods can be captured within the following framework. At each time step, define an attributed graph $G = \changemarker{(V, E)}$ with nodes $V(G) := \{ \text{all agents} \}$, edges $E(G) := \{ (i, j) ~|~ \text{agent $i$ is communicating with $j$} \}$, and for all agents $i$, the node $i$ is labeled with the observation of $i$. This graph is passed through a GNN $f$ which outputs values for each node and passes each resulting node value through the actor network of the corresponding agent. Assuming that the actor networks use shared weights \changemarker{(i.e.\ the same neural network is used for each actor)}, we can substitute them for a final GNN layer $M$, where $\fupdate^{\theta_{M}}(v, \sim) := P(v)$ and $P$ represents the shared actor network. We refer to communication methods that fall within this paradigm as \emph{graph decision networks} (GDNs). The framework is illustrated visually in Figure \ref{fig:theory:gdns}.

\begin{figure}
    \centering
    \caption{The Graph Decision Network (GDN) framework}
    \label{fig:theory:gdns}
    
    \includegraphics[width=0.9\linewidth]{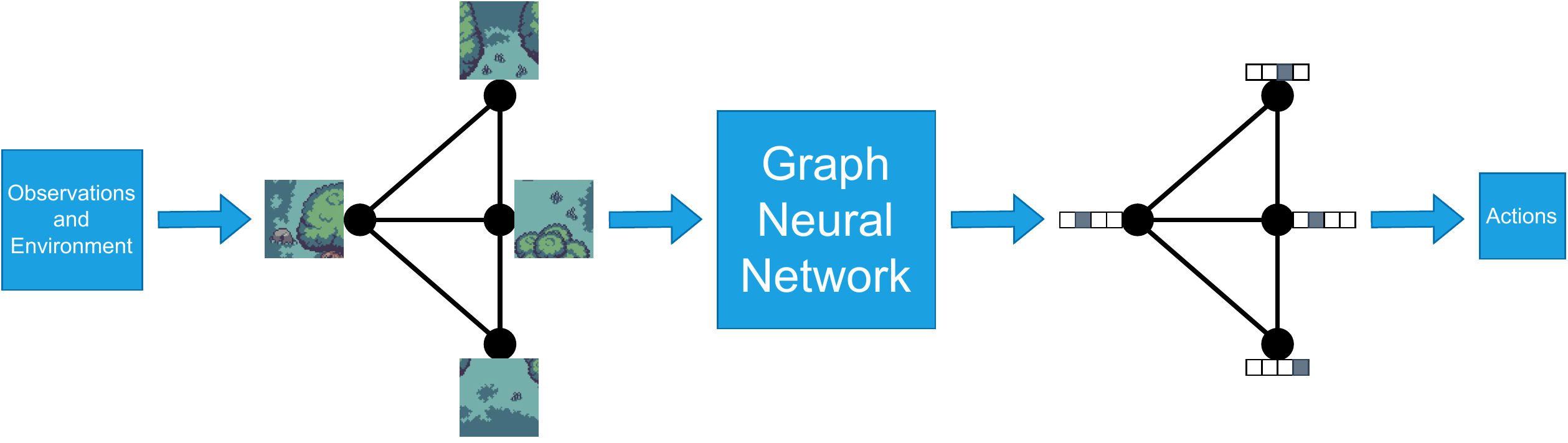}
\end{figure}

Given the above assumption of shared weights, any GDN simply reduces to a GNN node labelling problem, where the correct label for a given node is the corresponding actor network output that collectively maximizes the joint reward (or the individual reward, depending on how the agent is trained). Going forward, we only deal with such GDNs (ones with a shared actor network). In the case of stochastic policies \citep{foerster2018counterfactual, yu2021surprising}, the target labels are parameters of output distributions, instead of atomic actions. This applies to both discrete and continuous distributions. Note that this is not how RL agents are actually trained (i.e. they use reward signals, not supervised learning).

However, given that this paper aims to analyze expressivity, we argue that it does not matter how the GDN is trained. All that expressivity is concerned with is the \emph{ability} of a model to produce a certain output, not how the training paradigm causes the model to converge to the solution. All we have to know is that there are ``optimal'' actor network outputs for each agent, under some metric of optimality, and then we can reason about the ability of the model to provide these outputs. Scenarios with heterogeneous agents can still be considered within this paradigm, by allocating a portion of the observations to indicate the agent type (e.g. through a one-hot encoding) \citep{terry2020revisiting}. Models with recurrent networks also fall within the paradigm, where the hidden or cell states for the networks can be considered as part of the agent observations.

Due to the reduction of GDNs to a GNN node labelling problem, GDNs suffer from the same expressivity limits as GNNs, about which there is a plethora of work \citep{barcelo2020logical, chen2020can, garg2020generalization, loukas2019graph, loukas2020hard, morris2019weisfeiler, nt2019revisiting, oono2019graph, xu2018powerful}. These are expanded upon in Appendix \ref{sec:related_work}. For our analysis, we focus particularly on ways to achieve universal Weisfeiler-Lehman expressivity, but note that the above reduction unlocks many tools for reasoning about the expressivity of GDNs.

\subsection{Desired Properties of MARL Communication}
Whilst conventional GDNs cannot capture functions with expressive power beyond 1-WL \citep{morris2019weisfeiler, xu2018powerful}, recent GNN architectures have been proposed to achieve expressivity beyond 1-WL, even ones which are able to express any equivariant graph function. We can use these insights to construct more expressive GDNs. However, we note that classes of models which always yield equivariant functions are not necessarily desirable, since they cannot break symmetries between agents when required. Many MARL environments require agents to coordinate, needing some joint action to solve the task. However, if agents have identical observations and communication graph structure in a pure GDN framework, there is no way for them to disambiguate between each other and distribute the required actions amongst themselves. For a simple example, consider a setting where two agents have identical observations but must take opposite actions -- then the only way for them to solve the environment is to communicate in such a way that they can break this symmetry and take different actions from one another. This example is illustrated in Figure \ref{fig:theory:symmetry_breaking}.

\begin{figure}
    \centering
    \caption{A simple example of symmetry breaking}
    \label{fig:theory:symmetry_breaking}
    
    \includegraphics[width=0.9\linewidth]{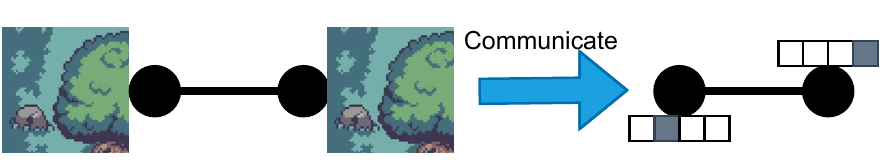}
\end{figure}

More formally, since GNNs are equivariant graph functions, GDNs are equivariant functions on the agent observations and communication graph structure. This means that agents within the same \emph{graph orbit} will always produce the same output. For a graph $G$, consider two nodes $u, v \in G$. If there exists an automorphism $\alpha$ of $G$ such that $\alpha(u) = v$, then $u$ and $v$ are said to be \emph{similar nodes}. The relation \emph{is similar to} forms an equivalence relation on the nodes of G. Each equivalence class is called an \emph{orbit}. Intuitively, every node in an orbit ``has the same structure''. We denote the set of all orbits of $G$ by $R(G)$: this forms a partition of $V(G)$.

\begin{theorem} \label{thm:equi_same}
Given a GDN $f$, observations $O = \{ o_1, ..., o_n \}$, and communication graph $G$ such that nodes $i$ and $j$ are similar in $G$ and $o_i = o_j$, then it holds that $f(O)_i = f(O)_j$.
\end{theorem}

Full proofs for all theorems in this paper can be found in Appendix \ref{app:proofs}. We formally state the desired behaviour of GDNs -- which we refer to as \emph{symmetry breaking} -- that would enable them to solve such coordination problems. Given a graph $G$ with orbits $R(G)$, a GDN $g$ ought to be able to produce, or \emph{target}, a \changemarker{multiset} of labels $A_k$ for each orbit $r_k$:

$$\forall r_k \in R(G), \{ g(G)_i ~|~ i \in r_k \} = A_k, $$

where $g(G)_i$ is the output of $g(G)$ for agent $i$. Thus, ideally, MARL communication methods should possess all the following properties: (1) universal expressivity for equivariant graph functions, (2) symmetry breaking for coordination problems, and (3) computational efficiency. We apply two existing GNN augmentations to GDNs to achieve this, both of which come with minimal extra computational cost. \changemarker{In the following section, we provide theorems which prove that the first two properties are satisfied by these augmentations.}

\subsection{Expressive Graph Decision Networks}

\paragraph{Random Node Initialization}
\citet{sato2021random} propose augmenting GNNs with \emph{random node initialization} (RNI), where for each node in the input graph, a number of randomly sampled values are concatenated to the original node attribute. For all graphs / nodes, the random values are sampled from the same distribution. \citet{abboud2020surprising} prove that such GNNs are universal and can approximate any permutation invariant / equivariant graph function. Technically, random initialization breaks the node invariance in GNNs, since the result of the message passing will depend on the structure of the graph as well as the values of the random initializations. However, when one views the model as computing a random variable, the random variable is still invariant when using RNI. In expectation, the mean of random features will be used for GNN predictions, and is the same across each node. However, the variability of the random samples allows the GNN to discriminate between nodes that have different random initializations, breaking the 1-WL upper bound.

\citet{abboud2020surprising} formally state and prove a universal approximation result for invariant graph functions. They note that it can be extended to equivariant functions, which is what GDNs are. As such, we adapt and state the theorem for equivariant functions. Let $G_n$ be the class of all $n$-node graphs. Let $f: G_n \to \mathbb{R}^{n}$, a graph function which outputs a real value for each node in $V(G)$. We say that a randomized function $X$ that associates with every graph $G \in G_n$ a sequence of random variables $X_1(G), X_2(G), ..., X_{n}(G)$, one for each node, is an $(\epsilon, \delta)$-\emph{approximation} of $f$ if for all $G \in G_n$ it holds that $\forall i \in \{ 1, 2, ..., n \}$, $\text{Pr}(|f(G)_i - X_i(G)| \leq \epsilon) \geq 1 - \delta$, where $f(G)_i$ is the output of $f(G)$ for node $i$. Note that a GNN $h$ with RNI computes such functions $X$. If $X$ is computed by $h$, we say that $h$ $(\epsilon, \delta)$-\emph{approximates} $f$. We can now state the following theorem:

\begin{theorem} \label{thm:rni_equivariant}
Let $n \geq 1$ and let $f: G_n \to \mathbb{R}^{n}$ be equivariant. Then for all $\epsilon, \delta > 0$, there is a GNN with RNI that $(\epsilon, \delta)$-\emph{approximates} $f$.
\end{theorem}

Such GNNs are also able to solve symmetry-breaking coordination problems by using RNI to disambiguate between otherwise identical agents. To formally state this property, we need to define $(\epsilon,\delta)$-approximation for sets. We say that two \changemarker{multisets} $A, B$ containing random variables are $(\epsilon,\delta)$-equal, denoted $A \cong_{\epsilon,\delta} B$, if there exists a bijection $\tau: A \to B$ such that $\forall a \in A,~ \text{Pr}(| a - \tau(a) | \leq \epsilon) \geq 1 - \delta$.

\begin{theorem} \label{thm:rni_coordination}
Let $n \geq 1$ and consider a set $T$, where each $(G, A) \in T$ is a graph-labels pair, such that $G \in G_n$ and there is a \changemarker{multiset} of target labels $A_k \in A$ for each orbit $r_k \in R(G)$, with $| A_k | = | r_k |$. Then for all $\epsilon, \delta > 0$ there is a GNN with RNI $g$ which satisfies:

$$\forall (G, A) \in T ~~ \forall r_k \in R(G), \{ g(G)_i ~|~ i \in r_k \} \cong_{\epsilon,\delta} A_k $$
\end{theorem}

In the GDN case, adding RNI means concatenating noise to the agent observations, thus achieving universal approximation and enabling the solving of symmetry-breaking coordination problems.

\paragraph{Unique Node Identifiers}

\citet{dasoulas2019coloring} augment GNNs with a coloring scheme to define \emph{colored local iterative procedure} (CLIP). They use colors to differentiate otherwise identical node attributes, with $k$-CLIP corresponding to $k$ different colorings being sampled and maximized over. They prove theoretically that when maximizing over all such possible colorings, $\infty$-CLIP can represent any invariant graph function.

Assigning nodes unique IDs is equivalent to 1-CLIP, since this guarantees that every node with identical attributes will have a unique ``color'': its particular unique ID. Therefore, we can leverage the universality result for 1-CLIP (Theorem 4 in \citep{dasoulas2019coloring}), which states that with any given degree of precision, 1-CLIP can approximate any invariant graph function. However, \citet{dasoulas2019coloring} note that such solutions may be difficult to converge to and require a large number of training steps in practice. Intuitively, this is because the GNN has to learn to deal with $n!$ permutations of unique IDs. Similarly to the RNI case, we extend their theorem to equivariant functions.

\begin{theorem} \label{thm:clip_universal}
Let $n \geq 1$ and let $f: G_n \to \mathbb{R}^{n}$ be equivariant. Then for all $\epsilon > 0$, there is a GNN with unique node IDs that $\epsilon$-\emph{approximates} $f$.
\end{theorem}

Such GNNs can also solve symmetry-breaking coordination problems, in a similar way to ones with RNI. We say that two \changemarker{multisets} $A, B$, which do not contain random variables, are $\epsilon$-equal, denoted $A \cong_{\epsilon} B$, if there exists a bijection $\tau: A \to B$ such that $\forall a \in A,~ | a - \tau(a) | \leq \epsilon$.

\begin{theorem} \label{thm:clip_coordination}
Let $n \geq 1$ and consider a set $T$, where each $(G, A) \in T$ is a graph-labels pair, such that $G \in G_n$ and there is a \changemarker{multiset} of target labels $A_k \in A$ for each orbit $r_k \in R(G)$, with $| A_k | = | r_k |$. Then for all $\epsilon > 0$ there is a GNN with unique node IDs $g$ which satisfies:

$$\forall (G, A) \in T ~~ \forall r_k \in R(G), \{ g(G)_i ~|~ i \in r_k \} \cong_{\epsilon} A_k $$
\end{theorem}

In the GDN case, this means that by giving each agent a unique ID in its observations, we can achieve universal approximation and enabling the solving of symmetry-breaking coordination problems.
\section{Experiments}
\subsection{Methods} \label{sec:methods}
\paragraph{Baselines}
For evaluation, we adopt a diverse selection of MARL communication methods which fall under the GDN paradigm. These are shown in Table \ref{tab:baseline_models}, along with the communication graph structure, agent model, and GNN architecture. We use the code provided by \citet{jiang2018graph, niu2021multi} as starting points. All of the implementations are extended to support multiple rounds of message-passing and the baselines are augmented with the ability for their communication to be masked by the environment (e.g.\ based on distance or obstacles in the environment). We fix the number of message-passing rounds to be 4 and otherwise use the original models and hyperparameters from \citep{jiang2018graph, niu2021multi}. Full experiment and hyperparameter details can be found in Appendix \ref{app:experiments}, and full results are shown in Appendix \ref{app:results}.

\begin{table}
  \caption{Architecture of the Baselines}
  \label{tab:baseline_models}
  \centering
  \begin{tabular}{llll}
    \toprule
    Name & Communication Graph & Agents & GNN Architecture \\
    \midrule
    CommNet \citep{sukhbaatar2016learning} & Complete (or environment-based) & Recurrent A2C & Sum Aggregation \\
    IC3Net \citep{singh2018learning} & Complete + Gating & Recurrent A2C & Sum Aggregation \\
    TarMAC \citep{das2019tarmac} & Complete + Learned Soft Edges & Recurrent A2C & GAT \\
    T-IC3Net \citep{singh2018learning, das2019tarmac} & Gating + Learned Soft Edges & Recurrent A2C & GAT \\
    MAGIC \citep{niu2021multi} & Learned & Recurrent A2C & GAT \\
    DGN \citep{jiang2018graph} & Environment-based & Q-network & GCN \\
    \bottomrule
  \end{tabular}
\end{table}

\paragraph{Environments}
\textbf{Predator-Prey} \citep{das2019tarmac, li2020deep, liu2020multi, niu2021multi, singh2018learning} and \textbf{Traffic Junction} \citep{das2019tarmac, li2020deep, liu2020multi, niu2021multi, singh2018learning, sukhbaatar2016learning} are common MARL communication benchmarks. In Predator-Prey, predator agents are tasked with capturing prey and in Traffic Junction, agents need to successfully navigate a traffic intersection (full descriptions of each environment are given in Appendix \ref{app:environments}). We perform evaluations on these benchmarks to test how well our universally expressive GDN models perform when there is not necessarily a benefit to having communication expressivity beyond 1-WL. We also introduce two new environments, Drone Scatter and Box Pushing, to respectively test symmetry-breaking and communication expressivity beyond 1-WL.

\textbf{Drone Scatter} consists of 4 drones in a homogeneous field surrounded by a fence. Their goal is to move around and find a target hidden in the field, which they can only notice when they get close to. The drones do not have GPS and can only see directly beneath them using their cameras, as well as observing their last action. The best way for them to locate the target is to split up and search in different portions of the field, despite them all having the same observations; thus, they are given rewards for splitting up.

\textbf{Box Pushing} consists of 10 robots in a 12x12 construction site, which has boxes within that need to be moved to the edge of the site: the clearing area. Robots attach themselves to boxes before they can move them; when attached, robots can no longer see around themselves. Free-roaming robots can communicate with any other free-roaming robots, but attached robots can only communicate with the robots directly adjacent to them. The environment either spawns with one large box or two small boxes and agents spawn already attached. 4 attached robots all moving in the same direction are needed to move a small box, and 8 all power moving in the same direction to move a large box. To solve the environment, the robots need to be able to communicate with each other to figure out which type of box they are on and all push correctly, at the same time, and in the same direction. Since the communication graphs corresponding to the scenarios with small and large boxes are 1-WL indistinguishable, communication beyond 1-WL is needed to optimally solve the environment.

\paragraph{Evaluation Procedure}
We augment baseline communication methods with RNI and unique IDs to perform our evaluations. Agent IDs are represented by one-hot encodings and ``0.25 RNI'' refers to 25\% of the observation space being randomly initialized. We sample each RNI value uniformly from $[-1, 1]$. For each scenario and for every baseline communication method, we compare 4 models: the baseline without modifications, the baseline augmented with unique IDs for each agent, the baseline augmented with 0.75 RNI, and finally 0.25 RNI. The only exceptions are the Drone Scatter evaluations, where 0.25 RNI is not used since the observation space is not large enough, and the Drone Scatter experiments using stochastic evaluation, where DGN is not used since it does not support stochastic evaluation.

For each run, corresponding to a random initialization (one seed) of the model in question, we perform periodic evaluations during training. Each epoch consists of 5000 training episodes, after which 100 evaluation episodes are used to report aggregate metric scores, yielding an evaluation score for the model after every epoch. Following an established practice in MARL evaluation \citep{foerster2018counterfactual, hu2019simplified, papoudakis2020benchmarking, saeed2021domain, weber2022remember, yu2021surprising, zhao2022dqmix}, we take the value of a metric for a run to be the \emph{best} value achieved during training, so that our metrics are robust against runs which converge at some point and then degrade in performance as they continue to train. In such cases, one would use the parameters from the best performing model found during training for real-world evaluation; thus, that performance makes more sense to report than the model performance once training has finished. We utilize 10 seeds for Box Pushing experiments and 5 for all others. For each scenario, metric, baseline communication method, and variant thereof, we report the mean metric value across all seeds and a 95\% confidence interval. To calculate the confidence interval, we assume a normal distribution and compute the interval as $1.96 \times \text{SEM}$ (standard error of the mean). Finally, for all Box Pushing experiments, we make use of a form of hybrid imitation learning to help deal with exceptionally sparse rewards (full details are given in Appendix \ref{app:hybrid_learning}).

\subsection{Results} \label{sec:results}
\paragraph{Benchmark Environments}
Experimental results on the benchmark environments are shown in Table \ref{tab:results:tj_easy} for Easy Traffic Junction, Table \ref{tab:results:pp} for Predator-Prey, and Table \ref{tab:results:tj_medium} for Medium Traffic Junction. In general, unique IDs tends to perform comparably to the baseline. The only exception to this is for IC3Net on Medium Traffic Junction, where unique IDs struggle.

0.75 RNI is categorically the worst method, consistently getting outperformed by all other methods and only coming out on top for MAGIC on Medium Traffic Junction, which is not significant due to the instability of that set of results. 75\% of observations being randomly initialized appears far too much for the system to be able to learn effective policies. However, universality results still hold for lower ratios of RNI.

0.25 RNI exhibits strong performance on Easy Traffic Junction, always solving the environment and almost always outperforming the baseline. However, on sparse-reward problems (such as Predator-Prey and, to a lesser extent, Medium Traffic Junction) RNI methods typically take longer to converge than the baseline and unique IDs, and 0.25 RNI can struggle to reach the performance of baseline methods.  This aligns with \citet{abboud2020surprising}'s observation that GNNs with RNI take significantly longer to converge than normal GNNs. This is only is exacerbated in a MARL setting with sparse reward signals, where slow convergence is expected regardless of the RNI augmentation. Indeed, on all examples, RNI methods typically take longer to converge than the baselines and unique IDs.

Overall, we conclude that both unique IDs and 0.25 RNI achieve sufficient performance on the benchmarks to qualify them for use, especially given that the extra expressivity they provide is not strictly necessary. With respect to the different baselines, we note that simple baselines such as CommNet work the best when the optimal policy is also simple, such as for Easy Traffic Junction, but that more sophisticated baselines outperform them on the complex environments. We also note the very unstable performance of MAGIC for the Traffic Junction environments.

%
%
\begin{table}
\small 
\caption{Mean and 95\% confidence interval for Easy Traffic Junction across all baselines}
\label{tab:results:tj_easy}
\centering
\begin{tabular}{@{}rrrrrrrrrrrrrr@{}}
    \toprule
    Baseline & Metric & \textbf{Baseline} && \textbf{Unique IDs} && \textbf{0.75 RNI} && \textbf{0.25 RNI} \\
    \midrule
    
CommNet & Success & \pmb{$ 1 \pm 0 $} && \pmb{$ 1 \pm 0 $} && \pmb{$ 1 \pm 0 $} && \pmb{$ 1 \pm 0 $} \\
DGN & Success & $ 0.987 \pm 0 $ && $ 0.99 \pm 0 $ && $ 0.848 \pm 0.15 $ && \pmb{$ 0.996 \pm 0 $} \\
IC3Net & Success & \pmb{$ 1 \pm 0 $} && \pmb{$ 1 \pm 0 $} && \pmb{$ 1 \pm 0 $} && $ 0.986 \pm 0.02 $ \\
MAGIC & Success & $ 0.634 \pm 0.11 $ && $ 0.764 \pm 0.13 $ && $ 0.684 \pm 0.11 $ && \pmb{$ 0.787 \pm 0.09 $} \\
TarMAC & Success & $ 0.994 \pm 0.01 $ && \pmb{$ 1 \pm 0 $} && $ 0.933 \pm 0.04 $ && \pmb{$ 1 \pm 0 $} \\
T-IC3Net & Success & \pmb{$ 1 \pm 0 $} && $ 0.998 \pm 0 $ && $ 0.94 \pm 0.04 $ && $ 0.974 \pm 0.04 $ \\

    \bottomrule
\end{tabular}
\end{table}

%
%
\begin{table}
\small 
\caption{Mean and 95\% confidence interval for Predator-Prey across all baselines}
\label{tab:results:pp}
\centering
\begin{tabular}{@{}rrrrrrrrrrrrrr@{}}
    \toprule
    Baseline & Metric & \textbf{Baseline} && \textbf{Unique IDs} && \textbf{0.75 RNI} && \textbf{0.25 RNI} \\
    \midrule
    
CommNet & Success & $ 0.88 \pm 0.03 $ && \pmb{$ 0.908 \pm 0.02 $} && $ 0.194 \pm 0.02 $ && $ 0.476 \pm 0.05 $ \\
DGN & Success & $ 0.014 \pm 0 $ && $ 0.016 \pm 0 $ && $ 0.026 \pm 0.03 $ && \pmb{$ 0.032 \pm 0.01 $} \\
IC3Net & Success & \pmb{$ 0.952 \pm 0 $} && $ 0.93 \pm 0.02 $ && $ 0.454 \pm 0.08 $ && $ 0.933 \pm 0.02 $ \\
MAGIC & Success & \pmb{$ 0.892 \pm 0.02 $} && $ 0.888 \pm 0.05 $ && $ 0.112 \pm 0.03 $ && $ 0.451 \pm 0.09 $ \\
TarMAC & Success & $ 0.169 \pm 0.09 $ && \pmb{$ 0.24 \pm 0.11 $} && $ 0.068 \pm 0.01 $ && $ 0.086 \pm 0.02 $ \\
T-IC3Net & Success & \pmb{$ 0.938 \pm 0.02 $} && \pmb{$ 0.938 \pm 0.01 $} && $ 0.27 \pm 0.02 $ && $ 0.913 \pm 0.02 $ \\

    \bottomrule
\end{tabular}
\end{table}

%
%
\begin{table}
\small 
\caption{Mean and 95\% confidence interval for Medium Traffic Junction across all baselines}
\label{tab:results:tj_medium}
\centering
\begin{tabular}{@{}rrrrrrrrrrrrrr@{}}
    \toprule
    Baseline & Metric & \textbf{Baseline} && \textbf{Unique IDs} && \textbf{0.75 RNI} && \textbf{0.25 RNI} \\
    \midrule
    
CommNet & Success & $ 0.761 \pm 0.31 $ && \pmb{$ 0.793 \pm 0.33 $} && $ 0.046 \pm 0 $ && $ 0.614 \pm 0.11 $ \\
DGN & Success & \pmb{$ 1 \pm 0 $} && \pmb{$ 1 \pm 0 $} && $ 0.062 \pm 0 $ && $ 0.619 \pm 0.4 $ \\
IC3Net & Success & \pmb{$ 0.971 \pm 0.04 $} && $ 0.804 \pm 0.1 $ && $ 0.588 \pm 0.03 $ && $ 0.855 \pm 0.13 $ \\
MAGIC & Success & $ 0.551 \pm 0.28 $ && $ 0.526 \pm 0.33 $ && \pmb{$ 0.734 \pm 0.21 $} && $ 0.4 \pm 0.35 $ \\
TarMAC & Success & \pmb{$ 0.064 \pm 0 $} && $ 0.052 \pm 0 $ && $ 0.05 \pm 0 $ && $ 0.054 \pm 0.01 $ \\
T-IC3Net & Success & $ 0.89 \pm 0.17 $ && $ 0.909 \pm 0.08 $ && $ 0.362 \pm 0.18 $ && \pmb{$ 0.962 \pm 0.02 $} \\

    \bottomrule
\end{tabular}
\end{table}

\begin{figure}
    \centering
    \caption{\changemarker{Training curves for IC3Net and CommNet on the benchmark communication environments}}
    \label{fig:results:benchmarks}
    
    \includegraphics[width=0.32\linewidth]{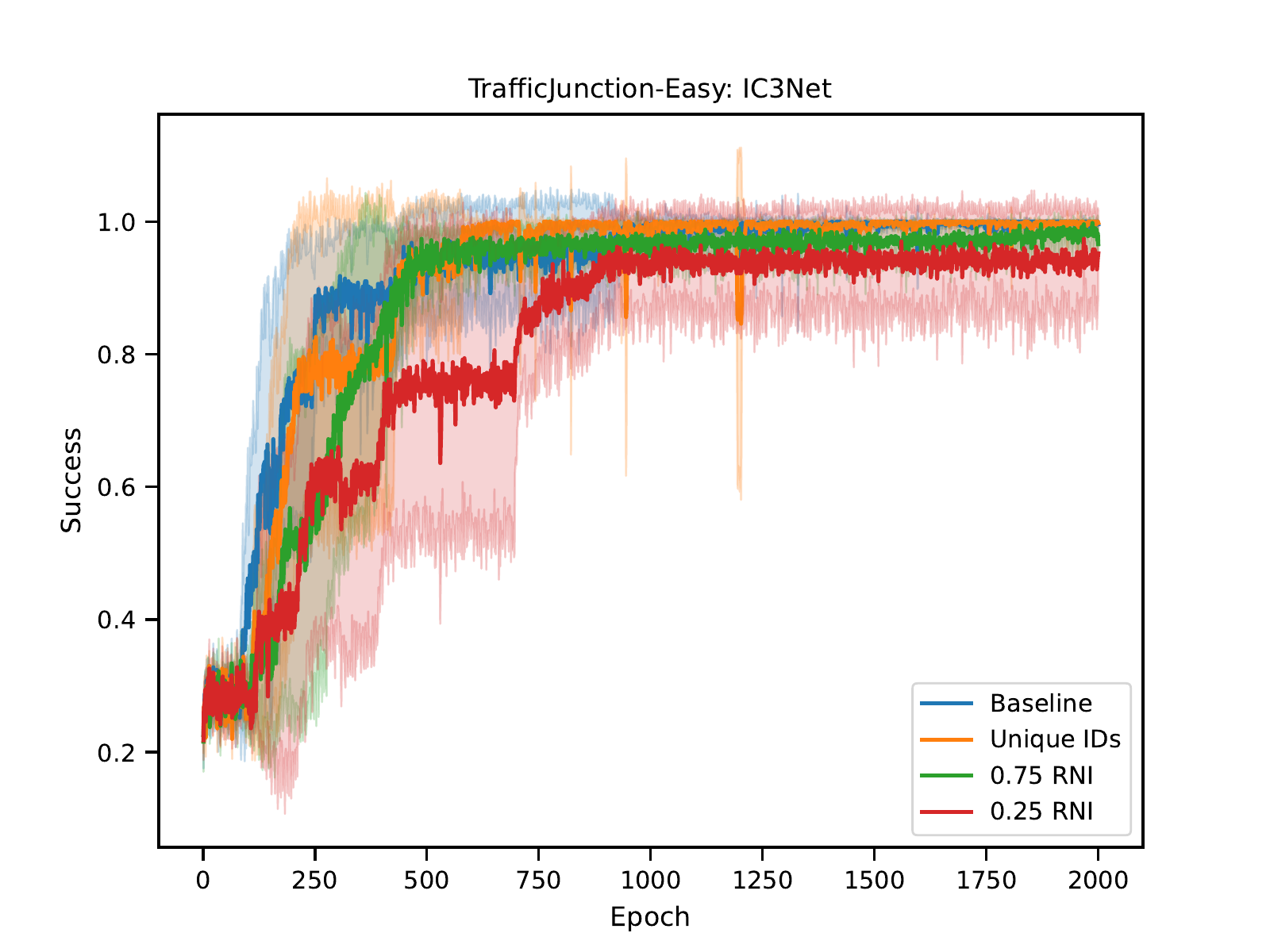}
    \includegraphics[width=0.32\linewidth]{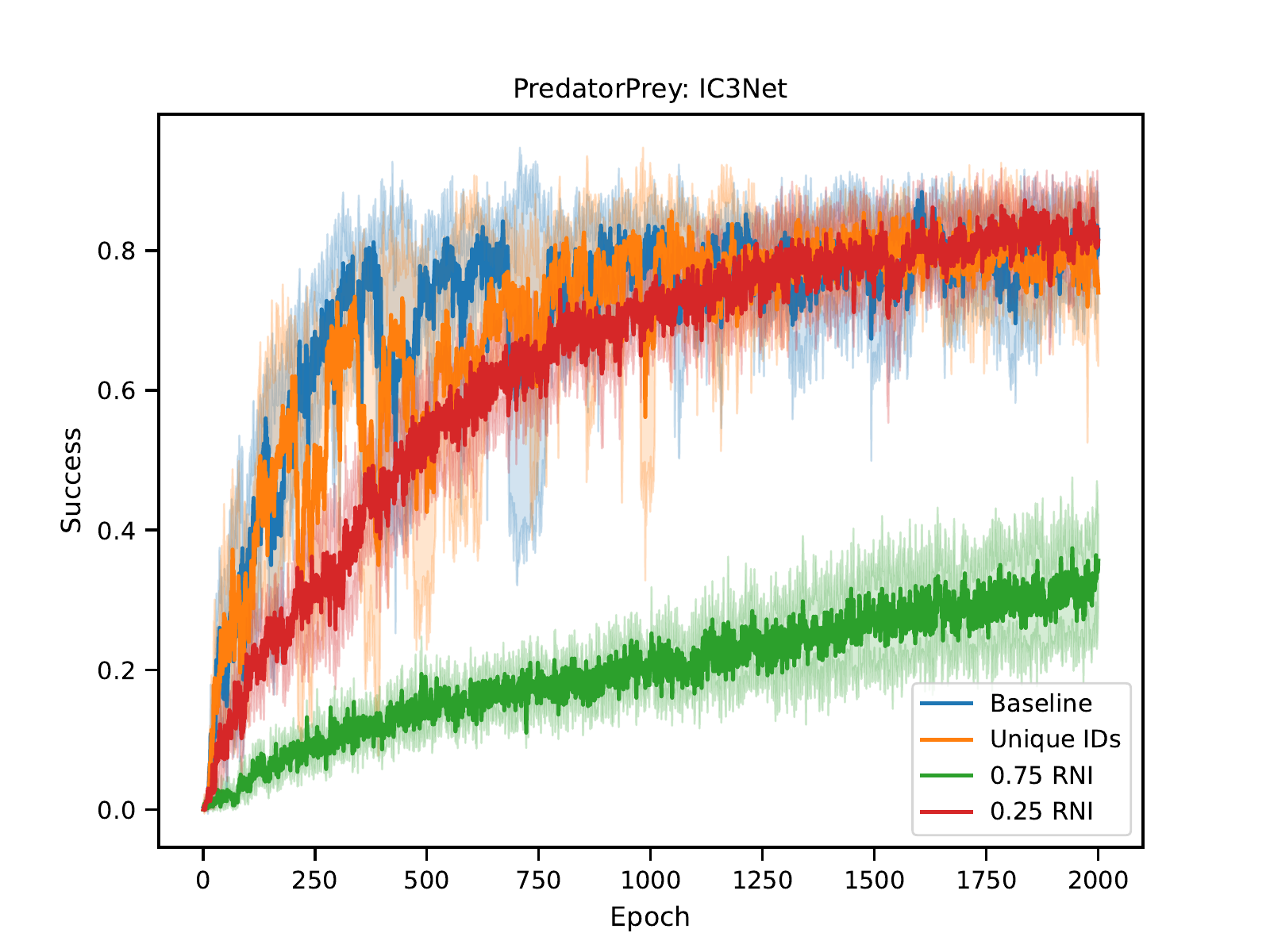}
    \includegraphics[width=0.32\linewidth]{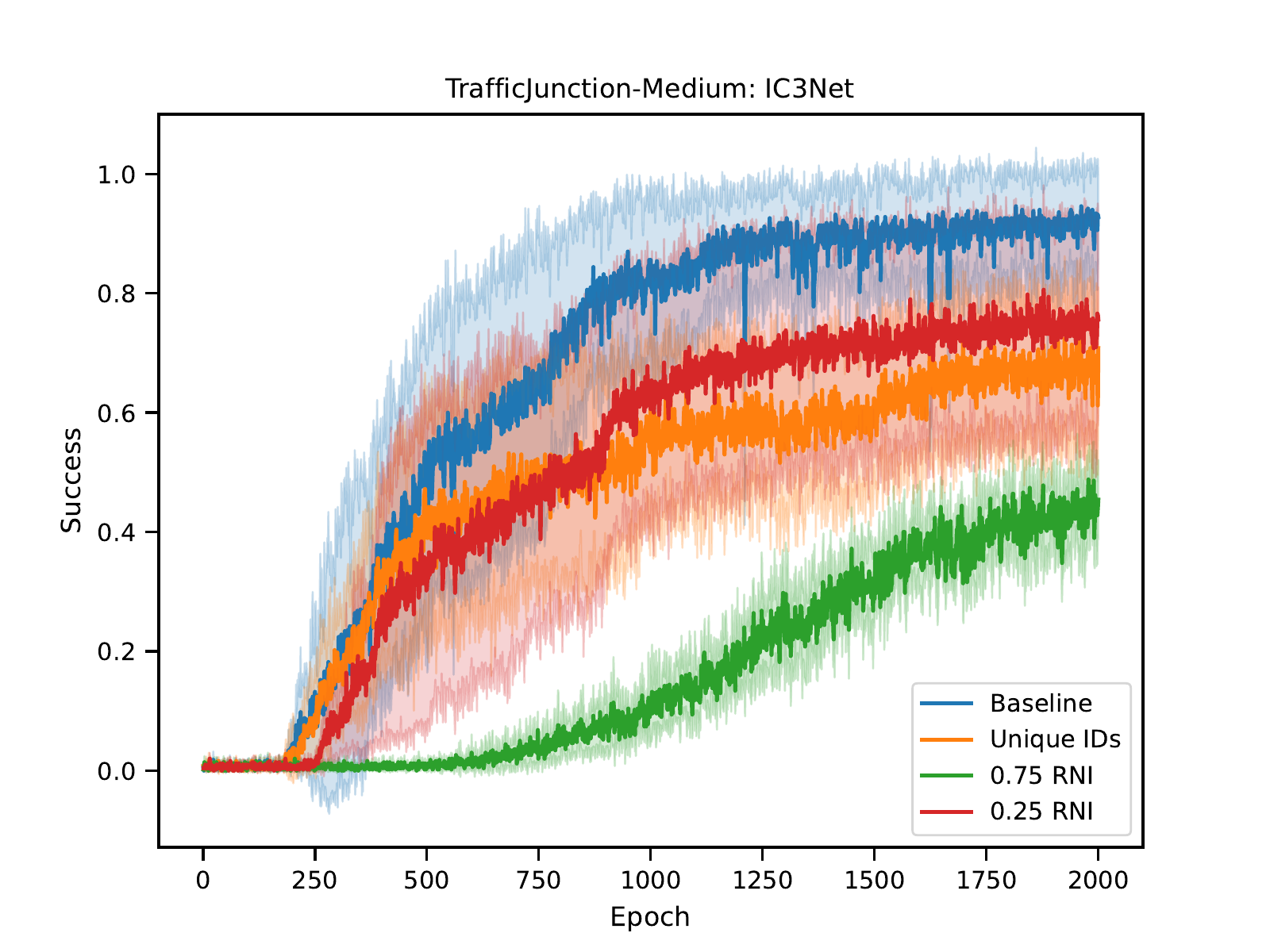}

    \includegraphics[width=0.32\linewidth]{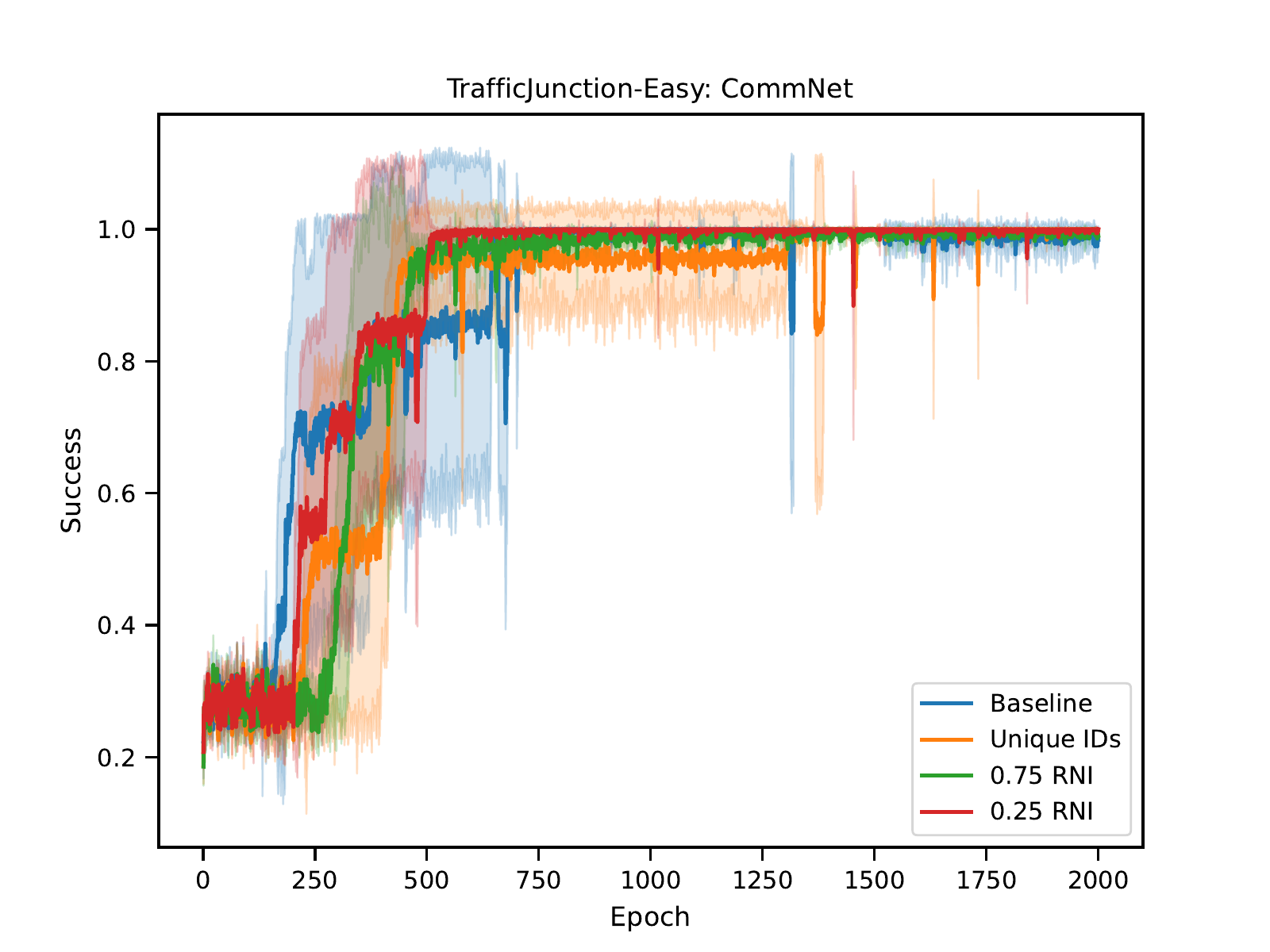}
    \includegraphics[width=0.32\linewidth]{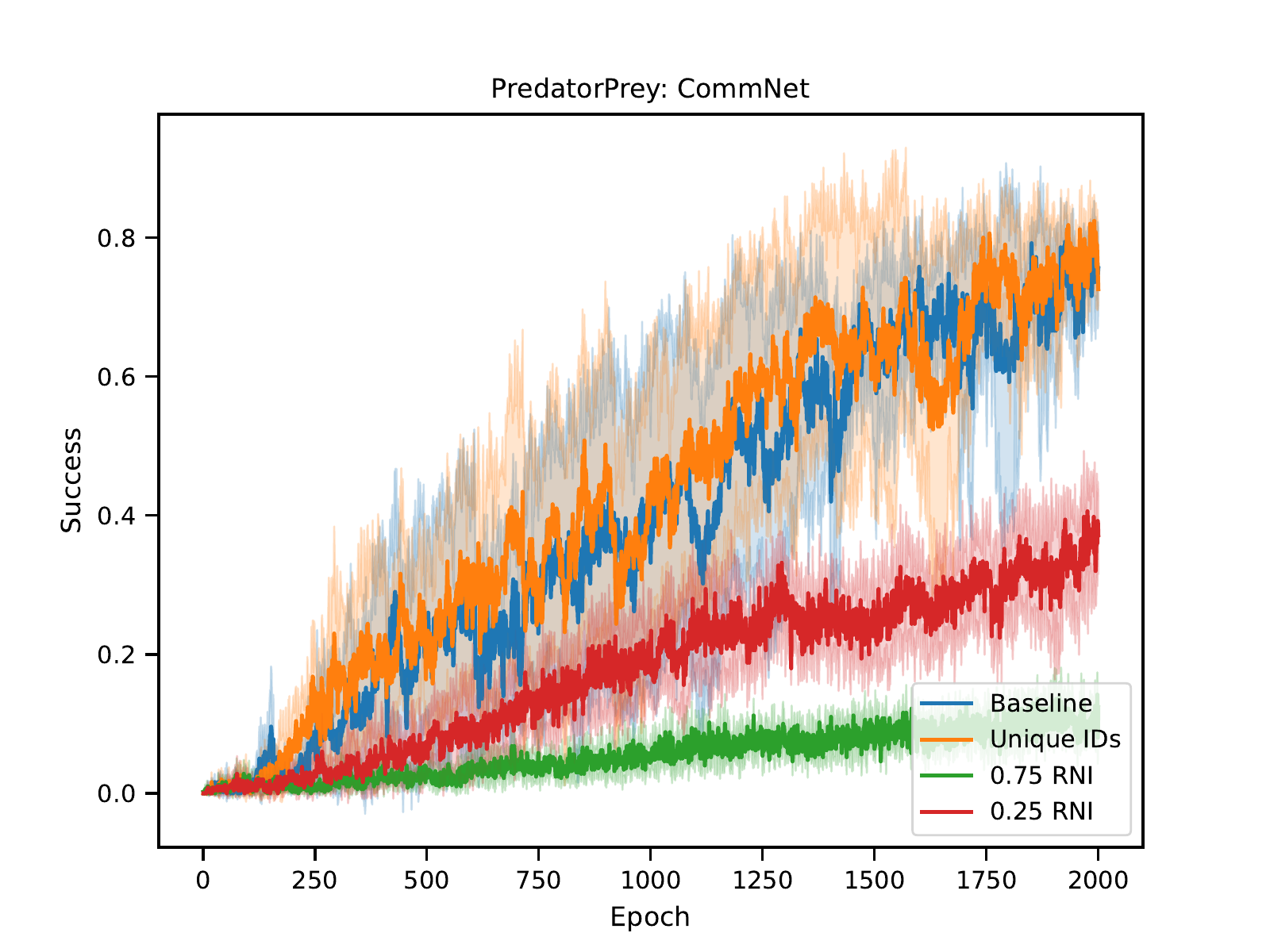}
    \includegraphics[width=0.32\linewidth]{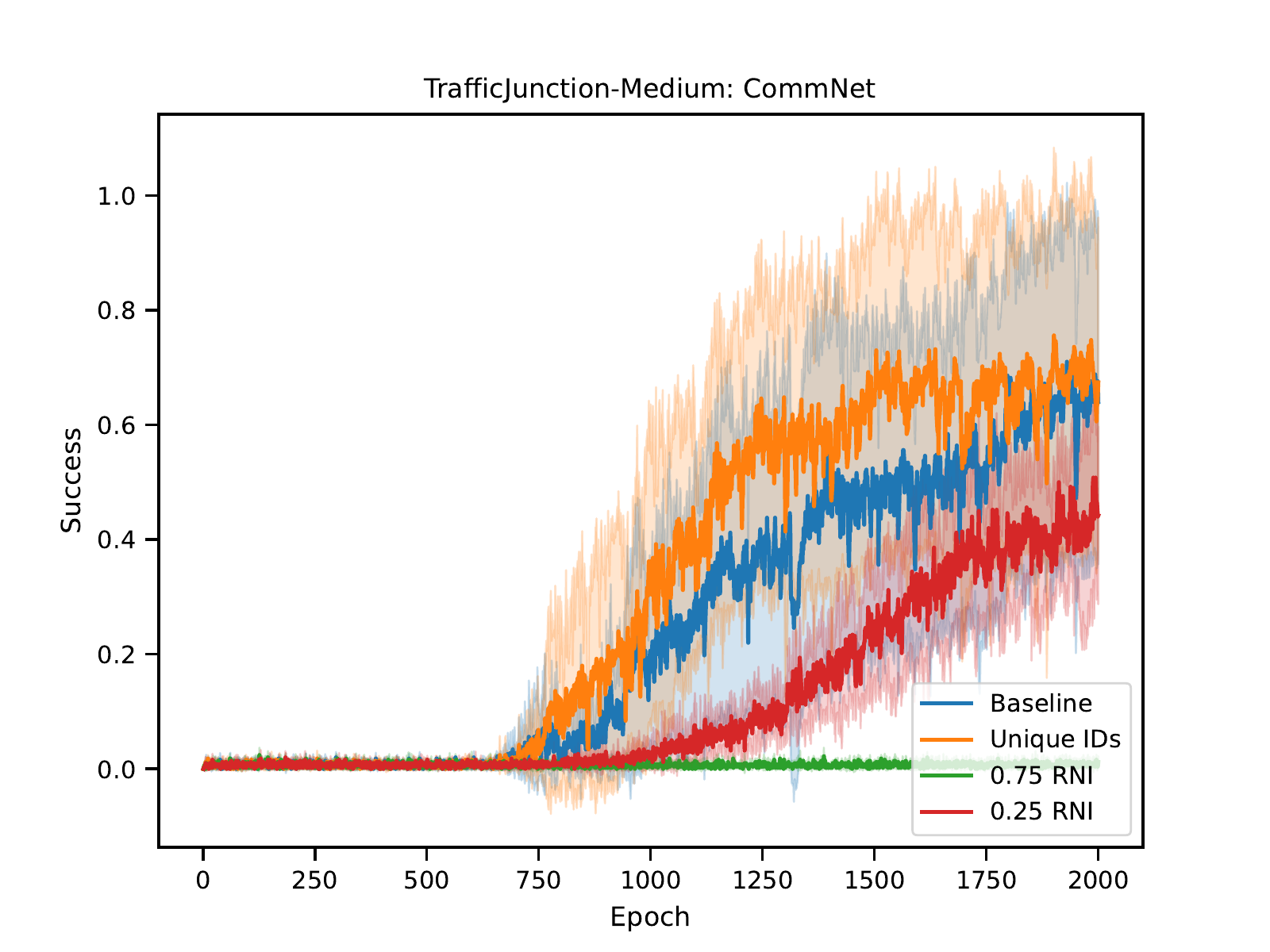}
\end{figure}

\paragraph{Weisfeiler-Lehman Expressivity}
Results for the Box Pushing environment are shown in Table \ref{tab:results:bp}. 0.25 RNI is the clear winner, achieving the top performance across almost all baselines and only slightly worse performance in the others. It is never outperformed by the baseline. This indicates that when communication expressivity beyond 1-WL is helpful for solving the environment, using RNI is the clear choice. Across all but one method, unique IDs also outperformed the baseline, demonstrating the benefit of having higher expressivity. However, unique IDs tend to yield less stable solutions and less effective policies than 0.25 RNI. \changemarker{We postulate that it is easier for agents to overfit on the particular unique IDs used, since they are deterministically assigned. On the other hand, using RNI encourages the agents to learn policies which respect the permutation invariance between agents since agents will receive different random observation augmentations at each time step.}

For completeness, we note that expressivity beyond 1-WL is not strictly needed to solve the Box Pushing environment, as demonstrated by several baselines achieving a ``ratio cleared'' score greater than 0.5, meaning they learned to sometimes clear both types of boxes. This is because the environment can be solved, albeit inefficiently, by recurrent policies which learn to alternate actions between ``normal'' and ``power'' moves each time step. Such policies are guaranteed to move the box every 2 time steps.

%
%
\begin{table}
\small 
\caption{Mean and 95\% confidence interval for Box Pushing across all baselines}
\label{tab:results:bp}
\centering
\begin{tabular}{@{}rrrrrrrrrrrrrr@{}}
    \toprule
    Baseline & Metric & \textbf{Baseline} && \textbf{Unique IDs} && \textbf{0.75 RNI} && \textbf{0.25 RNI} \\
    \midrule
    
CommNet & Ratio Cleared & $ 0.786 \pm 0.08 $ && \pmb{$ 0.829 \pm 0.08 $} && $ 0.768 \pm 0.08 $ && $ 0.795 \pm 0.09 $ \\
DGN & Ratio Cleared & $ 0.603 \pm 0 $ && $ 0.756 \pm 0.06 $ && \pmb{$ 0.958 \pm 0 $} && $ 0.957 \pm 0.01 $ \\
IC3Net & Ratio Cleared & $ 0.49 \pm 0.15 $ && $ 0.617 \pm 0.14 $ && $ 0.34 \pm 0.18 $ && \pmb{$ 0.676 \pm 0.06 $} \\
MAGIC & Ratio Cleared & $ 0.958 \pm 0.04 $ && $ 0.985 \pm 0.01 $ && $ 0.975 \pm 0.04 $ && \pmb{$ 0.998 \pm 0 $} \\
TarMAC & Ratio Cleared & $ 0.629 \pm 0.14 $ && $ 0.578 \pm 0.11 $ && $ 0.662 \pm 0.06 $ && \pmb{$ 0.679 \pm 0.06 $} \\
T-IC3Net & Ratio Cleared & $ 0.558 \pm 0.13 $ && $ 0.596 \pm 0.11 $ && $ 0.458 \pm 0.18 $ && \pmb{$ 0.643 \pm 0.15 $} \\

    \bottomrule
\end{tabular}
\end{table}

\paragraph{Symmetry Breaking}
Results for the Drone Scatter experiments are shown in Table \ref{tab:results:ds_stochastic} and Table \ref{tab:results:ds_greedy}, for agents with stochastic and greedy evaluation respectively. Across all of them, unique IDs exhibits consistently superior performance than 0.75 RNI and the baseline, since it deterministically breaks the symmetry between agents and allows them to split up easily to solve the environment. In the stochastic case, 0.75 RNI performs similarly to the baseline, but is markedly superior to the baseline in the greedy case. The comparison to a purely random agent indicates that the models are learning much more effective policies than just moving around at random.

Baseline methods with stochastic evaluation achieve consistently higher pairwise distances than their greedy counterparts, meaning they are learning to split up to find the target. They are capable of this due to a combination of 3 things: a stochastic policy, recurrent networks, and agents observing their last actions. Initially, agents cannot differentiate between each other, and all produce the same action distribution. However, if the distribution is diverse, then they are expected to produce different actions since they sample randomly from this distribution, which are observed in the next time step. The different observations lead to different hidden states in the recurrent networks, effectively changing the observations between agents and allowing them to differentiate between each other. This is not the case for greedy action selection, which is common when doing policy evaluation in MARL, where the actions chosen will always be the same.

%
%
\begin{table}
\small 
\caption{Mean and 95\% confidence interval for Drone Scatter across all baselines except DGN, including a purely random agent. Stochastic evaluation}
\label{tab:results:ds_stochastic}
\centering
\begin{tabular}{@{}rrrrrrrrrrrr@{}}
    \toprule
    Baseline & Metric & \textbf{Baseline} && \textbf{Unique IDs} && \textbf{0.75 RNI} \\
    \midrule
    
CommNet & Pairwise Distance & $ 11.34 \pm 0.9 $ && \pmb{$ 12.08 \pm 1.12 $} && $ 8.687 \pm 1.4 $ \\
& Steps Taken & $ 11.5 \pm 0.26 $ && \pmb{$ 9.767 \pm 0.32 $} && $ 11.74 \pm 1.39 $ \\
IC3Net & Pairwise Distance & $ 9.108 \pm 1.45 $ && \pmb{$ 13.3 \pm 0.71 $} && $ 10.99 \pm 0.38 $ \\
& Steps Taken & $ 11.94 \pm 0.84 $ && \pmb{$ 10.13 \pm 0.25 $} && $ 11.66 \pm 0.22 $ \\
MAGIC & Pairwise Distance & $ 7.693 \pm 1.47 $ && \pmb{$ 12.59 \pm 1 $} && $ 7.216 \pm 0.76 $ \\
& Steps Taken & $ 13.05 \pm 0.58 $ && \pmb{$ 11.12 \pm 1.05 $} && $ 13.54 \pm 0.21 $ \\
TarMAC & Pairwise Distance & $ 7.448 \pm 0.89 $ && \pmb{$ 10.26 \pm 0.69 $} && $ 8.486 \pm 0.36 $ \\
& Steps Taken & $ 13.49 \pm 0.1 $ && \pmb{$ 10.7 \pm 0.35 $} && $ 12.85 \pm 0.57 $ \\
T-IC3Net & Pairwise Distance & $ 8.891 \pm 0.27 $ && \pmb{$ 12.9 \pm 0.78 $} && $ 9.552 \pm 0.57 $ \\
& Steps Taken & $ 12.28 \pm 0.6 $ && \pmb{$ 10.33 \pm 0.46 $} && $ 12.22 \pm 0.82 $ \\

    \midrule
    Random & Pairwise Distance & $ 5.8 \pm 0.02 $ && -- && -- \\
    & Steps Taken & $ 17.39 \pm 0.04 $ && -- && -- \\

    \bottomrule
\end{tabular}
\end{table}

%
%
\begin{table}
\small 
\caption{Mean and 95\% confidence interval for Drone Scatter across all baselines. Greedy evaluation}
\label{tab:results:ds_greedy}
\centering
\begin{tabular}{@{}rrrrrrrrrrrr@{}}
    \toprule
    Baseline & Metric & \textbf{Baseline} && \textbf{Unique IDs} && \textbf{0.75 RNI} \\
    \midrule
    
CommNet & Pairwise Distance & $ 8.849 \pm 0.63 $ && \pmb{$ 13.28 \pm 1.27 $} && $ 8.589 \pm 1.35 $ \\
& Steps Taken & $ 13.79 \pm 0.12 $ && \pmb{$ 9.554 \pm 0.33 $} && $ 12.62 \pm 1.19 $ \\
DGN & Pairwise Distance & $ 3.221 \pm 0.18 $ && \pmb{$ 4.427 \pm 0.67 $} && $ 3.706 \pm 0.83 $ \\
& Steps Taken & $ 13.36 \pm 0.15 $ && \pmb{$ 13.27 \pm 0.21 $} && $ 13.46 \pm 0.14 $ \\
IC3Net & Pairwise Distance & $ 7.69 \pm 1.03 $ && \pmb{$ 14.09 \pm 0.54 $} && $ 11 \pm 0.86 $ \\
& Steps Taken & $ 13.25 \pm 0.4 $ && \pmb{$ 10.14 \pm 0.2 $} && $ 11.42 \pm 0.48 $ \\
MAGIC & Pairwise Distance & $ 6.61 \pm 1.28 $ && \pmb{$ 12.58 \pm 0.6 $} && $ 7.107 \pm 1.59 $ \\
& Steps Taken & $ 13.27 \pm 0.18 $ && \pmb{$ 11.84 \pm 0.68 $} && $ 13.61 \pm 0.25 $ \\
TarMAC & Pairwise Distance & $ 8.666 \pm 0.28 $ && \pmb{$ 12.09 \pm 0.73 $} && $ 8.999 \pm 0.94 $ \\
& Steps Taken & $ 13.73 \pm 0.21 $ && \pmb{$ 11.01 \pm 0.86 $} && $ 12.19 \pm 0.82 $ \\
T-IC3Net & Pairwise Distance & $ 7.28 \pm 0.69 $ && \pmb{$ 13.51 \pm 0.98 $} && $ 10.87 \pm 1.17 $ \\
& Steps Taken & $ 13.96 \pm 0.26 $ && \pmb{$ 10.63 \pm 0.66 $} && $ 11.73 \pm 0.54 $ \\

    \bottomrule
\end{tabular}
\end{table}

\begin{figure}
    \centering
    \caption{\changemarker{Training curves for IC3Net (top) and CommNet / DGN (bottom) on Box Pushing (left), Drone Scatter with stochastic evaluation (middle), and Drone Scatter with greedy evaluation (right)}}
    \label{fig:results:designed_envs}
    
    \includegraphics[width=0.32\linewidth]{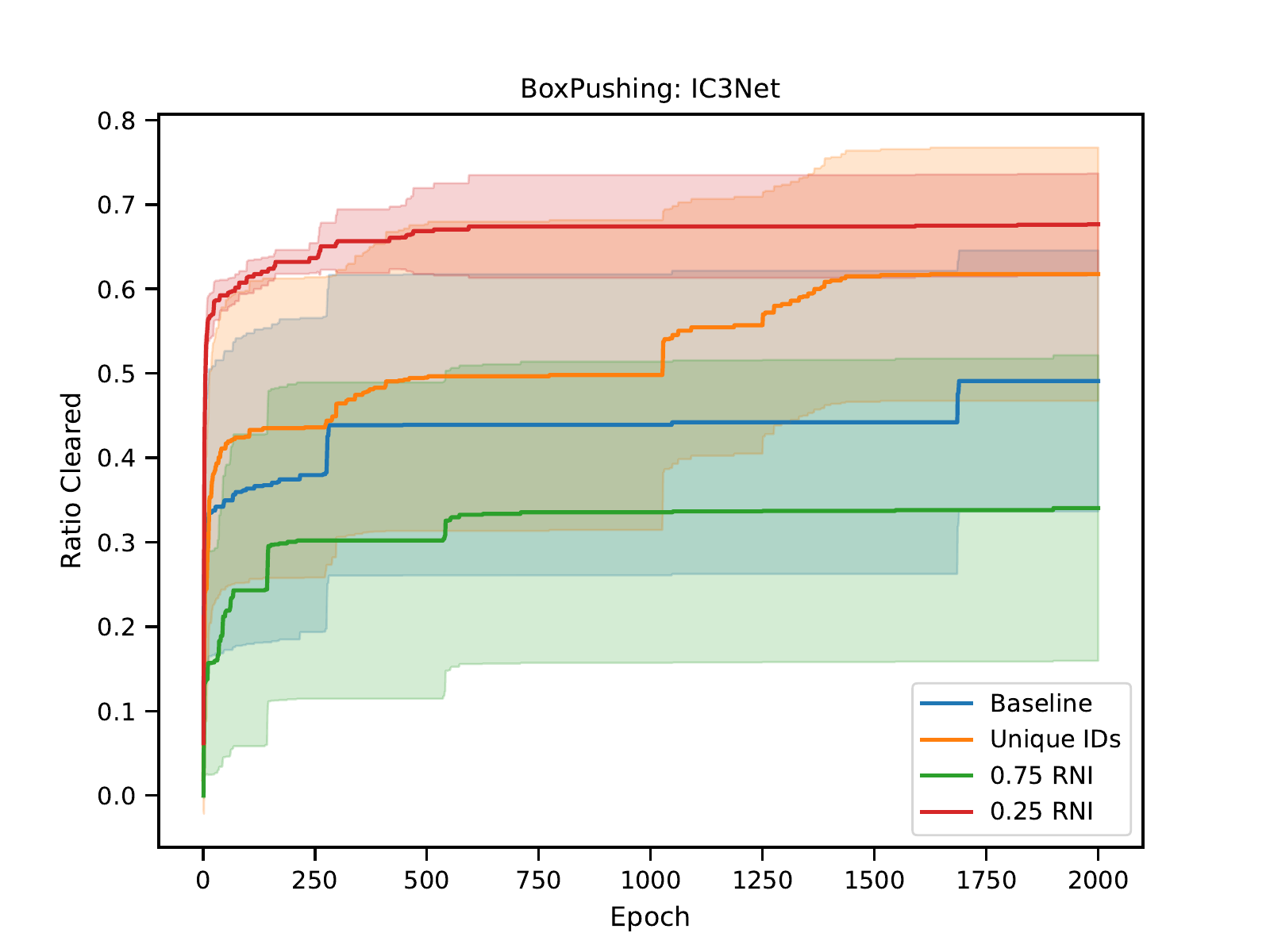}
    \includegraphics[width=0.32\linewidth]{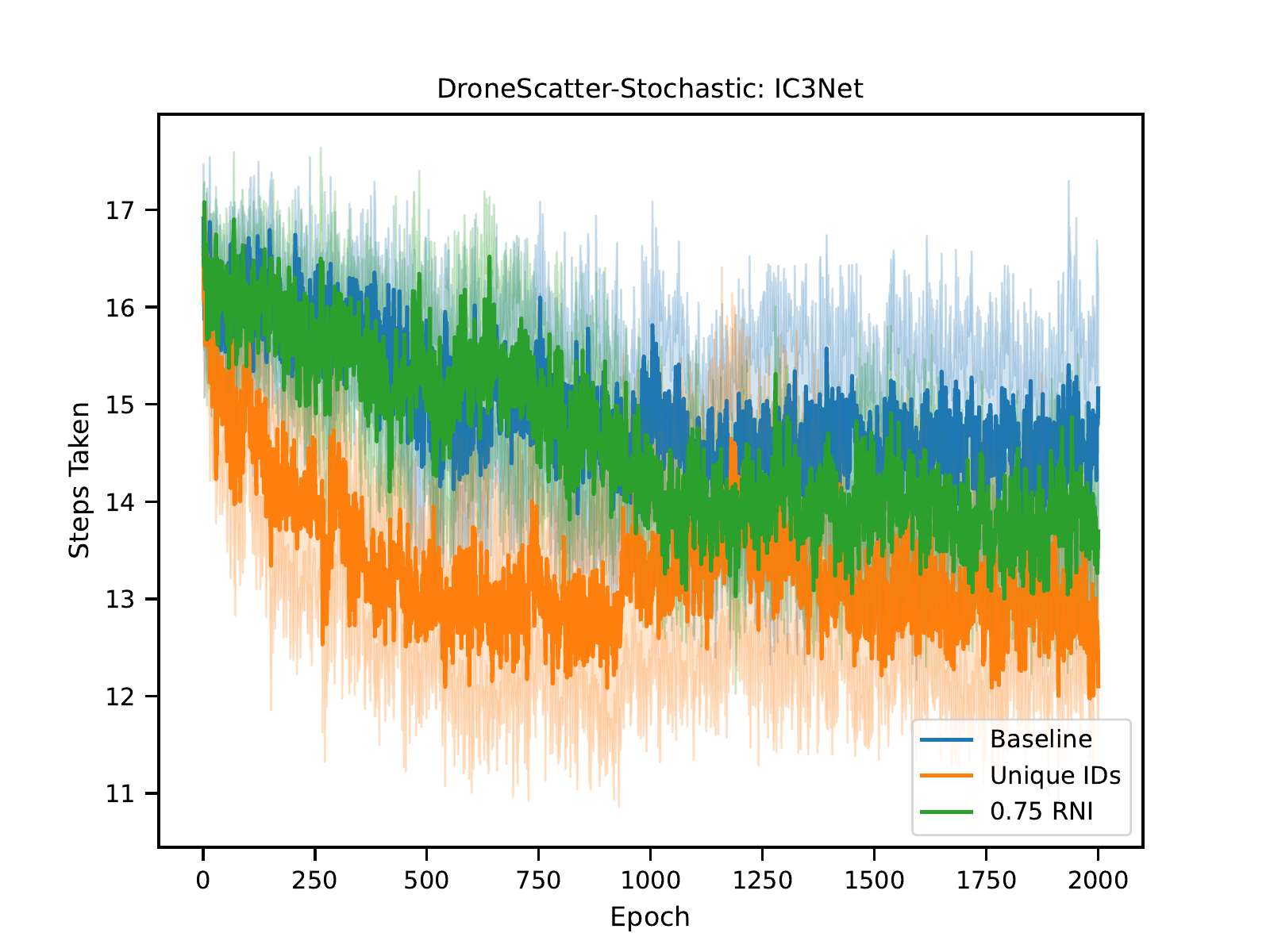}
    \includegraphics[width=0.32\linewidth]{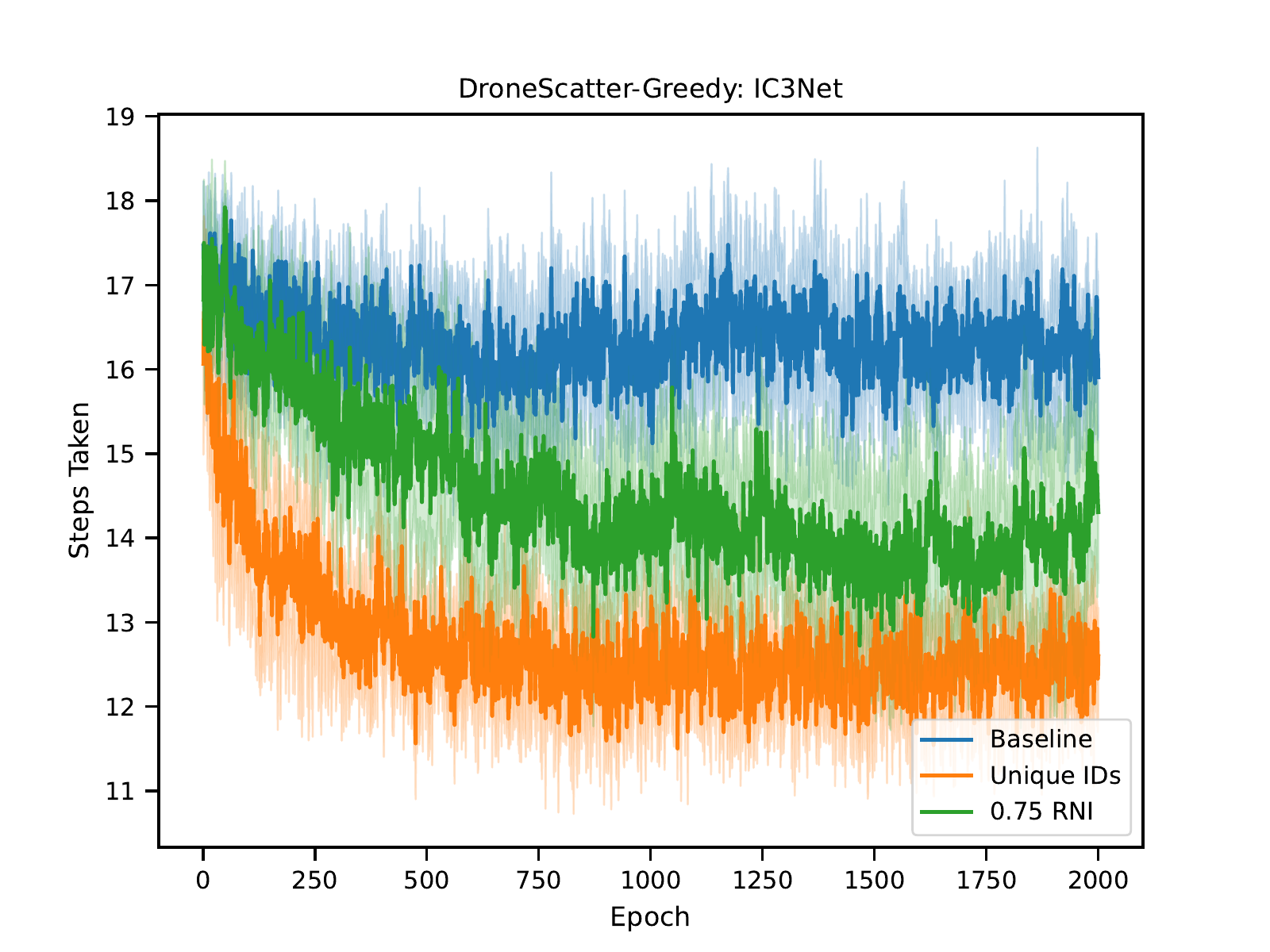}

    \includegraphics[width=0.32\linewidth]{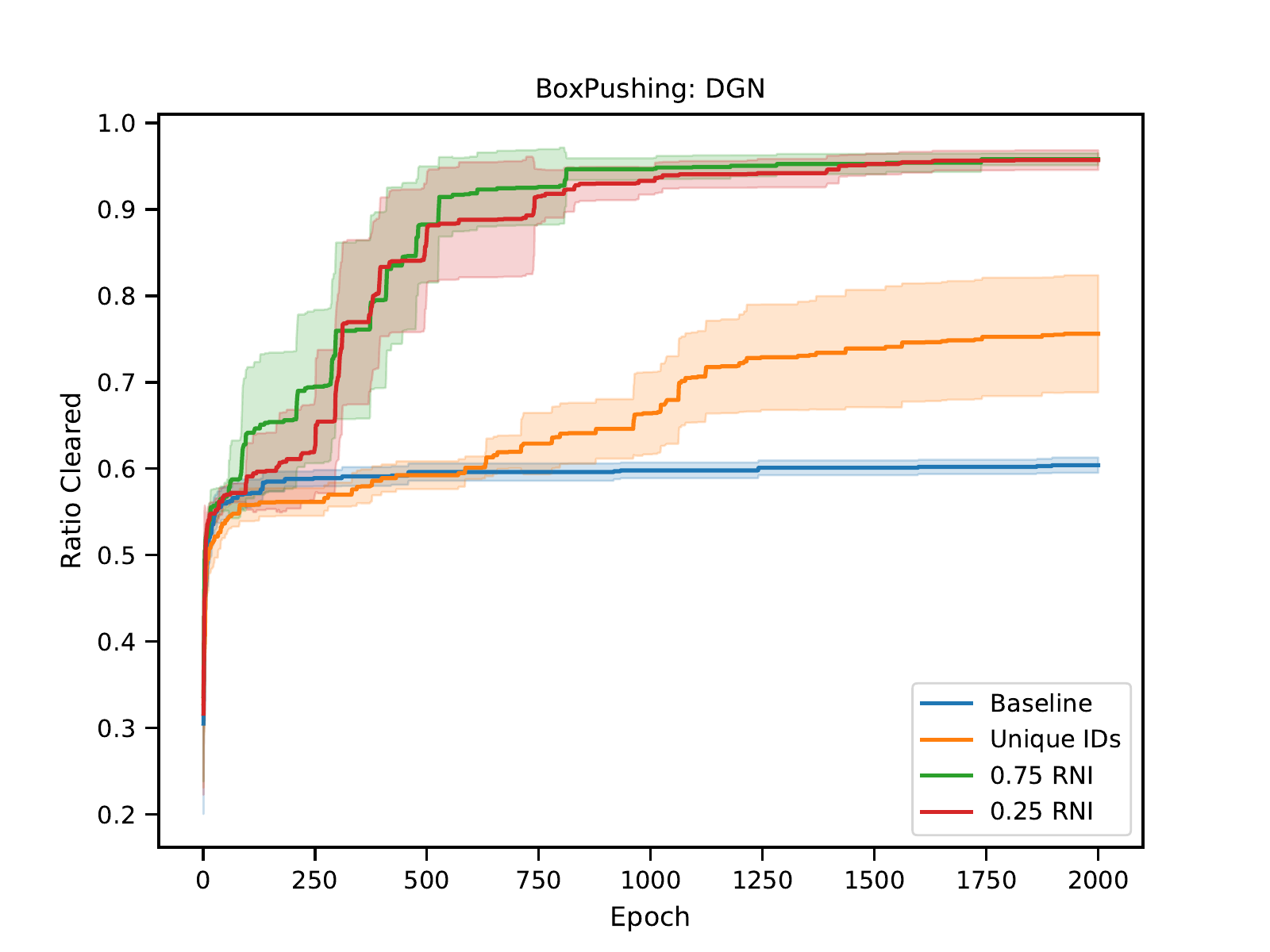}
    \includegraphics[width=0.32\linewidth]{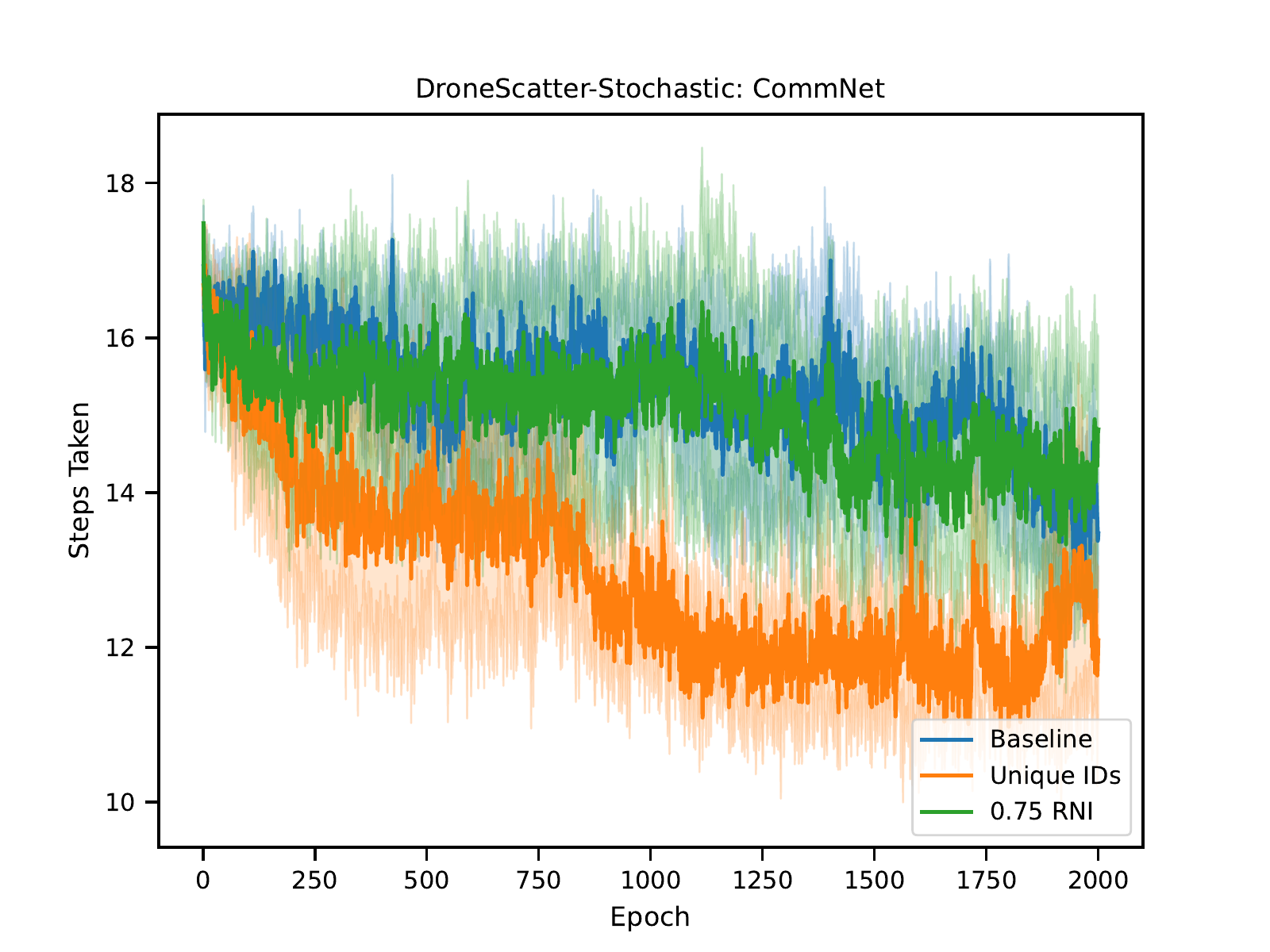}
    \includegraphics[width=0.32\linewidth]{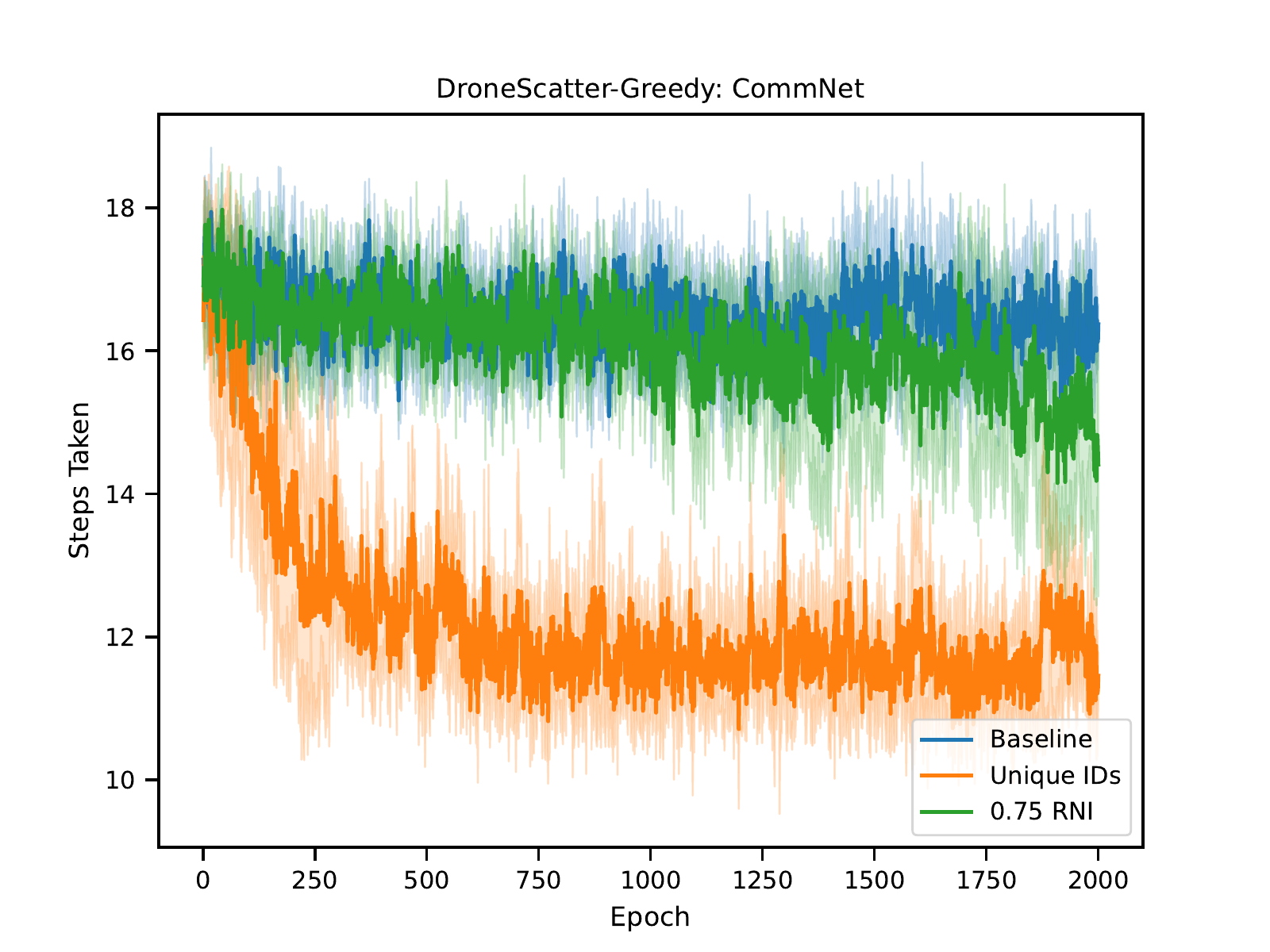}
\end{figure}
\section{Conclusion}
We introduce GDNs, a framework for MARL communication, and formally show how it corresponds to the node labelling problem in GNNs. Our theoretical contributions use this observation to demonstrate that existing MARL communication methods relying on conventional GNN architectures have provably limited expressivity. Driven by this, we prove how augmenting agent observations with unique IDs or random noise yields universally expressive invariant communication in MARL, whilst also providing desirable properties such as being able to perform symmetry-breaking: targeting arbitrary sets of joint actions for identical agents.

Experimentally, we compare these augmentations across 6 different MARL communication baselines that fall within the GDN paradigm, using 3 benchmark communication environments and 2 tasks designed to separately test expressivity and symmetry-breaking. Ultimately, we find that, whilst unique IDs or smaller RNI augmentations can typically be applied without detriment on standard environments, they do not readily provide improved performance either. However, on environments where more complex coordination is required, these augmentations are essential for strong performance. With RNI and unique IDs being best suited to environments requiring increased expressivity and symmetry-breaking, respectively, it is interesting to note that no single method emerges which we can recommend as the \emph{de facto} choice for MARL practitioners. This suggests that a more complete picture of the relationship between communication expressivity and downstream performance on relevant tasks remains an open question for future research. Furthermore, insights into GNN architectures can be leveraged in GDNs, which opens many promising avenues for future work in MARL communication.

\newpage


\bibliographystyle{plainnat}
\bibliography{ref}

\appendix

\newpage

%
%
\section{Full Background} \label{app:background}
In this appendix, for technical and theoretical completeness, we expand upon information given in the background of the paper.

\subsection{Multi-Agent Reinforcement Learning}
We consider the case of Decentralized Partially Observable Markov Decision Processes (Dec-POMDPs) \citep{littman1994markov} augmented with communication between agents. In a Dec-POMDP, at each timestep $t$ every agent $i \in \{ 1, ..., N \}$ gets a local observation $o^t_i$, takes an action $a^t_i$, and receives a reward $r^t_i$. The objective is to maximize the agent rewards over the actions. We consider two agent paradigms: value-based and actor-critic.

In value-based methods, such as MADQN \citep{tampuu2017multiagent}, the aim is to learn a function $Q_\theta$ with parameters $\theta$ that estimates the value of taking an action $a_i$ after observing $o_i$. Such methods are often trained using a replay buffer, to which tuples $(O, A, O', R)$ are added, where $O = \{ o_1, ..., o_N \}$ is the set of observations, $A$ is the set of actions, $O'$ is the set of next observations, and $R$ is the set of rewards. If the environment includes communication, the underlying communication graph $C$ can also be included in the replay buffer. These experiences are added to the buffer whilst the agents interact with the environment. To collect diverse experiences, methods such as $\epsilon$-greedy exploration \citep{tokic2010adaptive} can be used. For training, random minibatches of size $B$ are sampled from the buffer, and loss similar to the following is minimized:

$$ L(\theta) =  \frac{1}{B} \sum^B_{b=1} \frac{1}{N} \sum^N_{i=1} (y_i - Q_\theta(o_i, a_i))^2 $$

where $y_i = r_i + \gamma \text{ max}_{a'} Q_{\theta'}({o_i}', {a_i}')$. In this formula, $Q_{\theta'}$ is referred to as the \emph{target network}, and its parameters $\theta'$ are updated softly or intermittently from $\theta$ during training.

In actor-critic methods such as MADDPG \citep{lowe2017multi}, the aim is to learn a policy function $\pi_\theta$ that maps observations onto distributions over actions, where the action most likely to maximize the reward is assigned the highest probability. The policy gradient is estimated by the following:

$$ \nabla_\theta J(\theta) = \mathbb{E}_{o \sim \rho^\pi,~a \sim \pi_\theta} [\sum_{t=1}^{T} \nabla_\theta \log \pi_\theta(a_t | o_t) (R_t - V(o_t))] $$

where $\rho^\pi$ is the observation distribution, $\pi_\theta$ is the policy distribution, $R_t = \sum_{t'=t}^T \gamma^{t'-t} r(s_{t'}, a_{t'})$ is the discounted reward, and $V$ is a learned value function, used to decrease the variance of the estimated policy gradient.

For brevity, we collectively refer to the policy network / value function or Q-network as the \emph{actor network}. Often in MARL, instead of learning an actor network for each agent, a single shared network will be used for all agents. For example, COMA \citep{foerster2018counterfactual}, Q-Mix \citep{rashid2018qmix}, and Mean Field RL \citep{yang2018mean} all share parameters in their neural networks. This parameter sharing typically yields faster and more stable training \citep{gupta2017cooperative}.

\subsection{Communication in MARL}
Many environments require agents to coordinate to solve tasks. Communication is crucial for enabling this. \citet{foerster2016learning, sukhbaatar2016learning} were among the first to propose learned communication in multi-agent reinforcement learning. Since then, many different methods for communication in MARL have been proposed \citep{zhu2022survey}. When communicating, there is a helpful structural inductive bias which can be used: the order in which incoming messages from other agents are processed should not affect the outcome. More formally: agents ought to be \emph{permutation invariant} when using incoming messages. A function $\rho$ is permutation invariant if for any input $X = (x_1, x_2, ..., x_k)$ and any permutation $\sigma$ on $X$, $\rho(X) = \rho(\sigma \circ X)$.

In Section \ref{sec:gdn}, we define Graph Decision Networks (GDNs) using the framework of GNNs, a neural architecture which respects permutation invariance between nodes when doing message passing.

Many of the most successful models for MARL communication fall within this paradigm, including CommNet \citep{sukhbaatar2016learning}, IC3Net \citep{singh2018learning}, GA-Comm \citep{liu2020multi}, MAGIC \citep{niu2021multi}, Agent-Entity Graph \citep{agarwal2019learning}, IP \citep{qu2020intention}, TARMAC \citep{das2019tarmac}, IMMAC \citep{sun2021intrinsic}, DGN \citep{jiang2018graph}, VBC \citep{zhang2019efficient}, MAGNet \citep{malysheva2018deep}, and TMC \citep{zhang2020succinct}. Other models such as ATOC \citep{jiang2018learning} and BiCNet \citep{peng2017multiagent} do not fall within the paradigm since they use LSTMs for combining messages, which are not permutation invariant, and models such as RIAL, DIAL \citep{foerster2016learning}, ETCNet \citep{hu2020event}, and SchedNet \citep{kim2019learning} do not since they used a fixed message-passing structure.

\subsection{Graph Neural Networks}
A graph consists of nodes and edges connecting them. Nodes and edges can have attributes (also referred to as features or labels), which often take the form of real vectors. Formally, we define an attributed graph $G$ as a triple $(V, E, a)$, where $V(G)$ is a finite set of nodes, $E(G) \subseteq \{ (u, v) ~|~ u, v \in V(G) \}$ is a set of directed edges, and $a: V(G) \cup E(G) \to \mathbb{R}^d$ is an attribute function where $d > 0$. For $w \in V(G) \cup E(G)$, $a(w)$ is the attribute of $w$. Undirected graphs are ones in which $E(G)$ is a symmetric relation on $V(G)$. Graphs are found in many different areas of application, leading to a plethora of research of how to learn using graph structured data \citep{barabasi2004network, easley2010networks, simonovsky2017dynamic}.

When considering functions operating on graphs, it is sensible to demand permutation invariance and equivariance. Intuitively: the output of any function on a graph should not depend on the order of the nodes (invariance), and if the function provides outputs for each node, then re-ordering the nodes of the input graph should be equivalent to applying the same re-ordering to the output values (equivariance). Formally, let $S(V(G))$ be the set of all permutations of $V(G)$, $D$ a set of graphs, and $L$ a set of potential output attributes (e.g. $\mathbb{R}^3$). Then a function $f: D \to L$ is \emph{invariant} if:

$$ \forall \text{ graphs } G \in D, \forall \text{ permutations } \sigma \in S(V(G)), f(G) = f(\sigma \circ G) $$

A function $f: D \to L^{|V(G)|}$ is likewise \emph{equivariant} if instead $f(\sigma \circ G) = \sigma \circ f(G)$.

GNNs belong to a class of neural methods that operate upon graphs. The term is often used to refer to a large variety of models; in this paper, we define the term ``Graph Neural Networks'' to correspond to the definition of Message Passing Neural Networks (MPNNs) by \citet{gilmer2017neural}, which is the most common GNN architecture. Notable instances of this architecture include Graph Convolutional Networks (GCNs) \citep{duvenaud2015convolutional}, GraphSAGE \citep{hamilton2017inductive}, and Graph Attention Networks (GATs) \citep{velivckovic2017graph}. A GNN consists of multiple message-passing layers, where each layer aggregates node attribute information to every node from its neighbours in the graph, and then uses the aggregated information with the current node attribute to assign a new value to the attribute, passing all updated node attributes to the next GNN layer. GNNs are often augmented with global readouts, where in each layer we also aggregate a feature vector for the whole graph and use it in conjunction with the local aggregations \citep{barcelo2020logical}. For layer $m$ and node $i$ with current attribute $v_i^m$, the new attribute $v_i^{m+1}$ is computed as

$$ v_i^{m+1} := \fupdate^{\theta_m}(v_i^{m},~ \faggr^{\theta_m'}(\{ v_j^m ~|~ j \in N(i) \}),~ \fread^{\theta_m''}(\{ v_j^m ~|~ j \in V(G) \}) ) $$

where $N(i)$ is all nodes with edges connecting to $i$ and $\theta_m, \theta_m', \theta_m''$ are the (possibly trainable) parameters of the update, aggregation, and readout functions for layer $m$. Parameters may be shared between layers, e.g. $\theta_0 = \theta_1$. The functions $\faggr^{\theta_m'}, \fread^{\theta_m''}$ are permutation invariant. A global graph feature can be computed by having a final \emph{readout layer} which aggregates all node attributes into a single feature. Importantly, GNNs are invariant / equivariant graph functions.

\subsection{Expressivity and Related Work} \label{sec:related_work}
\citet{morris2019weisfeiler, xu2018powerful} concurrently showed that any GNN cannot be more powerful than the 1-WL graph-coloring algorithm in terms of distinguishing non-isomorphic graphs, meaning that there are pairs of non-isomorphic graphs $G_1, G_2$ which for any GNN $f$, $f(G_1) = f(G_2)$. \citet{morris2019weisfeiler} also define $k$-GNNs, a class of higher-order GNNs which have the same expressive power as the $k$-WL algorithm. \citet{chen2020can} prove that GNNs cannot count certain types of sub-structures and that certain higher-order GNNs can.

\citet{garg2020generalization} prove that some important graph properties cannot be computed by GNNs, and also provide data dependent generalization bounds for GNNs. \citet{nt2019revisiting} show that GNNs only perform low-pass filtering on attributes and investigate their resilience to noise in the features. \citet{barcelo2020logical} prove a direct correspondence between GNNs and Boolean classifers expressed in the first-order logic fragment FOC$_2$.

\citet{loukas2019graph} demonstrates how GNNs lose expressivity when their depth and width are restricted. \citet{loukas2020hard} further analyzes the expressive power of GNNs with respect to their \emph{communication capacity}, a measure of how much information the nodes of a network can exchange during message-passing. \citet{oono2019graph} analyze the expressive power of GNNs as the number of layers tends to infinity, proving that under certain conditions, the output will carry no information other than node degrees and connected components.

In the space of expressivity for reinforcement learning, \citet{dong2020expressivity} compare model-free reinforcement learning with the model-based approaches with respect to the expressive power of neural networks for policies and Q-functions. \citet{castellini2019representational} empirically evaluate the representational power of different value-based RL network architectures using a series of one-shot games. The simplistic games capture many of the issues that arise in the multi-agent setting, such as the lack of an explicit coordination mechanism, which provides good motivation for the inclusion of communication.

\subsection{Weisfeiler-Lehman Expressivity}
1-WL \citep{weisfeiler1968reduction} is a graph coloring algorithm that tests if two graphs are non-isomorphic by iteratively re-coloring the nodes. Given an initial graph coloring corresponding to the node labels, in each iteration, two nodes with the same color get assigned different colors if the number of identically colored neighbors is not equal. If, at some point, the number of nodes assigned a specific color is different across the two graphs, the algorithm asserts that the graphs are not isomorphic. However, there are non-isomorphic graphs which the algorithm will not recognize as non-isomorphic, e.g. in Figure \ref{fig:1wl_graphs_full}. \citet{morris2019weisfeiler, xu2018powerful} proved that for any two non-isomorphic graphs indistinguishable by 1-WL, there is no GNN that can produce different outputs for the two graphs. Furthermore, there is a fundamental link between this graph separation power and function approximation power. \citet{chen2019equivalence} proved that a class of models can separate all graphs if and only if it can approximate any continuous invariant function.

There are several GNN architectures which are designed to go beyond 1-WL expressivity. \citet{morris2019weisfeiler} propose $k$-GNNs, which are equivalent to the $k$-WL algorithm. \citet{morris2021weisfeiler} also show the link between $k$-order equivariant graph networks (EGNs) \citep{kondor2018covariant} and the $k$-WL algorithm. Other attempts also include using unique node IDs and passing matrix features \citep{vignac2020building}, relational pooling \citep{murphy2018janossy}, and random dropouts \citep{papp2021dropgnn}. \citet{morris2021weisfeiler} provide an overview of many such higher-order models. However, many of these models do not scale well and are computationally infeasible in practice.

\begin{figure}
    \centering
    \caption{A pair of graphs indistinguishable by 1-WL}
    \label{fig:1wl_graphs_full}

    \begin{tikzpicture}
        \node[main node] (1) {};
        \node[main node] (2) [right = of 1]  {};
        \node[main node] (3) [right = of 2] {};
        \node[main node] (4) [right = of 3] {};
        \node[main node] (5) [below = of 4] {};
        \node[main node] (6) [left = of 5] {};
        \node[main node] (7) [left = of 6] {};
        \node[main node] (8) [left = of 7] {};
    
        \path[draw,thick]
        (1) edge node {} (2)
        (2) edge node {} (3)
        (3) edge node {} (4)
        (4) edge node {} (5)
        (5) edge node {} (6)
        (6) edge node {} (7)
        (7) edge node {} (8)
        (8) edge node {} (1)
        ;
        \begin{scope}[xshift=7cm]
        \node[main node, fill=orange] (1) {};
        \node[main node, fill=orange] (2) [right = of 1]  {};
        \node[main node, fill=orange] (3) [below = of 2] {};
        \node[main node, fill=orange] (4) [below = of 1] {};
    
        \path[draw,thick]
        (1) edge node {} (2)
        (2) edge node {} (3)
        (3) edge node {} (4)
        (4) edge node {} (1)
        ;
        \end{scope}
        
        \begin{scope}[xshift=10cm]
        \node[main node, fill=orange] (1) {};
        \node[main node, fill=orange] (2) [right = of 1]  {};
        \node[main node, fill=orange] (3) [below = of 2] {};
        \node[main node, fill=orange] (4) [below = of 1] {};
    
        \path[draw,thick]
        (1) edge node {} (2)
        (2) edge node {} (3)
        (3) edge node {} (4)
        (4) edge node {} (1)
        ;
        \end{scope}
    \end{tikzpicture}
\end{figure}
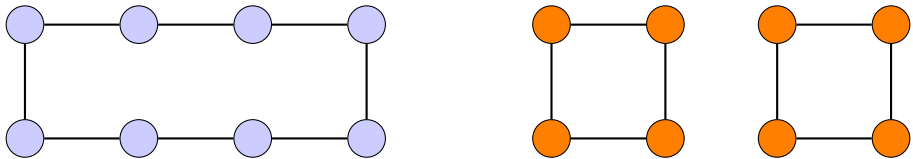

\subsection{Random Node Initialization}
\citet{sato2021random} propose augmenting GNNs with \emph{random node initialization} (RNI), where for each node in the input graph, a number of randomly sampled values are concatenated to the original node attribute. For all graphs / nodes, the random values are sampled from the same distribution. \citet{abboud2020surprising} prove that such GNNs are universal and can approximate any permutation invariant graph function. Technically, random initialization breaks the node invariance in GNNs, since the result of the message passing will depend on the structure of the graph as well as the values of the random initializations. However, when one views the model as computing a random variable, the random variable is still invariant when using RNI. In expectation, the mean of random features will be used for GNN predictions, and is the same across each node. However, the variability of the random samples allow the GNN to discriminate between nodes that have different random initializations, breaking the 1-WL upper bound.

The above is formally described by \citet{abboud2020surprising} as follows. Let $G_n$ be the class of all $n$-node graphs (i.e. graphs that consist of at most $n$ nodes) and let $f: G_n \to \mathbb{R}$. We say that a randomized function $X$ that associates with every graph $G \in G_n$ a random variable $X(G)$ is an $(\epsilon, \delta)$-\emph{approximation} of $f$ if for all $G \in G_n$ it holds that $\text{Pr}(|f(G) - X(G)| \leq \epsilon) \geq 1 - \delta$. Note that an MPNN $N$ with
RNI computes such functions $X$. If $X$ is computed by $N$, we say that $N$ $(\epsilon, \delta)$-\emph{approximates} $f$.

They state the following theorem:

\begin{theorem*} \label{thm:rni_universal}
Let $n \geq 1$ and let $f: G_n \to \mathbb{R}$ be invariant. Then for all $\epsilon, \delta > 0$, there is an MPNN with RNI that $(\epsilon, \delta)$-\emph{approximates} $f$.
\end{theorem*}

\subsection{Other Conditions for Expressivity}
For another method of analyzing GDN expressivity through the lens of GNNs, consider the work of \citet{loukas2020hard}, which defines the \emph{communication capacity / complexity} $c_g$ of GNNs, a measure of how much information the nodes can exchange during message passing. The following intuitive theorem is proven:

\begin{theorem*}
Let $f$ be an MPNN with $d$ layers, where each has width $w_m$ (attribute size), message size $a_m$ (output of aggregation), and a global state of size $\gamma_m$ (output of global readout). For any disjoint partition of $V$ into $V_a, V_b$, where $\text{cut}(V_a, V_b)$ is the size of the smallest cut separating $V_a, V_b$:

$$ c_g \leq \text{cut}(V_a, V_b) \sum_{m=1}^{d} \text{min}(a_m, w_m) + \sum_{m=1}^{d} \gamma_m $$
\end{theorem*}

\citet{loukas2020hard} prove that MPNNs with sub-quadratic and sub-linear capacity (with respect to the number of nodes) cannot compute the isomorphism class of graphs and trees respectively, demonstrating that capacity is an important consideration for GNN expressivity. This is also supported empirically. When designing a GNN model, we have no control over the structure of the input graphs. Thus, all features apart from $\text{cut}(V_a, V_b)$ are important considerations for the architecture. In GDNs, this corresponds to the message sizes $m$ and the number of rounds of message passing $d$. We provide the practical recommendation that when choosing $\{ m, d \}$, one ought to scale them such that $m \cdot d \in \Omega(n^2)$, where $n$ is the number of agents. In environments where communication is limited by range or obstacles, the number of rounds of message passing is also important when it comes to increasing an agent's \emph{receptive field}: from how many edges away information is propagated to the agent.

%
%
\newpage
\section{Proofs} \label{app:proofs}

Before we dive into the proofs, here follows a few brief notes / clarifications on the theory outlined in the paper.

\paragraph{Neural Network Function Approximation}
GNNs consist of the composition of update, aggregate, and readout functions, all of which are approximated by neural networks. Thus, given that neural networks are only universal approximators for \emph{continuous} functions, if standard neural architectures are used within GNNs, then GNNs can only ever approximate continuous functions. As such, if standard neural networks are used, then all universal approximation results in this paper require the additional assumption that the function being approximated is continuous.

\paragraph{GNN Vector Targeting}
For all theorems and proofs that utilize $\mathbb{R}$, such as equivariant graph functions with codomain $\mathbb{R}^n$, note that the scalar $\mathbb{R}$ can be replaced with the vector $\mathbb{R}^k$ for any $k > 1$, whilst maintaining correctness.

To demonstrate this, let $k > 1$ and consider an equivariant graph function $f: G_n \to (\mathbb{R}^k)^n$ that we are trying to approximate. We can instead approximate $k$ different functions $f_1, f_2, ..., f_k$ with $f_j: G_n \to \mathbb{R}^n$, such that $\forall G \in G_n ~ \forall i \in \{ 1, 2, ..., n \}, ~ f(G)_i = (f_1(G)_i, f_2(G)_i, ..., f_n(G)_i)$. These functions can be approximated using the theoretical results we currently have, with GNNs $g_1, g_2, ..., g_k$. We will simulate the application of all of these GNNs using a single GNN.

Thus, $f$ can also be approximated by the following construction: use an initial GNN layer with the update function defined such that $\fupdate(v) := (v, v, ..., v)$, where $v$ is transformed into $k$ copies of itself. Then, define the update, aggregation, and readout functions of each layer using the GNNs $g_1, g_2, ..., g_k$, conditioning each one on a portion of the node attributes. For example, the update function of layer $m$ for node $i$ on attribute $(v_1, v_2, ..., v_k)$ is defined by $\fupdate^{\theta_m}(v_1, v_2, ..., v_k) := ( g_{1 \text{ update}}^{\theta_m}(v_1), ..., g_{k \text{ update}}^{\theta_m}(v_k) )$.

This construction simulates the application of $f_1, f_2, ..., f_k$ in parallel and thus approximates $f$.

\paragraph{Recurrent GNNs}
In Section \ref{sec:gdn}, it is stated that models with recurrent networks also fall within the GDN paradigm, where the hidden or cell states for the networks can be considered as part of the agent observations. In this section, we briefly demonstrate how this can be done.

More concretely, consider a scenario with $n$ agents, where each GNN layer $m \in R \subseteq \{ 1, 2, ..., M \} $ uses hidden or cell states $C^{m-1} := \{ c_1^{m-1}, c_2^{m-1}, ..., c_n^{m-1} \}$. In this case, $R$ represents the list of layers that use recurrent networks. In a non-recurrent GNN, if layer $m$ is expressed as a function, it takes as input only the node attributes from the previous layer: $V^{m-1}$. In the recurrent case, it also takes $C^{m-1}$ as input: $m(V^{m-1}, C^{m-1})$.

To express this in terms of a non-recurrent GNN, instead modify the initial node attributes $V^0 = \{ v_1^0, v_2^0, ..., v_n^0 \}$ such that for each node $i$, $v_i^0 := (v_i^0, c_i^1, c_i^2, ..., c_i^{m'})$, where $m' := $ max value in $R$. Then in each layer $m \in \{ 1, 2, ..., M \} $ of the GNN and for each node $i$, only update $v_i^{m-1}$ to $v_i^m$ in $(v_i^{m-1}, c_i^1, c_i^2, ..., c_i^{m'})$, leaving all other portions of the attribute as-is. As input to the update, if $m \in R$ then use only $v_i^{m-1}$ and $c_i^m$, and otherwise use only $v_i^{m-1}$.

The above-described non-recurrent GNN computes the exact same output as the recurrent GNN it is emulating, by simply pulling hidden or cell states from a portion of the initial node attributes that has been set aside for them.

\newpage
\subsection{Theorem \ref{thm:equi_same}}
\begin{theorem*}
Given a GDN $f$, observations $O = \{ o_1, ..., o_n \}$, and communication graph $G$ such that nodes $i$ and $j$ are similar in $G$ and $o_i = o_j$, then it holds that $f(O)_i = f(O)_j$.
\end{theorem*}

\begin{proof}
Since $f$ is a GDN, there exists a GNN $g$ whose output when operating on the graph $G'$, equal to $G$ augmented with initial node attributes of $O$, coincides with that of $f$ on $O$.

Since $i$ and $j$ are similar in $G$ and $o_i = o_j$, $i$ and $j$ are similar in $G'$ (using the same automorphism). Also $g$ is a GNN, so $g$ is an equivariant function on $G'$. Thus

$$\forall \text{ graphs } G, \forall \text{ permutations } \sigma \in S(V(G)), f(\sigma \circ G) = \sigma \circ f(G)$$

Define $\sigma := (i ~ j)$, the permutation that maps $i \to j$ and $j \to i$, while mapping all other nodes to themselves. Note that since $i$ and $j$ are similar, $\sigma \circ G' = G'$. Furthermore, $\sigma = \sigma^{-1}$. So

$$f(O)_i = g(G')_i = (\sigma^{-1} \circ \sigma \circ g(G'))_i = (\sigma^{-1} \circ g(\sigma \circ G'))_i = g(\sigma \circ G')_j = g(G')_j = f(O)_j$$
\end{proof}

\subsection{Theorem \ref{thm:rni_equivariant}}
\begin{theorem*}
Let $n \geq 1$ and let $f: G_n \to \mathbb{R}^{n}$ be equivariant. Then for all $\epsilon, \delta > 0$, there is a GNN with RNI that $(\epsilon, \delta)$-\emph{approximates} $f$.
\end{theorem*}

\begin{proof}
For this proof, we assume the reader to be familiar with the proofs in the appendix of \citet{abboud2020surprising}, as we make use of their definitions, notation, lemmas, and proofs.

\citet{abboud2020surprising} state and prove the following. Let $G_n$ be the class of all $n$-node graphs (i.e., graphs that consist of at most $n$ nodes) and let $f: G_n \to \mathbb{R}$. We say that a randomized function $X$ that associates with every graph $G \in G_n$ a random variable $X(G)$ is an $(\epsilon, \delta)$-\emph{approximation} of $f$ if for all $G \in G_n$ it holds that $\text{Pr}(|f(G) - X(G)| \leq \epsilon) \geq 1 - \delta$. Note that an MPNN $N$ with
RNI computes such functions $X$. If $X$ is computed by $N$, we say that $N$ $(\epsilon, \delta)$-\emph{approximates} $f$.

\begin{theorem*}
Let $n \geq 1$ and let $f: G_n \to \mathbb{R}$ be invariant. Then for all $\epsilon, \delta > 0$, there is an MPNN with RNI that $(\epsilon, \delta)$-\emph{approximates} $f$.
\end{theorem*}

We extend this theorem to equivariant functions. Recall that we define the following. Let $G_n$ be the class of all $n$-node graphs. Let $f: G_n \to \mathbb{R}^{n}$, a graph function which outputs a real value for each node in $V(G)$. We say that a randomized function $X$ that associates with every graph $G \in G_n$ a sequence of random variables $X_1(G), X_2(G), ..., X_{n}(G)$, one for each node, is an $(\epsilon, \delta)$-\emph{approximation} of $f$ if for all $G \in G_n$ it holds that $\forall i \in \{ 1, 2, ..., n \}$, $\text{Pr}(|f(G)_i - X_i(G)| \leq \epsilon) \geq 1 - \delta$, where $f(G)_i$ is the output of $f(G)$ for node $i$. Note that a GNN $h$ with RNI computes such functions $X$. If $X$ is computed by $h$, we say $h$ $(\epsilon, \delta)$-\emph{approximates} $f$.

We adapt the proof of \citet{abboud2020surprising}, shown in their appendix, to correspond to equivariant functions instead. Notice that equivariant functions have their output at the node level instead of the graph level, so instead of identifying graphs with $C^2$-sentences that have no input variables, we identify a graph and a node in the graph by a 1-variable formula $\phi(v)$, where $v$ identifies the node.

Lemma A.3 from \citet{abboud2020surprising} proves that for every individualized colored graph $G$ there
is a $C^2$-sentence $\chi_G$ that identifies $G$. Thus, for every individualized colored graph $G$ and node $u$, the formula $\phi_{G, u}(v) := \chi_G \land \text{Node}_u(v)$ identifies $G$ and the node $u$, where $\text{Node}_u(v)$ is a Boolean function that is only true when $u=v$. In fact, $\phi_{G, u}(v) := \phi_u(v) := \text{Node}_u(v)$ already identifies $G$ by identifying the exact node.

We can similarly adapt Lemma A.4 and state the following:
\begin{lemma*}
Let $h: \mathcal{G}_{n, k} \to \{0, 1\}^n$ be an equivariant Boolean function. Then there exists a 1-variable formula $\phi_h(v)$ such that for all $G \in \mathcal{G}_{n, k}$ and all $v \in G$ it holds that $[[\phi_h(v)]](G) = h(G)_v$.
\end{lemma*}

To prove this, let $\mathcal{V} \subseteq \{ V(G) ~|~ G \in \mathcal{G}_{n, k} \}$ be the subset consisting of all nodes $u$ with $h(G)_u = 1$, where $G$ is the graph such that $u \in V(G)$. Then let

$$ \phi_h(v) := \bigvee\limits_{u \in \mathcal{V}} \phi_u(v) $$

We eliminate duplicates in the disjunction. Since, up to isomorphism, the class $\mathcal{G}_{n, k}$ is finite, and the number of nodes in each graph is upper-bounded by $n$, the disjunction over $\mathcal{V}$ is finite and hence $\phi_h(v)$ is well-defined.

We adapt Corollary A.1 in the same way as Lemma A.4. Lemma A.5 can be used as-is to show that RNI yields individualized colored graphs with high probability. From here, the remainder of the proof works analogously, substituting in equivariant functions for invariant ones and $\phi_h(v)$ for $\psi_h$.

\end{proof}

\subsection{Theorem \ref{thm:rni_coordination}}

\begin{theorem*}
Let $n \geq 1$ and consider a set $T$, where each $(G, A) \in T$ is a graph-labels pair, such that $G \in G_n$ and there is a \changemarker{multiset} of target labels $A_k \in A$ for each orbit $r_k \in R(G)$, with $| A_k | = | r_k |$. Then for all $\epsilon, \delta > 0$ there is a GNN with RNI $g$ which satisfies:

$$\forall (G, A) \in T ~~ \forall r_k \in R(G), \{ g(G)_i ~|~ i \in r_k \} \cong_{\epsilon,\delta} A_k $$
\end{theorem*}

\begin{proof}
Recall that we say two \changemarker{multisets} $A, B$ containing random variables are $(\epsilon,\delta)$-equal, denoted $A \cong_{\epsilon,\delta} B$, if there exists a bijection $\tau: A \to B$ such that $\forall a \in A,~ \text{Pr}(| a - \tau(a) | \leq \epsilon) \geq 1 - \delta$.


We will define a GNN with RNI $g$ by construction which satisfies the property required in the theorem. We do this in 3 intuitive steps:
\begin{enumerate}
    \item Define a GNN with RNI $f$ that, for each node, outputs a unique identifier for the orbit of the node and the original RNI value given to the node. Such a GNN exists because the function it is approximating is equivariant.
    \item Append $n$ identical layers onto $f$, each of which identifies the node containing the highest RNI value, gives that node a value from the target \changemarker{multiset} of labels corresponding to its orbit, marks off that particular value as claimed, and sets its RNI value to be small.
    \item Append a final layer which extracts only the target labels from the node attributes; these were given to the nodes by the preceding $n$ layers.
\end{enumerate}

First, notice that there exists a GNN with RNI $f: G_n \to (\mathbb{R}^4)^n$ that approximates the outputs $(n_i, r_i, 0, 0)$ for each node $i$ in the input graph $G$, where $n_i$ is the random noise initially added by RNI (before any message passing) and $r_i$ is a unique value corresponding to the graph orbit of $i$. Formally: $r_i = r_j \iff $ $i$ and $j$ are in the same orbit of the same graph (up to isomorphism). Put another way: $r_i = r_j \iff $ there exists an isomorphism $\alpha: G_i \to G_j$ (the graphs containing nodes $i$ and $j$, respectively) such that $\alpha(i) = j$.

Such a GNN $f$ exists because the function it is approximating is equivariant, allowing us to use Theorem \ref{thm:rni_equivariant}. Without loss of generality, we assume that RNI values are sampled from the interval $(0, 1)$ and that we only augment each node with one RNI value.

We define a GNN with RNI $g$ using $f$ as a starting point: we will append further message-passing layers to $f$. Append $n$ identical message-passing layers to $f$, each of which is defined as follows. Each node attribute in these layers will be a tuple $(n_i, r_i, c_i, t_i)$, where $c_i$ is used as a counter and $t_i$ is used to store the eventual node output value, corresponding to some target label. For this proof, we assume that target labels $A_k$ come from $\mathbb{R}$, but note that the proof is easily extended to vectors from $\mathbb{R}$ instead. Furthermore, we allow for $A_k$ to be a multiset (i.e. with repeated elements). Define $\fread$ by

$$ \fread(\{ (n_j, r_j, c_i, t_i) ~|~ j \in V(G) \}) := \text{argmax}_{(n_j, r_j, c_j, t_j) ~\forall j \in V(G)} ~n_j $$

In other words, $\fread$ extracts the tuple containing the maximum value of $n_j$ in the graph. Such a unique maximum exists with probability $1$, since finitely many RNI values are sampled from an infinite distribution. Do not define $\faggr$.

Define $\fupdate$ on the output $(n_j, r_j, c_j, t_j)$ of $\fread$ and the current node value $(n_i, r_i, c_i, t_i)$.

\[
  \fupdate( (n_i, r_i, c_i, t_i), (n_j, r_j, c_j, t_j) ) :=
  \begin{cases}
    (n_i, r_i, c_i, t_i) & \text{if } r_i \neq r_j \\
    (n_i, r_i, c_i + 1, t_i) & \text{if } r_i = r_j \text{ and } n_i \neq n_j \\
    (0, r_i, c_i + 1, (A_k)_{c_i}) & \text{if } r_i = r_j \text{ and } n_i = n_j
  \end{cases}
\]

In the above, $(A_k)_{c_i}$ denotes treating the \changemarker{multiset} of target labels $A_k$ as a sequence and retrieving the element with index $c_i$ (first index is 0). The above has access to $A_k$ since it can uniquely identify the input graph and orbit of the node using $r_i$, by how $r_i$ is defined.

As a consequence of its definition, the update function will retrieve one value from the target labels at a time, updating exactly one node to store this value. If another node $j$ within the same orbit as $i$ is being updated with this value, then the counter of $i$ is incremented to track that a value from the outputs has been claimed. Since the RNI value $n_i$ is set to 0, it ensures that each node will be given exactly one target label after $n$ rounds of message passing.

After appending the above $n$ message-passing layers, append one final layer with only an update function that extracts only the target labels, defined by

$$ \fupdate( (0, r_i, c_i, t_i) ) := t_i $$

The above construction of $g$ satisfies the probability ($\delta$) and approximation ($\epsilon$) requirements of the property stated in the theorem, since $f$ is an $(\epsilon, \delta)$-approximation, a unique maximum RNI value exists for each graph with probability 1, and all other required operations can be $\epsilon$-approximated by neural networks. The exact bijection used to map between the output and target \changemarker{multisets} will depend on the RNI values of the nodes, since they determine the order in which target values are assigned to node attributes in the construction. Whilst our provided construction requires at least $n+1$ message-passing layers, more efficient constructions exist using more complex readout functions than simple maximisation. The final update layer can also be merged with the previous layer. However, we presented the above construction due to its simplicity and how easy it is to understand the mechanism.

\end{proof}

\newpage
\subsection{Theorem \ref{thm:clip_universal}}

\begin{theorem*}
Let $n \geq 1$ and let $f: G_n \to \mathbb{R}^{n}$ be equivariant. Then for all $\epsilon > 0$, there is a GNN with unique node IDs that $\epsilon$-\emph{approximates} $f$.
\end{theorem*}

\begin{proof}
For this proof, we assume the reader to be familiar with the theorems and proofs of \citet{dasoulas2019coloring}, particularly their Theorem 4. First, we need to prove that GNNs with unique node IDs are equivalent to 1-CLIP, which is $k$-CLIP with $k=1$.

CLIP is defined as a 3-step process, the first of which is assigning colours to nodes. In CLIP, the essential part of each colouring chosen is that all nodes with the same attributes will be assigned different colours. By representing colours with unique node IDs (using a one-hot encoding) we ensure that the above property holds, since all nodes will be assigned different colours. Since we are considering 1-CLIP, we set $k=1$ and only sample one colouring, which is the particular set of unique IDs we assign.

Step 2 of CLIP is just standard GNN message-passing on our created coloured graph. Step 3 of CLIP maximizes over all possible colourings, of which we have only one, so the maximization can be dropped. This yields a standard GNN global readout layer. Thus, GNNs with unique node IDs are equivalent to 1-CLIP. Theorem 4 of \citet{dasoulas2019coloring} states the universality of 1-CLIP for invariant functions, which we provide here as a Lemma:

\begin{lemma*}
The 1-CLIP algorithm with one local iteration is a random representation whose expectation is a universal representation of the space $\textbf{Graph}_m$ of graphs with node attributes.
\end{lemma*}

They further state that for any colouring, 1-CLIP returns an $\epsilon$-approximation of the target function and that, given sufficient training, the variance can be reduced to an arbitrary precision. The above only applies to invariant functions because it is a representation of the space $\textbf{Graph}_m$, which is defined using invariance by permutation of the labels. Note that when defining $\textbf{Graph}_m$ in this paper, we use $n_{\text{max}} := n$, instead of just considering some arbitrary large $n_{\text{max}}$.

To apply this theorem to equivariant functions, we need to consider the space $\textbf{Node}_m$, which we define to be the set of all nodes from graphs in $\textbf{Graph}_m$. $\textbf{Node}_m$ is Hausdorff as a trivial consequence of $\textbf{Graph}_m$ being Hausdorff, using the same quotient space of orbits. If we can separate this space in a continuous and concatenable way, then we can utilize Corollary 1 of \citet{dasoulas2019coloring} to show that it is universal. To do this, consider a GNN $f$ which separates the space $\textbf{Graph}_m$, which we know exists due to Theorem 4 of \citet{dasoulas2019coloring}. We define a new GNN $g$ using $f$ as a starting point. Substitute the final global readout layer $M$ of $f$ for a new layer with the same readout function, but have the output of the readout be assigned to every node. Formally, define

$$ \fupdate^{\theta_M}(v_i^{M},~ \faggr^{\theta_M'}(\{ v_j^M ~|~ j \in N(i) \}),~ \fread^{\theta_M''}(\{ v_j^M ~|~ j \in V(G) \}) ) := \fread^{\theta_M''}(\{ v_j^M ~|~ j \in V(G) \}), $$

where $\fread^{\theta_M''}$ is the former global readout layer. Then, change the update functions of each layer of $f$ such that a portion of each node attribute is reserved for the unique ID of the node, and do not use or change this ID in each layer of $f$. Then, after the final layer $M$, the attribute $v_i^{M+1}$ of each node $i$ will be

$$ v_i^{M+1} = (u_i, r_i) := (u_i, \fread^{\theta_M''}(\{ v_j^M ~|~ j \in V(G) \})), $$ 

where $u_i$ is the unique ID given to node $i$. Since $f$ separates $\textbf{Graph}_m$, $r_i$ uniquely identifies the graph provided in the input. Furthermore, $u_i$ uniquely identifies each node in the input graph. Thus, $v_i^{M+1} = (u_i, r_i)$ uniquely identifies every node in the space $\textbf{Node}_m$, meaning that $g$ yields a separable representation of $\textbf{Node}_m$. Furthermore, $g$ is continuous and concatenable by its construction, so we can apply Corollary 1 of \citet{dasoulas2019coloring} to state that $g$ is universal.

\end{proof}

\subsection{Theorem \ref{thm:clip_coordination}}

\begin{theorem*}
Let $n \geq 1$ and consider a set $T$, where each $(G, A) \in T$ is a graph-labels pair, such that $G \in G_n$ and there is a \changemarker{multiset} of target labels $A_k \in A$ for each orbit $r_k \in R(G)$, with $| A_k | = | r_k |$. Then for all $\epsilon > 0$ there is a GNN with unique node IDs $g$ which satisfies:

$$\forall (G, A) \in T ~~ \forall r_k \in R(G), \{ g(G)_i ~|~ i \in r_k \} \cong_{\epsilon} A_k $$
\end{theorem*}

\begin{proof}
Recall that two \changemarker{multisets} $A, B$ without random variables are $\epsilon$-equal, denoted $A \cong_{\epsilon} B$, if there exists a bijection $\tau: A \to B$ such that $\forall a \in A,~ | a - \tau(a) | \leq \epsilon$.

The proof for this theorem is by construction, where the construction is nearly identical to the one used in the proof of Theorem \ref{thm:rni_coordination}. A GNN $f$ exists due to Theorem \ref{thm:clip_universal} instead of Theorem \ref{thm:rni_equivariant}. Unique ID values are used instead of RNI values. One-hot encodings can be maximised over in a similar way, by treating them as binary numbers. A unique maximum unique ID will always exist by definition.

The construction otherwise proceeds analogously, except that it presents an $\epsilon$-approximation of each target \changemarker{multiset} instead of an $(\epsilon, \delta)$-approximation, since no randomness is used.
\end{proof}

%
%
\newpage
\section{Experiments} \label{app:experiments}
In this appendix, we provide full details about our experiments for reproducibility.

\subsection{Baseline Communication Methods} \label{app:baselines}

For evaluation, we adopt a diverse selection of MARL communication methods which fall under the GDN paradigm. These are shown in Table \ref{tab:baseline_models_full}, along with the respective paradigm (whether the method simply falls within GDNs or whether GNNs are explicitly used for communication), the MARL paradigm, and communication graph structure. We use the code provided by \citet{niu2021multi, jiang2018graph} as starting points. The \href{https://github.com/jiechuanjiang/pytorch_DGN}{code} of \citet{jiang2018graph} uses an MIT license and the \href{https://github.com/CORE-Robotics-Lab/MAGIC}{code} of \citet{niu2021multi} does not have one. All of the implementations are extended to be able to support multiple rounds of message-passing and the baselines are augmented with the ability for their communication to be masked by the environment (e.g. based on distance or obstacles in the environment).

\citet{sukhbaatar2016learning} define CommNet, which has a single, basic, learnable communication channel. They define it in such a way that agents can enter and leave the communication range of other agents. It maps directly onto a GDN approach where mean is used for aggregation. \citet{singh2018learning} define IC3Net, which operates in a similar manner to CommNet, except the communication graph is complete and communication is controlled by gating, meaning each agent can decide whether or not to broadcast to another agent. \citet{das2019tarmac} define TarMAC, where a soft attention mechanism is used to decide how much of a message an agent should process. This implicitly yields a complete communication graph, which a graph attention network (GAT) is able to model. TarMAC is also extended to use IC3Net's reward and communication structure, which we refer to as \emph{T-IC3Net}.

\citet{jiang2018graph} define DGN, which operates on graphs that arise deterministically from the environment (e.g. based on agent proximity). The model consists of an encoding layer from the observations, two convolutional layers that use multi-head dot-product attention as the convolutional kernel, and a shared Q-network between all agents. There are skip connections between the convolutional layers. Note that among our chosen methods, DGN is the only value-based one. \citet{niu2021multi} define MAGIC, which learns to construct a communication graph and then uses GNNs to operate on the graph. The Scheduler learns which agents should communicate with each other and outputs a communication graph. The Message Processor then uses GATs for multiple rounds of communication on the graph.

\begin{table}
  \caption{Architecture of the Baselines}
  \label{tab:baseline_models_full}
  \centering
  \begin{tabular}{llll}
    \toprule
    Name & Communication Graph & MARL Paradigm & GNN Usage \\
    \midrule
    CommNet \citep{sukhbaatar2016learning} & Complete (or environment-based) & Recurrent A2C & Implicit \\
    IC3Net \citep{singh2018learning} & Complete + Gating & Recurrent A2C & Implicit \\
    TarMAC \citep{das2019tarmac} & Complete + Learned Soft Edges & Recurrent A2C & Implicit GAT \\
    T-IC3Net \citep{singh2018learning, das2019tarmac} & Gating + Learned Soft Edges & Recurrent A2C & Implicit GAT \\
    MAGIC \citep{niu2021multi} & Learned & Recurrent A2C & Explicit GAT \\
    DGN \citep{jiang2018graph} & Environment-based & Q-network & Explicit GCN \\
    \bottomrule
  \end{tabular}
\end{table}

\subsection{Environments} \label{app:environments}
Predator-Prey \citep{singh2018learning, das2019tarmac, liu2020multi, li2020deep, niu2021multi} and Traffic Junction \citep{sukhbaatar2016learning, singh2018learning, das2019tarmac, liu2020multi, li2020deep, niu2021multi} are common MARL communication benchmarks. We perform evaluations on them to test how well our universally expressive GDN models perform when there is not necessarily a benefit to having communication expressivity beyond 1-WL. We also introduce two new environments, Drone Scatter and Box Pushing, to respectively test symmetry-breaking and communication expressivity beyond 1-WL.

\textbf{Predator-Prey}, introduced by \citet{singh2018learning}, consists of predators (agents) with limited vision trying to find stationary prey. They can communicate with each other within a range of 5 and at each time step move one grid cell in any cardinal direction. An episode is deemed a success if all agents have found and are sitting on top of the prey. We utilize the ``cooperative'' reward setting of the environment, meaning that reward is given at each time step proportional to the number of agents on the prey. The environment is demonstrated in Figure \ref{fig:pp_image_full}.

\begin{figure}
    \centering
    \includegraphics[scale=0.4]{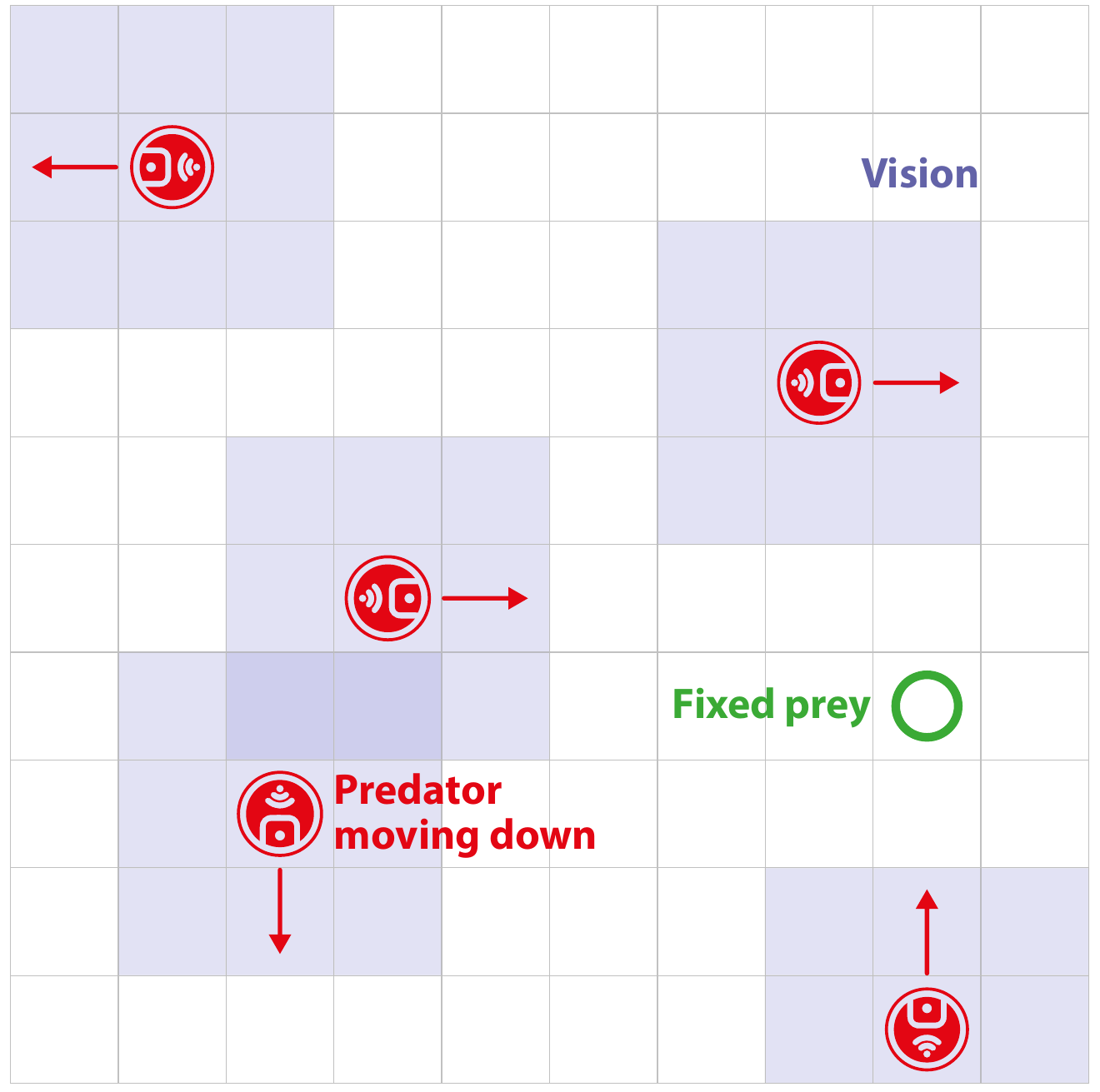}
    \caption{PredatorPrey Environment}
    \label{fig:pp_image_full}
\end{figure}

\textbf{Traffic Junction}, introduced by \citet{sukhbaatar2016learning}, consists of intersecting roads with cars (agents) driving along them. The agents have limited vision and need to communicate to avoid collisions; an episode is deemed a success if it had no collisions. Each car can communicate with any other car within a range of 3. At each time step, cars enter the environment with a given probability, provided the number of cars in the environment does not exceed the allowed maximum. At each step, a car can either ``gas'' or ``break'', leading to it either moving forward one cell on its route or remaining stationary. We utilize both the ``Easy'' and ``Medium'' versions of the environment, which respectively consist of two intersecting 1-way roads and two intersecting 2-way roads. The environment is demonstrated in Figure \ref{fig:tj_image_full}.

\begin{figure}
    \centering
    \includegraphics[scale=0.4]{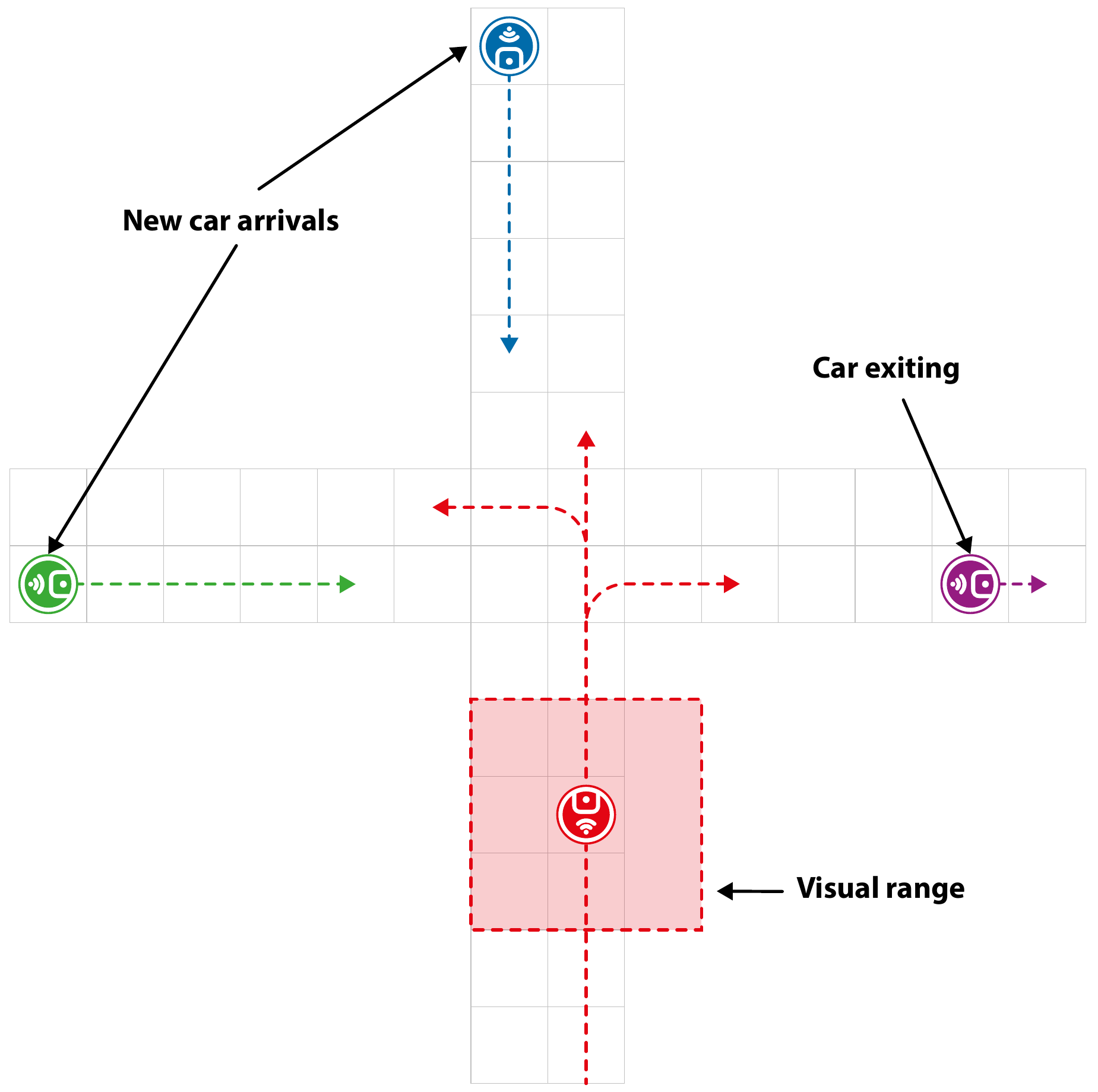}
    \caption{Medium TrafficJunction Environment}
    \label{fig:tj_image_full}
\end{figure}

The \textbf{Drone Scatter} environment is a grid environment, which we design to test the ability of communication models to perform symmetry-breaking. It consists of 4 drones in a 20x20 homogeneous field. The outer lines of the field are marked by fences. The drones can move any of the 4 cardinal directions at each time step. Their goal is to move around and find a target hidden in the field, which they can only notice when they get close to. The drones do not have GPS and can only see directly beneath them using their cameras. They also know which action they took in the last time step. The best way for them to locate the target is to split up and search in different portions of the field. Thus, in the ``easy'' version of the environment, they are encouraged to do this by being given a reward based on how far away they are from the rest of the drones. They are also always given a reward for finding the target. The environment is demonstrated in Figure \ref{fig:ds_image_full}.

\begin{figure}
    \centering
    \includegraphics[scale=0.4]{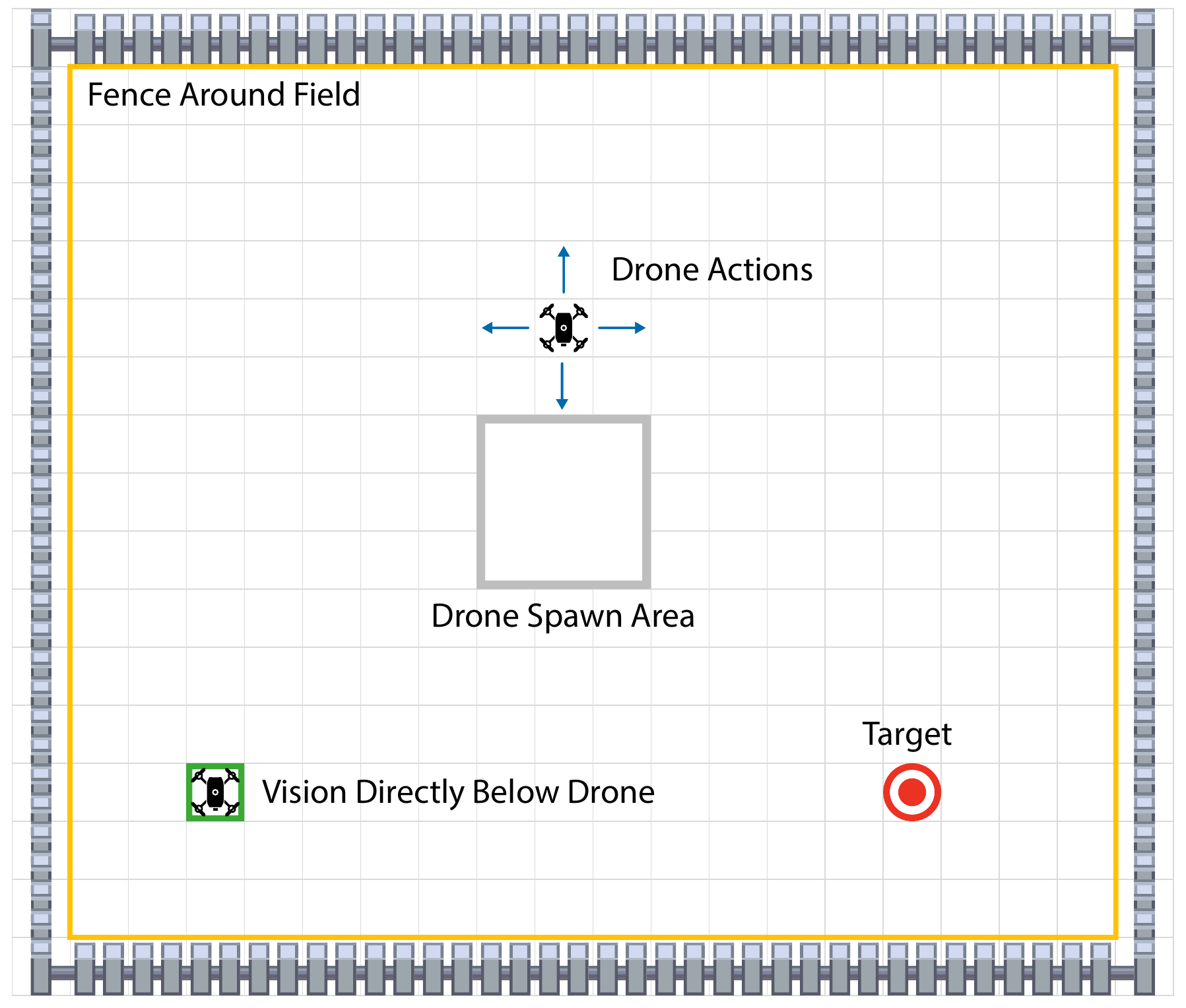}
    \caption{DroneScatter Environment}
    \label{fig:ds_image_full}
\end{figure}

The \textbf{Box Pushing} environment is also a grid environment, which we design to test communication expressivity beyond 1-WL. It has 10 robots in a 12x12 construction site. The out-most 3 grid cells on every side represent the ``clearing area'', into which robots need to clear the boxes from the central area of the site. Robots can see the cells next to them. Robots attach themselves to boxes before they can move them. When they are attached, robots cannot see around them any more. Free-roaming robots can communicate with any other free-roaming robots, but attached robots can only communicate with the robots directly adjacent to them. The environment either spawns with one large box or two small boxes. 4 attached robots are needed to move a small box. 8 attached robots are needed to move a large box. The agents have 9 possible actions: stay, move in 1 of the 4 cardinal directions, or power move in 1 of the 4 cardinal directions. A small box only moves if all attached agents move in the same direction. A large box only moves if all attached agents power move in the same direction. Robots are penalised for exerting themselves without moving the box, and are rewarded for moving the box closer to the clearing area or clearing the box. Once a box has been cleared, it is removed from the site and the agents are free to continue moving around the environment.

To solve the environment, the robots need to be able to communicate with each other to figure out which type of box they are on and all push correctly, at the same time, and in the same direction. Since the communication graphs corresponding to the scenarios with small and large boxes are 1-WL indistinguishable, communication beyond 1-WL is needed to properly solve the environment. In the ``easy'' version, robots spawn already attached to the boxes. The environment is demonstrated in Figure \ref{fig:bp_image_full}.

\begin{figure}
    \centering
    \includegraphics[scale=0.45]{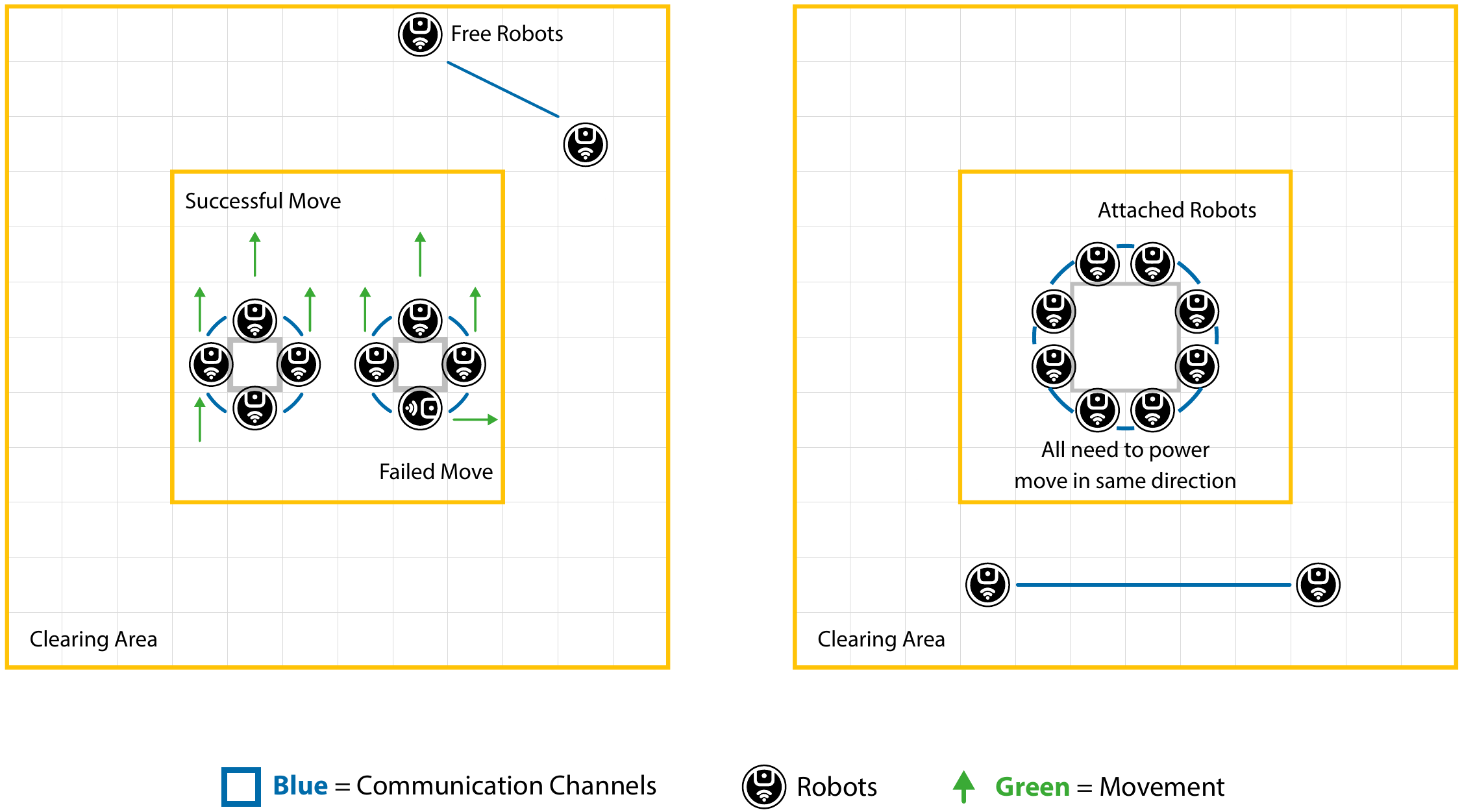}
    \caption{BoxPushing Environment}
    \label{fig:bp_image_full}
\end{figure}

\subsection{Hybrid Imitation Learning} \label{app:hybrid_learning}
We use hybrid imitation learning for all Box Pushing experiments to solve the issue of the exceptionally sparse rewards. Since agents are only given a reward when they all take the same action and thus move the box, random policies will struggle to ever obtain a meaningful reward signal during exploration. When agents are attached to a large box, the probability of this happening at a single time step is $\nicefrac{4}{9} \times (\nicefrac{1}{9})^7 = 9.29 \mathrm{e}{-8}$, meaning $\approx 10^8$ time steps of experience are needed before any reward can be expected.

\citet{hester2018deep} propose using expert demonstrations for training and \citet{subramanian2016exploration} propose using some expert demonstrations to help exploration. Inspired by this, during training, we interleave 100 expert experiences for every 500 experiences collected by the agents in the environment. While this is stable for the value-based method DGN, doing so for the A2C methods leads to very unstable performance during training. \citet{wang2016sample} propose several ways to improve such training of A2C methods, but we choose not to implement them as it is not the focus of this paper.

\subsection{Model and Environment Hyperparameters}

\subsubsection{Fixed Model Hyperparameters}
Model hyperparameters that are fixed across all experiments are shown in Table \ref{tab:params:all}, along with which group they belong to, their fixed values, and their descriptions.

\begin{table}
\small 
\caption{Fixed model parameters for all experiments}
\label{tab:params:all}
\centering
\begin{tabular}{@{}llll@{}}
    \toprule
    \textbf{Group} & \textbf{Parameter} & \textbf{Value} & \textbf{Description} \\
    \midrule
    
    Training & epoch\_size & 10 & Number of update iterations in an epoch \\
    & batch\_size & 500 & Number of steps before each update (per thread) \\
    & nprocesses & 1 & How many processes to run \\
    
    \midrule
    
    DGN Training & update\_interval & 5 & How many episodes between model update steps \\
    & train\_steps & 5 & How many times to train the model in a training step \\
    & dgn\_batch\_size & 128 & Batch size  \\
    & epsilon\_start & 1 & Epsilon starting value \\
    & epsilon\_min & 0.1 & Minimum epsilon value \\
    & buffer\_capacity & 40000 & Capacity of the replay buffer \\
    
    \midrule
    
    Model & hid\_size & 128 & Hidden layer size \\
    & qk\_hid\_size & 16 & Key and query size for soft attention \\
    & value\_hid\_size & 32 & Value size for soft attention \\
    & recurrent & True & Make the A2C model recurrent in time \\
    & num\_evals & 10 & Number of evaluation runs for each training iteration \\
    & env\_graph & True & Whether the environment masks communication \\
    & comm\_passes & 4 & Number of comm passes per step over the model \\
    
    \midrule
    
    Optimization & gamma & 1 & Discount factor \\
    & normalize\_rewards & False & Normalize rewards in each batch \\
    & lrate & 0.001 & Learning rate \\
    & entr & 0 & Entropy regularization coefficient \\
    & value\_coeff & 0.01 & Coefficient for value loss term \\
    
    \midrule
    
    A2C Models & comm\_mode & avg & Mode for communication tensor calculation \\
    & comm\_mask\_zero & False & Mask all communication \\
    & mean\_ratio & 1 & How much cooperation? 1 means fully cooperative \\
    & rnn\_type & MLP & Type of RNN to use [LSTM | MLP] \\
    & detach\_gap & 10 & Detach hidden and cell states for RNNs at this interval \\
    & comm\_init & uniform & How to initialise comm weights [uniform | zeros] \\
    & hard\_attn & False & Whether to use hard attention: action - talk | silent \\
    & comm\_action\_one & False & Whether to always talk \\
    & advantages\_per\_action & False & Whether to multiply action log prob with advantages \\
    & share\_weights & True & Share model parameters between agents \\
    
    \midrule
    
    MAGIC & directed & True & Whether the communication graph is directed \\
    & self\_loop\_type & 1 & Self loop type in the GAT layers (1: with self loop) \\
    & gat\_num\_heads & 4 & Number of heads in GAT layers except the last one \\
    & gat\_num\_heads\_out & 1 & Number of heads in output GAT layer \\
    & gat\_hid\_size & 32 & Hidden size of one head in GAT \\
    & message\_decoder & True & Whether use the message decoder \\
    & gat\_normalize & False & Whether to normalize the GAT coefficients \\
    & ge\_num\_heads & 4 & Number of heads in the GAT encoder \\
    & gat\_encoder\_normalize & False & Normalize the coefficients in the GAT encoder \\
    & use\_gat\_encoder & False & Whether use the GAT encoder \\
    & gat\_encoder\_out\_size & 64 & Hidden size of output of the GAT encoder \\
    & graph\_complete & False & Whether the communication graph is complete \\
    & learn\_different\_graphs & False & Learn a new communication graph each round \\
    & message\_encoder & False & Whether to use the message encoder \\

    \bottomrule
\end{tabular}
\end{table}

\subsubsection{TrafficJunction-Easy}
Model hyperparameters for the Easy Traffic Junction experiments are shown in Table \ref{tab:params:easy_tj}, along with which group they belong to, their values (sometimes a set of values), and their descriptions.

\begin{table}
\small 
\caption{Hyperparameters for Easy Traffic Junction experiments}
\label{tab:params:easy_tj}
\centering
\begin{tabular}{@{}llp{3.4cm}l@{}}
    \toprule
    \textbf{Group} & \textbf{Parameter} & \textbf{Value(s)} & \textbf{Description} \\
    \midrule
    
    Environment & difficulty & easy & Difficulty level [easy | medium | hard] \\
    & dim & 6 & Dimension of box (i.e length of road) \\
    & env\_name & traffic\_junction & Environment name \\
    & max\_steps & 20 & Force to end the game after this many steps \\
    & nagents & 5 & Number of agents \\
    & vision & 1 & Vision of car \\
    & add\_rate\_min & 0.3 & Min probability to add car (till curr. start) \\
    & add\_rate\_max & 0.3 & Max rate at which to add car \\
    & curr\_start & 0 & Start making harder after this epoch \\
    & curr\_end & 0 & When to make the game hardest \\
    & vocab\_type & bool & Type of location vector to use [bool | scalar] \\
    & comm\_range & 3 & Agent communication range \\
    
    \midrule
    
    Model & epsilon\_step & $2\mathrm{e}{-5}$ & Amount to subtract from epsilon each episode \\
    & model & [commnet, tarmac, ic3net, tarmac\_ic3net, dgn, magic] & Which baseline model to use \\
    & num\_epochs & 2000 & Number of training epochs \\
    & rni & [0.75, 0.25, 0, 1] & RNI ratio. 0 for none. 1 for unique IDs \\
    & seed & [1, 2, 3, 4, 5] & Random seed \\

    \bottomrule
\end{tabular}
\end{table}

\subsubsection{PredatorPrey}
Model hyperparameters for the Predator-Prey experiments are shown in Table \ref{tab:params:pp}, along with which group they belong to, their values (sometimes a set of values), and their descriptions.

\begin{table}
\small 
\caption{Hyperparameters for Predator-Prey experiments}
\label{tab:params:pp}
\centering
\begin{tabular}{@{}llp{3.4cm}l@{}}
    \toprule
    \textbf{Group} & \textbf{Parameter} & \textbf{Value(s)} & \textbf{Description} \\
    \midrule
    
    Environment & dim & 10 & Dimension of box (i.e side length) \\
    & env\_name & predator\_prey & Environment name \\
    & max\_steps & 40 & Force to end the game after this many steps \\
    & mode & cooperative & Reward mode \\
    & nagents & 5 & Number of agents \\
    & vision & 1 & Vision of predator \\
    & nenemies & 1 & Total number of preys in play \\
    & moving\_prey & False & Whether prey is fixed or moving \\
    & no\_stay & False & Whether predators have an action to stay \\
    & comm\_range & 5 & Agent communication range \\
    
    \midrule
    
    Model & epsilon\_step & $2\mathrm{e}{-5}$ & Amount to subtract from epsilon each episode \\
    & model & [commnet, tarmac, ic3net, tarmac\_ic3net, dgn, magic] & Which baseline model to use \\
    & num\_epochs & 2000 & Number of training epochs \\
    & rni & [0.75, 0.25, 0, 1] & RNI ratio. 0 for none. 1 for unique IDs \\
    & seed & [1, 2, 3, 4, 5] & Random seed \\

    \bottomrule
\end{tabular}
\end{table}

\subsubsection{TrafficJunction-Medium}
Model hyperparameters for the Medium Traffic Junction experiments are shown in Table \ref{tab:params:medium_tj}, along with which group they belong to, their values (sometimes a set of values), and their descriptions.

\begin{table}
\small 
\caption{Hyperparameters for Medium Traffic Junction experiments}
\label{tab:params:medium_tj}
\centering
\begin{tabular}{@{}llp{3.4cm}l@{}}
    \toprule
    \textbf{Group} & \textbf{Parameter} & \textbf{Value(s)} & \textbf{Description} \\
    \midrule
    
    Environment & difficulty & medium & Difficulty level [easy | medium | hard] \\
    & dim & 14 & Dimension of box (i.e length of road) \\
    & env\_name & traffic\_junction & Environment name \\
    & max\_steps & 40 & Force to end the game after this many steps \\
    & nagents & 10 & Number of agents \\
    & vision & 1 & Vision of car \\
    & add\_rate\_min & 0.3 & Min probability to add car (till curr. start) \\
    & add\_rate\_max & 0.3 & Max rate at which to add car \\
    & curr\_start & 0 & Start making harder after this epoch \\
    & curr\_end & 0 & When to make the game hardest \\
    & vocab\_type & bool & Type of location vector to use [bool | scalar] \\
    & comm\_range & 3 & Agent communication range \\
    
    \midrule
    
    Model & epsilon\_step & $2\mathrm{e}{-5}$ & Amount to subtract from epsilon each episode \\
    & model & [commnet, tarmac, ic3net, tarmac\_ic3net, dgn, magic] & Which baseline model to use \\
    & num\_epochs & 2000 & Number of training epochs \\
    & rni & [0.75, 0.25, 0, 1] & RNI ratio. 0 for none. 1 for unique IDs \\
    & seed & [1, 2, 3, 4, 5] & Random seed \\

    \bottomrule
\end{tabular}
\end{table}

\subsubsection{BoxPushing}
Model hyperparameters for the Box Pushing experiments are shown in Table \ref{tab:params:bp}, along with which group they belong to, their values (sometimes a set of values), and their descriptions.

\begin{table}
\small 
\caption{Hyperparameters for Box Pushing experiments}
\label{tab:params:bp}
\centering
\begin{tabular}{@{}lp{1.75cm}p{3.4cm}l@{}}
    \toprule
    \textbf{Group} & \textbf{Parameter} & \textbf{Value(s)} & \textbf{Description} \\
    \midrule
    
    Environment & difficulty & easy & Difficulty level. Easy: robots already attached \\
    & dim & 12 & Dimension of area (i.e. side length) \\
    & env\_name & box\_pushing & Environment name \\
    & max\_steps & 20 & Force to end the game after this many steps \\
    & nagents & 10 & Number of agents \\
    & vision & 1 & Vision of robot \\

    \midrule
    
    Model & epsilon\_step & $2\mathrm{e}{-5}$ & Amount to subtract from epsilon each episode \\
    & imitation & True & Whether to use hybrid imitation learning \\
    & model & [commnet, tarmac, ic3net, tarmac\_ic3net, dgn, magic] & Which baseline model to use \\
    & num\_epochs & 2000 & Number of training epochs \\
    & num\_\newline imitation\_\newline experiences & 100 & Number of experiences coming from imitation \\
    & num\_\newline normal\_\newline experiences & 500 & Number of normal policy experiences \\
    & rni & [0.75, 0.25, 0, 1] & RNI ratio. 0 for none. 1 for unique IDs \\
    & seed & [1, 2, 3, 4, 5, 6, 7, 8, 9, 10] & Random seed \\

    \bottomrule
\end{tabular}
\end{table}

\subsubsection{DroneScatter-Stochastic}
Model hyperparameters for the Drone Scatter experiments with stochastic evaluation are shown in Table \ref{tab:params:ds_stochastic}, along with which group they belong to, their values (sometimes a set of values), and their descriptions.

\begin{table}
\small 
\caption{Hyperparameters for Drone Scatter experiments with stochastic evaluation}
\label{tab:params:ds_stochastic}
\centering
\begin{tabular}{@{}lp{1.75cm}p{3.4cm}l@{}}
    \toprule
    \textbf{Group} & \textbf{Parameter} & \textbf{Value(s)} & \textbf{Description} \\
    \midrule
    
    Environment & difficulty & easy & Difficulty level. Easy: rewarded for splitting \\
    & dim & 20 & Dimension of field area (i.e. side length) \\
    & env\_name & drone\_scatter & Environment name \\
    & max\_steps & 20 & Force to end the game after this many steps \\
    & nagents & 4 & Number of agents \\
    & comm\_range & 10 & Agent communication range \\
    & find\_range & 3 & Agent distance to target to count as find \\
    & min\_target\_\newline distance & 3 & Min distance target can be from spawn area \\
    
    \midrule
    
    Model & epsilon\_step & $2\mathrm{e}{-5}$ & Amount to subtract from epsilon each episode \\
    & model & [commnet, tarmac, ic3net, tarmac\_ic3net, magic] & Which baseline model to use \\
    & num\_epochs & 2000 & Number of training epochs \\
    & rni & [0.75, 0, 1] & RNI ratio. 0 for none. 1 for unique IDs \\
    & seed & [1, 2, 3, 4, 5] & Random seed \\

    \bottomrule
\end{tabular}
\end{table}

\subsubsection{DroneScatter-Greedy}
Model hyperparameters for the Drone Scatter experiments with greedy evaluation are shown in Table \ref{tab:params:ds_greedy}, along with which group they belong to, their values (sometimes a set of values), and their descriptions.

\begin{table}
\small 
\caption{Hyperparameters for Drone Scatter experiments with greedy evaluation}
\label{tab:params:ds_greedy}
\centering
\begin{tabular}{@{}lp{1.75cm}p{3.4cm}l@{}}
    \toprule
    \textbf{Group} & \textbf{Parameter} & \textbf{Value(s)} & \textbf{Description} \\
    \midrule
    
    Environment & difficulty & easy & Difficulty level. Easy: rewarded for splitting \\
    & dim & 20 & Dimension of field area (i.e. side length) \\
    & env\_name & drone\_scatter & Environment name \\
    & max\_steps & 20 & Force to end the game after this many steps \\
    & nagents & 4 & Number of agents \\
    & comm\_range & 10 & Agent communication range \\
    & find\_range & 3 & Agent distance to target to count as find \\
    & min\_target\_\newline distance & 3 & Min distance target can be from spawn area \\
    
    \midrule
    
    Model & epsilon\_step & $2\mathrm{e}{-5}$ & Amount to subtract from epsilon each episode \\
    & greedy\_a2c\_\newline eval & True & Whether to evaluate A2C methods greedily \\
    & model & [commnet, tarmac, ic3net, tarmac\_ic3net, dgn, magic] & Which baseline model to use \\
    & num\_epochs & 2000 & Number of training epochs \\
    & rni & [0.75, 0, 1] & RNI ratio. 0 for none. 1 for unique IDs \\
    & seed & [1, 2, 3, 4, 5] & Random seed \\

    \bottomrule
\end{tabular}
\end{table}

%
%
\clearpage
\section{Full Results} \label{app:results}
In this appendix, full results from all experiments are shown. Our experiments were done in parallel on an internal cluster, using only CPUs. With regards to compute time, 127 days were used for Easy Traffic Junction, 149 for Predator-Prey, 186 for Medium Traffic Junction, 502 for Box Pushing, 76 for stochastic Drone Scatter, and 89 for greedy Drone Scatter. This comes to a total of 1129 days.

%
%
\subsection{Result Tables}
In this section, scores for all metrics across all experiments are shown in Table \ref{tab:full_results:tj_easy} (Easy Traffic Junction), Table \ref{tab:full_results:pp} (Predator-Prey), Table \ref{tab:full_results:tj_medium} (Medium Traffic Junction), Table \ref{tab:full_results:bp} (Box Pushing), Table \ref{tab:full_results:ds_stochastic} (Drone Scatter with stochastic evaluation), and Table \ref{tab:full_results:ds_greedy} (Drone Scatter with greedy evaluation).

%
%
\begin{table}
\small 
\caption{Mean and 95\% confidence interval for Easy TrafficJunction across all baselines}
\label{tab:full_results:tj_easy}
\centering
\begin{tabular}{@{}rrrrrrrrrrrrrr@{}}
    \toprule
    Baseline & Metric & \textbf{Baseline} && \textbf{Unique IDs} && \textbf{0.75 RNI} && \textbf{0.25 RNI} \\
    \midrule
    
CommNet & Success & \pmb{$ 1 \pm 0 $} && \pmb{$ 1 \pm 0 $} && \pmb{$ 1 \pm 0 $} && \pmb{$ 1 \pm 0 $} \\
& Reward & $ -1.7 \pm 0.01 $ && \pmb{$ -1.69 \pm 0 $} && $ -1.78 \pm 0.02 $ && $ -1.7 \pm 0.01 $ \\
DGN & Success & $ 0.987 \pm 0 $ && $ 0.99 \pm 0 $ && $ 0.848 \pm 0.15 $ && \pmb{$ 0.996 \pm 0 $} \\
& Reward & $ -4.48 \pm 0.79 $ && $ -4 \pm 0.17 $ && $ -9.35 \pm 5.57 $ && \pmb{$ -3.99 \pm 0.15 $} \\
IC3Net & Success & \pmb{$ 1 \pm 0 $} && \pmb{$ 1 \pm 0 $} && \pmb{$ 1 \pm 0 $} && $ 0.986 \pm 0.02 $ \\
& Reward & $ -1.71 \pm 0 $ && \pmb{$ -1.7 \pm 0.01 $} && $ -1.74 \pm 0.01 $ && $ -2.02 \pm 0.51 $ \\
MAGIC & Success & $ 0.634 \pm 0.11 $ && $ 0.764 \pm 0.13 $ && $ 0.684 \pm 0.11 $ && \pmb{$ 0.787 \pm 0.09 $} \\
& Reward & $ -16 \pm 1.67 $ && $ -15.8 \pm 1.69 $ && $ -15 \pm 2.05 $ && \pmb{$ -14.7 \pm 2.08 $} \\
TarMAC & Success & $ 0.994 \pm 0.01 $ && \pmb{$ 1 \pm 0 $} && $ 0.933 \pm 0.04 $ && \pmb{$ 1 \pm 0 $} \\
& Reward & $ -2.07 \pm 0.44 $ && \pmb{$ -1.72 \pm 0.02 $} && $ -3.59 \pm 1.28 $ && $ -1.76 \pm 0.04 $ \\
T-IC3Net & Success & \pmb{$ 1 \pm 0 $} && $ 0.998 \pm 0 $ && $ 0.94 \pm 0.04 $ && $ 0.974 \pm 0.04 $ \\
& Reward & \pmb{$ -1.74 \pm 0.01 $} && $ -1.79 \pm 0.11 $ && $ -3.16 \pm 1.03 $ && $ -2.25 \pm 0.91 $ \\

    \bottomrule
\end{tabular}
\end{table}

%
%
\begin{table}
\small 
\caption{Mean and 95\% confidence interval for PredatorPrey across all baselines}
\label{tab:full_results:pp}
\centering
\begin{tabular}{@{}rrrrrrrrrrrrrr@{}}
    \toprule
    Baseline & Metric & \textbf{Baseline} && \textbf{Unique IDs} && \textbf{0.75 RNI} && \textbf{0.25 RNI} \\
    \midrule
    
CommNet & Success & $ 0.88 \pm 0.03 $ && \pmb{$ 0.908 \pm 0.02 $} && $ 0.194 \pm 0.02 $ && $ 0.476 \pm 0.05 $ \\
& Reward & $ 23.15 \pm 0.92 $ && \pmb{$ 23.71 \pm 1.18 $} && $ 1.828 \pm 0.51 $ && $ 10.53 \pm 1.92 $ \\
DGN & Success & $ 0.014 \pm 0 $ && $ 0.016 \pm 0 $ && $ 0.026 \pm 0.03 $ && \pmb{$ 0.032 \pm 0.01 $} \\
& Reward & $ -6.8 \pm 0.71 $ && $ -7.84 \pm 0.26 $ && $ -7.69 \pm 0.91 $ && \pmb{$ -4.37 \pm 2.63 $} \\
IC3Net & Success & \pmb{$ 0.952 \pm 0 $} && $ 0.93 \pm 0.02 $ && $ 0.454 \pm 0.08 $ && $ 0.933 \pm 0.02 $ \\
& Reward & $ 22.54 \pm 1.19 $ && $ 22.99 \pm 0.52 $ && $ 10.19 \pm 1.47 $ && \pmb{$ 24.38 \pm 1.63 $} \\
MAGIC & Success & \pmb{$ 0.892 \pm 0.02 $} && $ 0.888 \pm 0.05 $ && $ 0.112 \pm 0.03 $ && $ 0.451 \pm 0.09 $ \\
& Reward & \pmb{$ 21.62 \pm 1.31 $} && $ 21.36 \pm 1.65 $ && $ -0.85 \pm 1.63 $ && $ 9.487 \pm 3.07 $ \\
TarMAC & Success & $ 0.169 \pm 0.09 $ && \pmb{$ 0.24 \pm 0.11 $} && $ 0.068 \pm 0.01 $ && $ 0.086 \pm 0.02 $ \\
& Reward & $ 0.323 \pm 3.65 $ && \pmb{$ 3.131 \pm 3.96 $} && $ -5.22 \pm 0.54 $ && $ -3.14 \pm 1.27 $ \\
T-IC3Net & Success & \pmb{$ 0.938 \pm 0.02 $} && \pmb{$ 0.938 \pm 0.01 $} && $ 0.27 \pm 0.02 $ && $ 0.913 \pm 0.02 $ \\
& Reward & \pmb{$ 23.77 \pm 1.03 $} && $ 23.24 \pm 0.43 $ && $ 4.725 \pm 0.97 $ && $ 22.79 \pm 0.46 $ \\

    \bottomrule
\end{tabular}
\end{table}

%
%
\begin{table}
\small 
\caption{Mean and 95\% confidence interval for Medium TrafficJunction across all baselines}
\label{tab:full_results:tj_medium}
\centering
\begin{tabular}{@{}rrrrrrrrrrrrrr@{}}
    \toprule
    Baseline & Metric & \textbf{Baseline} && \textbf{Unique IDs} && \textbf{0.75 RNI} && \textbf{0.25 RNI} \\
    \midrule
    
CommNet & Success & $ 0.761 \pm 0.31 $ && \pmb{$ 0.793 \pm 0.33 $} && $ 0.046 \pm 0 $ && $ 0.614 \pm 0.11 $ \\
& Reward & $ -50.4 \pm 36 $ && $ -68.9 \pm 73.8 $ && $ -168 \pm 7.13 $ && \pmb{$ -48.5 \pm 7.87 $} \\
DGN & Success & \pmb{$ 1 \pm 0 $} && \pmb{$ 1 \pm 0 $} && $ 0.062 \pm 0 $ && $ 0.619 \pm 0.4 $ \\
& Reward & \pmb{$ -62.7 \pm 0.09 $} && \pmb{$ -62.7 \pm 0.1 $} && $ -245 \pm 2.61 $ && $ -138 \pm 80.7 $ \\
IC3Net & Success & \pmb{$ 0.971 \pm 0.04 $} && $ 0.804 \pm 0.1 $ && $ 0.588 \pm 0.03 $ && $ 0.855 \pm 0.13 $ \\
& Reward & \pmb{$ -22.6 \pm 1.31 $} && $ -28.1 \pm 3.3 $ && $ -42.2 \pm 2.54 $ && $ -27.1 \pm 4.96 $ \\
MAGIC & Success & $ 0.551 \pm 0.28 $ && $ 0.526 \pm 0.33 $ && \pmb{$ 0.734 \pm 0.21 $} && $ 0.4 \pm 0.35 $ \\
& Reward & $ -132 \pm 61.1 $ && \pmb{$ -112 \pm 59.7 $} && $ -173 \pm 55.1 $ && $ -198 \pm 60.1 $ \\
TarMAC & Success & \pmb{$ 0.064 \pm 0 $} && $ 0.052 \pm 0 $ && $ 0.05 \pm 0 $ && $ 0.054 \pm 0.01 $ \\
& Reward & $ -187 \pm 24.4 $ && \pmb{$ -182 \pm 28.7 $} && $ -245 \pm 2.79 $ && $ -211 \pm 20.7 $ \\
T-IC3Net & Success & $ 0.89 \pm 0.17 $ && $ 0.909 \pm 0.08 $ && $ 0.362 \pm 0.18 $ && \pmb{$ 0.962 \pm 0.02 $} \\
& Reward & $ -26.4 \pm 7.2 $ && $ -24.9 \pm 2.48 $ && $ -94.2 \pm 65.8 $ && \pmb{$ -23.7 \pm 1.1 $} \\

    \bottomrule
\end{tabular}
\end{table}

%
%
\begin{table}
\small 
\caption{Mean and 95\% confidence interval for BoxPushing across all baselines}
\label{tab:full_results:bp}
\centering
\begin{tabular}{@{}rrrrrrrrrrrrrr@{}}
    \toprule
    Baseline & Metric & \textbf{Baseline} && \textbf{Unique IDs} && \textbf{0.75 RNI} && \textbf{0.25 RNI} \\
    \midrule
    
CommNet & Ratio Cleared & $ 0.786 \pm 0.08 $ && \pmb{$ 0.829 \pm 0.08 $} && $ 0.768 \pm 0.08 $ && $ 0.795 \pm 0.09 $ \\
& Reward & $ 3777 \pm 728 $ && \pmb{$ 4439 \pm 687 $} && $ 4313 \pm 514 $ && $ 4196 \pm 615 $ \\
DGN & Ratio Cleared & $ 0.603 \pm 0 $ && $ 0.756 \pm 0.06 $ && \pmb{$ 0.958 \pm 0 $} && $ 0.957 \pm 0.01 $ \\
& Reward & $ 3811 \pm 51.5 $ && $ 4127 \pm 278 $ && \pmb{$ 5536 \pm 57 $} && $ 5469 \pm 90.2 $ \\
IC3Net & Ratio Cleared & $ 0.49 \pm 0.15 $ && $ 0.617 \pm 0.14 $ && $ 0.34 \pm 0.18 $ && \pmb{$ 0.676 \pm 0.06 $} \\
& Reward & $ 2528 \pm 950 $ && $ 2990 \pm 864 $ && $ 1341 \pm 711 $ && \pmb{$ 3306 \pm 556 $} \\
MAGIC & Ratio Cleared & $ 0.958 \pm 0.04 $ && $ 0.985 \pm 0.01 $ && $ 0.975 \pm 0.04 $ && \pmb{$ 0.998 \pm 0 $} \\
& Reward & $ 5199 \pm 221 $ && $ 5322 \pm 214 $ && $ 5444 \pm 332 $ && \pmb{$ 5464 \pm 118 $} \\
TarMAC & Ratio Cleared & $ 0.629 \pm 0.14 $ && $ 0.578 \pm 0.11 $ && $ 0.662 \pm 0.06 $ && \pmb{$ 0.679 \pm 0.06 $} \\
& Reward & $ 3425 \pm 820 $ && $ 2961 \pm 729 $ && $ 3343 \pm 555 $ && \pmb{$ 3610 \pm 475 $} \\
T-IC3Net & Ratio Cleared & $ 0.558 \pm 0.13 $ && $ 0.596 \pm 0.11 $ && $ 0.458 \pm 0.18 $ && \pmb{$ 0.643 \pm 0.15 $} \\
& Reward & $ 2979 \pm 753 $ && \pmb{$ 3062 \pm 785 $} && $ 1917 \pm 763 $ && $ 2908 \pm 847 $ \\

    \bottomrule
\end{tabular}
\end{table}

%
%
\begin{table}
\small 
\caption{Mean and 95\% confidence interval for DroneScatter across all baselines except DGN, including a purely random agent. Stochastic evaluation}
\label{tab:full_results:ds_stochastic}
\centering
\begin{tabular}{@{}rrrrrrrrrrrr@{}}
    \toprule
    Baseline & Metric & \textbf{Baseline} && \textbf{Unique IDs} && \textbf{0.75 RNI} \\
    \midrule
    
CommNet & Pairwise Distance & $ 11.34 \pm 0.9 $ && \pmb{$ 12.08 \pm 1.12 $} && $ 8.687 \pm 1.4 $ \\
& Steps Taken & $ 11.5 \pm 0.26 $ && \pmb{$ 9.767 \pm 0.32 $} && $ 11.74 \pm 1.39 $ \\
& Reward & $ 269 \pm 18.6 $ && \pmb{$ 314.2 \pm 22 $} && $ 236 \pm 45.3 $ \\
IC3Net & Pairwise Distance & $ 9.108 \pm 1.45 $ && \pmb{$ 13.3 \pm 0.71 $} && $ 10.99 \pm 0.38 $ \\
& Steps Taken & $ 11.94 \pm 0.84 $ && \pmb{$ 10.13 \pm 0.25 $} && $ 11.66 \pm 0.22 $ \\
& Reward & $ 239.2 \pm 33.7 $ && \pmb{$ 316.7 \pm 11.3 $} && $ 260.9 \pm 8.62 $ \\
MAGIC & Pairwise Distance & $ 7.693 \pm 1.47 $ && \pmb{$ 12.59 \pm 1 $} && $ 7.216 \pm 0.76 $ \\
& Steps Taken & $ 13.05 \pm 0.58 $ && \pmb{$ 11.12 \pm 1.05 $} && $ 13.54 \pm 0.21 $ \\
& Reward & $ 205.6 \pm 21.4 $ && \pmb{$ 273.9 \pm 36.7 $} && $ 180 \pm 17.3 $ \\
TarMAC & Pairwise Distance & $ 7.448 \pm 0.89 $ && \pmb{$ 10.26 \pm 0.69 $} && $ 8.486 \pm 0.36 $ \\
& Steps Taken & $ 13.49 \pm 0.1 $ && \pmb{$ 10.7 \pm 0.35 $} && $ 12.85 \pm 0.57 $ \\
& Reward & $ 171.4 \pm 13.8 $ && \pmb{$ 270.4 \pm 7.81 $} && $ 200.2 \pm 14.8 $ \\
T-IC3Net & Pairwise Distance & $ 8.891 \pm 0.27 $ && \pmb{$ 12.9 \pm 0.78 $} && $ 9.552 \pm 0.57 $ \\
& Steps Taken & $ 12.28 \pm 0.6 $ && \pmb{$ 10.33 \pm 0.46 $} && $ 12.22 \pm 0.82 $ \\
& Reward & $ 219 \pm 37.2 $ && \pmb{$ 309 \pm 12 $} && $ 224.9 \pm 22.5 $ \\

    Random & Pairwise Distance & $ 5.8 \pm 0.02 $ && -- && -- \\
    & Steps Taken & $ 17.39 \pm 0.04 $ && -- && -- \\
    & Reward & $ 44.59 \pm 1.73 $ && -- && -- \\

    \bottomrule
\end{tabular}
\end{table}

%
%
\begin{table}
\small 
\caption{Mean and 95\% confidence interval for DroneScatter across all baselines. Greedy evaluation}
\label{tab:full_results:ds_greedy}
\centering
\begin{tabular}{@{}rrrrrrrrrrrr@{}}
    \toprule
    Baseline & Metric & \textbf{Baseline} && \textbf{Unique IDs} && \textbf{0.75 RNI} \\
    \midrule
    
CommNet & Pairwise Distance & $ 8.849 \pm 0.63 $ && \pmb{$ 13.28 \pm 1.27 $} && $ 8.589 \pm 1.35 $ \\
& Steps Taken & $ 13.79 \pm 0.12 $ && \pmb{$ 9.554 \pm 0.33 $} && $ 12.62 \pm 1.19 $ \\
& Reward & $ 170.6 \pm 6.72 $ && \pmb{$ 319.7 \pm 24.7 $} && $ 204.8 \pm 40.5 $ \\
DGN & Pairwise Distance & $ 3.221 \pm 0.18 $ && \pmb{$ 4.427 \pm 0.67 $} && $ 3.706 \pm 0.83 $ \\
& Steps Taken & $ 13.36 \pm 0.15 $ && \pmb{$ 13.27 \pm 0.21 $} && $ 13.46 \pm 0.14 $ \\
& Reward & $ 147.9 \pm 6.16 $ && \pmb{$ 154.4 \pm 6.63 $} && $ 149.3 \pm 5.38 $ \\
IC3Net & Pairwise Distance & $ 7.69 \pm 1.03 $ && \pmb{$ 14.09 \pm 0.54 $} && $ 11 \pm 0.86 $ \\
& Steps Taken & $ 13.25 \pm 0.4 $ && \pmb{$ 10.14 \pm 0.2 $} && $ 11.42 \pm 0.48 $ \\
& Reward & $ 186.6 \pm 8.87 $ && \pmb{$ 310.1 \pm 11 $} && $ 264.3 \pm 21.9 $ \\
MAGIC & Pairwise Distance & $ 6.61 \pm 1.28 $ && \pmb{$ 12.58 \pm 0.6 $} && $ 7.107 \pm 1.59 $ \\
& Steps Taken & $ 13.27 \pm 0.18 $ && \pmb{$ 11.84 \pm 0.68 $} && $ 13.61 \pm 0.25 $ \\
& Reward & $ 193.3 \pm 15.2 $ && \pmb{$ 222.7 \pm 26.6 $} && $ 161.3 \pm 23.9 $ \\
TarMAC & Pairwise Distance & $ 8.666 \pm 0.28 $ && \pmb{$ 12.09 \pm 0.73 $} && $ 8.999 \pm 0.94 $ \\
& Steps Taken & $ 13.73 \pm 0.21 $ && \pmb{$ 11.01 \pm 0.86 $} && $ 12.19 \pm 0.82 $ \\
& Reward & $ 139.1 \pm 4.43 $ && \pmb{$ 255.5 \pm 31.4 $} && $ 202.9 \pm 32.6 $ \\
T-IC3Net & Pairwise Distance & $ 7.28 \pm 0.69 $ && \pmb{$ 13.51 \pm 0.98 $} && $ 10.87 \pm 1.17 $ \\
& Steps Taken & $ 13.96 \pm 0.26 $ && \pmb{$ 10.63 \pm 0.66 $} && $ 11.73 \pm 0.54 $ \\
& Reward & $ 156.9 \pm 21.9 $ && \pmb{$ 278.6 \pm 32.1 $} && $ 240.3 \pm 25.6 $ \\

    \bottomrule
\end{tabular}
\end{table}

%
%
\clearpage
\subsection{Result Plots}
In this section, we provide the full result plots for all of our experiments. For all but the BoxPushing experiments, results are shown aggregated across 5 seeds, with a 95\% confidence interval. For the BoxPushing experiments, the hybrid imitation learning paradigm yielded unstable training for all A2C methods. Thus, to better visualize the results, for each seed, we first denote the performance at each epoch to be the maximum performance achieved so far. Then, we aggregate these runs across 10 seeds, showing the mean and a 95\% confidence interval.

%
%
\subsection{TrafficJunction-Easy}
\centering

\includegraphics[width=.49\linewidth]{results/TrafficJunction-Easy/TrafficJunction-Easy_commnet_success.pdf}
\includegraphics[width=.49\linewidth]{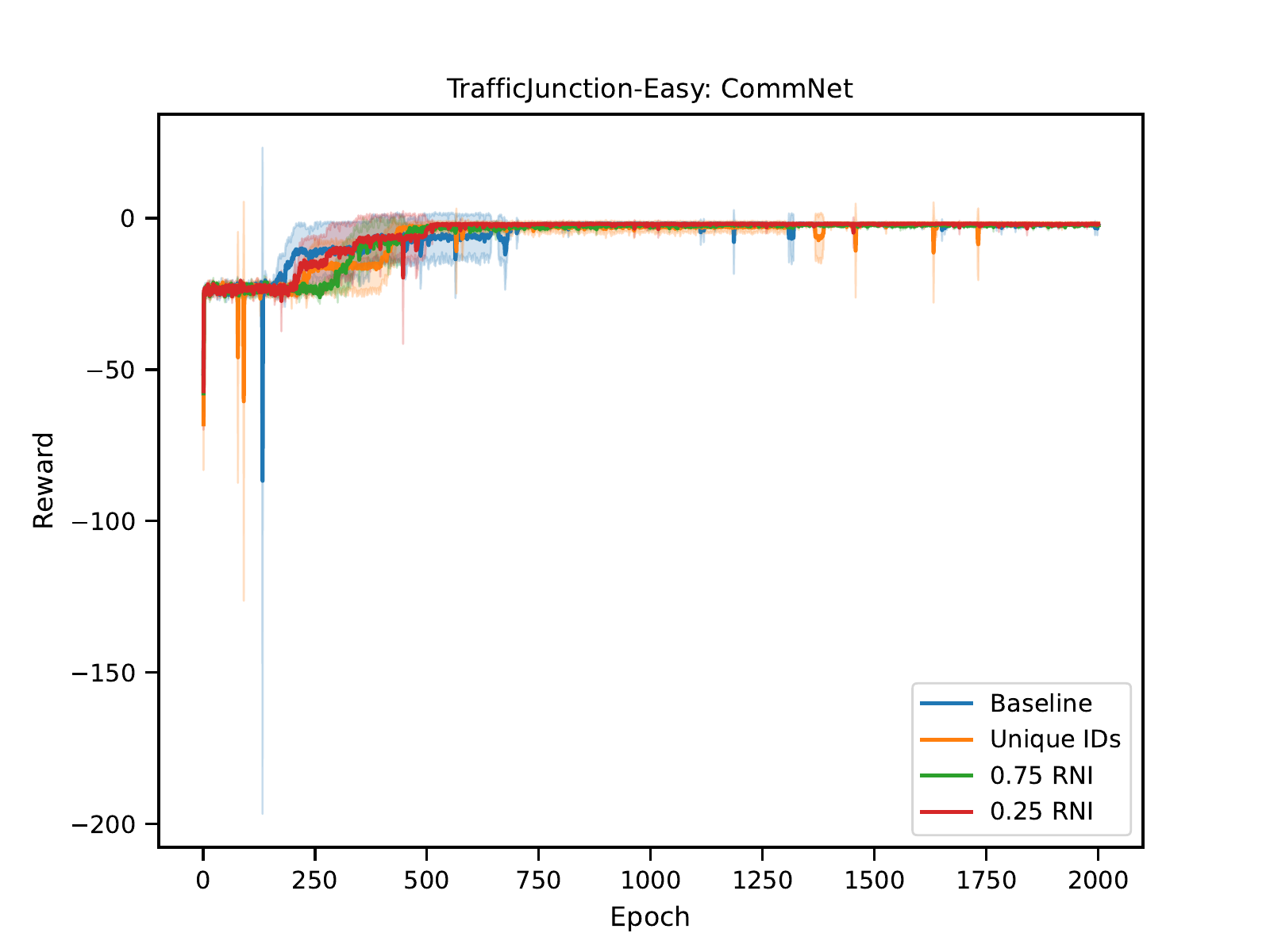}

\includegraphics[width=.49\linewidth]{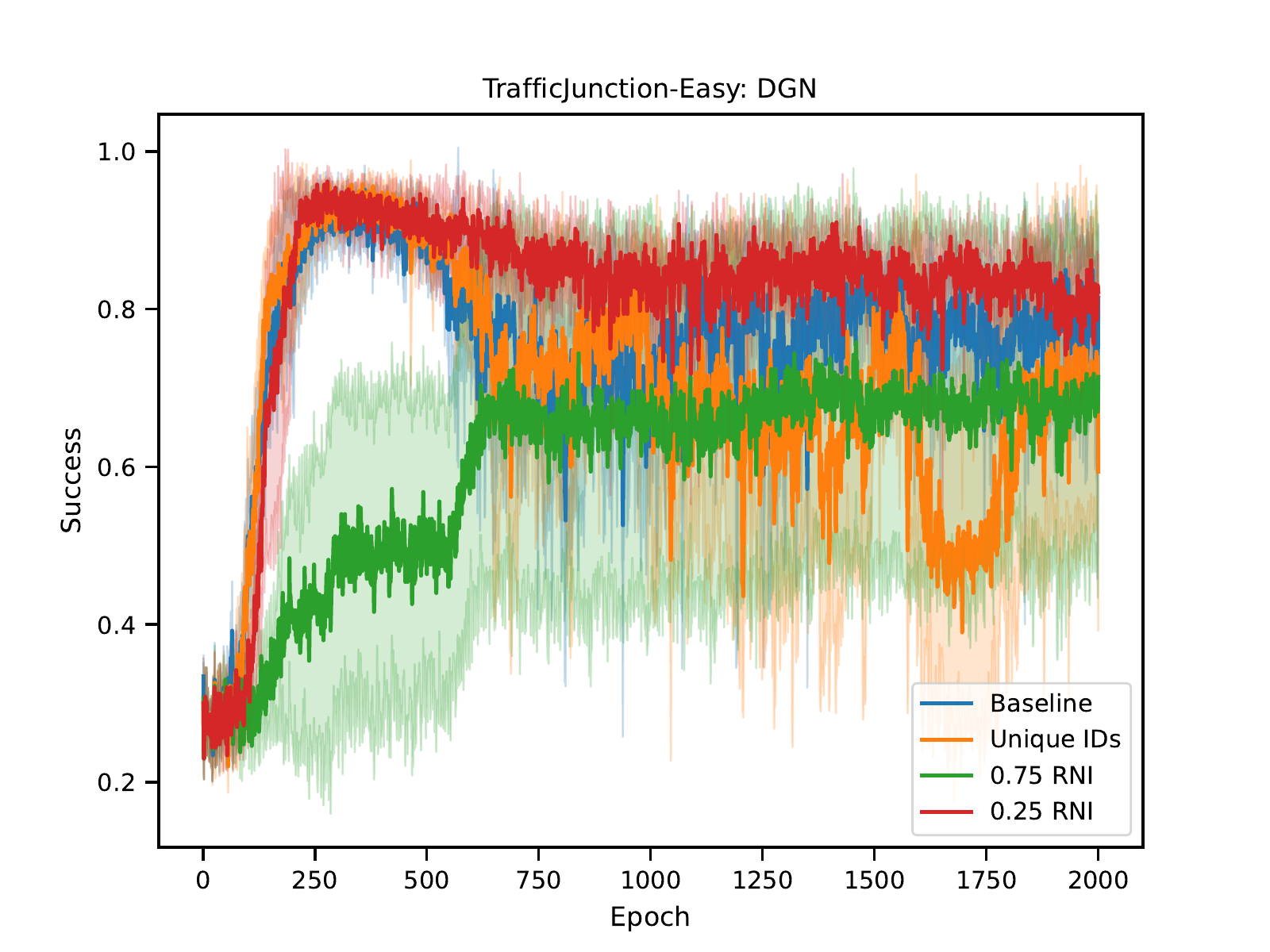}
\includegraphics[width=.49\linewidth]{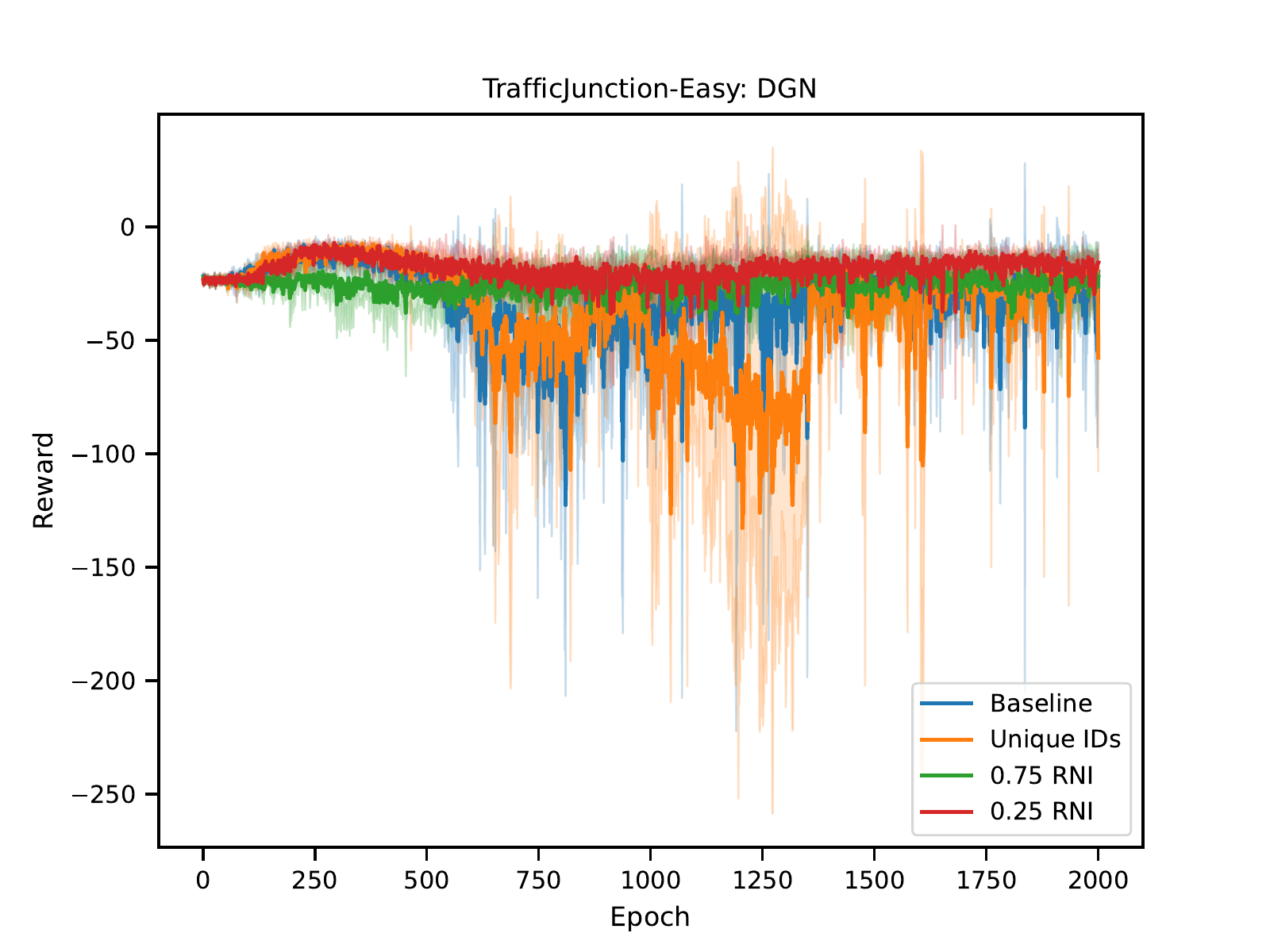}

\includegraphics[width=.49\linewidth]{results/TrafficJunction-Easy/TrafficJunction-Easy_ic3net_success.pdf}
\includegraphics[width=.49\linewidth]{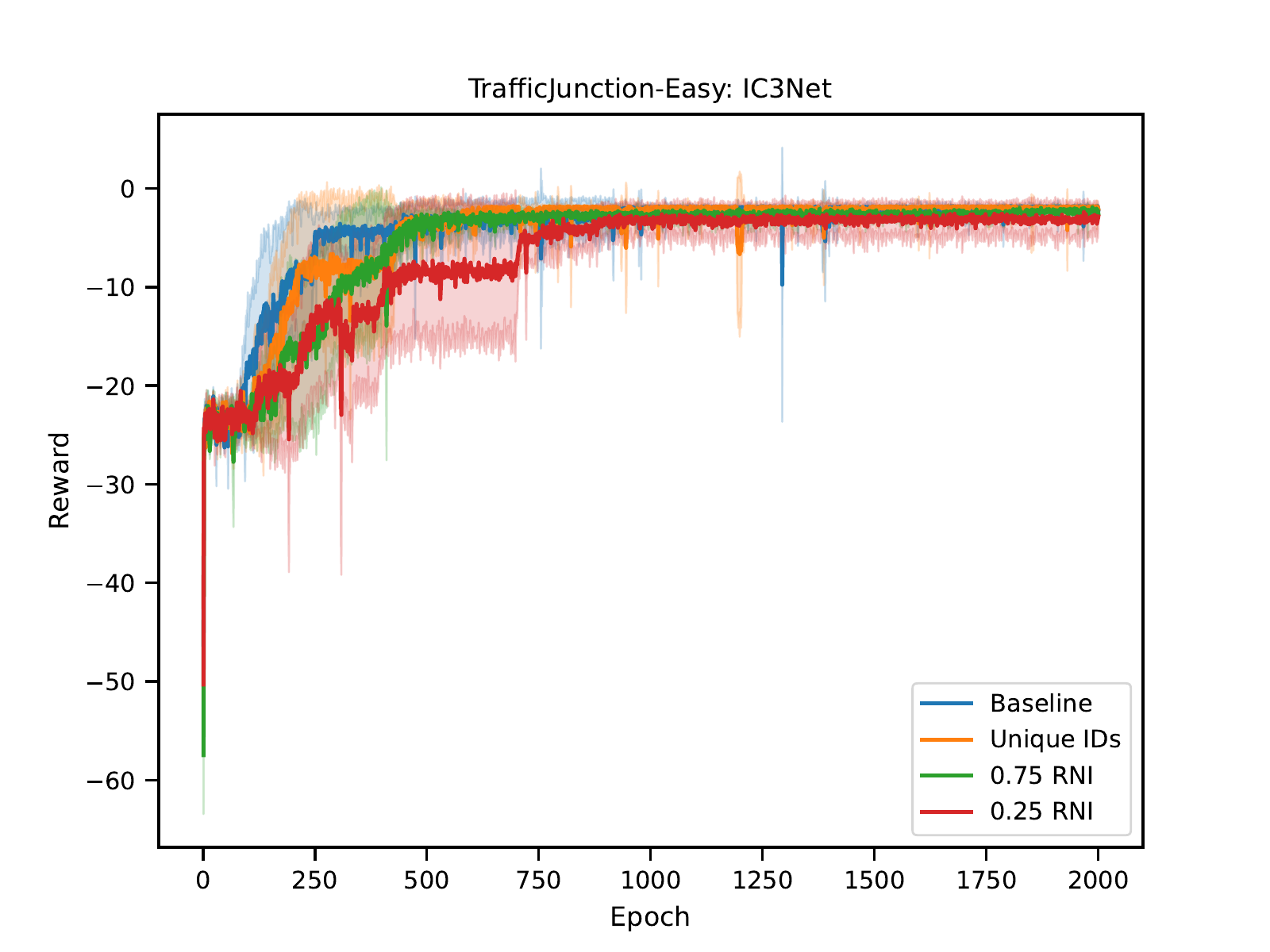}

\includegraphics[width=.49\linewidth]{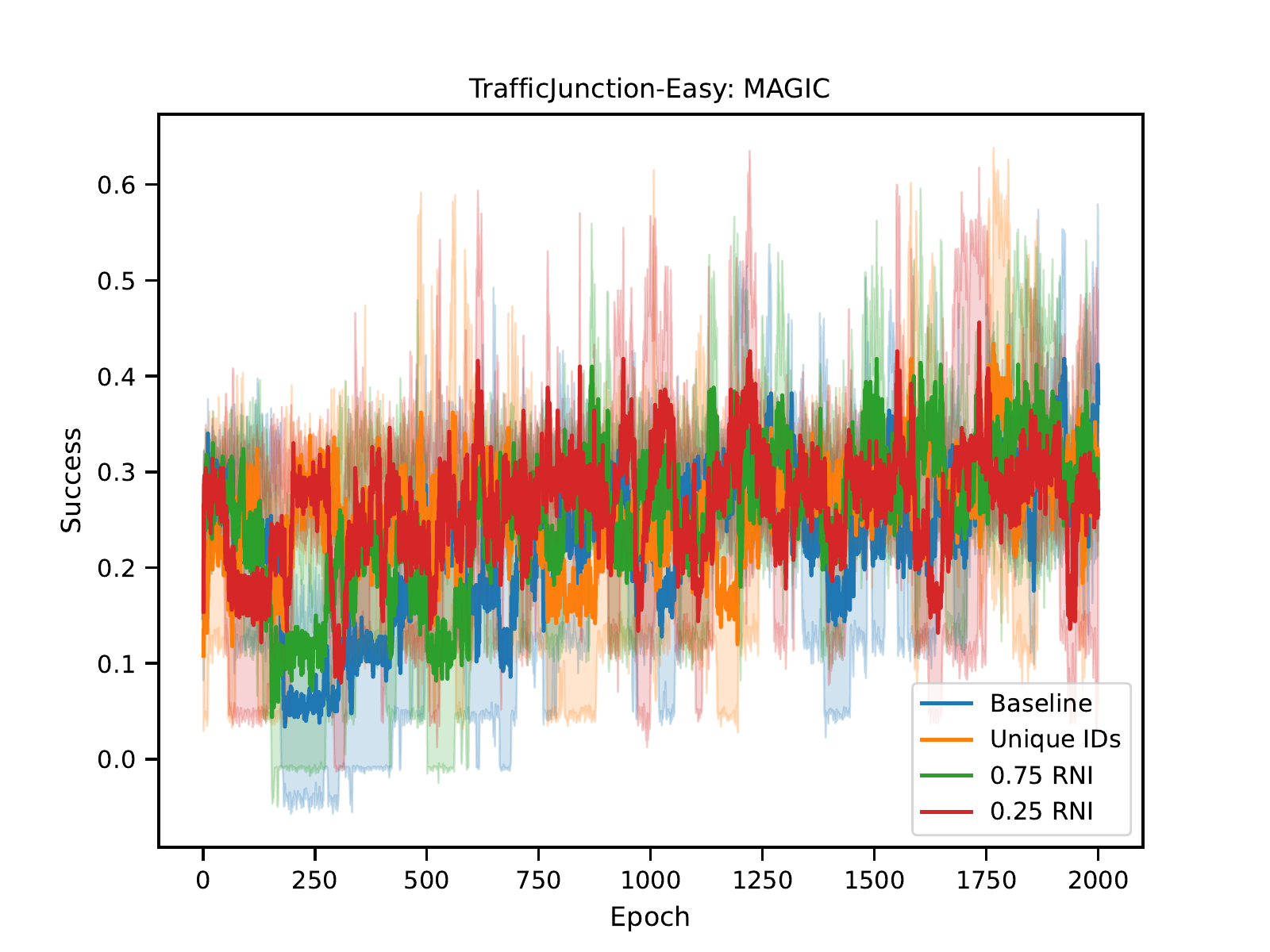}
\includegraphics[width=.49\linewidth]{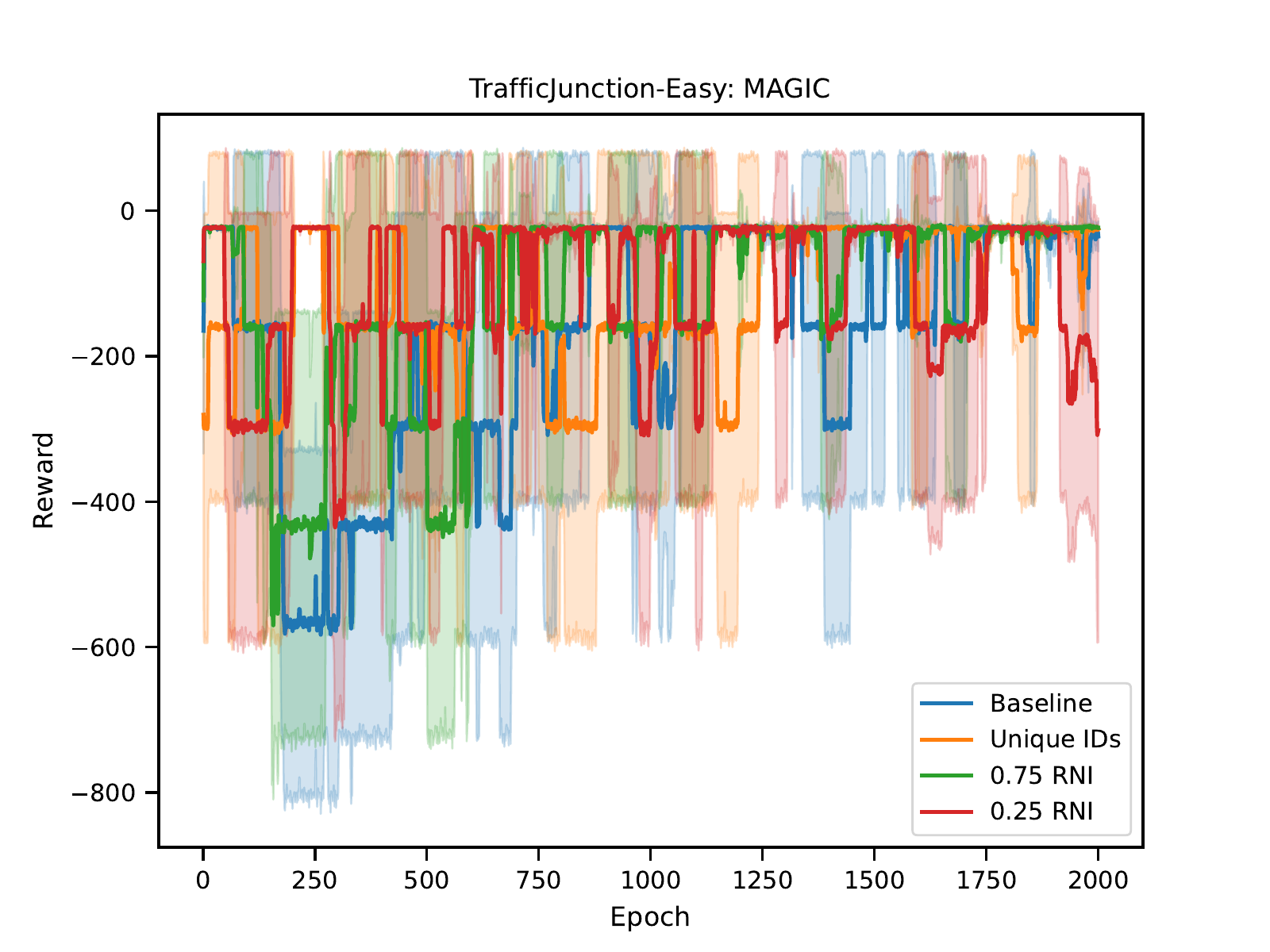}

\includegraphics[width=.49\linewidth]{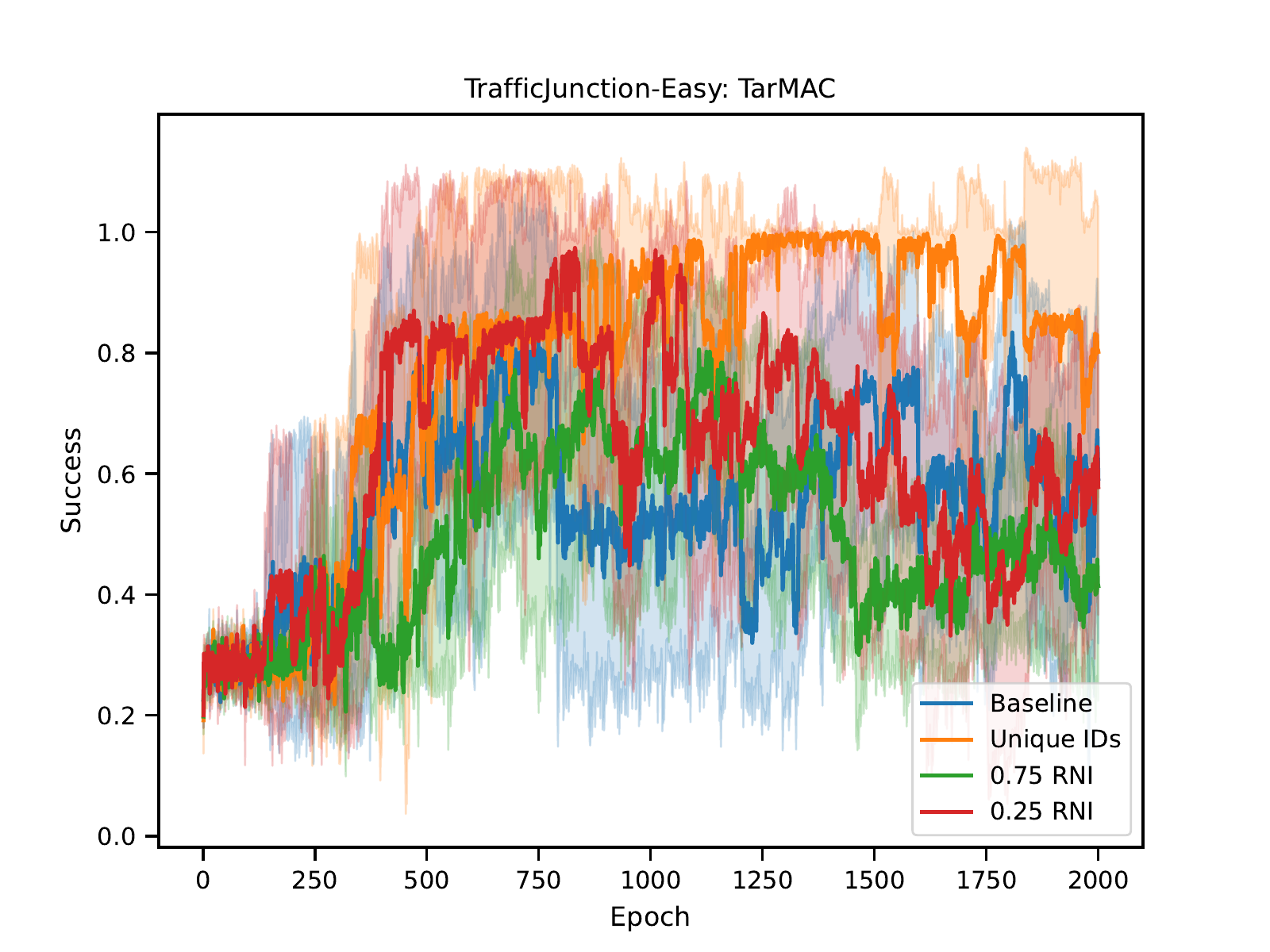}
\includegraphics[width=.49\linewidth]{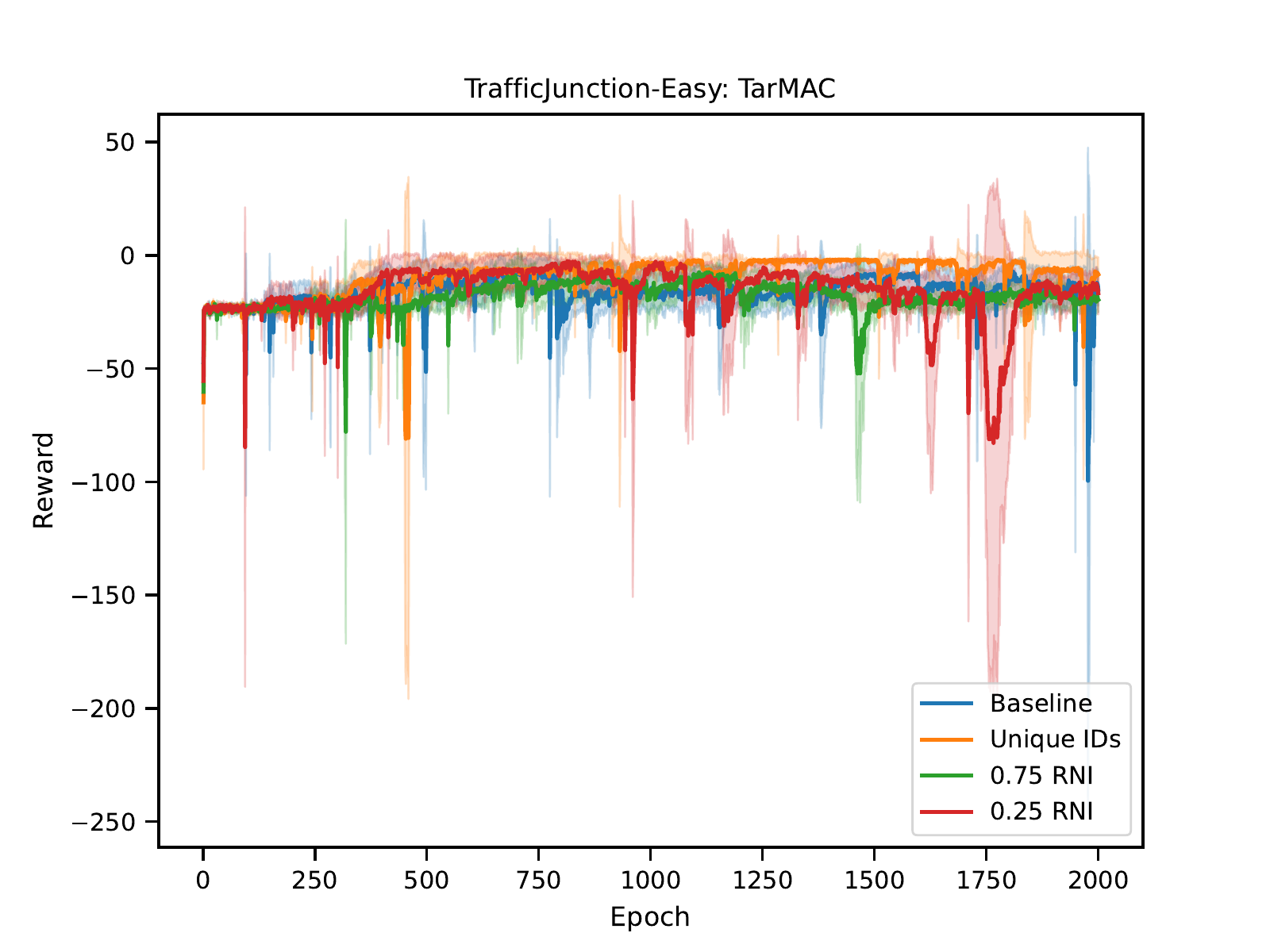}

\includegraphics[width=.49\linewidth]{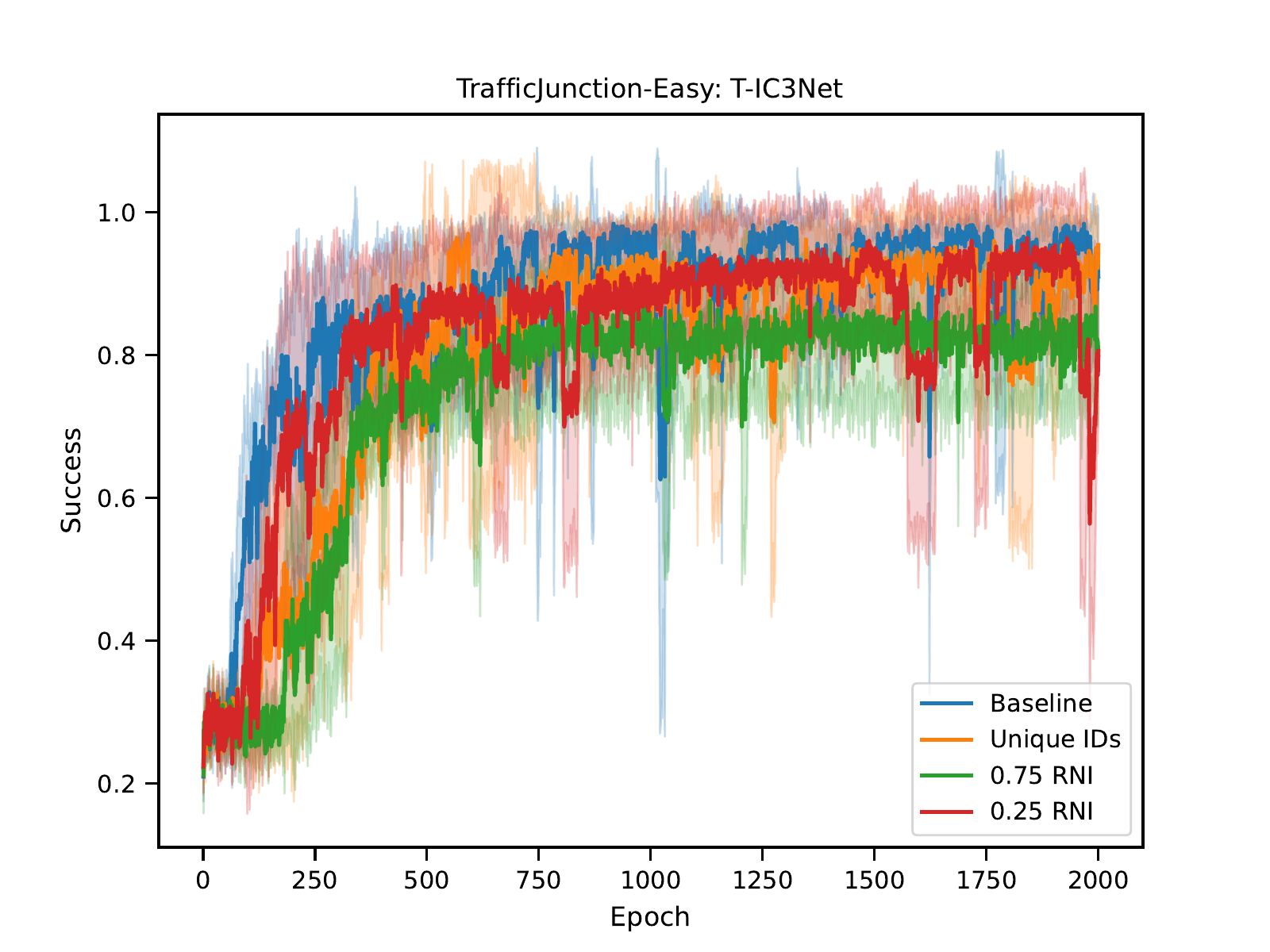}
\includegraphics[width=.49\linewidth]{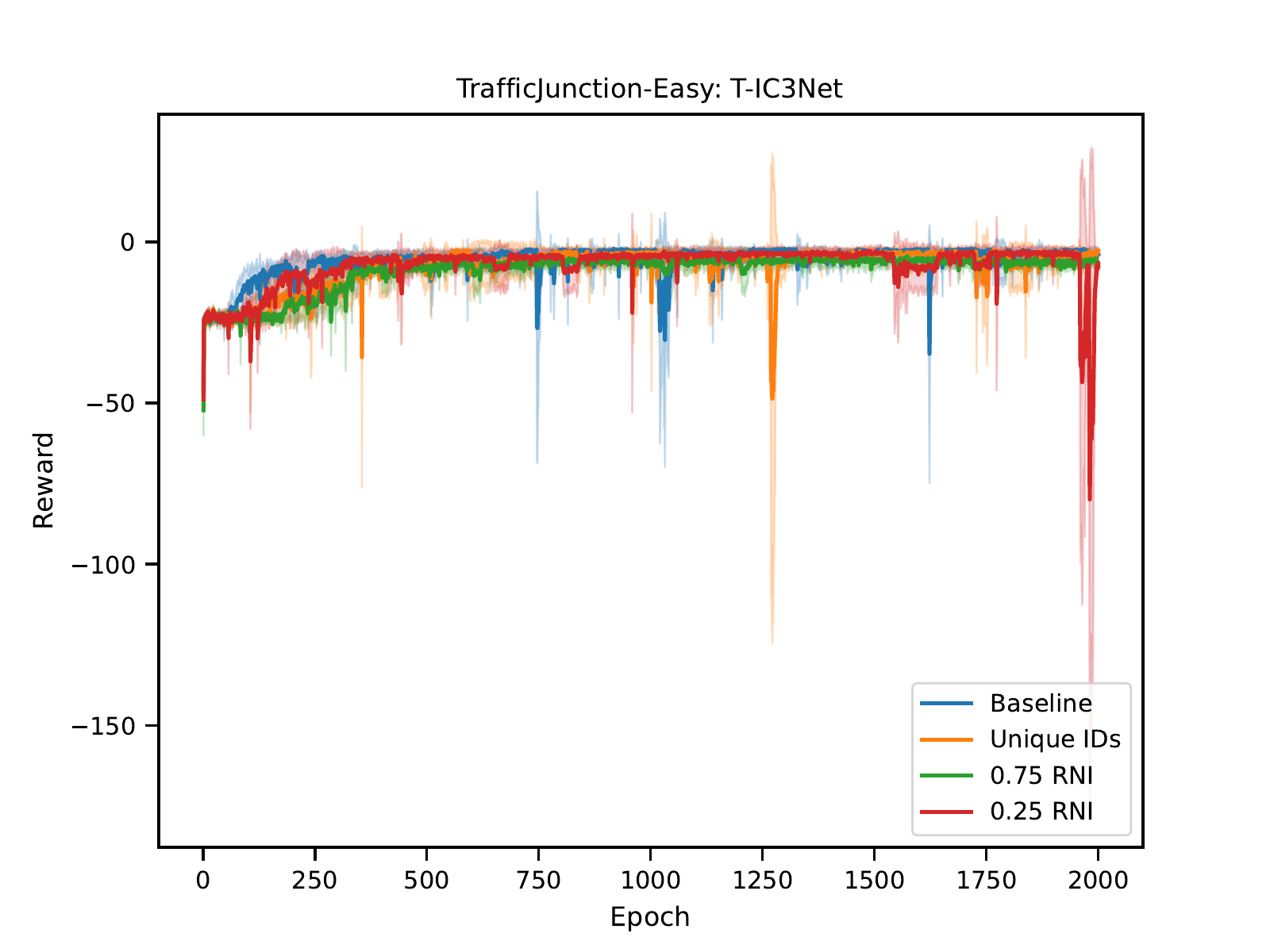}

%
%
\newpage
\subsection{PredatorPrey}
\centering

\includegraphics[width=.49\linewidth]{results/PredatorPrey/PredatorPrey_commnet_success.pdf}
\includegraphics[width=.49\linewidth]{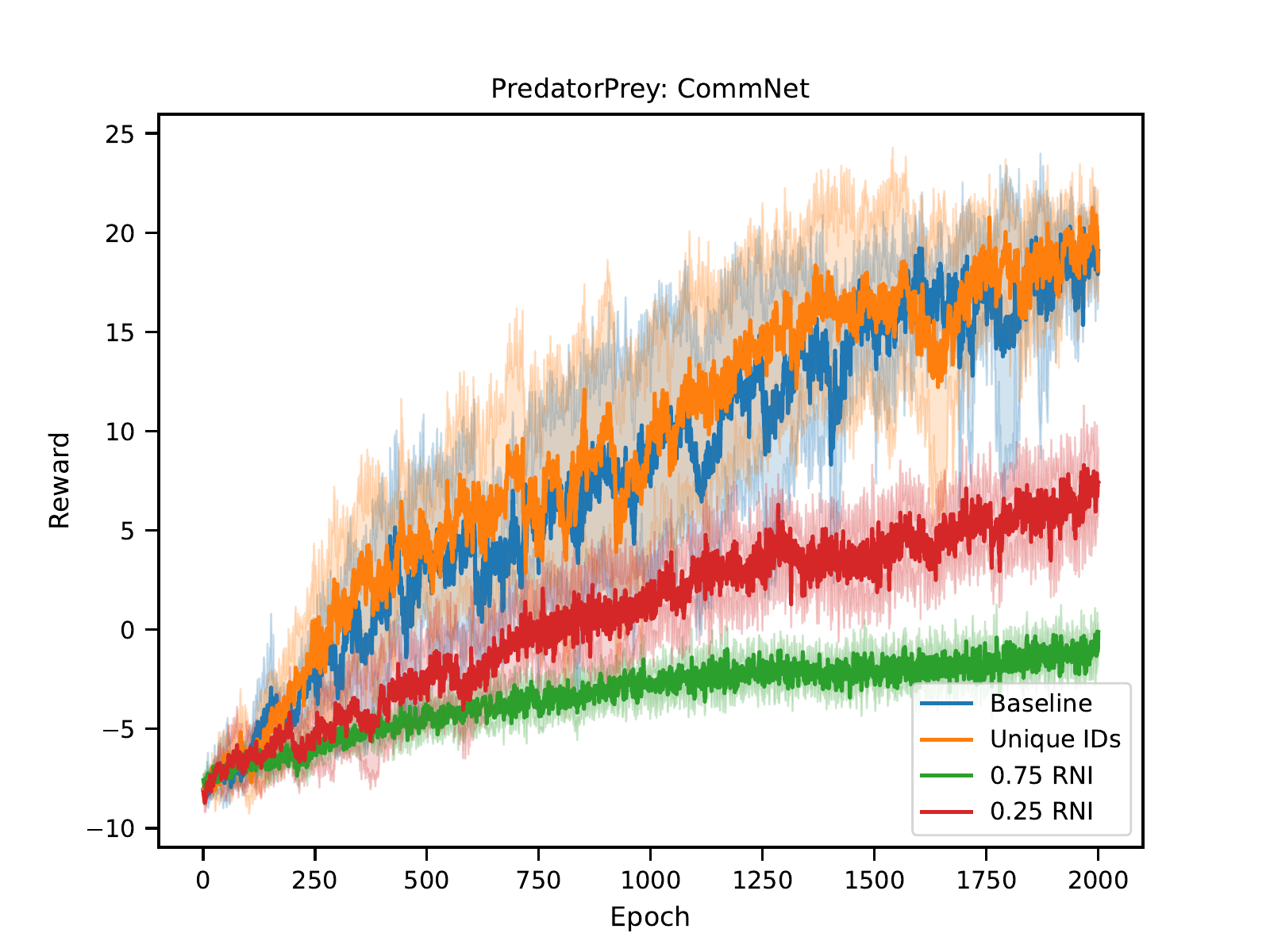}

\includegraphics[width=.49\linewidth]{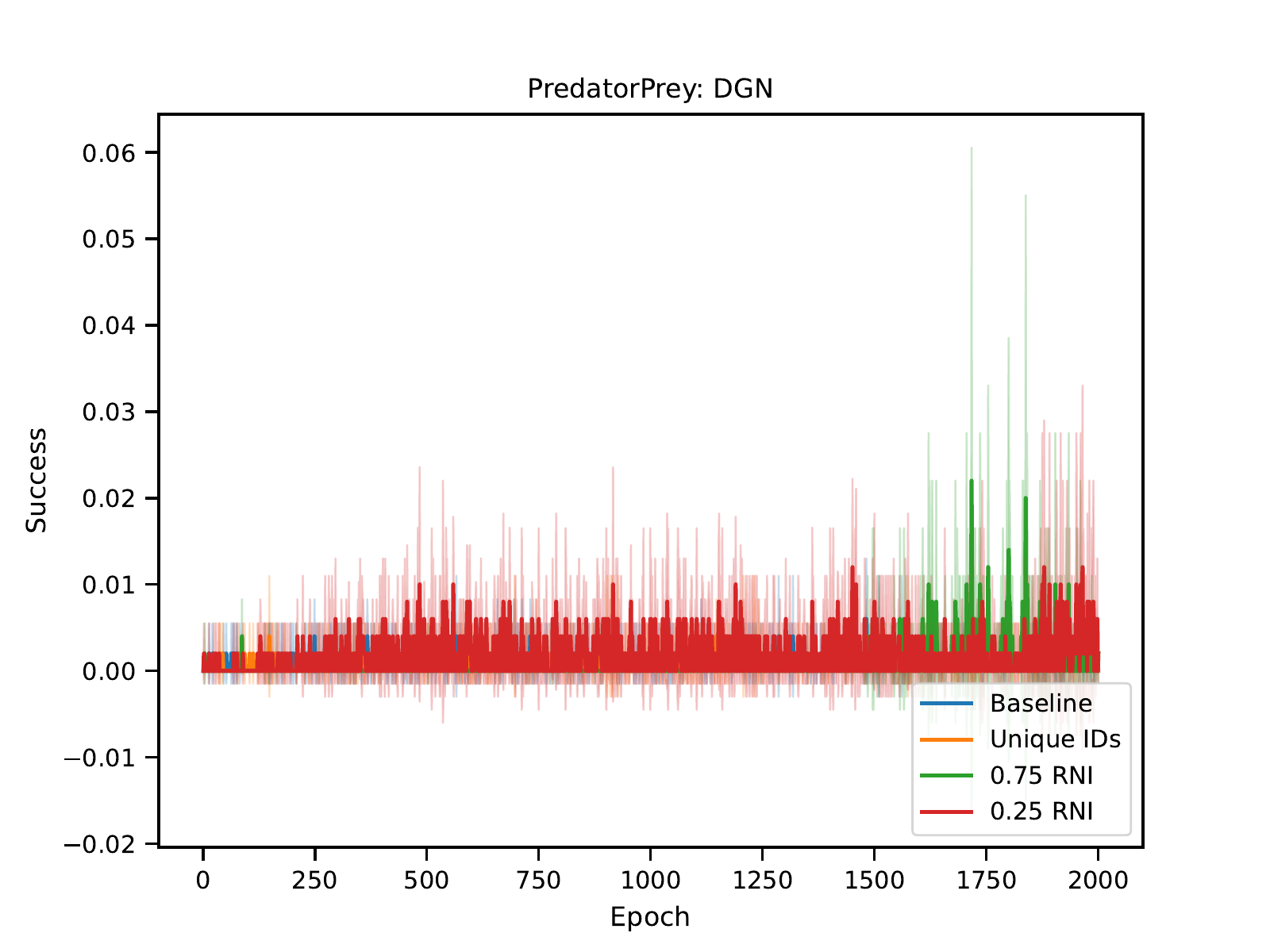}
\includegraphics[width=.49\linewidth]{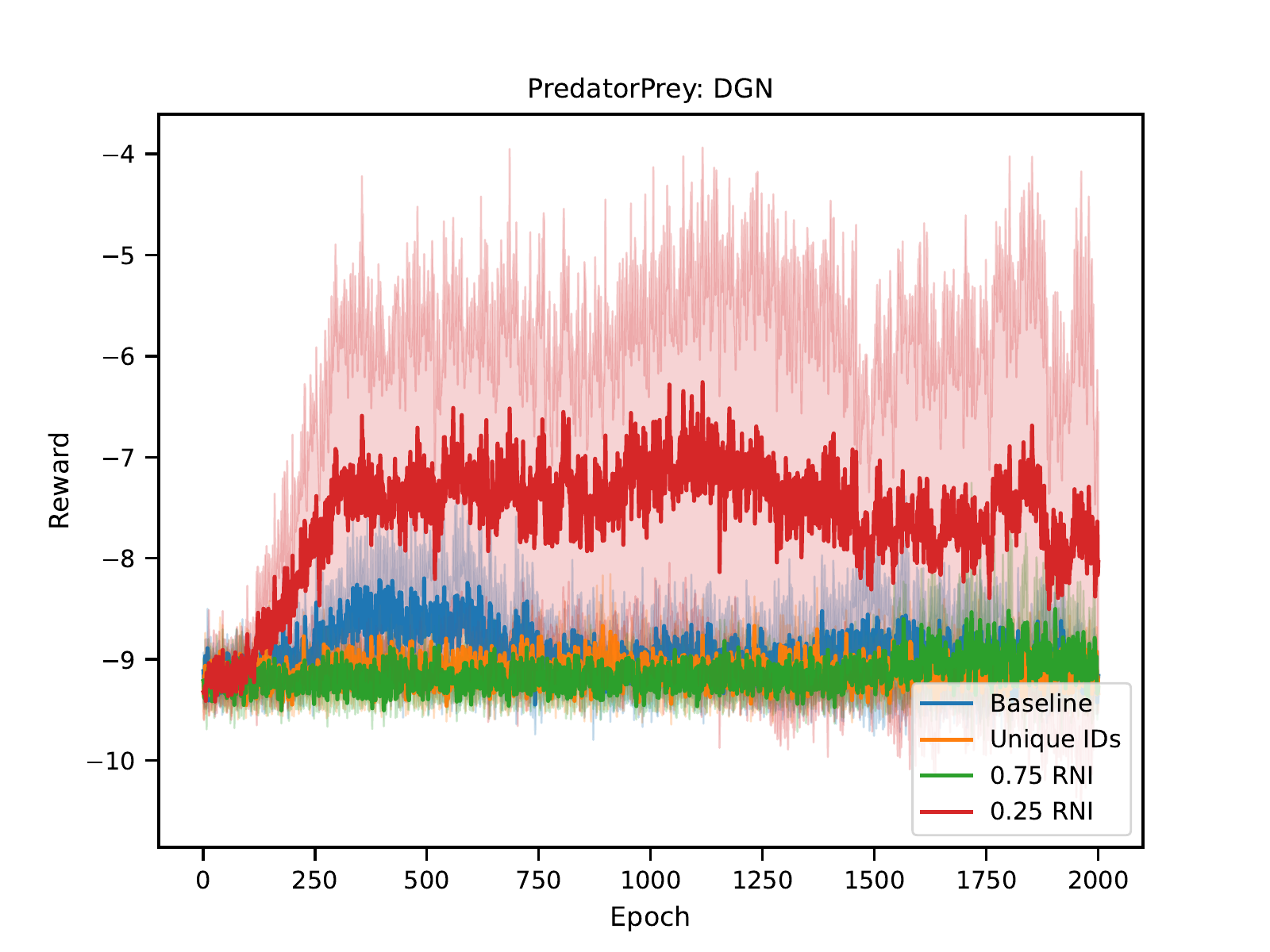}

\includegraphics[width=.49\linewidth]{results/PredatorPrey/PredatorPrey_ic3net_success.pdf}
\includegraphics[width=.49\linewidth]{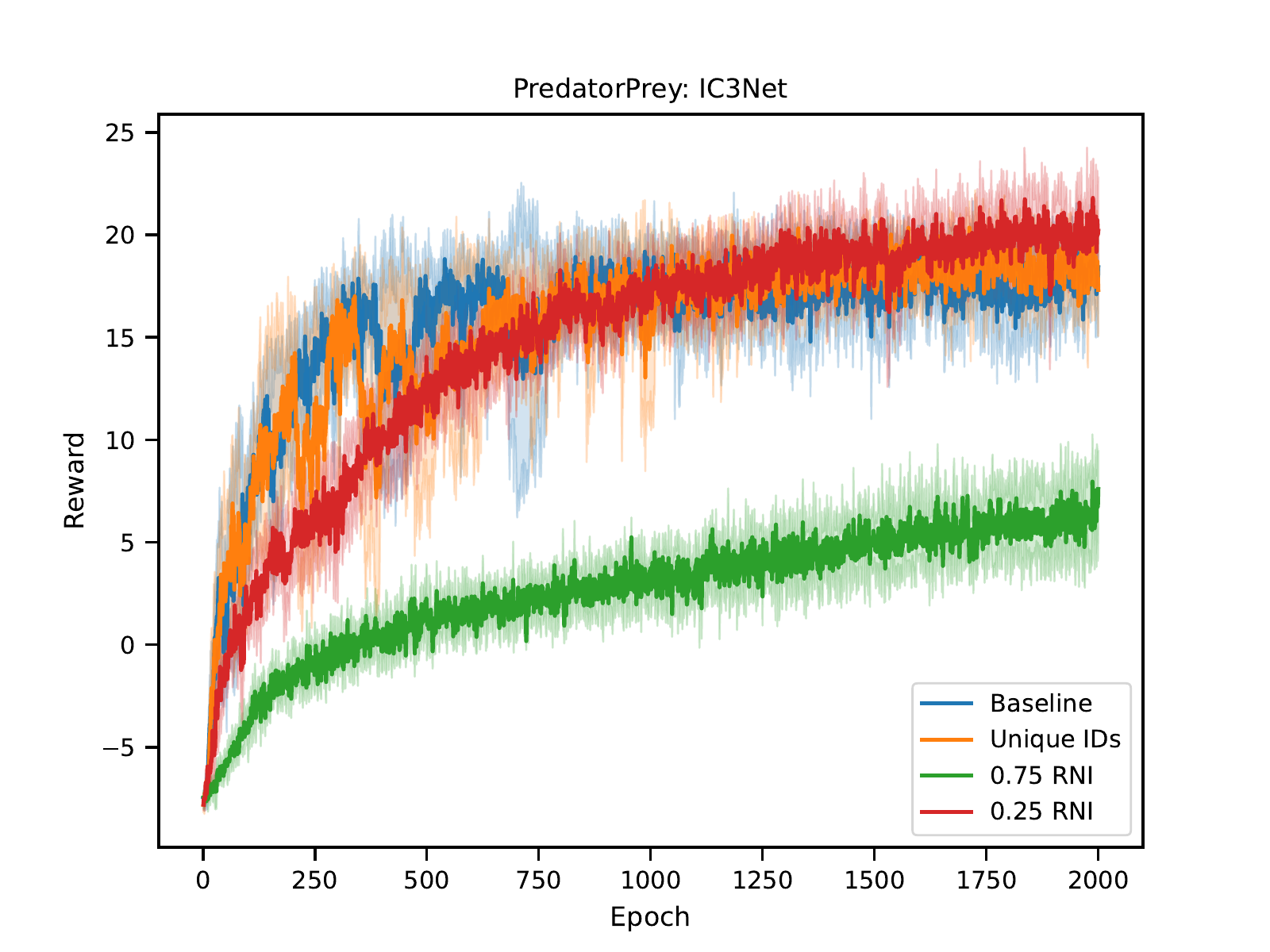}

\includegraphics[width=.49\linewidth]{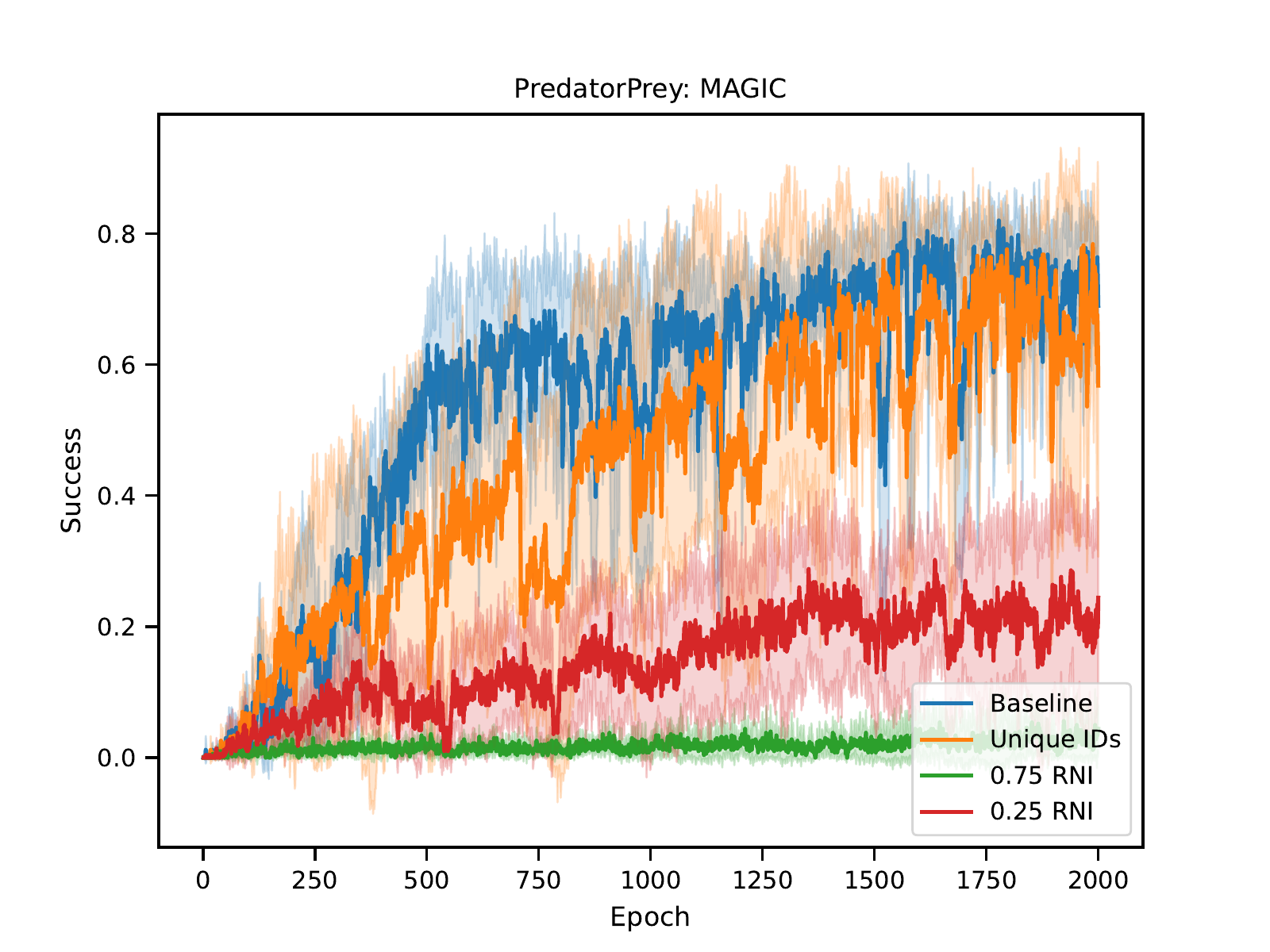}
\includegraphics[width=.49\linewidth]{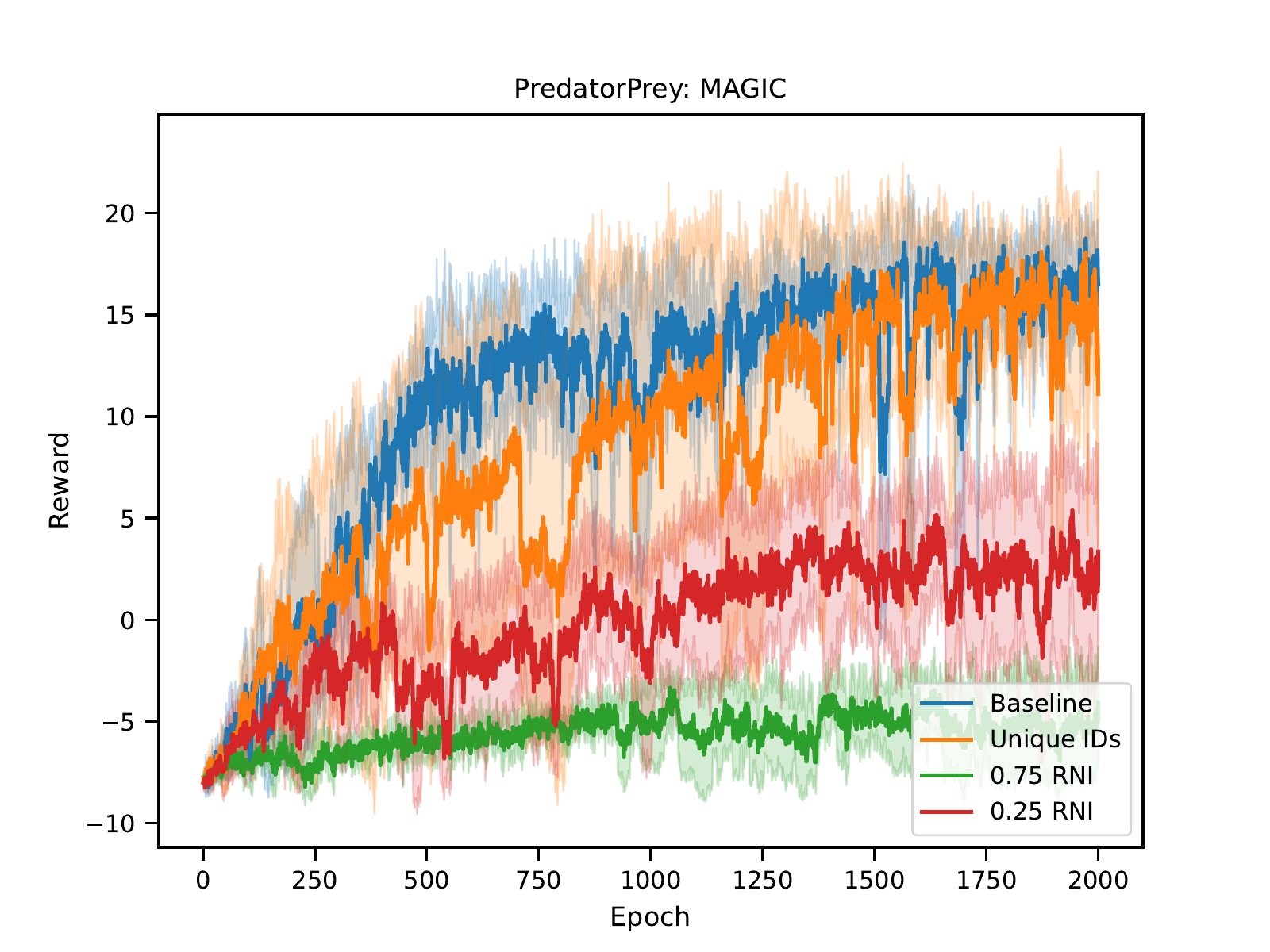}

\includegraphics[width=.49\linewidth]{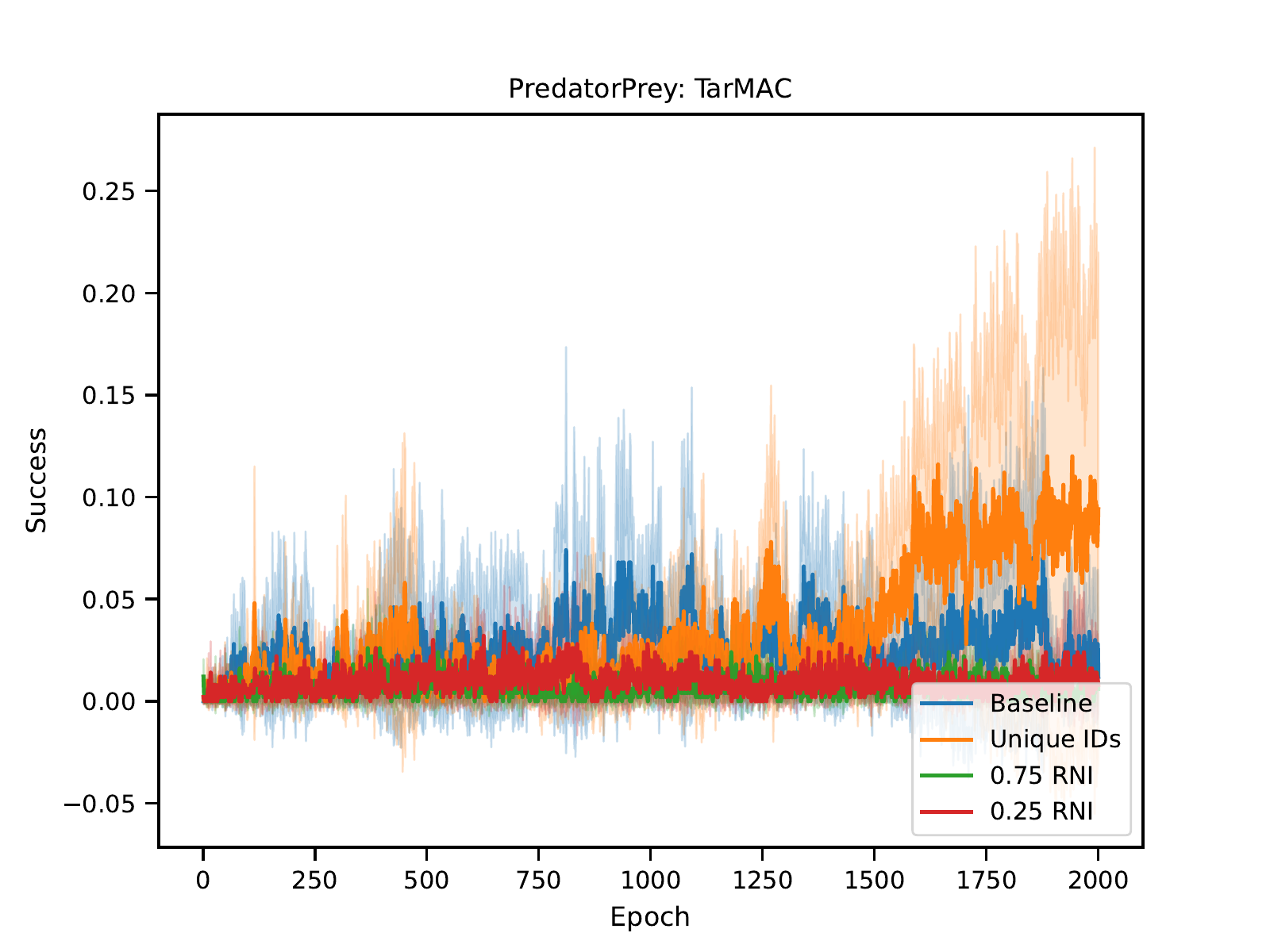}
\includegraphics[width=.49\linewidth]{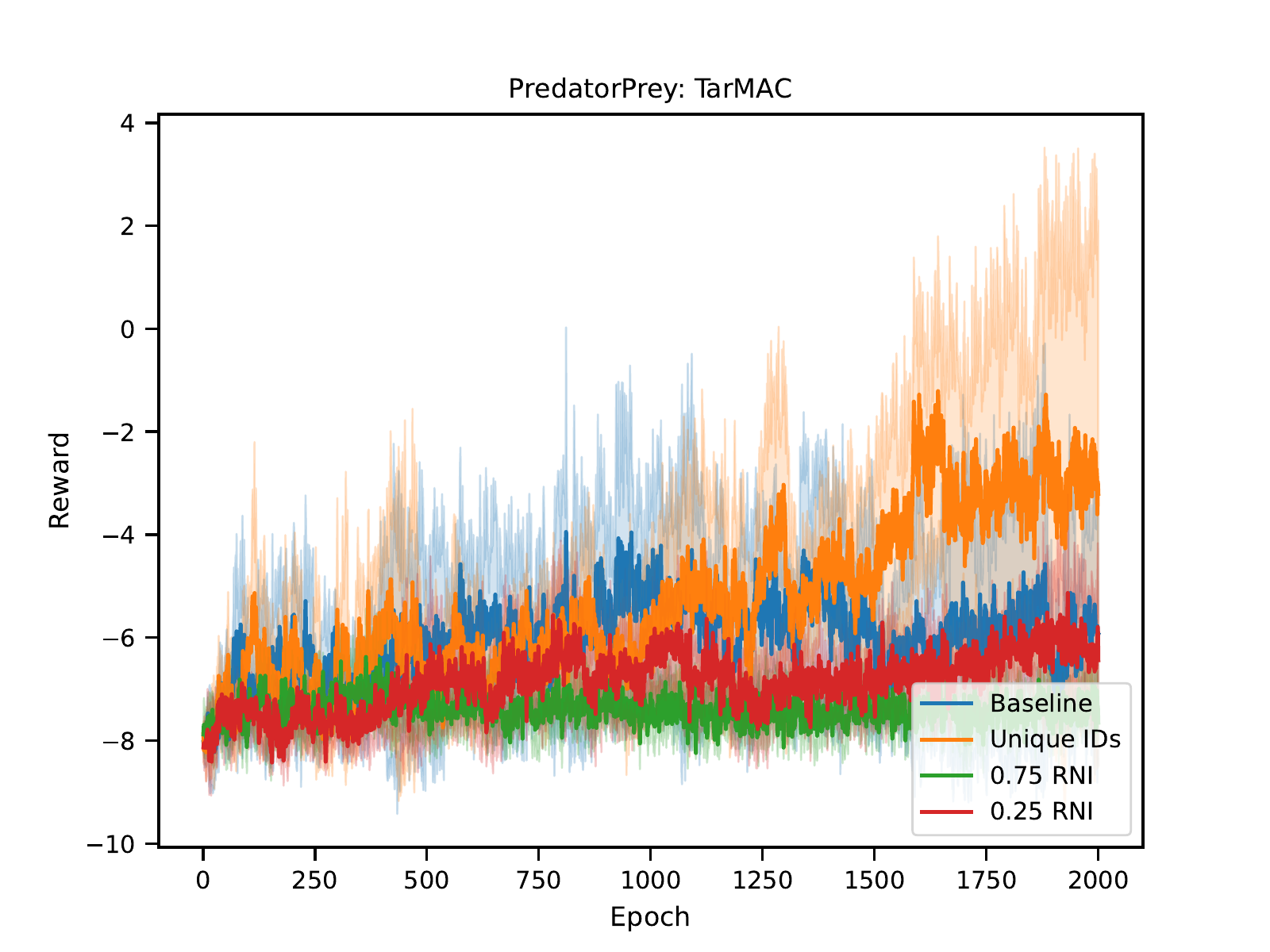}

\includegraphics[width=.49\linewidth]{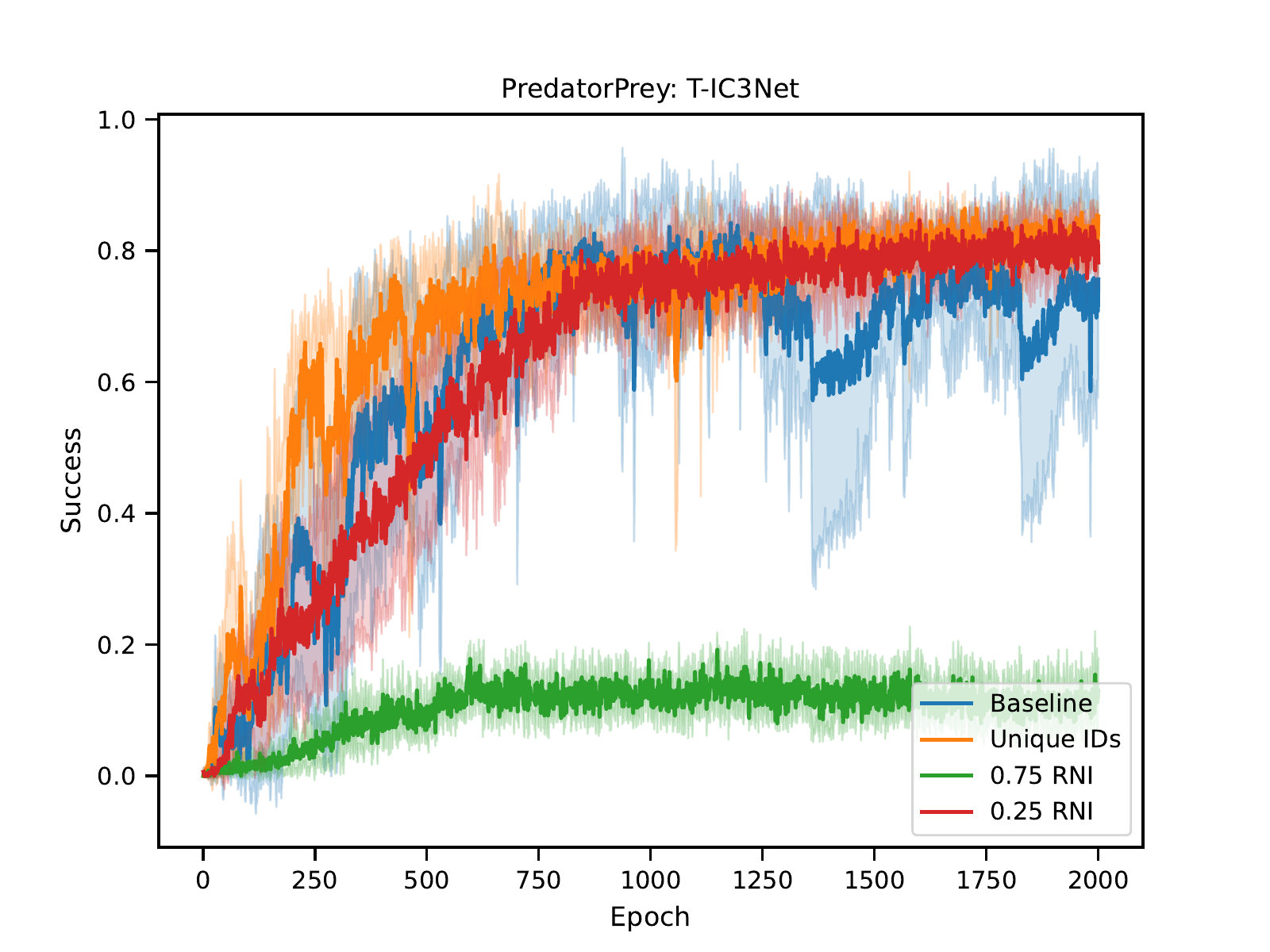}
\includegraphics[width=.49\linewidth]{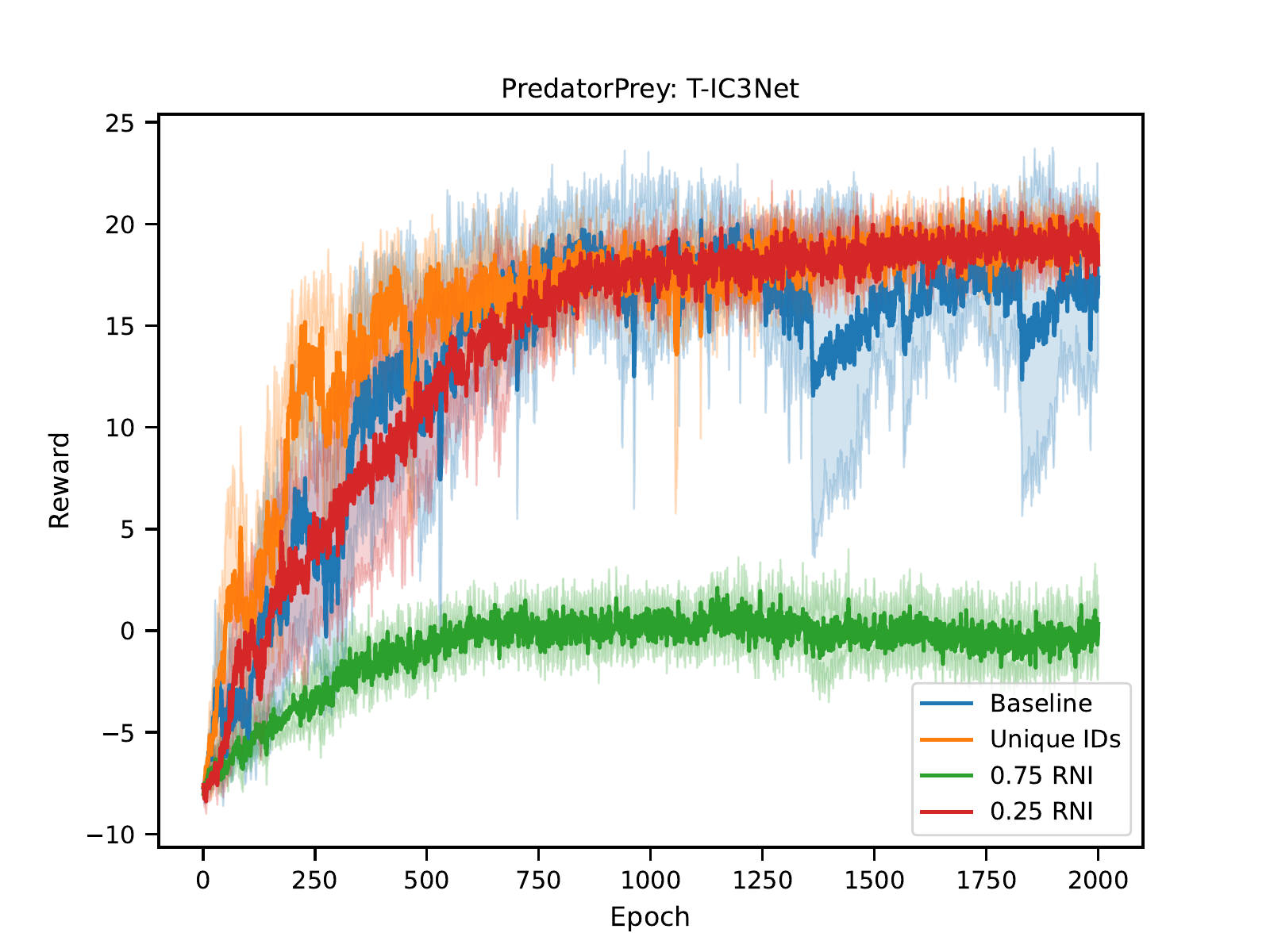}

%
%
\subsection{TrafficJunction-Medium}
\centering

\includegraphics[width=.49\linewidth]{results/TrafficJunction-Medium/TrafficJunction-Medium_commnet_success.pdf}
\includegraphics[width=.49\linewidth]{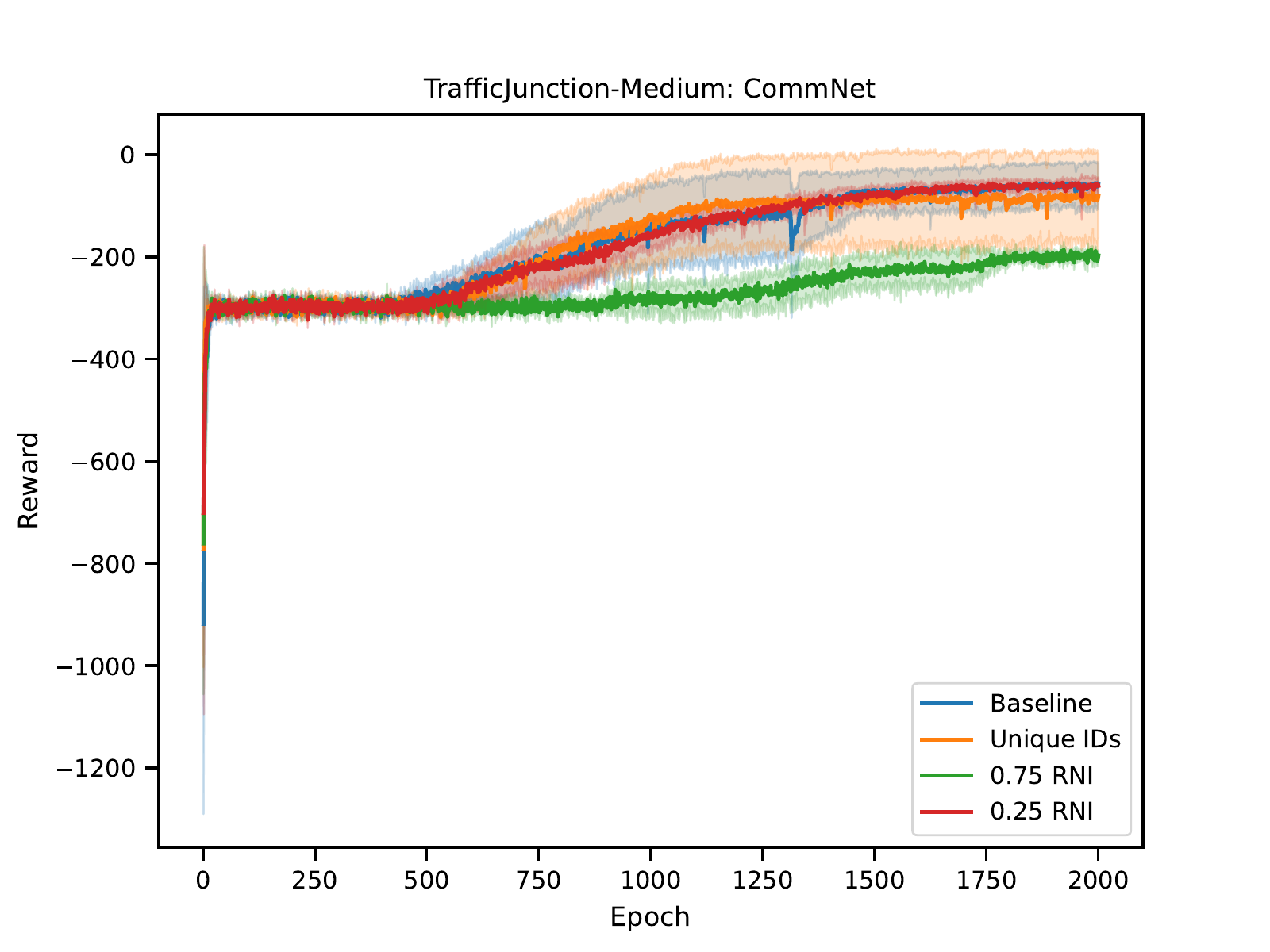}

\includegraphics[width=.49\linewidth]{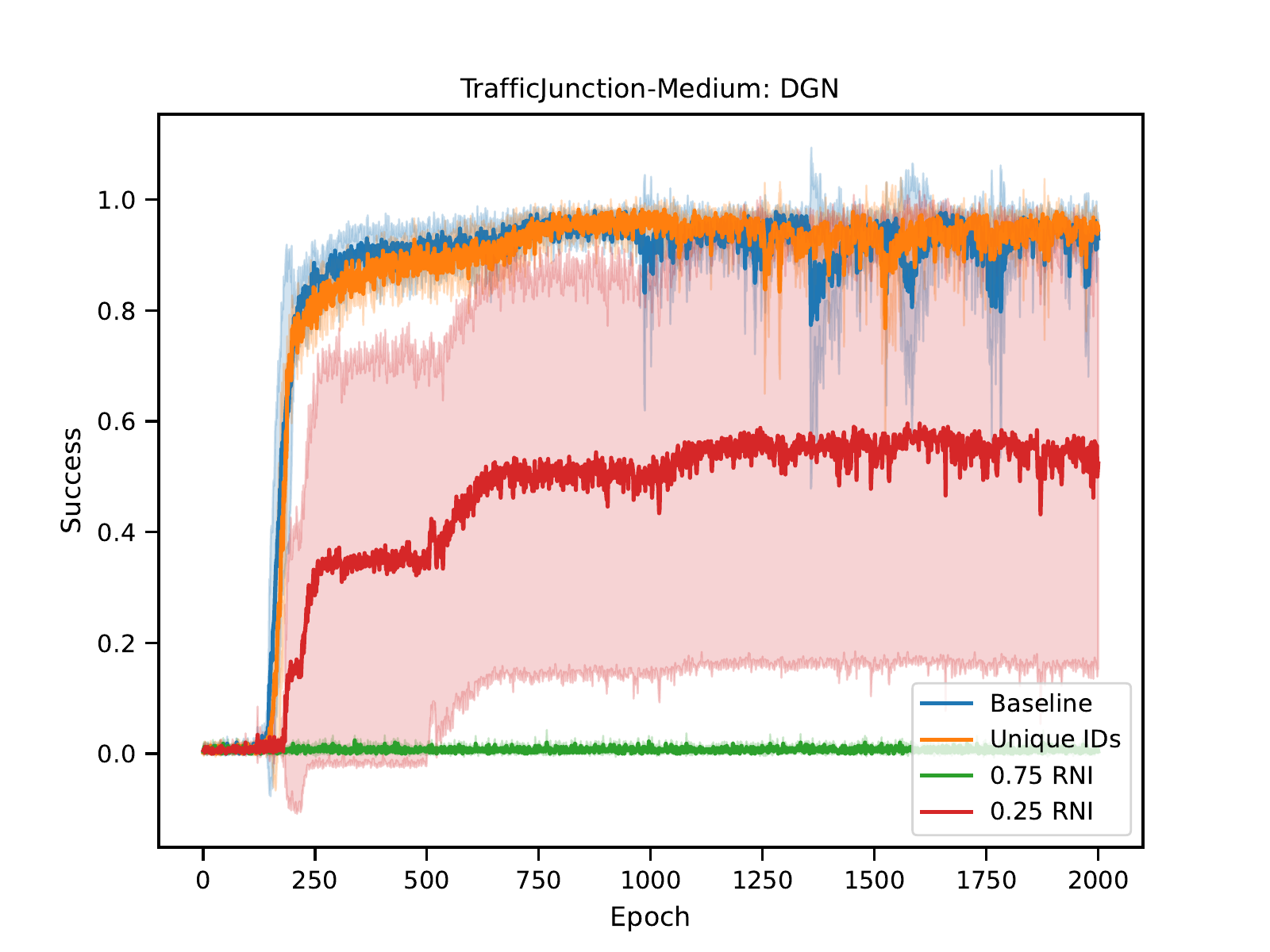}
\includegraphics[width=.49\linewidth]{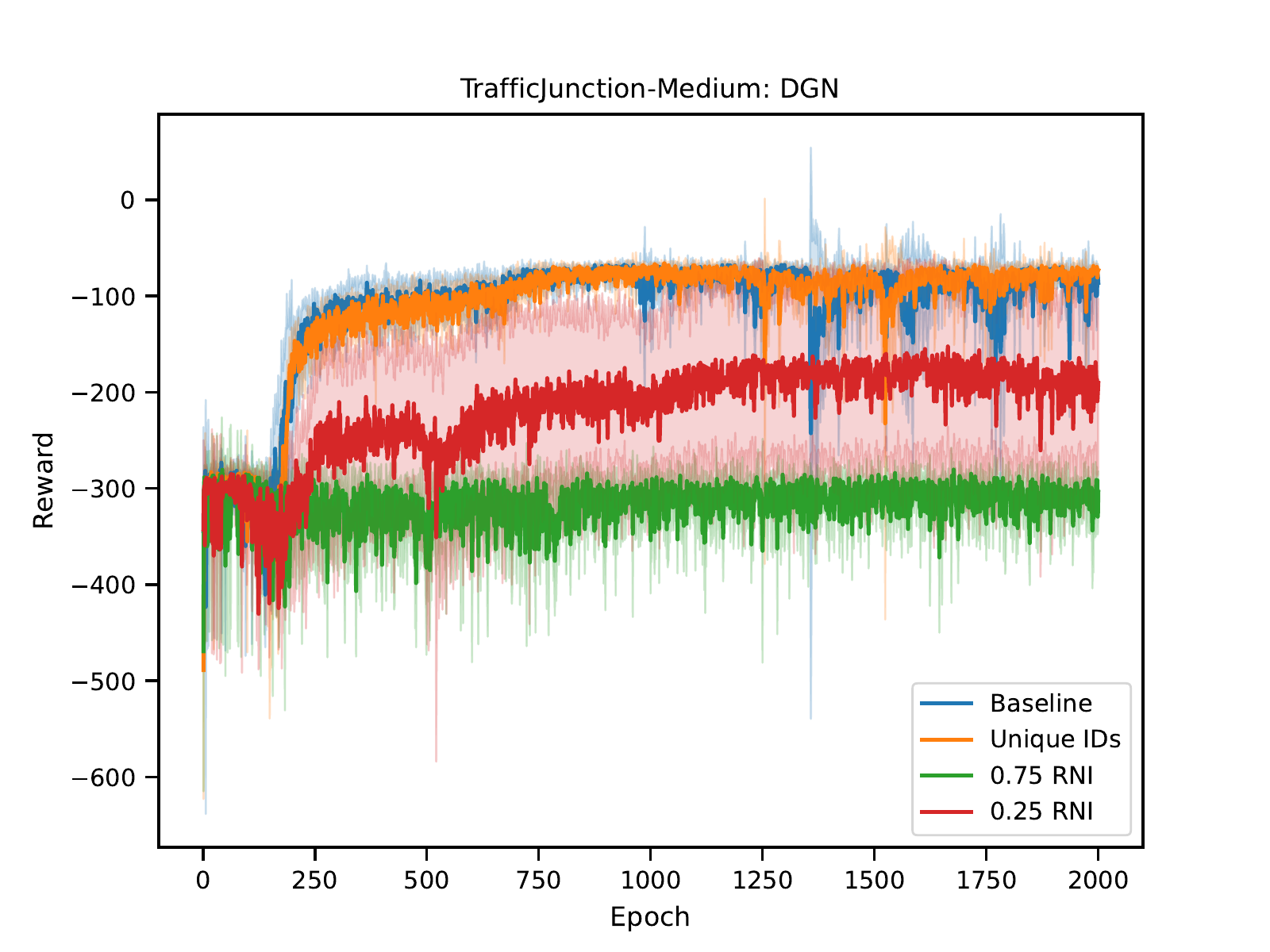}

\includegraphics[width=.49\linewidth]{results/TrafficJunction-Medium/TrafficJunction-Medium_ic3net_success.pdf}
\includegraphics[width=.49\linewidth]{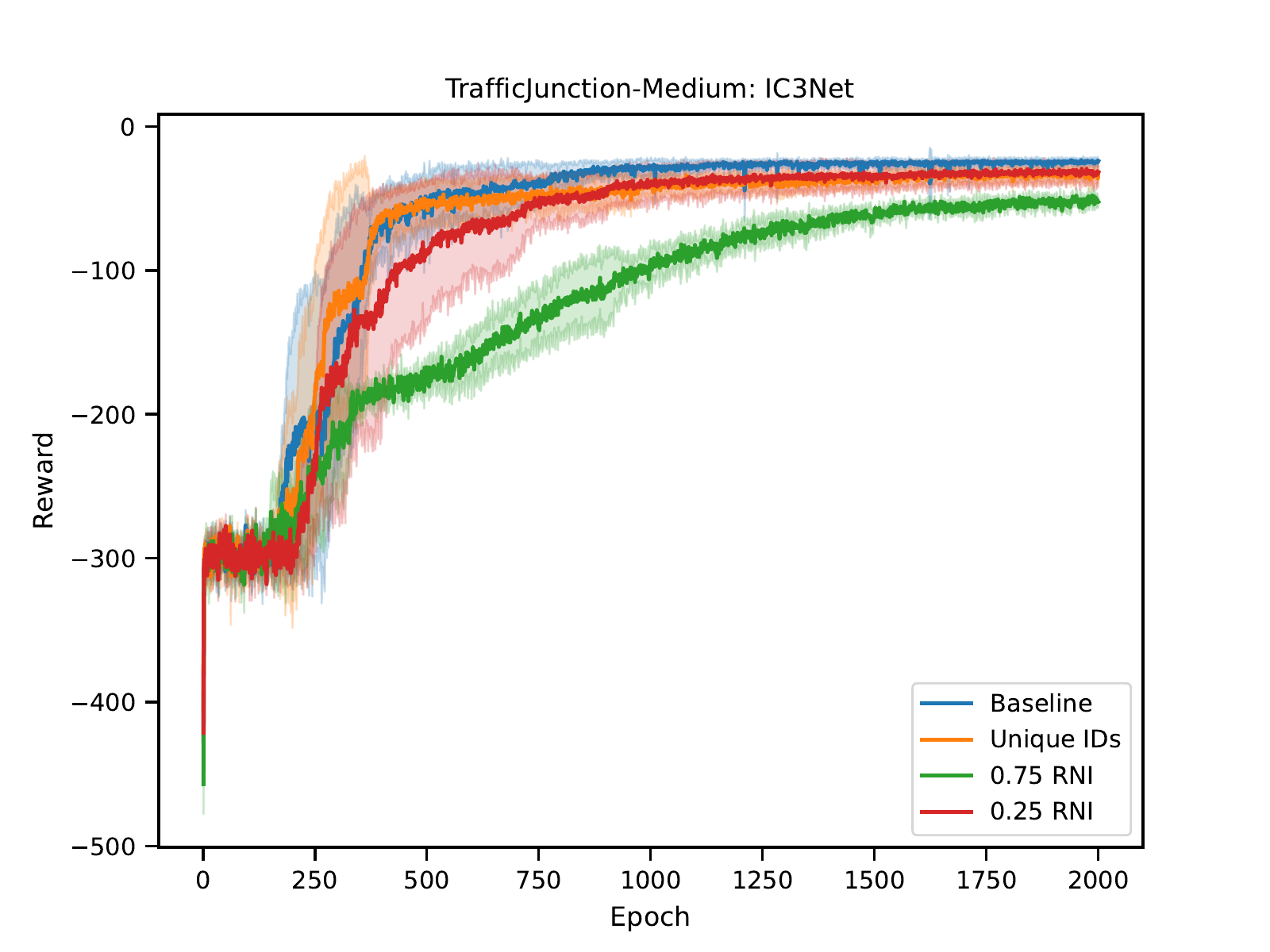}

\includegraphics[width=.49\linewidth]{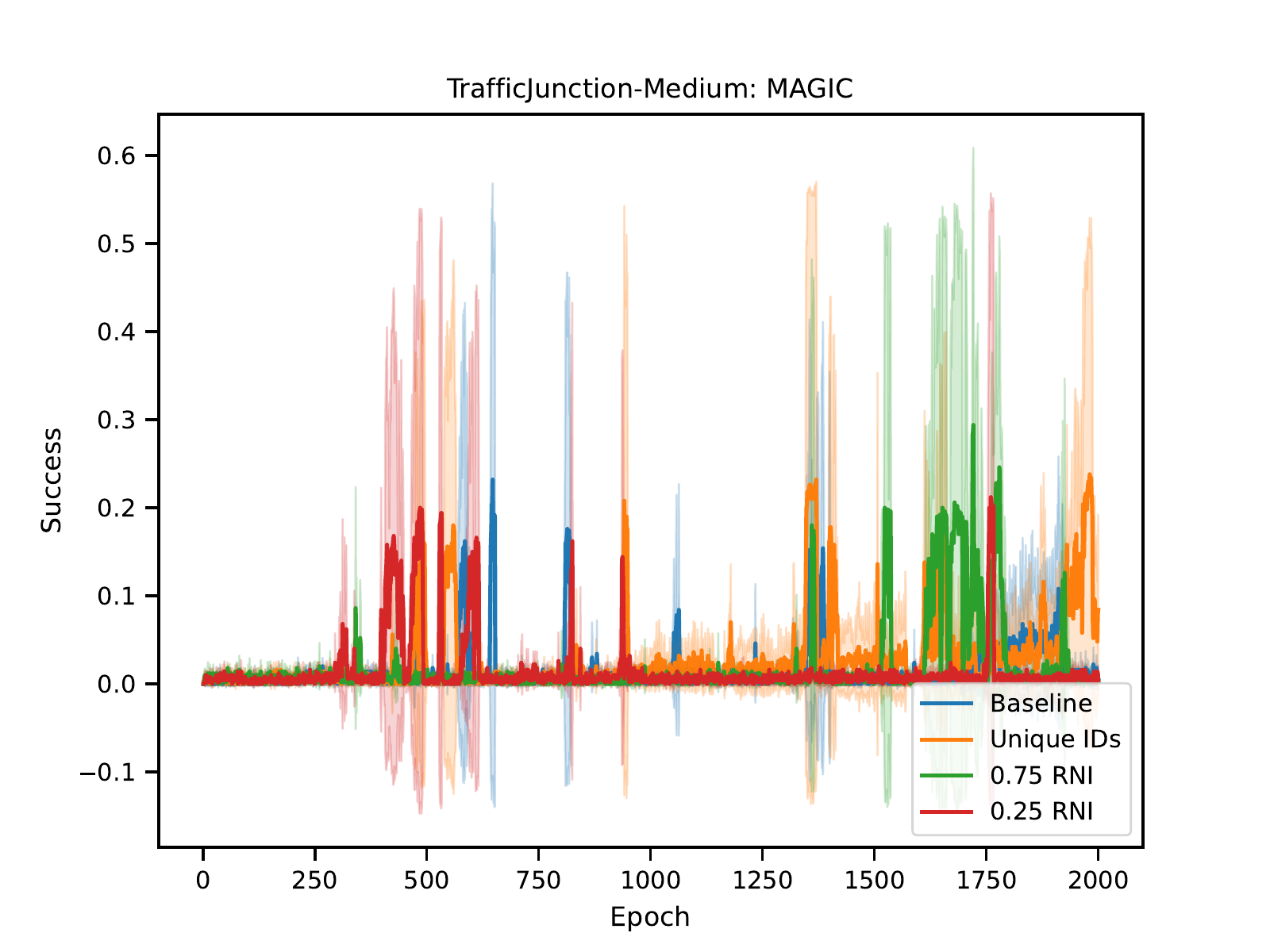}
\includegraphics[width=.49\linewidth]{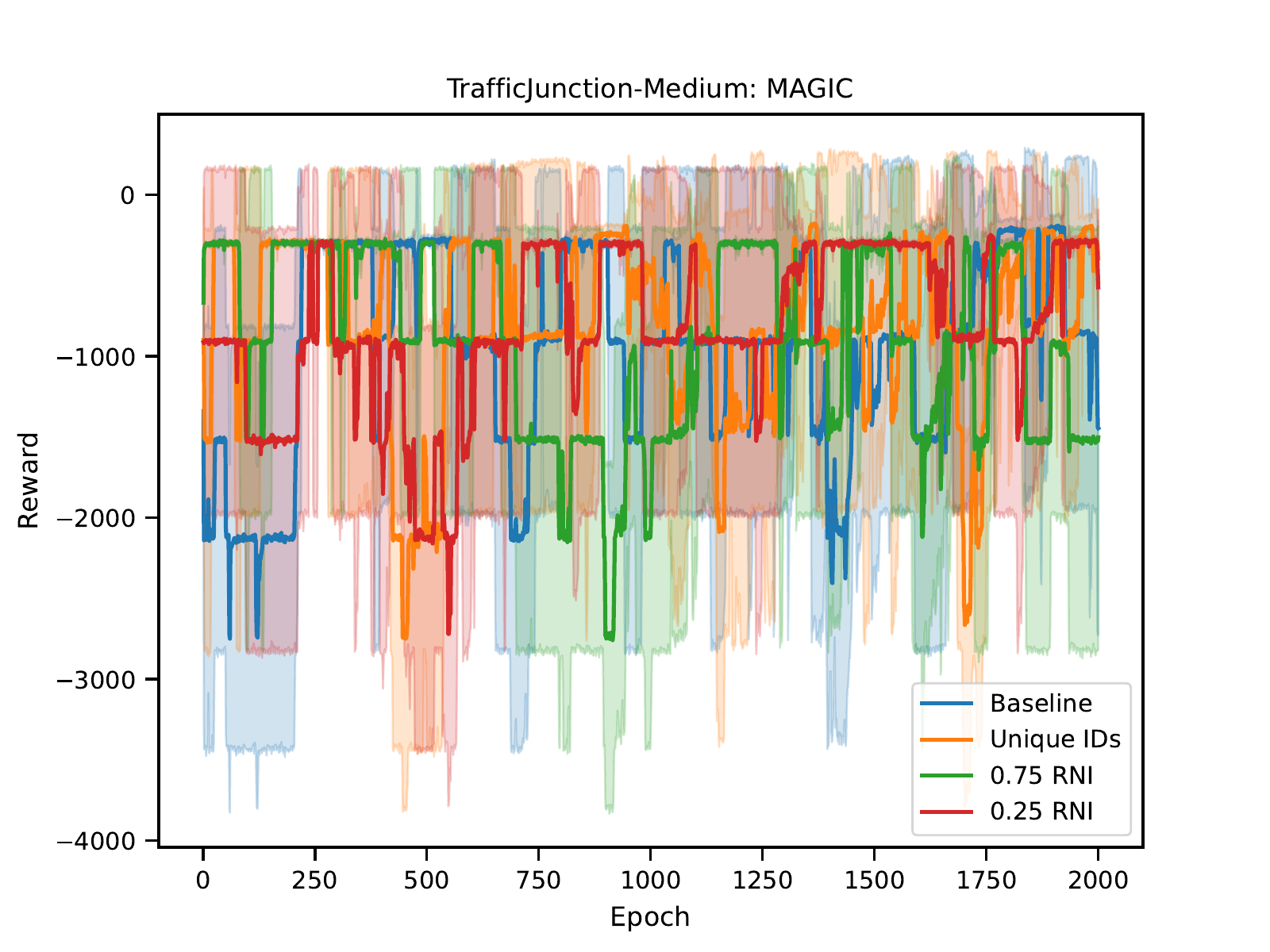}

\includegraphics[width=.49\linewidth]{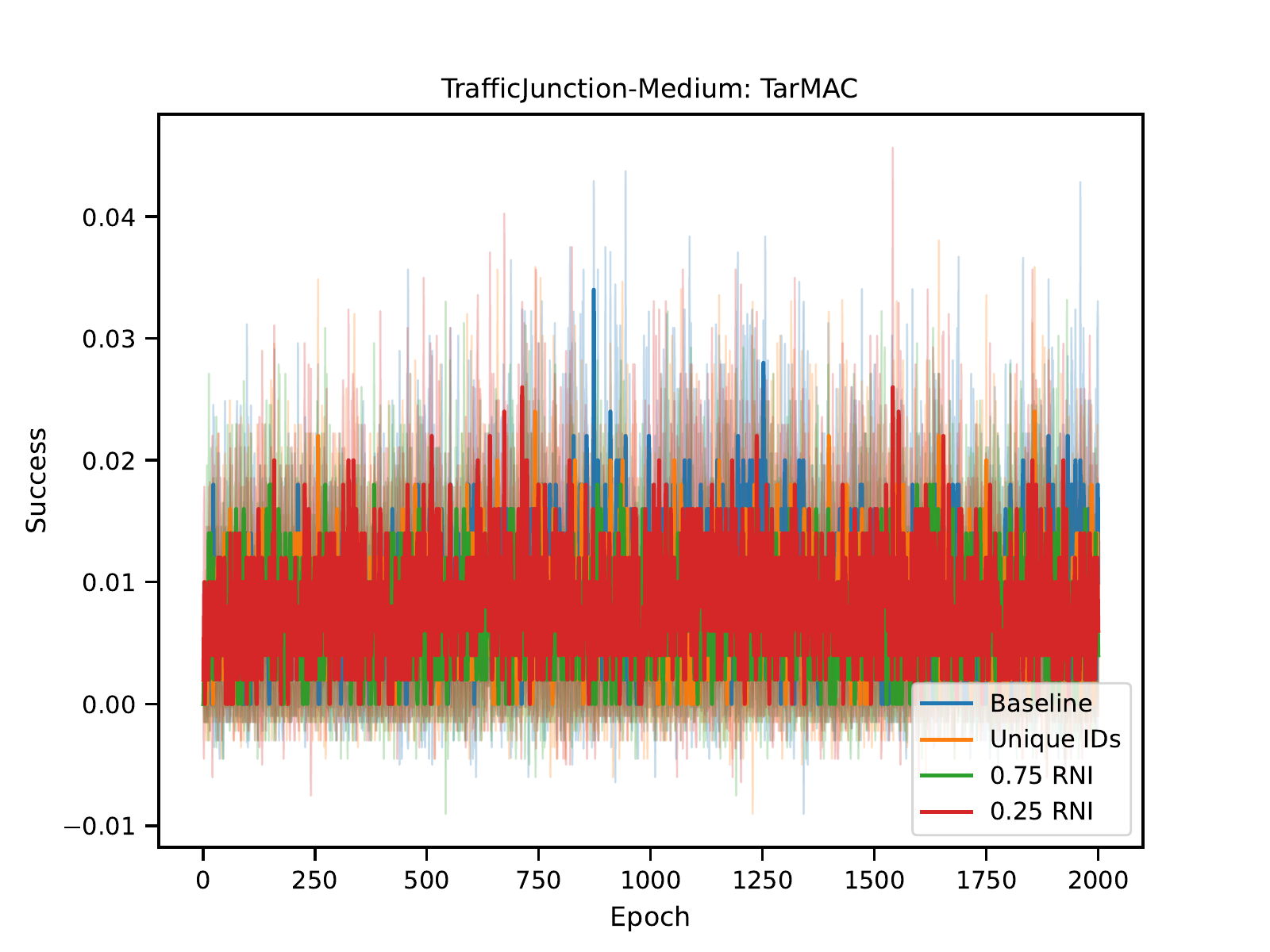}
\includegraphics[width=.49\linewidth]{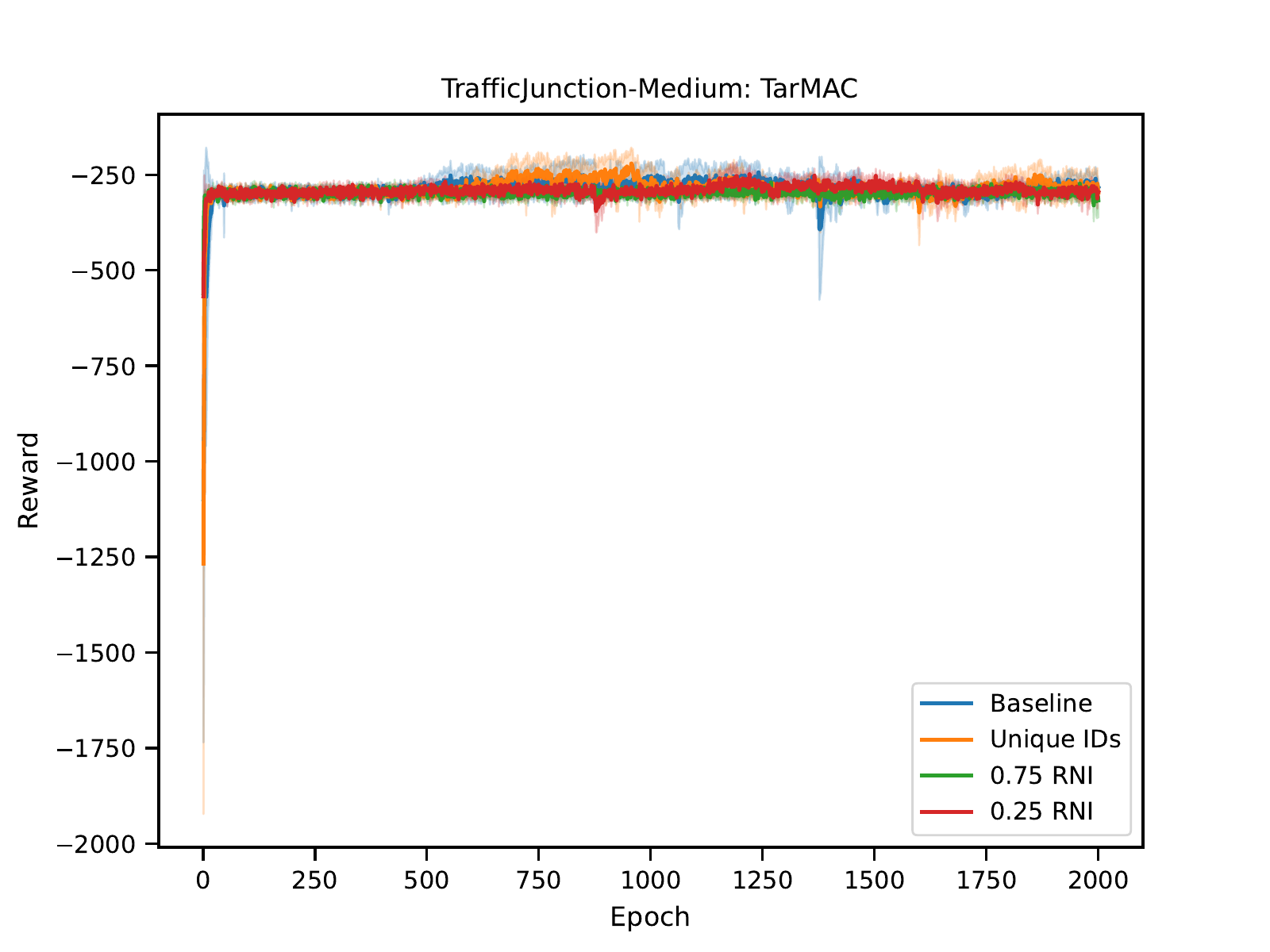}

\includegraphics[width=.49\linewidth]{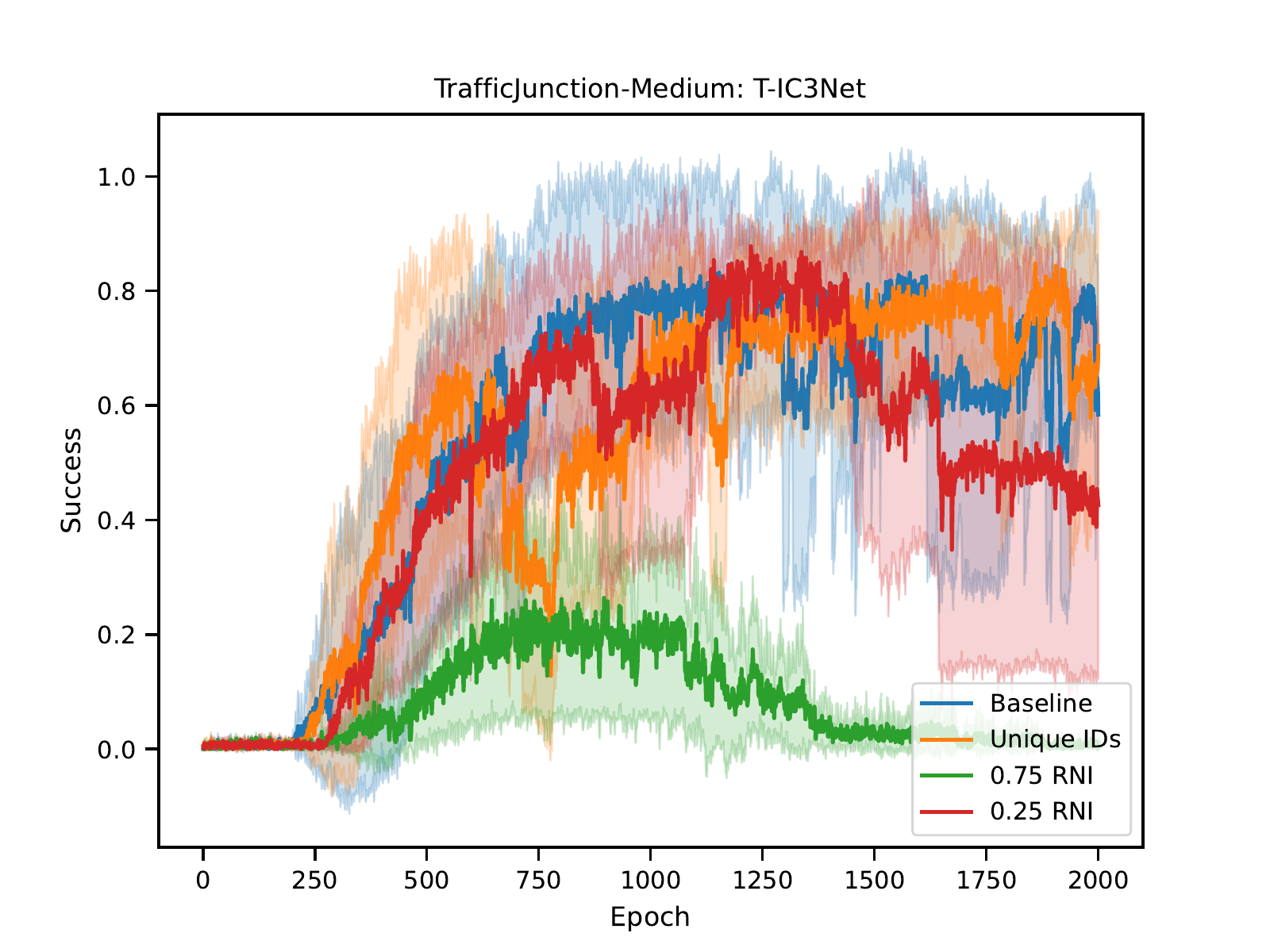}
\includegraphics[width=.49\linewidth]{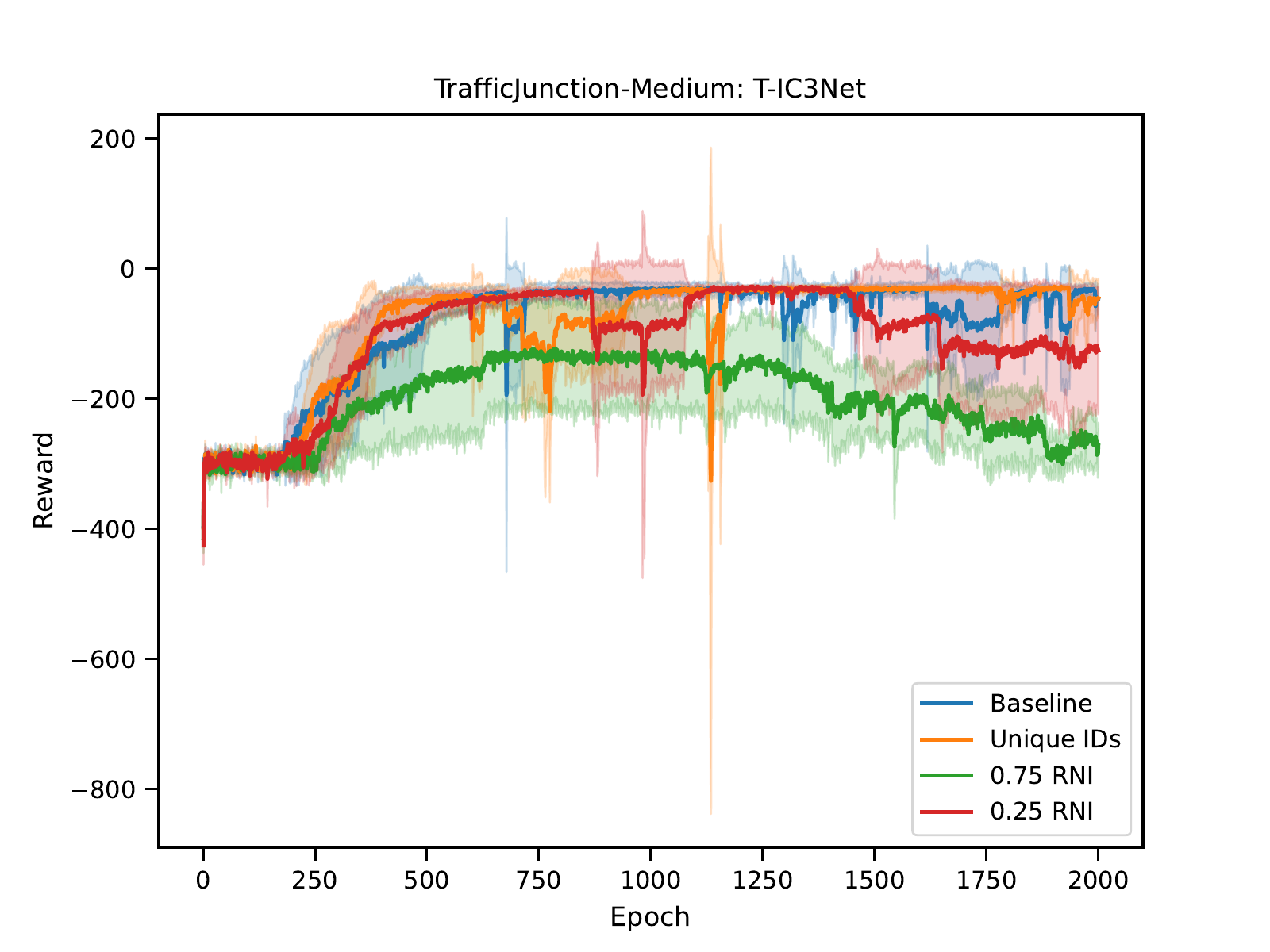}

%
%
\subsection{BoxPushing}
\centering

\includegraphics[width=.49\linewidth]{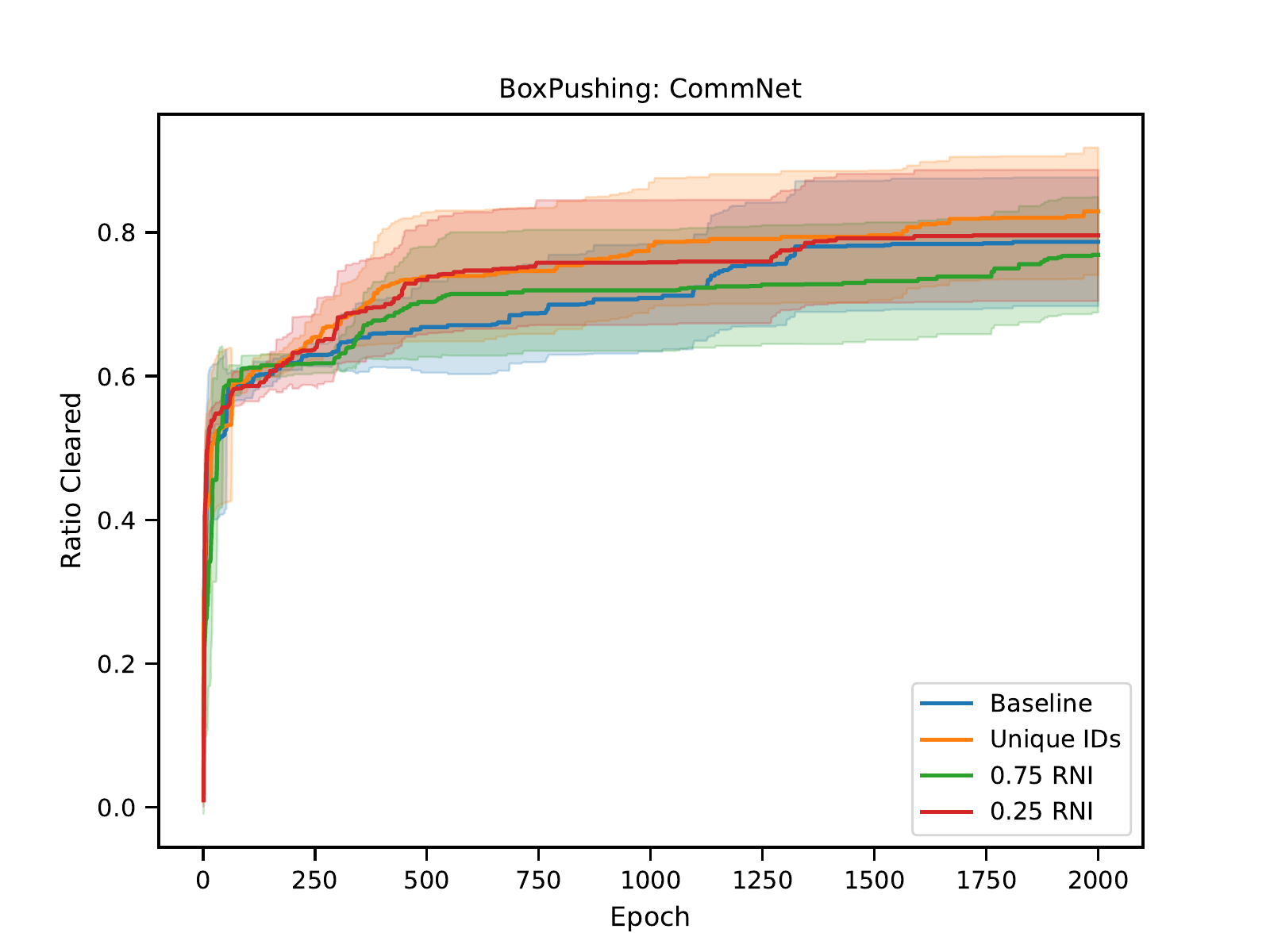}
\includegraphics[width=.49\linewidth]{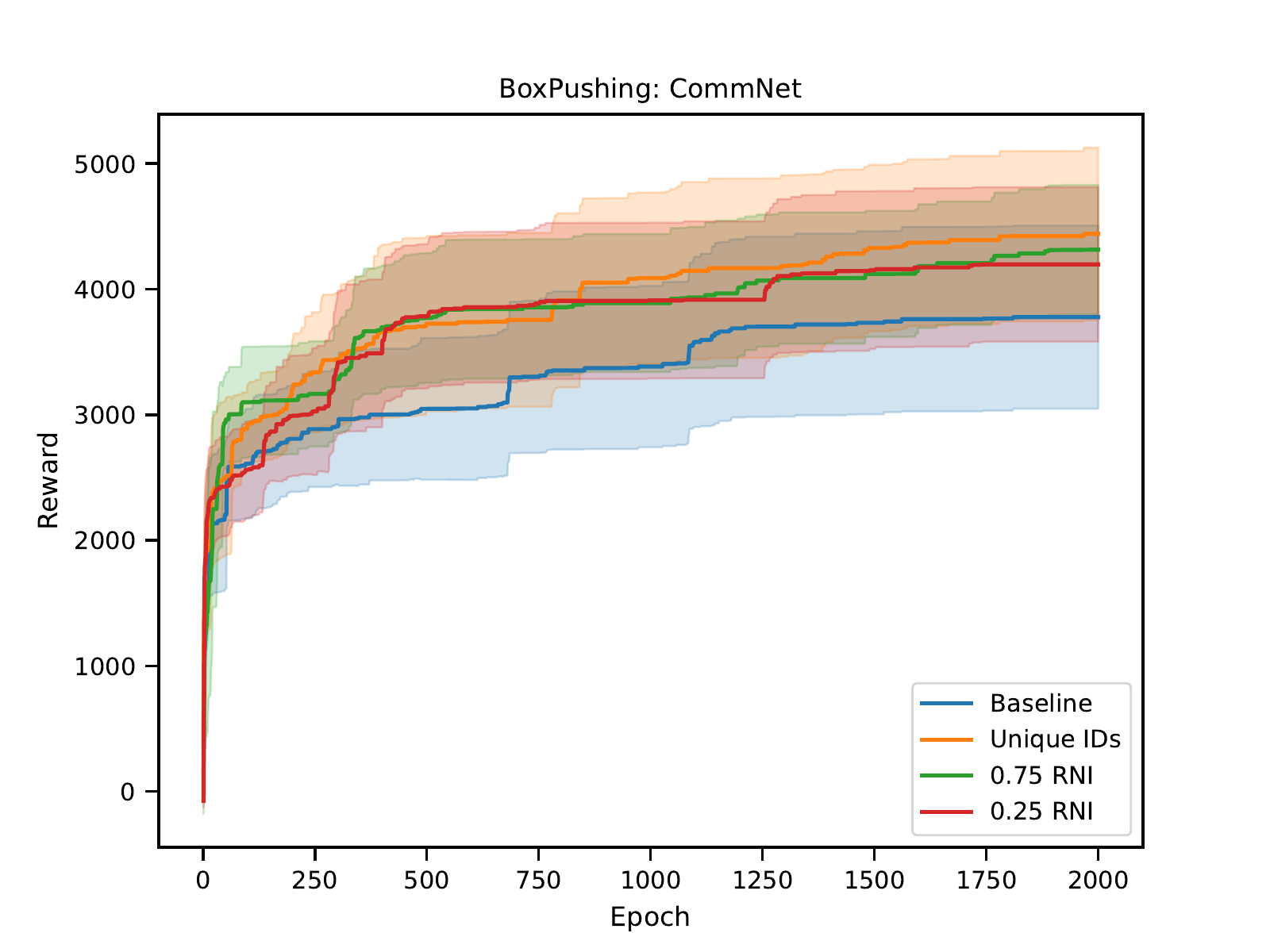}

\includegraphics[width=.49\linewidth]{results/BoxPushing/BoxPushing_dgn_ratio_boxes_cleared.pdf}
\includegraphics[width=.49\linewidth]{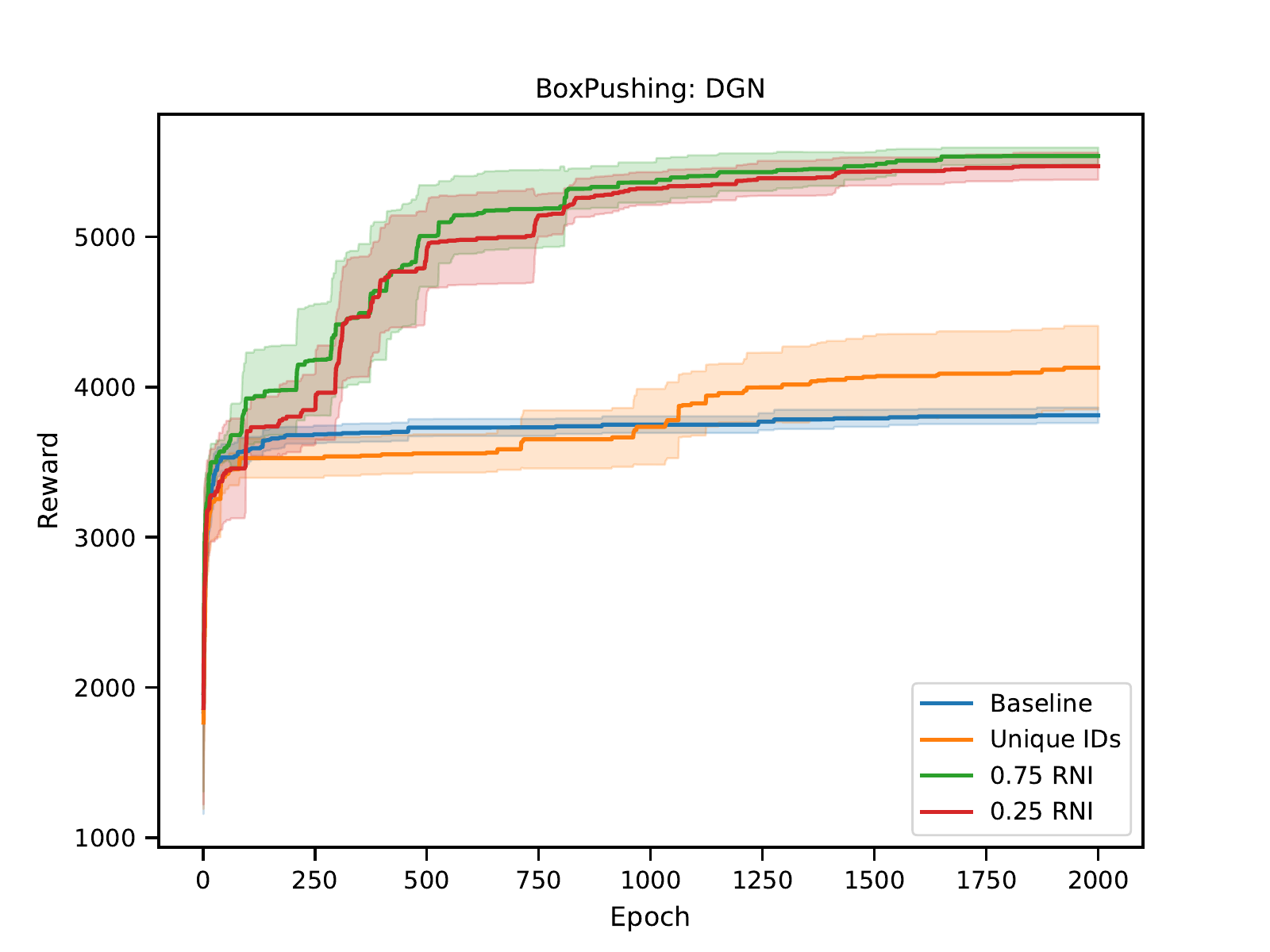}

\includegraphics[width=.49\linewidth]{results/BoxPushing/BoxPushing_ic3net_ratio_boxes_cleared.pdf}
\includegraphics[width=.49\linewidth]{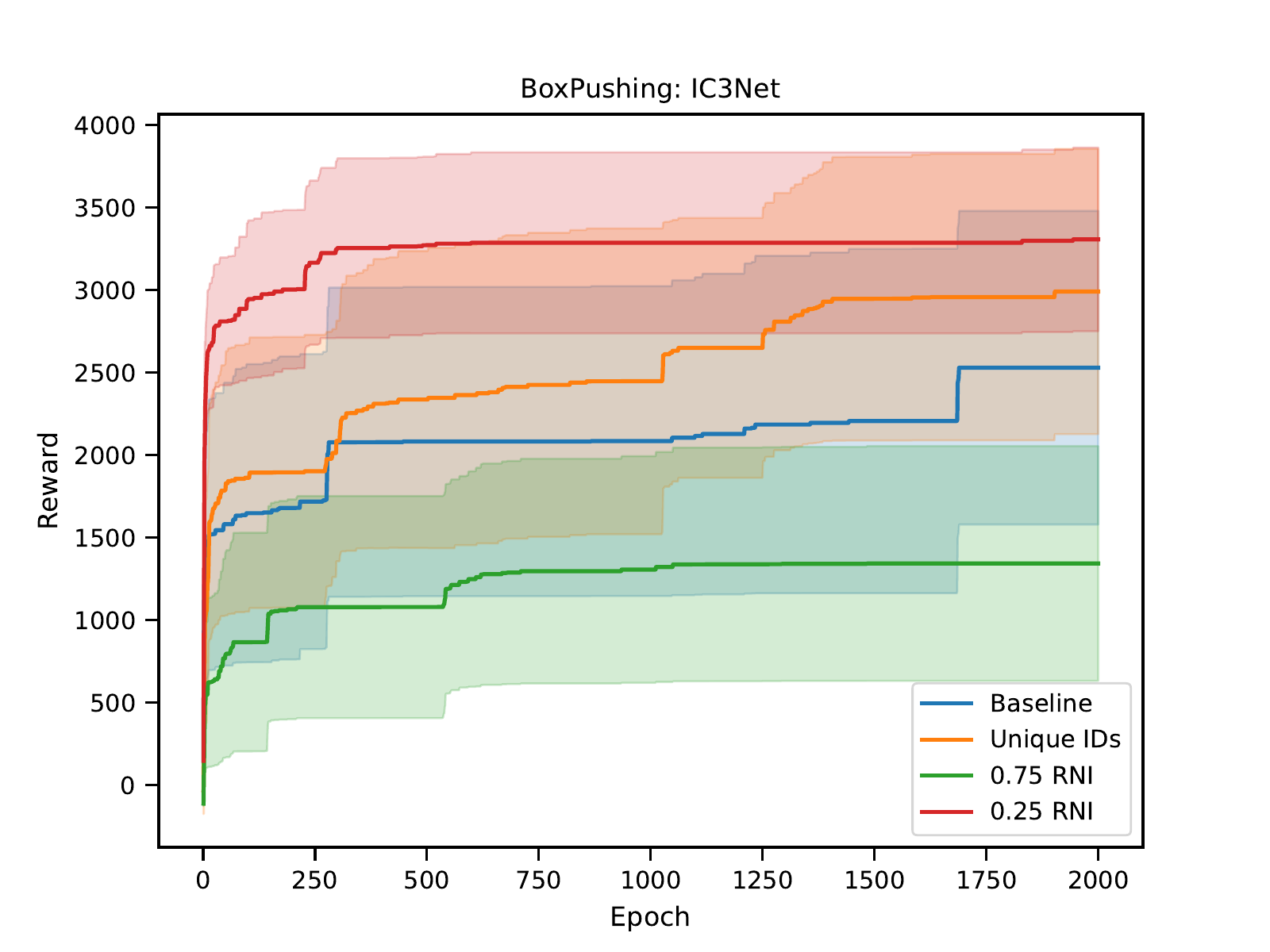}

\includegraphics[width=.49\linewidth]{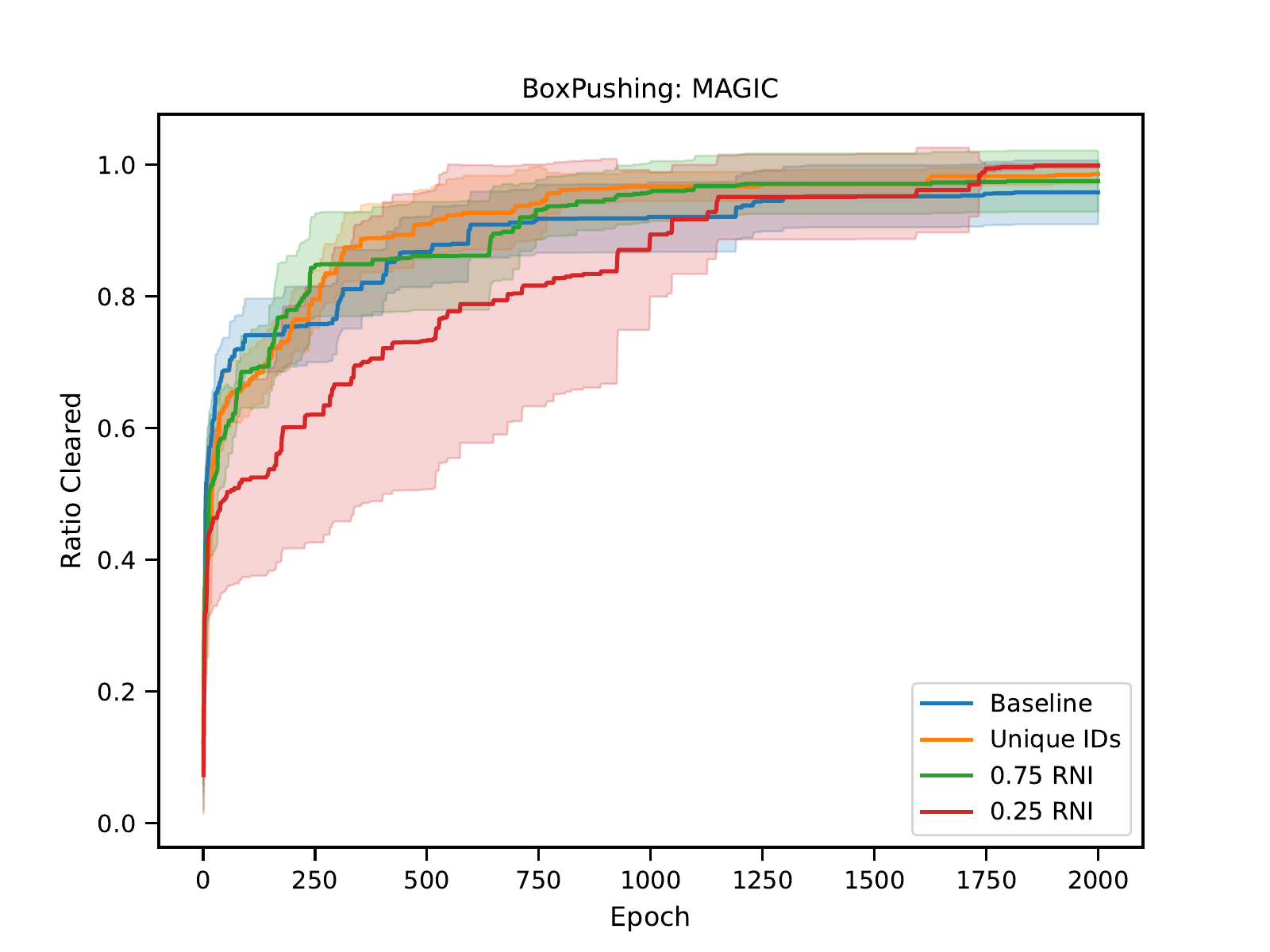}
\includegraphics[width=.49\linewidth]{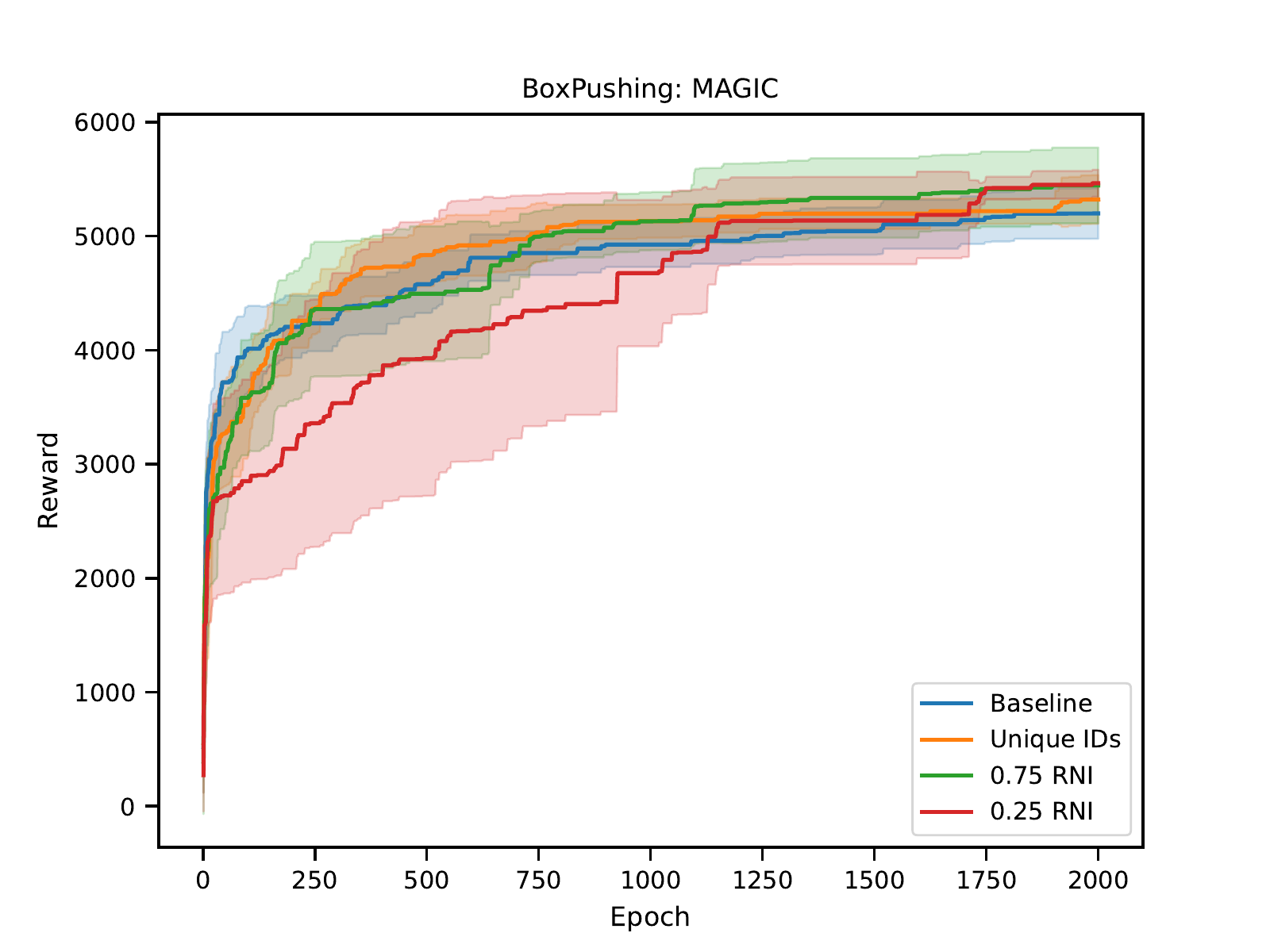}

\includegraphics[width=.49\linewidth]{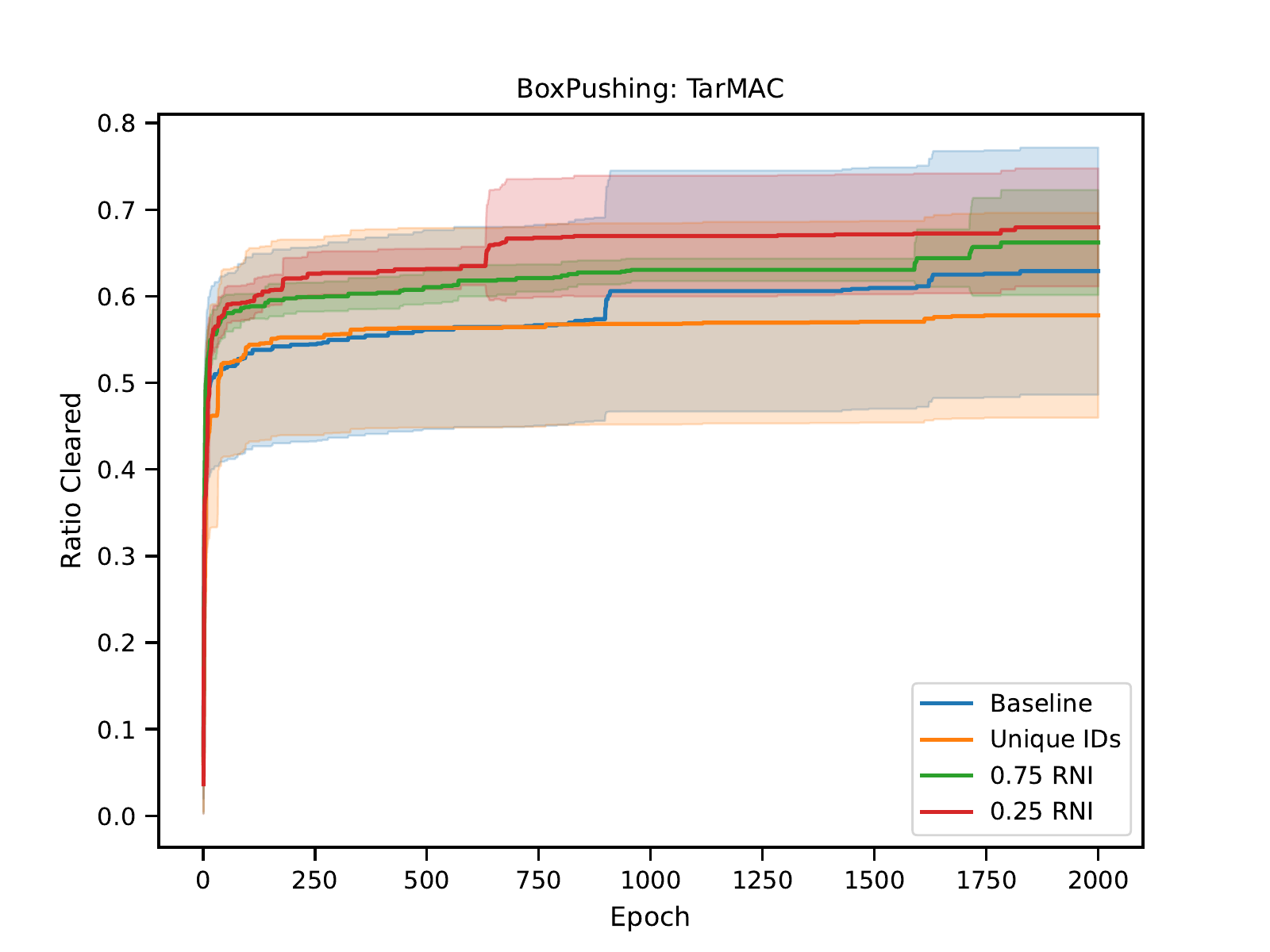}
\includegraphics[width=.49\linewidth]{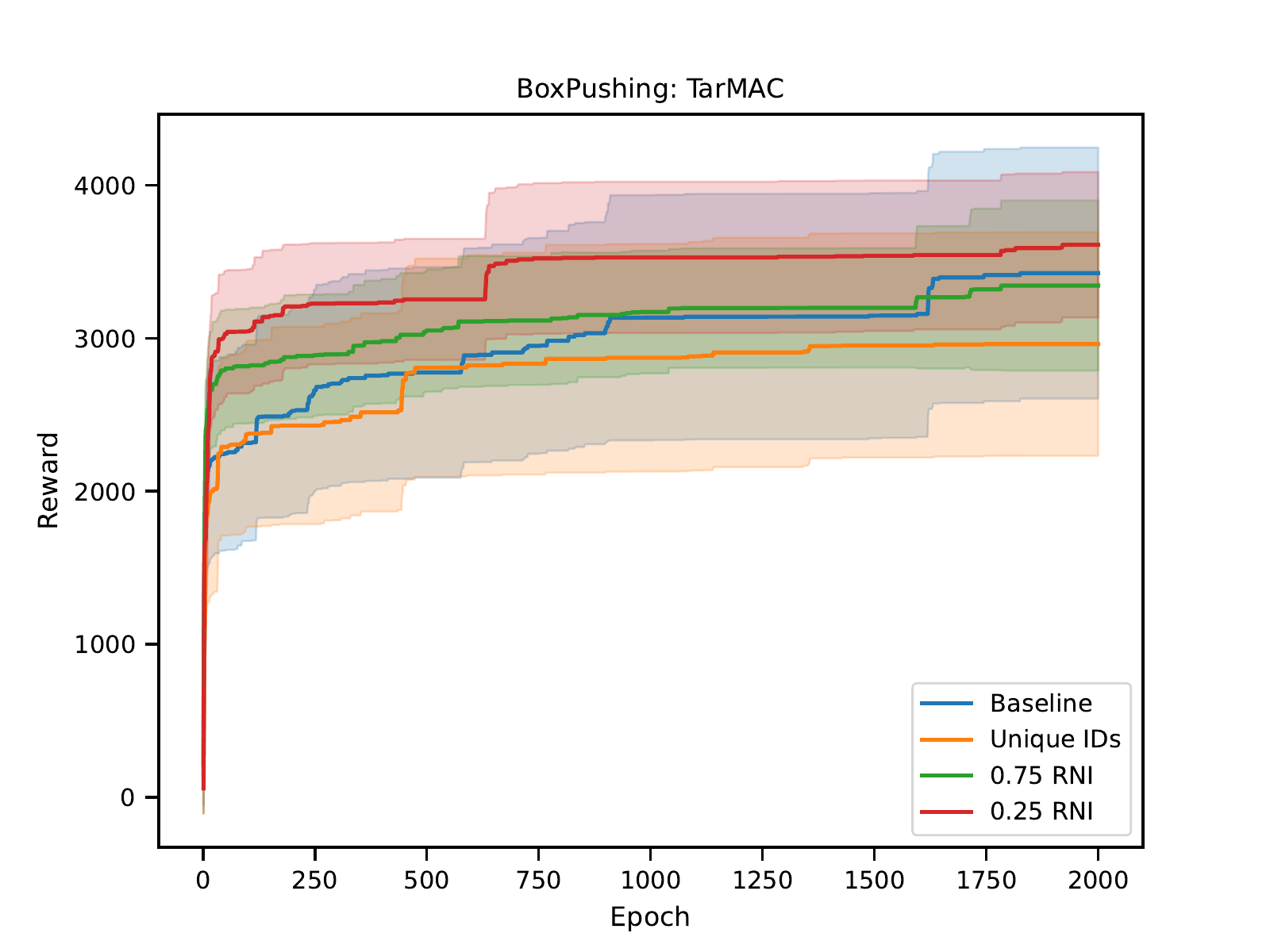}

\includegraphics[width=.49\linewidth]{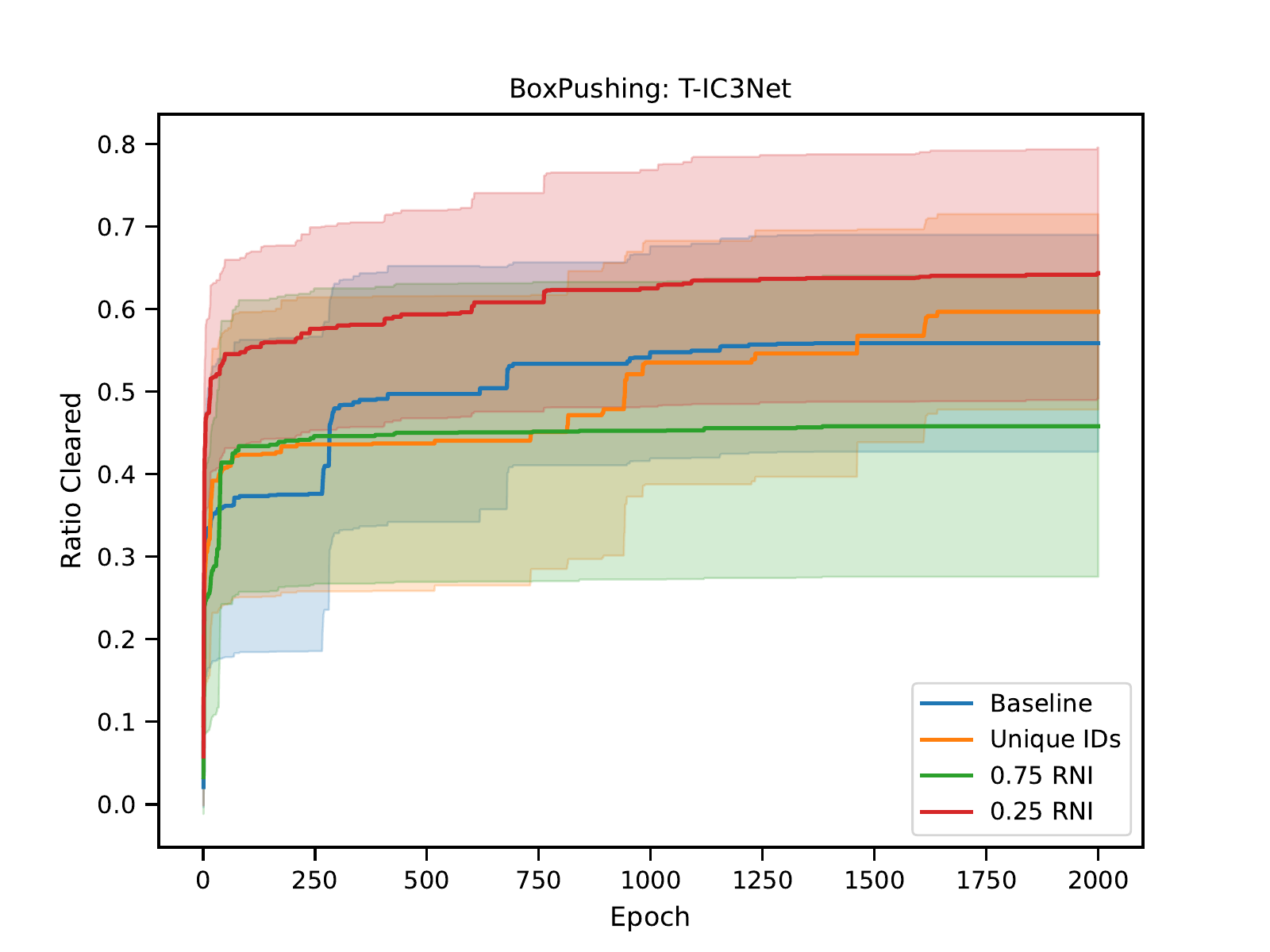}
\includegraphics[width=.49\linewidth]{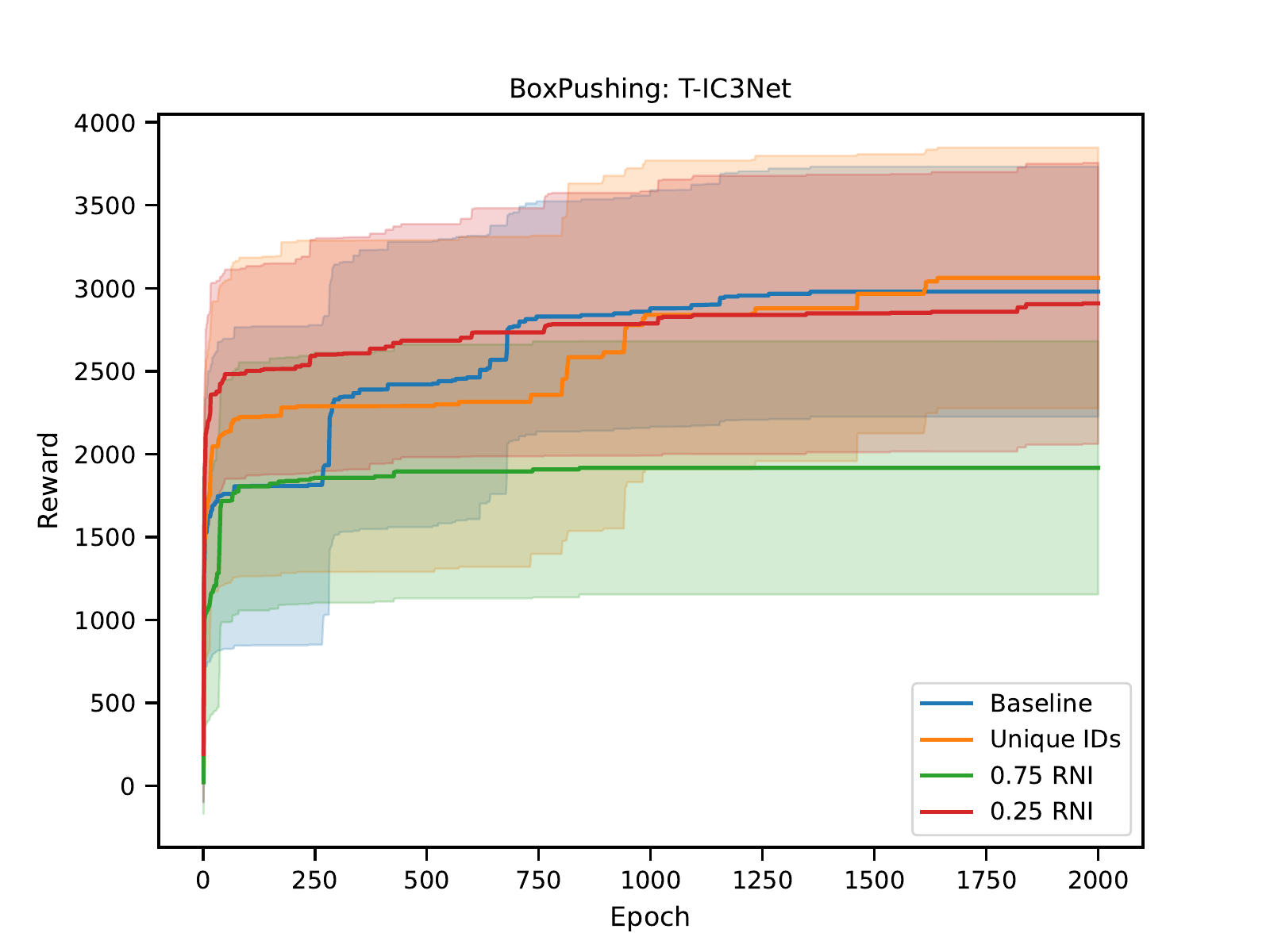}

%
%
\newpage
\subsection{DroneScatter with Greedy Evaluation}
\centering

\includegraphics[width=.32\linewidth]{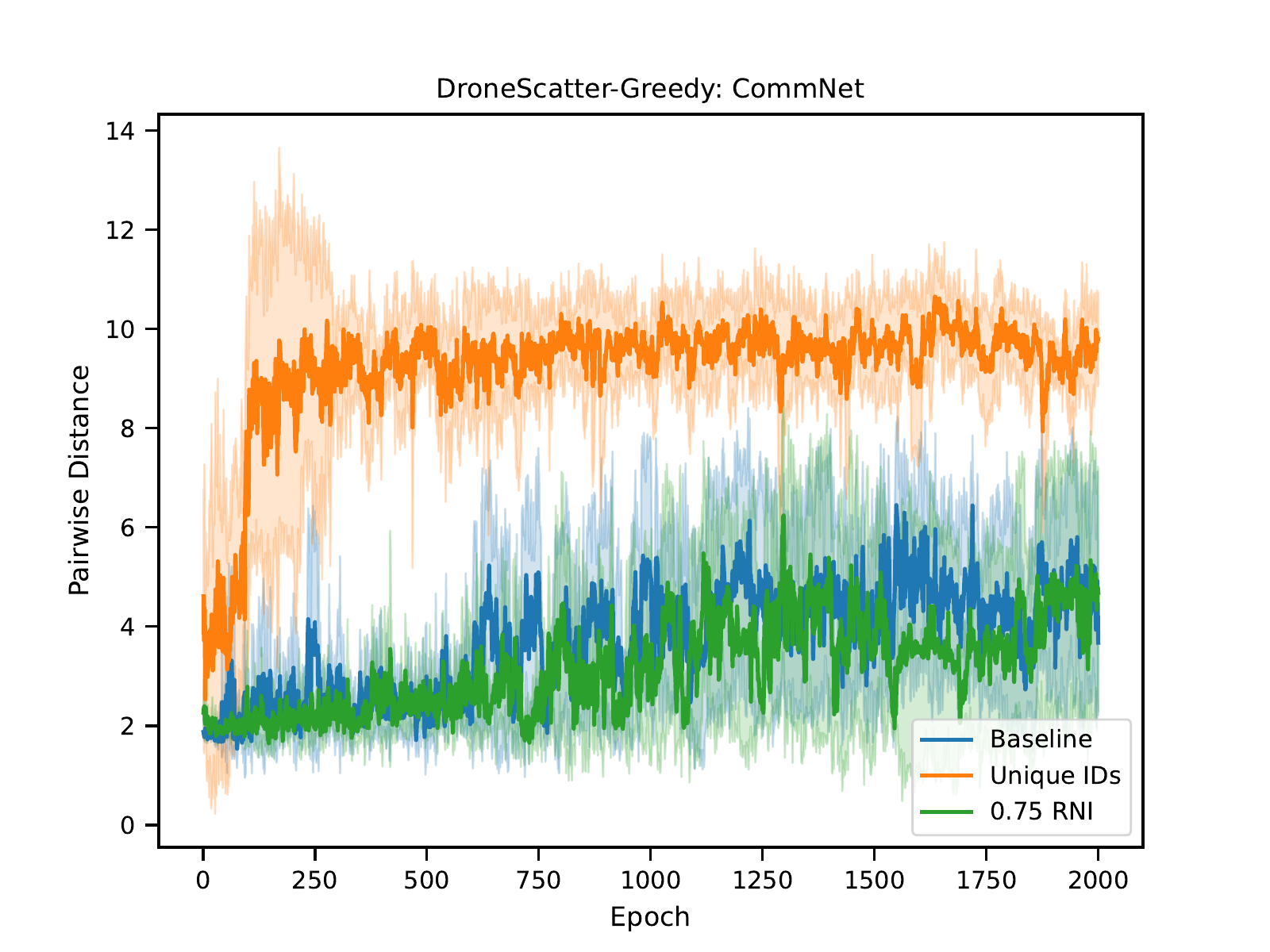}
\includegraphics[width=.32\linewidth]{results/DroneScatter-Greedy/DroneScatter-Greedy_commnet_steps_taken.pdf}
\includegraphics[width=.32\linewidth]{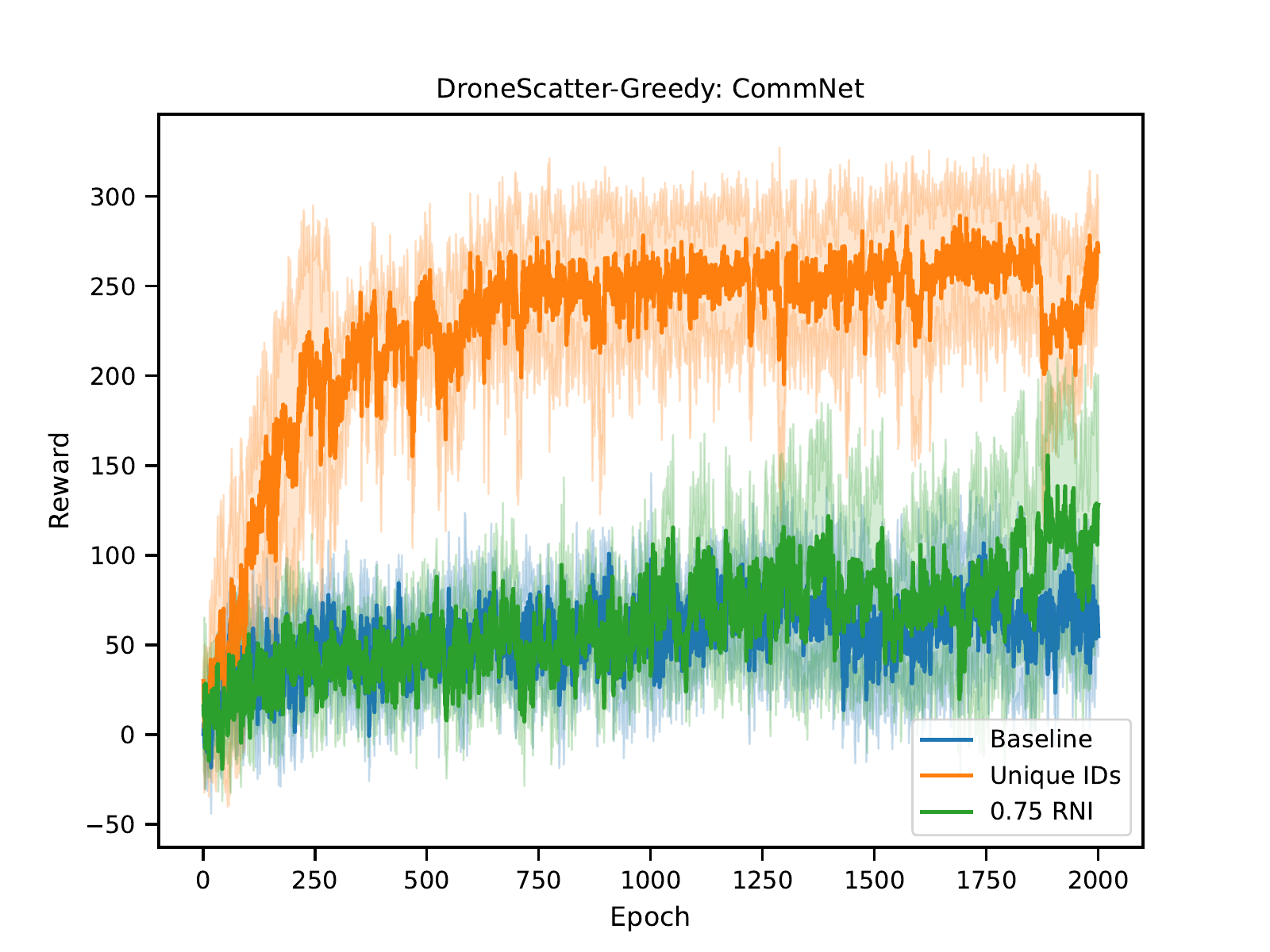}

\includegraphics[width=.32\linewidth]{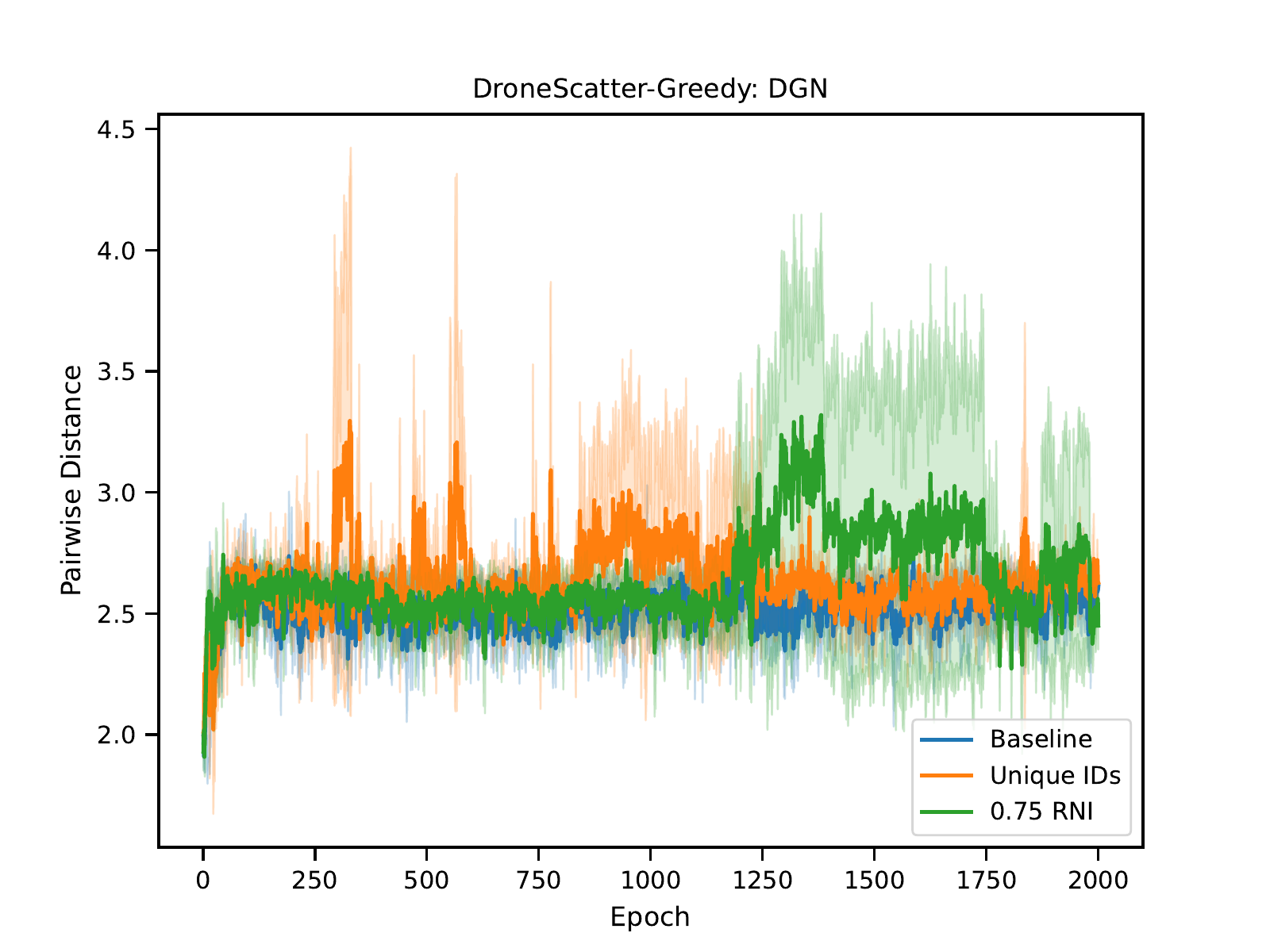}
\includegraphics[width=.32\linewidth]{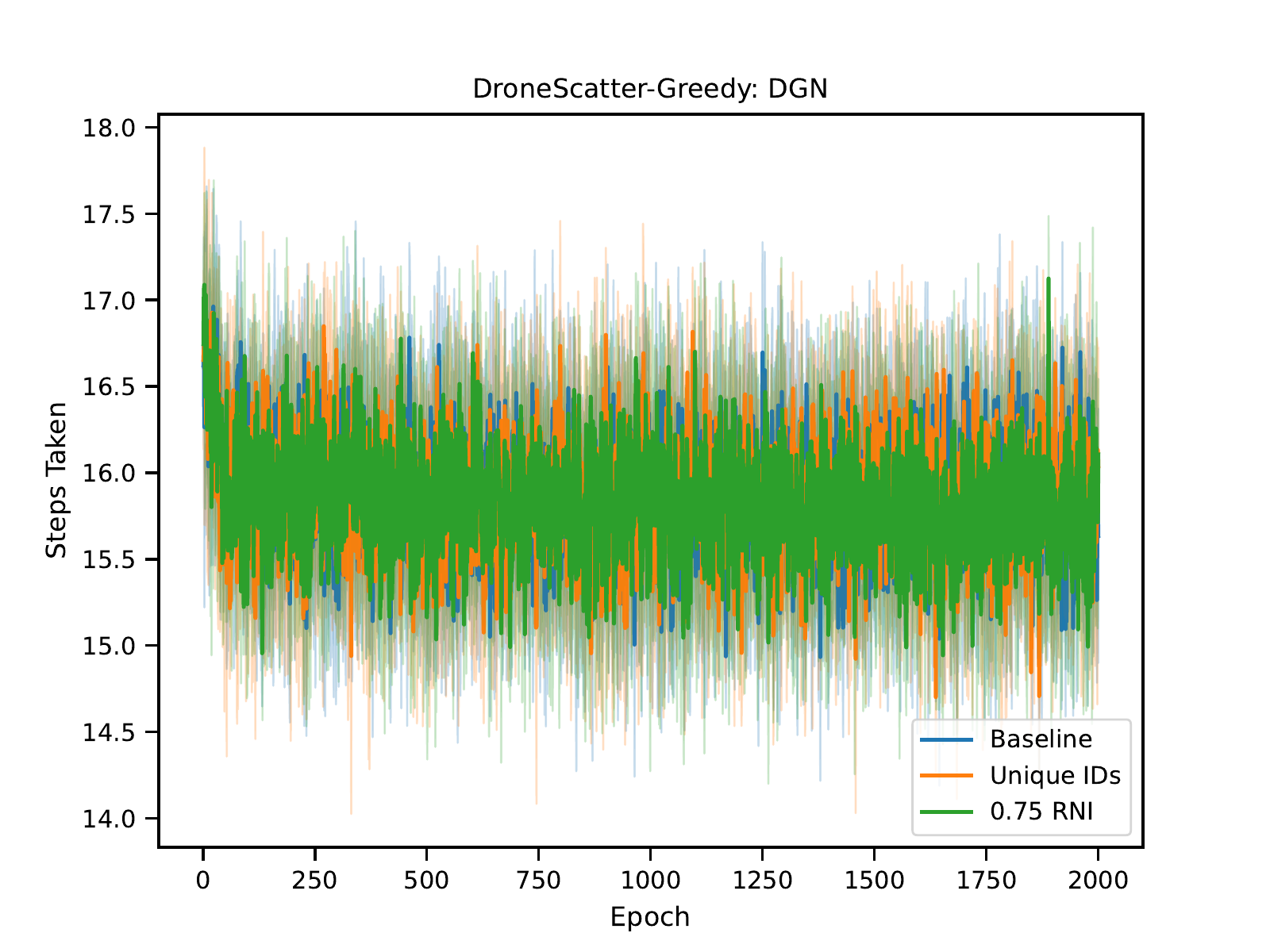}
\includegraphics[width=.32\linewidth]{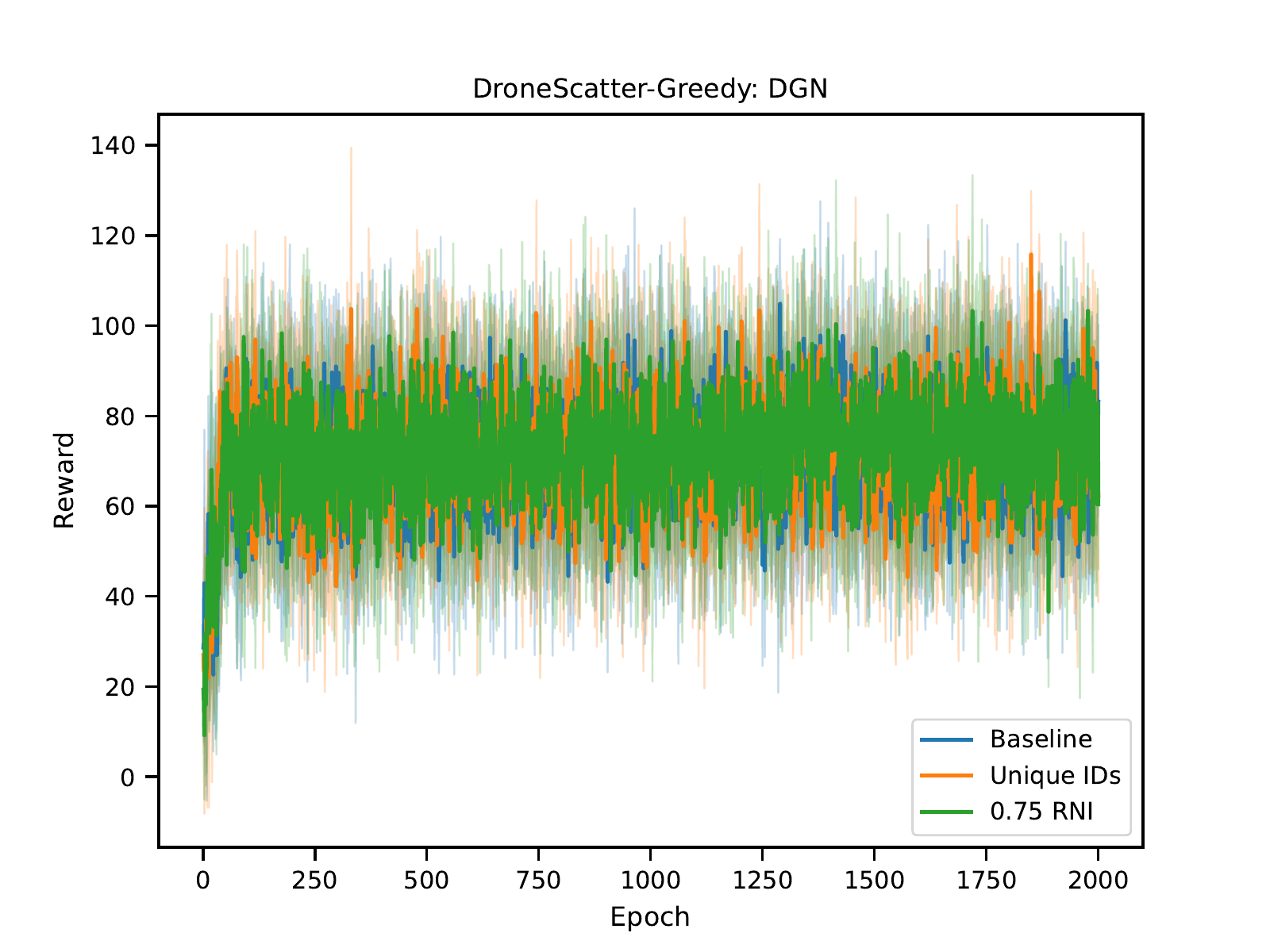}

\includegraphics[width=.32\linewidth]{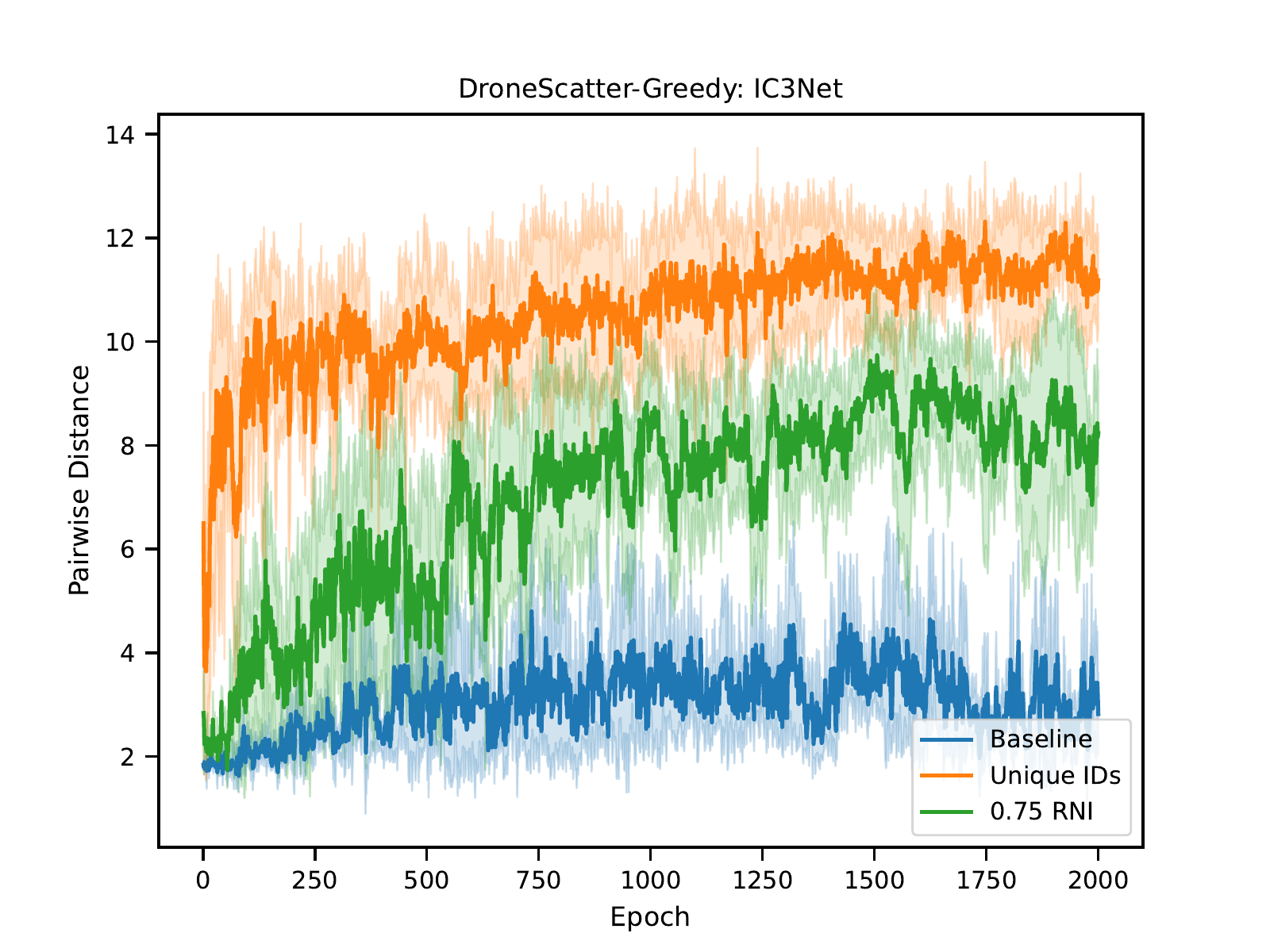}
\includegraphics[width=.32\linewidth]{results/DroneScatter-Greedy/DroneScatter-Greedy_ic3net_steps_taken.pdf}
\includegraphics[width=.32\linewidth]{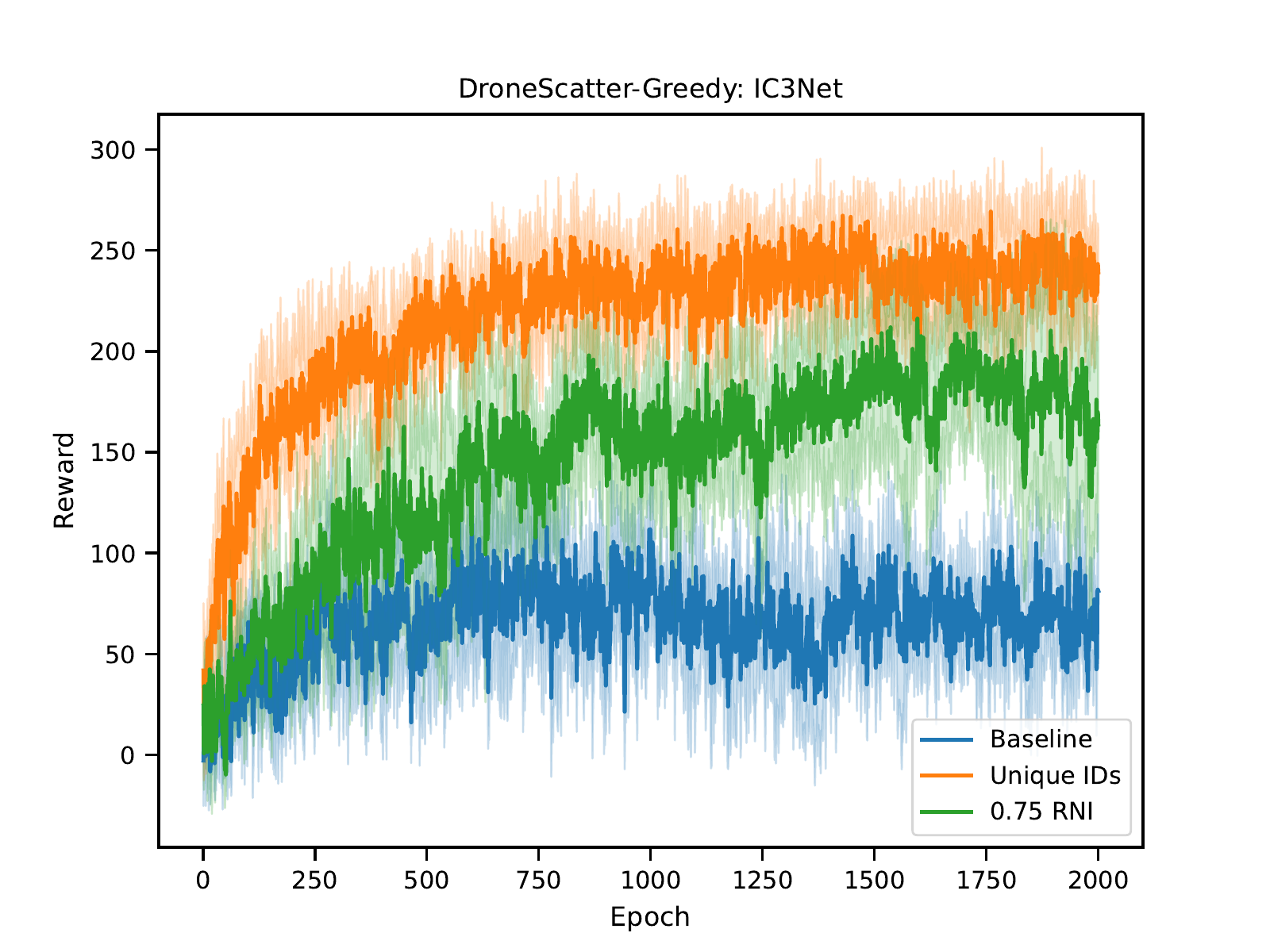}

\includegraphics[width=.32\linewidth]{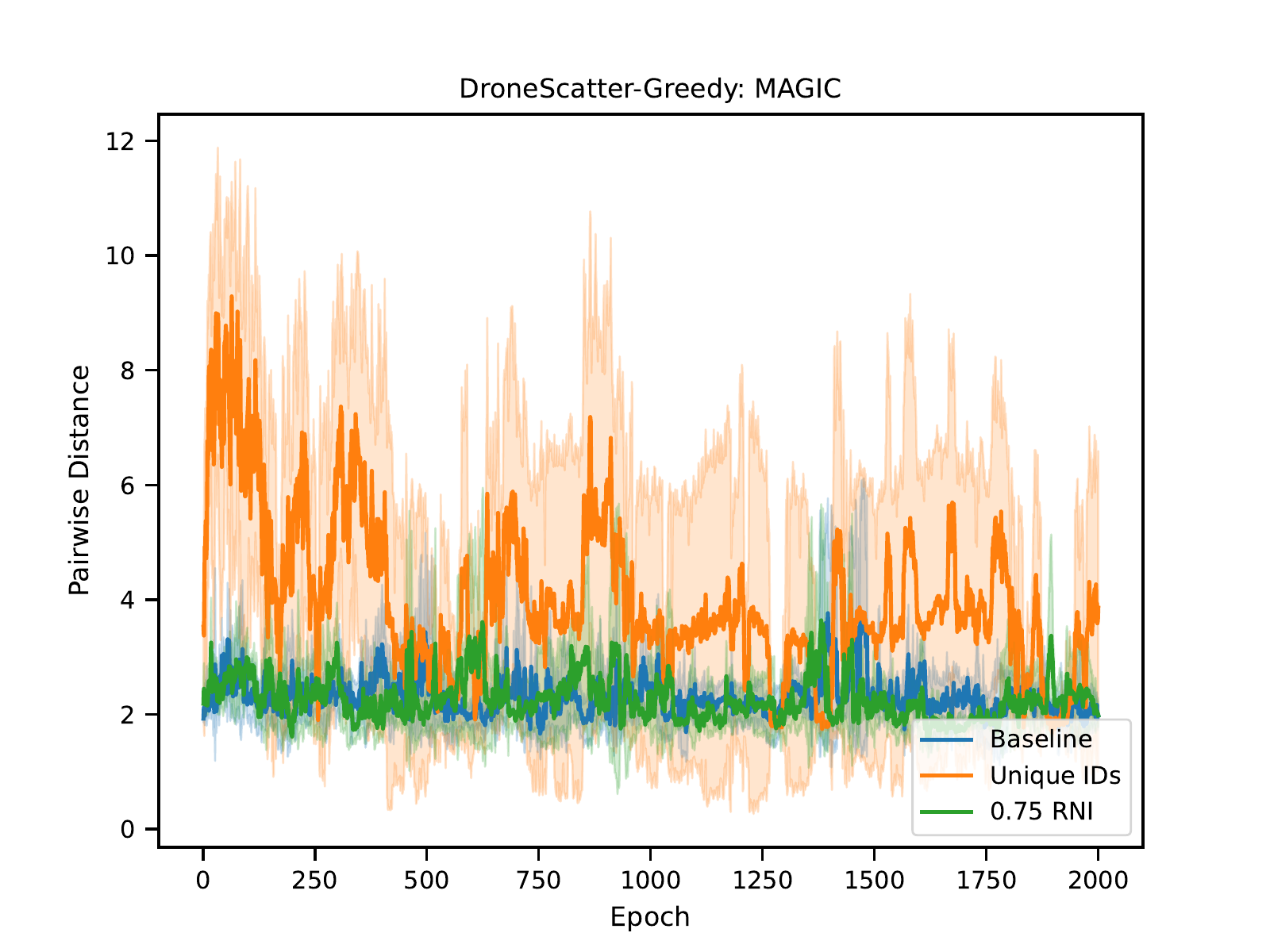}
\includegraphics[width=.32\linewidth]{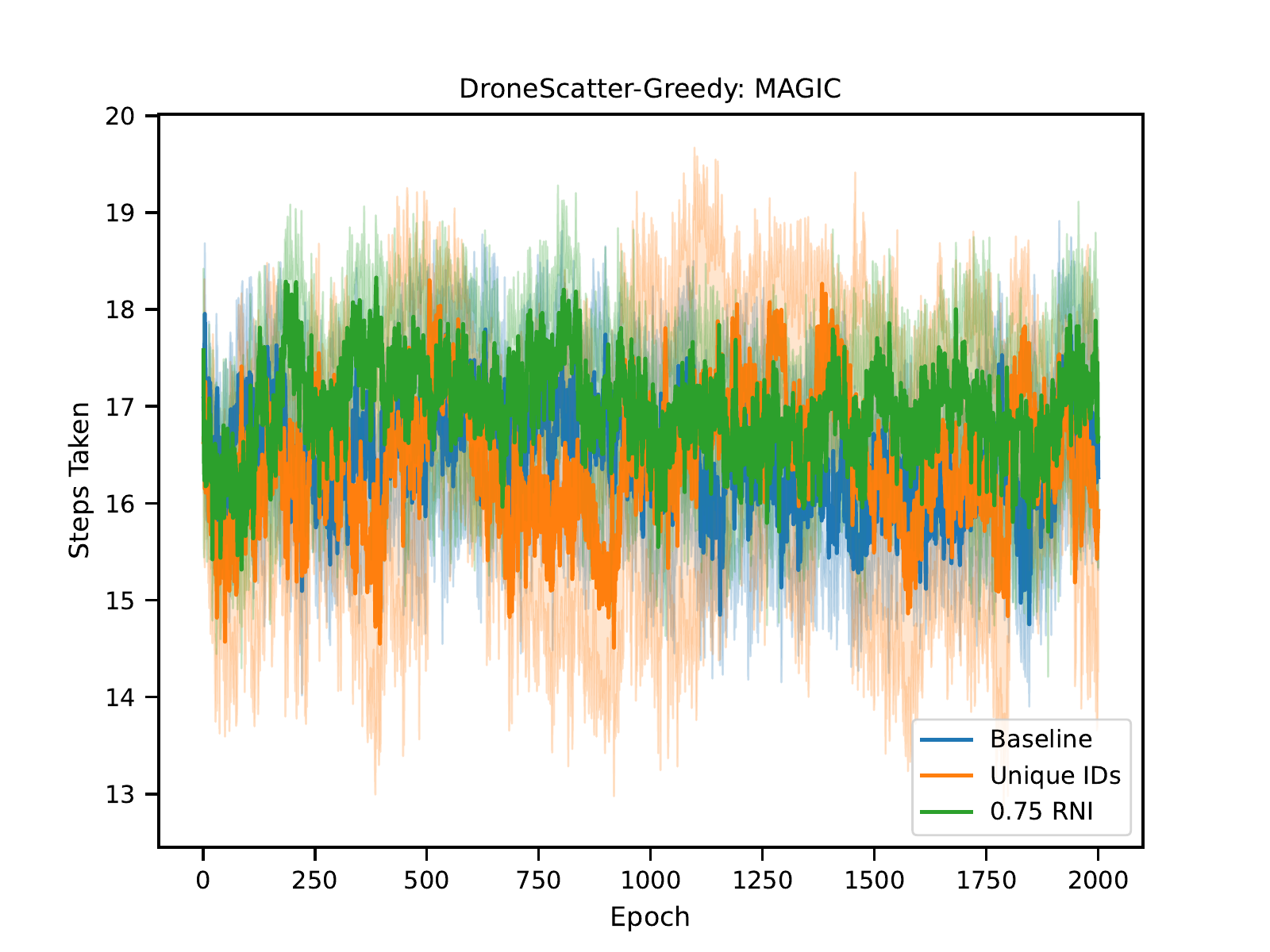}
\includegraphics[width=.32\linewidth]{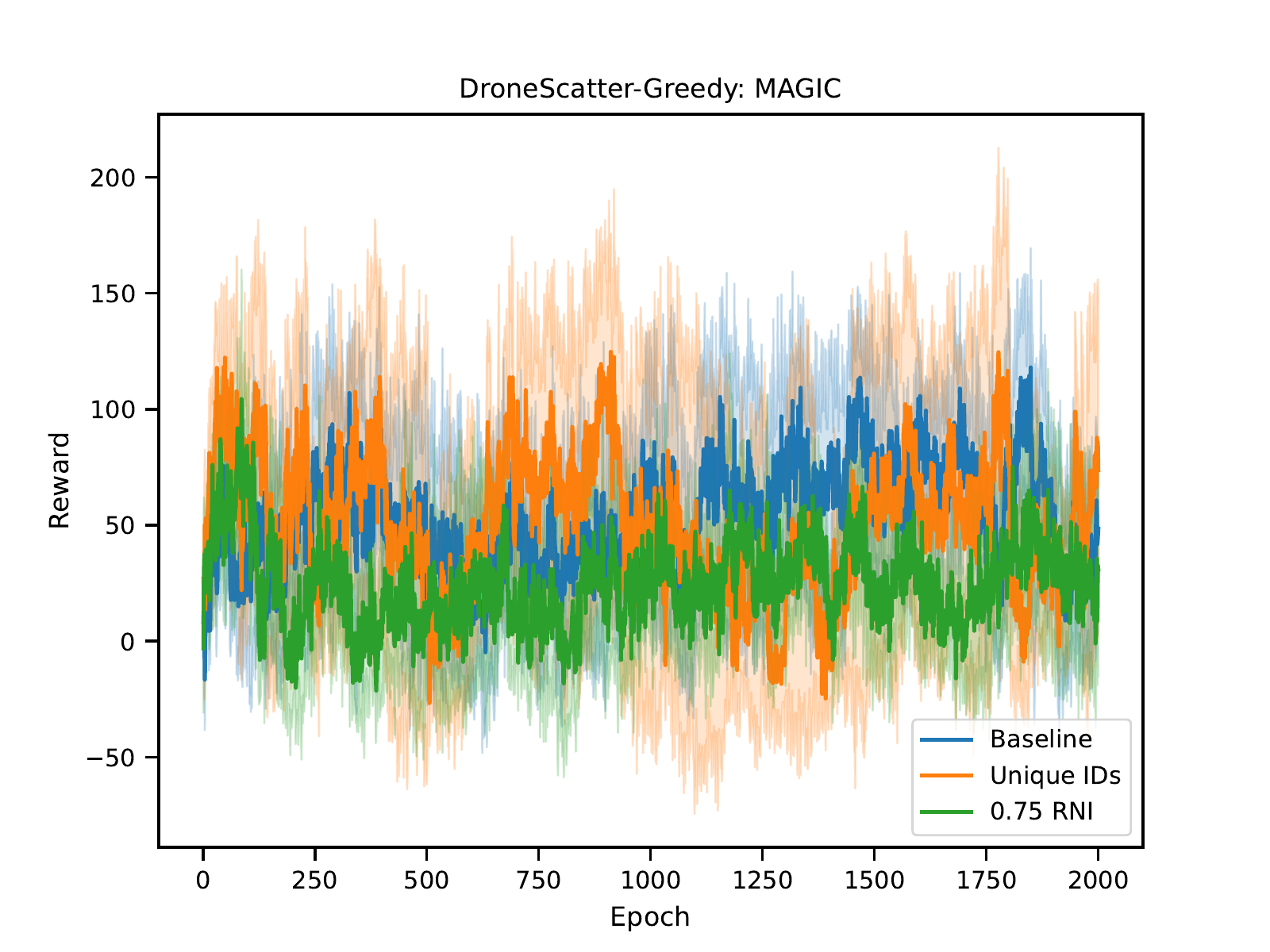}

\includegraphics[width=.32\linewidth]{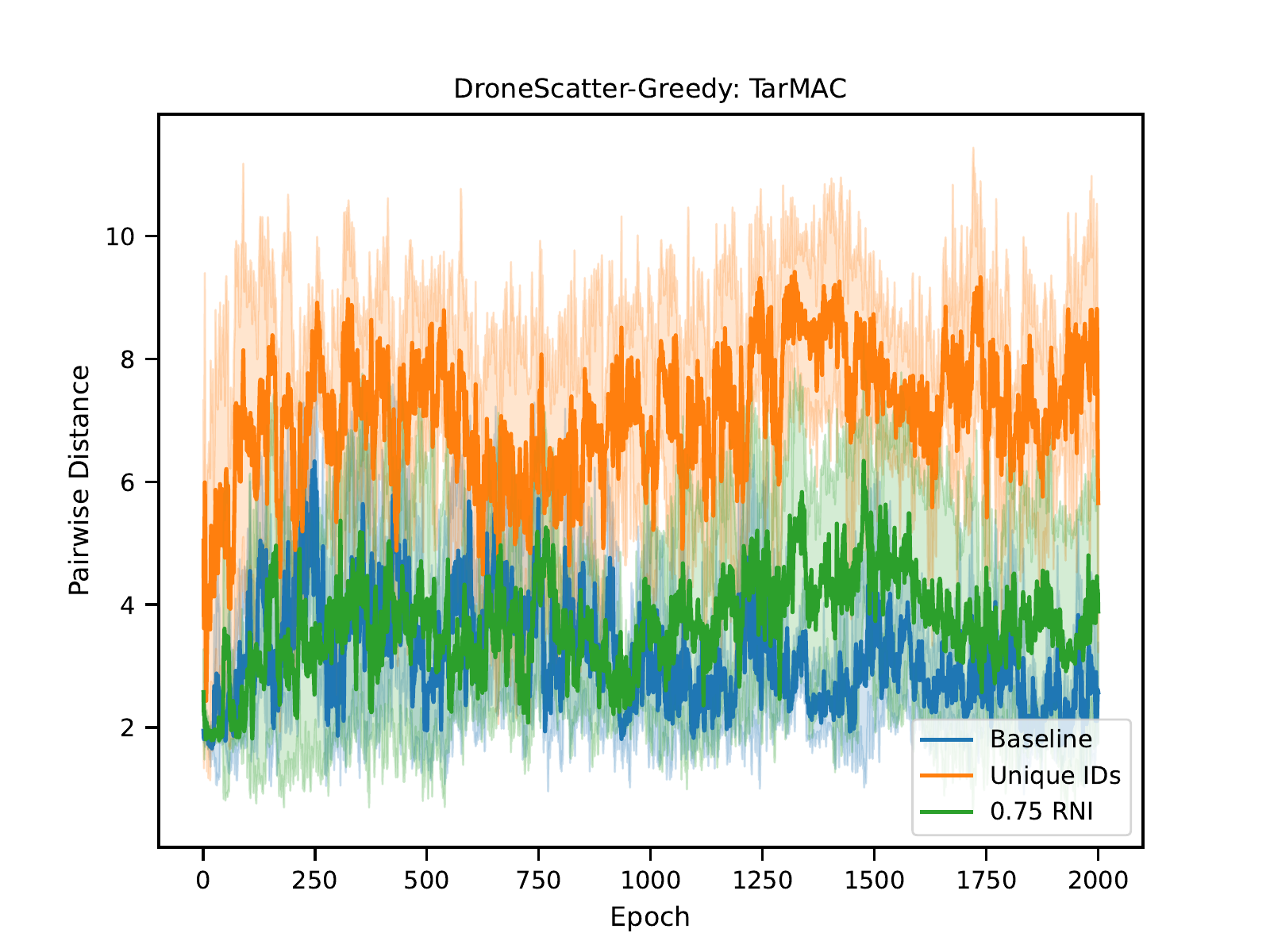}
\includegraphics[width=.32\linewidth]{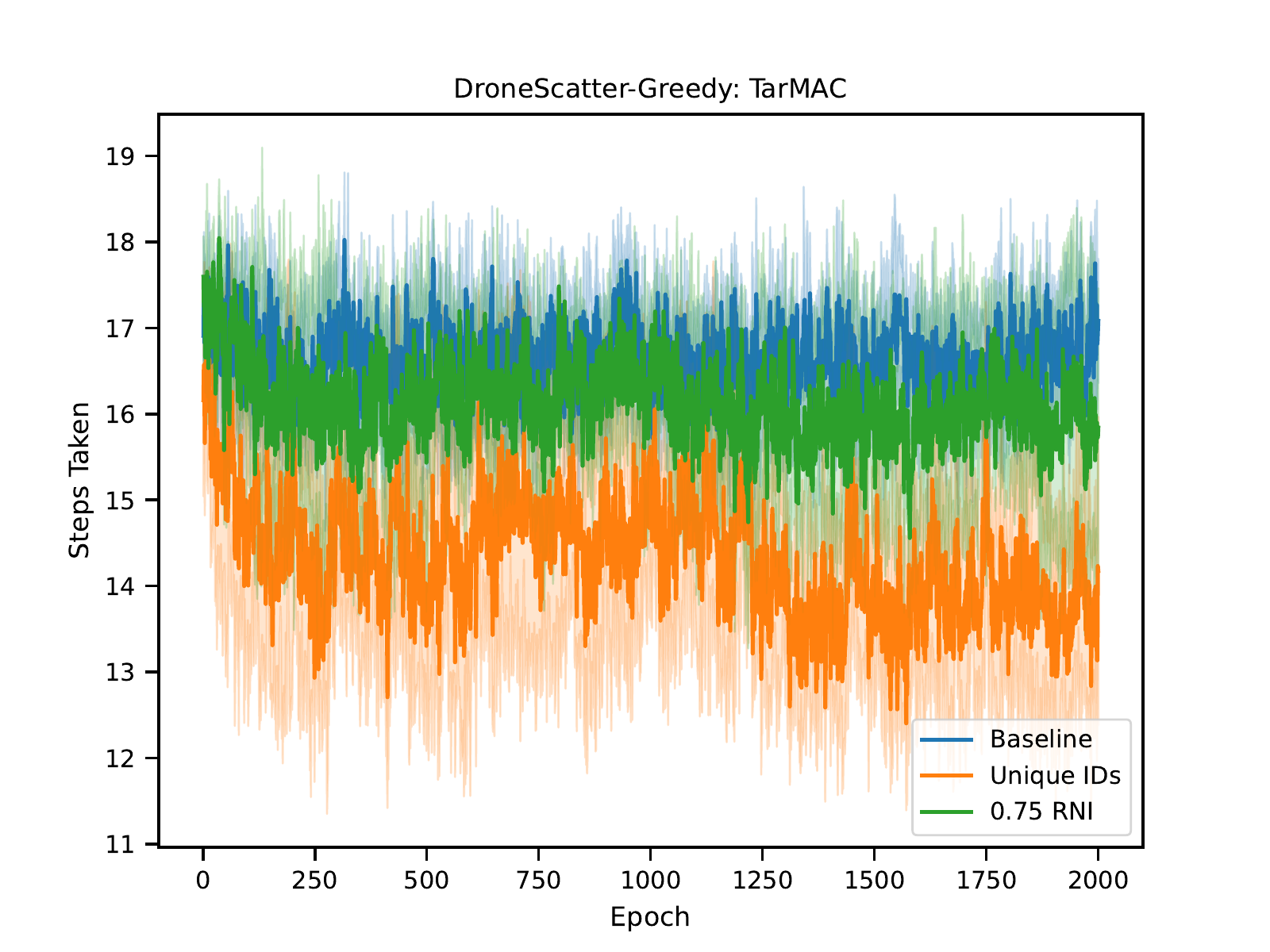}
\includegraphics[width=.32\linewidth]{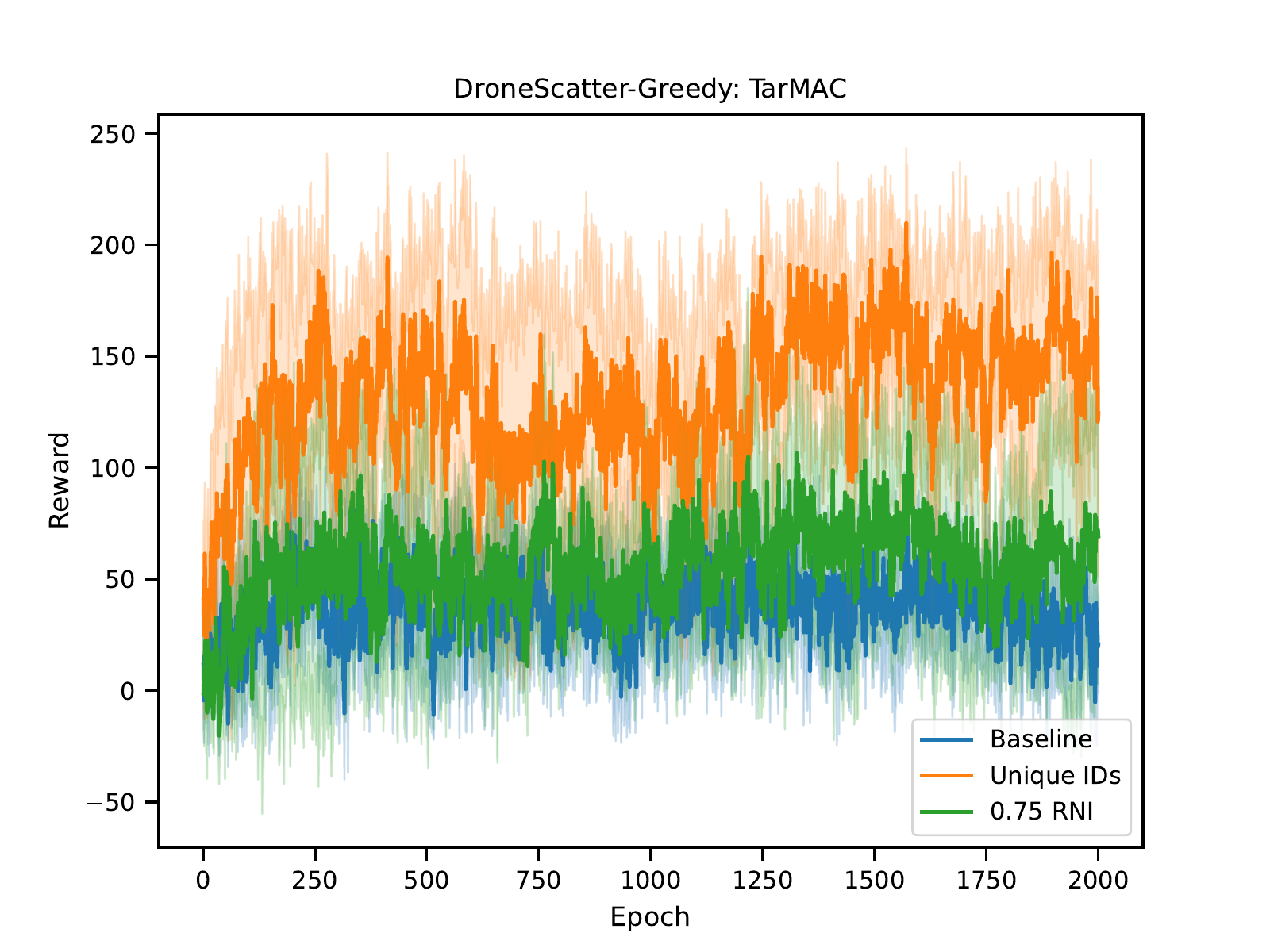}

\includegraphics[width=.32\linewidth]{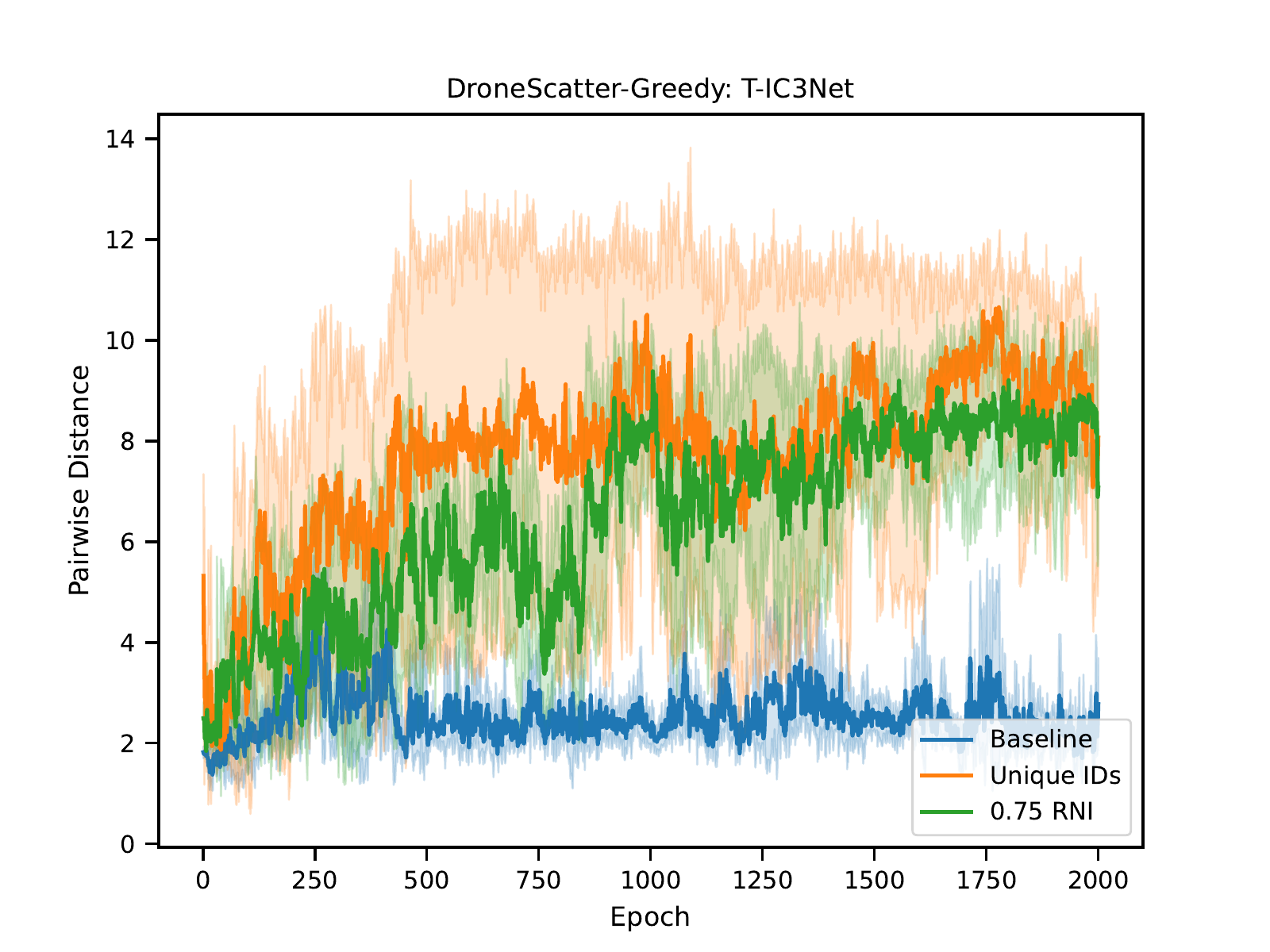}
\includegraphics[width=.32\linewidth]{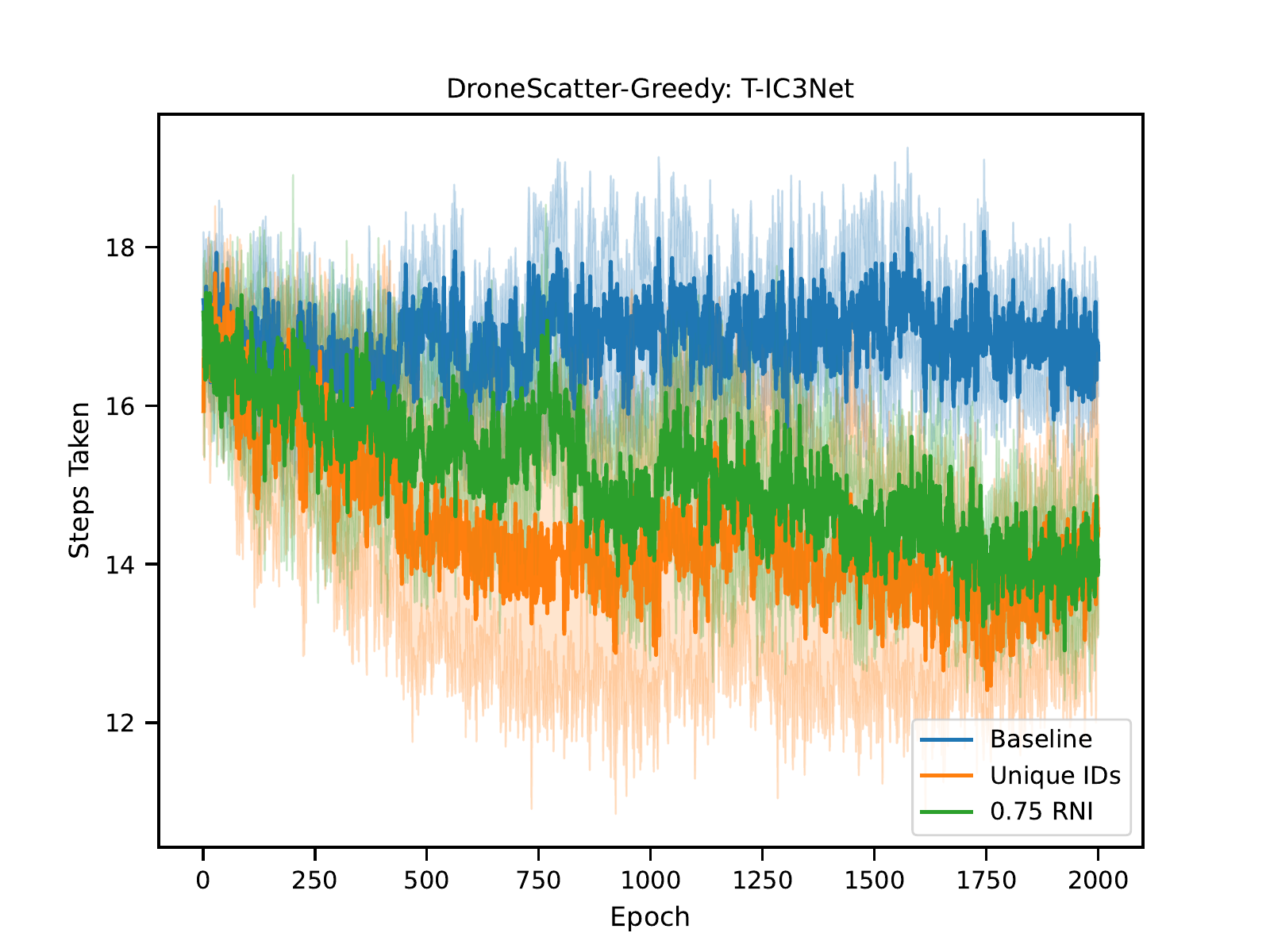}
\includegraphics[width=.32\linewidth]{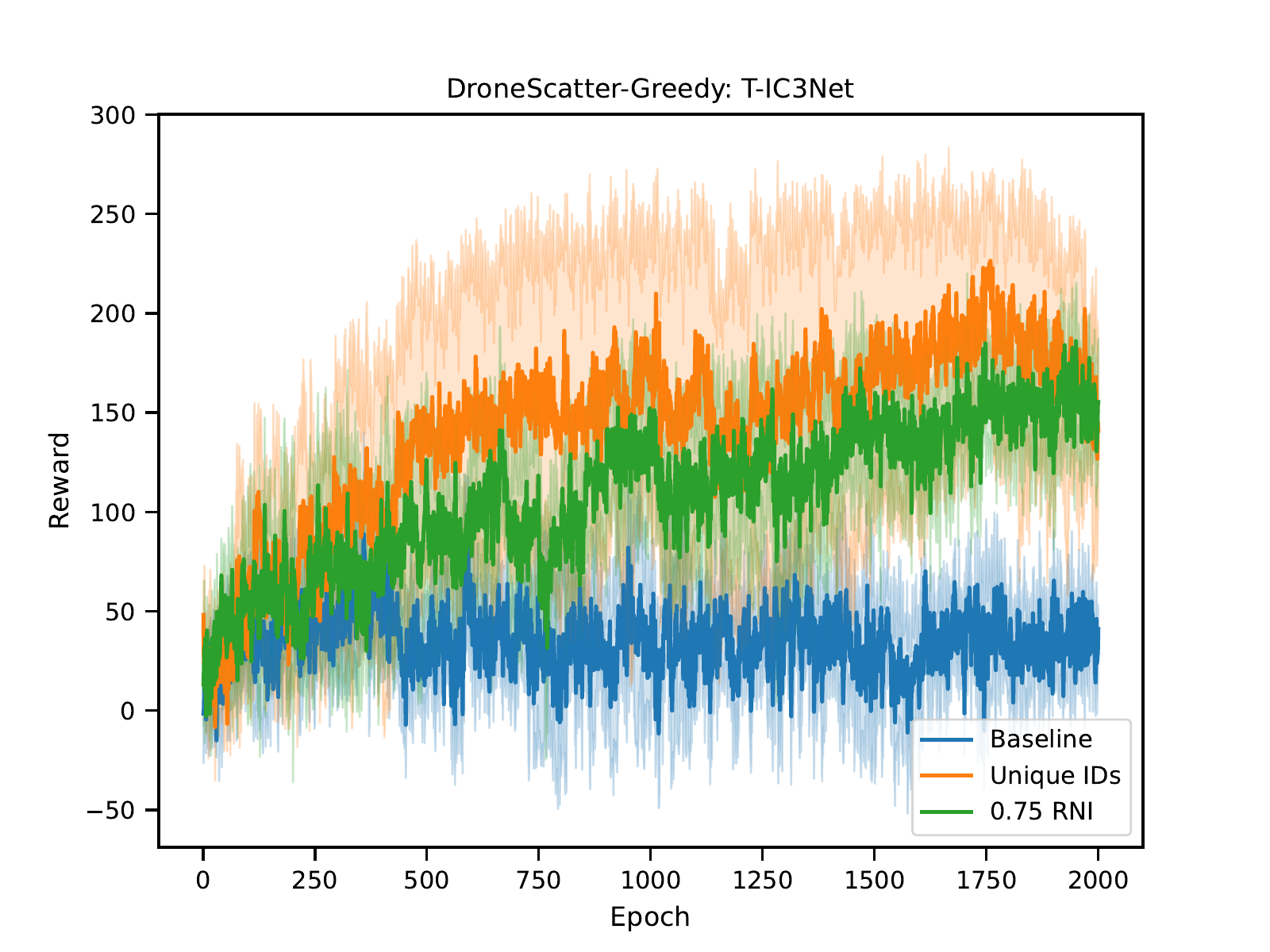}

%
%
\subsection{DroneScatter with Stochastic Evaluation}
\centering

\includegraphics[width=.32\linewidth]{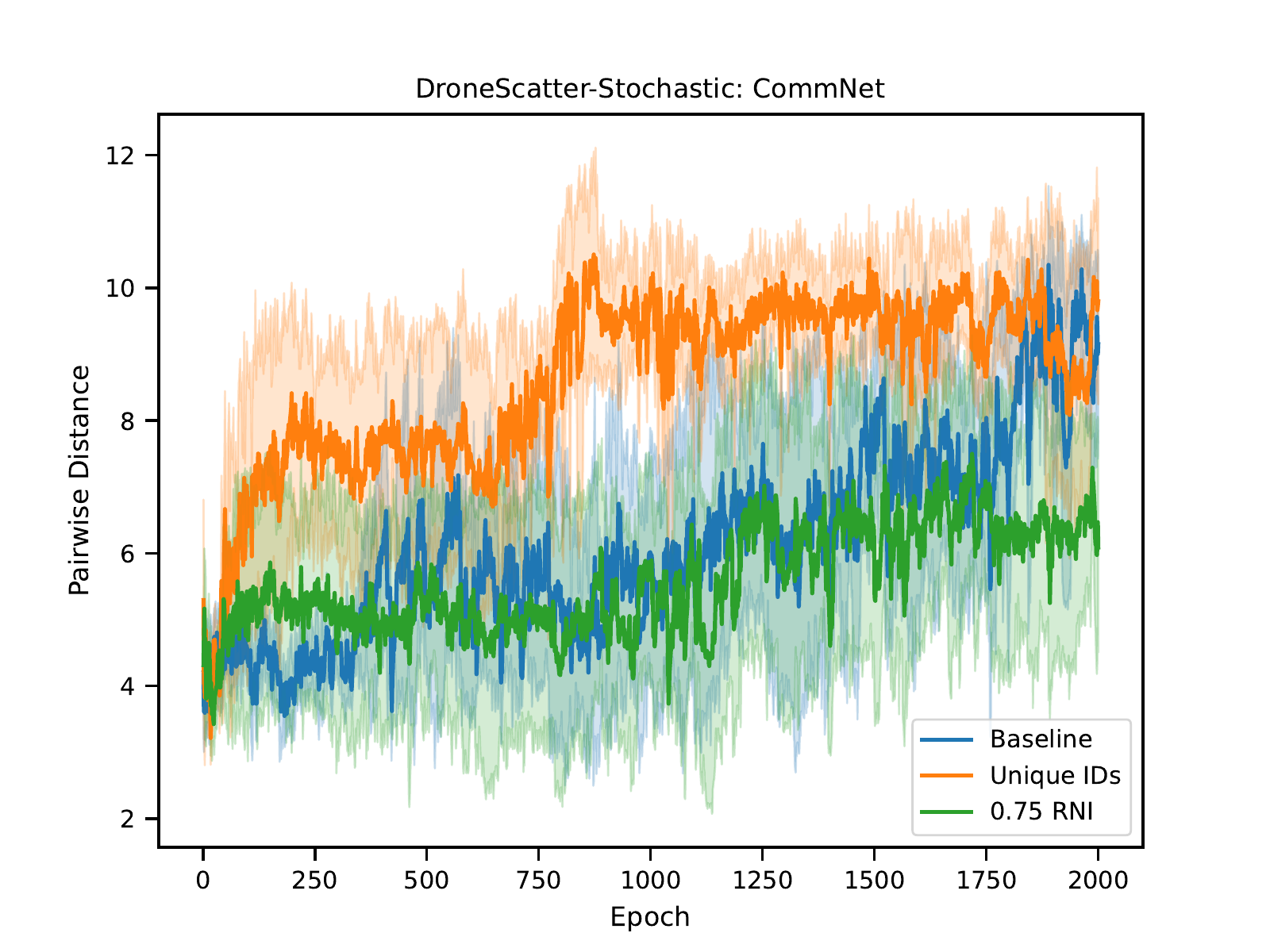}
\includegraphics[width=.32\linewidth]{results/DroneScatter-Stochastic/DroneScatter-Stochastic_commnet_steps_taken.pdf}
\includegraphics[width=.32\linewidth]{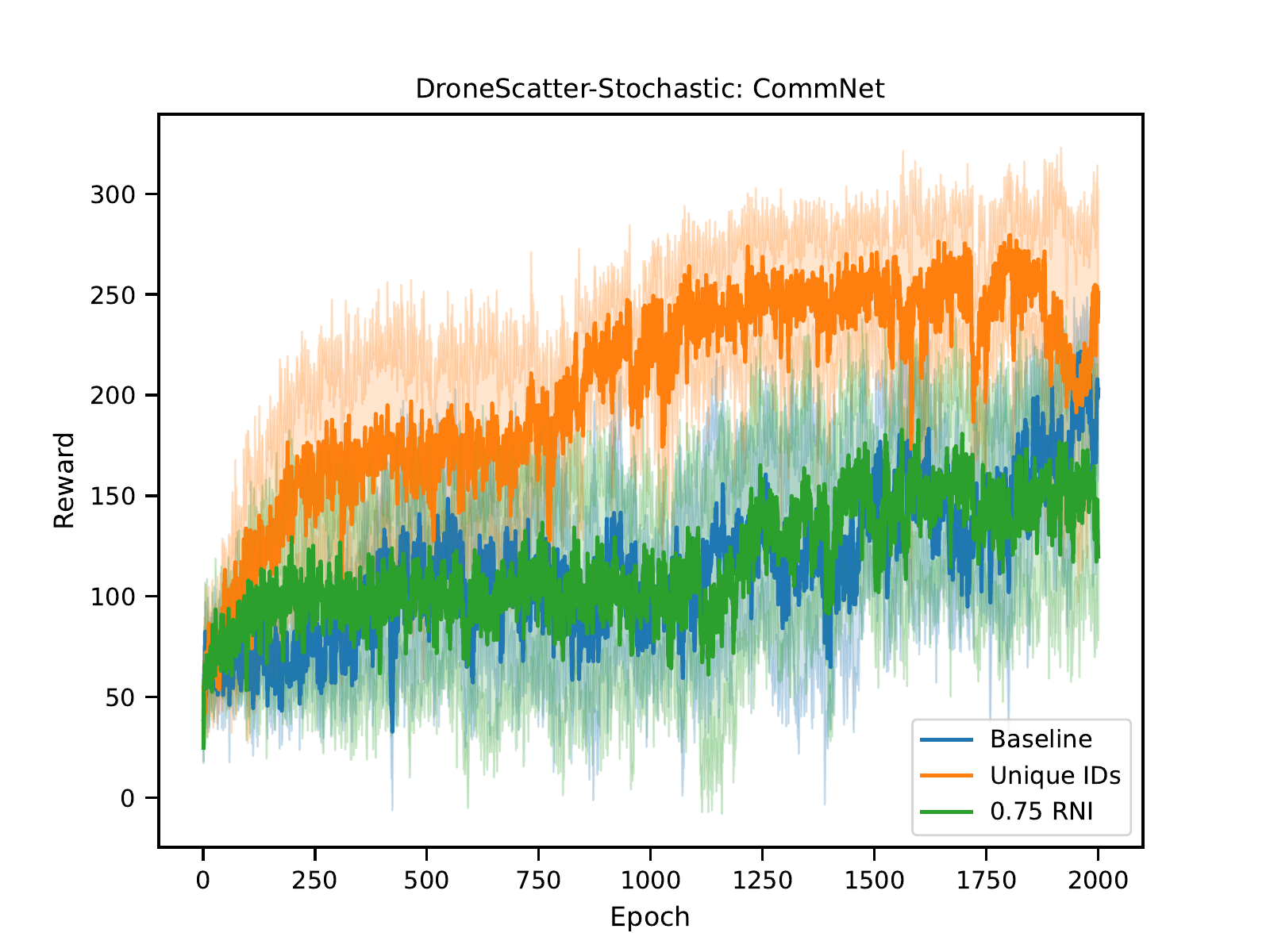}

\includegraphics[width=.32\linewidth]{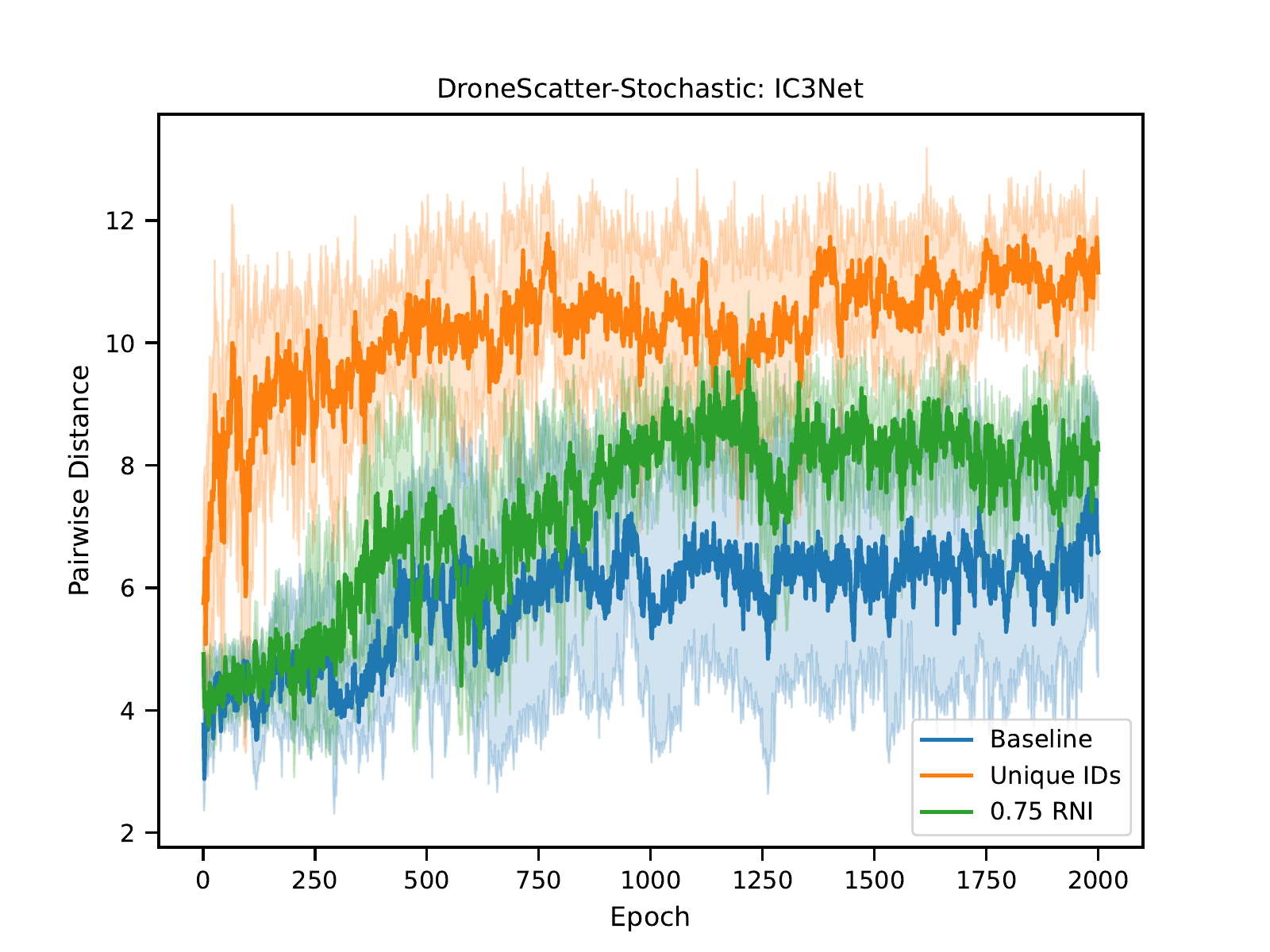}
\includegraphics[width=.32\linewidth]{results/DroneScatter-Stochastic/DroneScatter-Stochastic_ic3net_steps_taken.pdf}
\includegraphics[width=.32\linewidth]{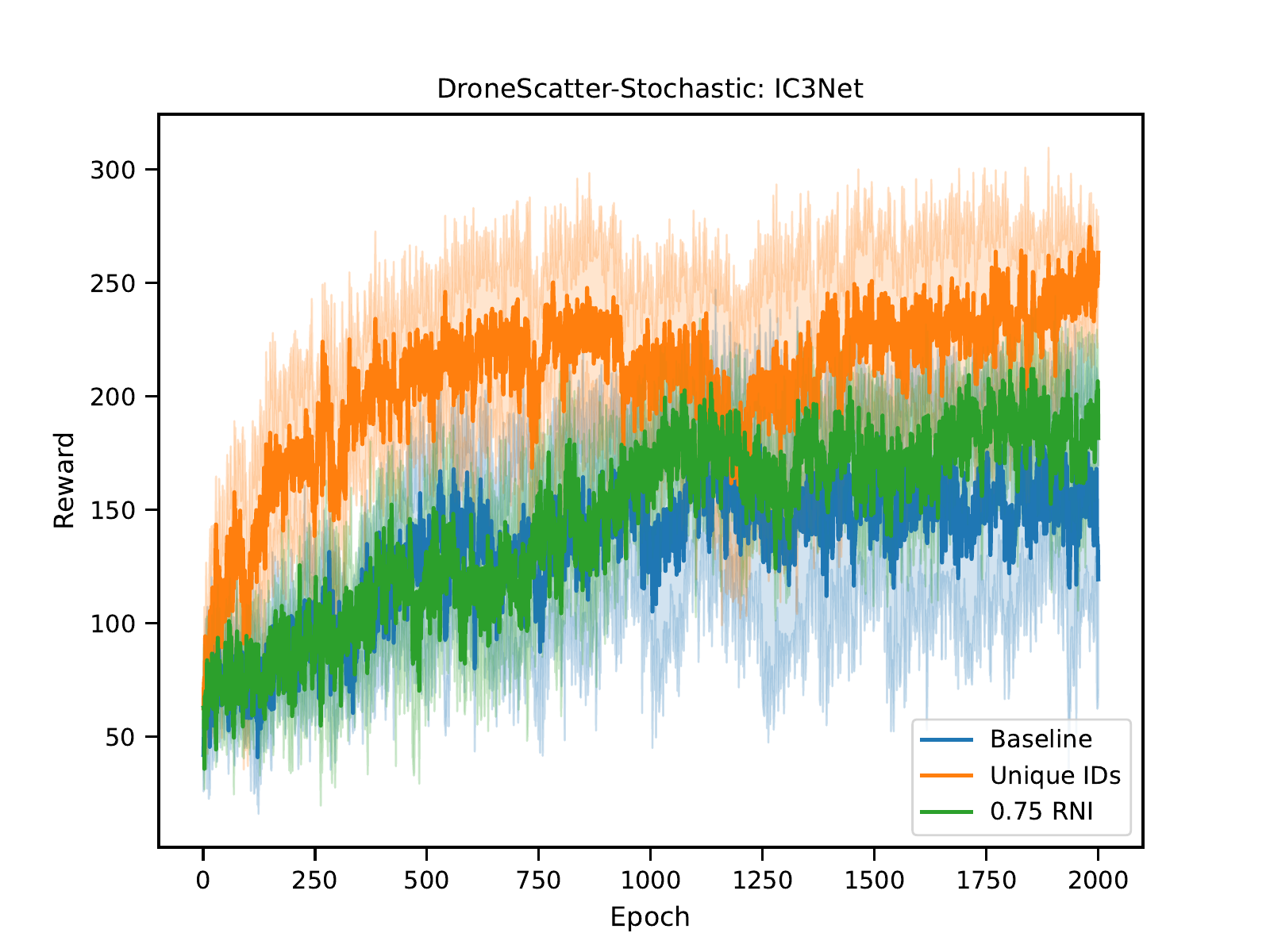}

\includegraphics[width=.32\linewidth]{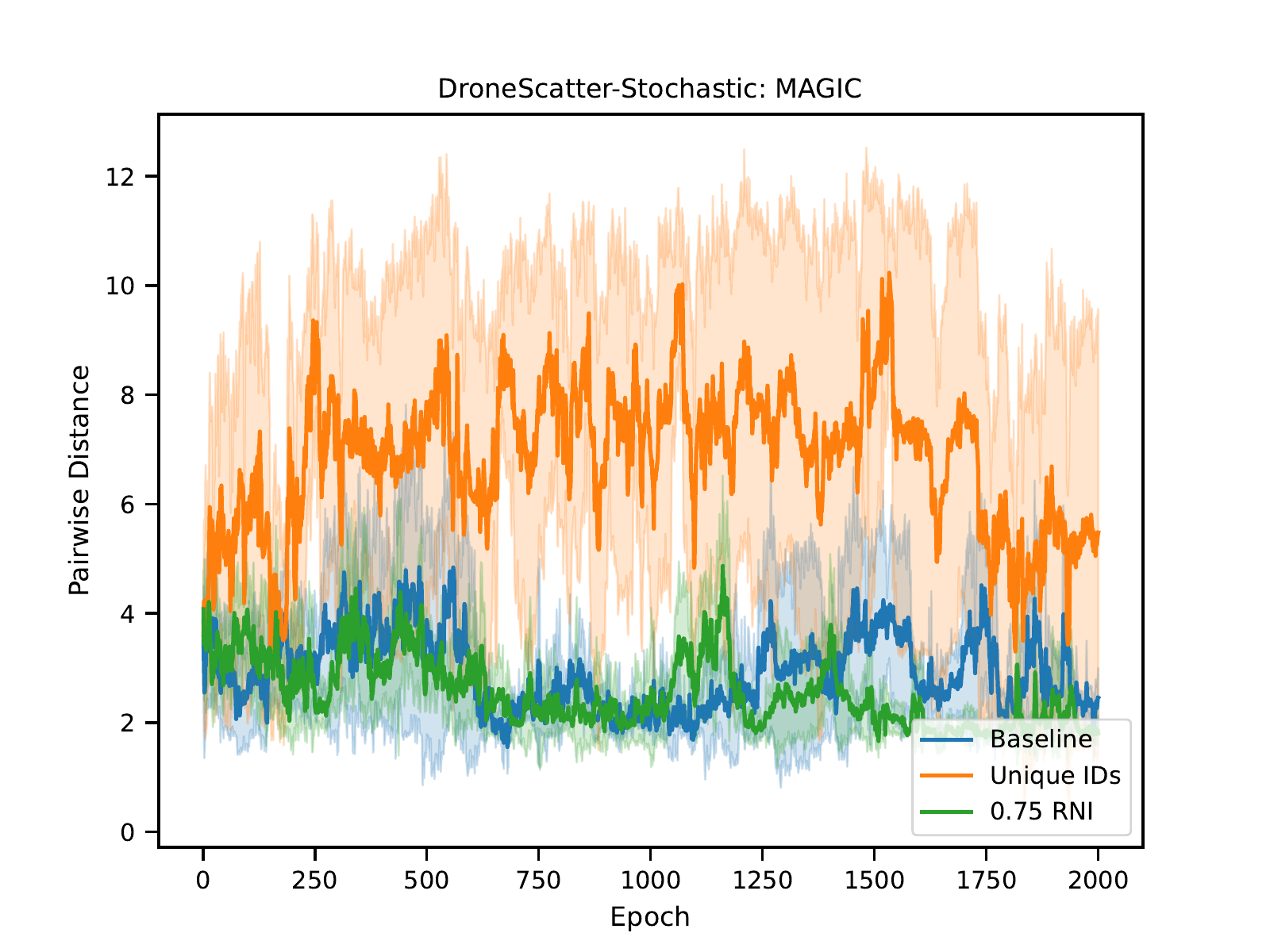}
\includegraphics[width=.32\linewidth]{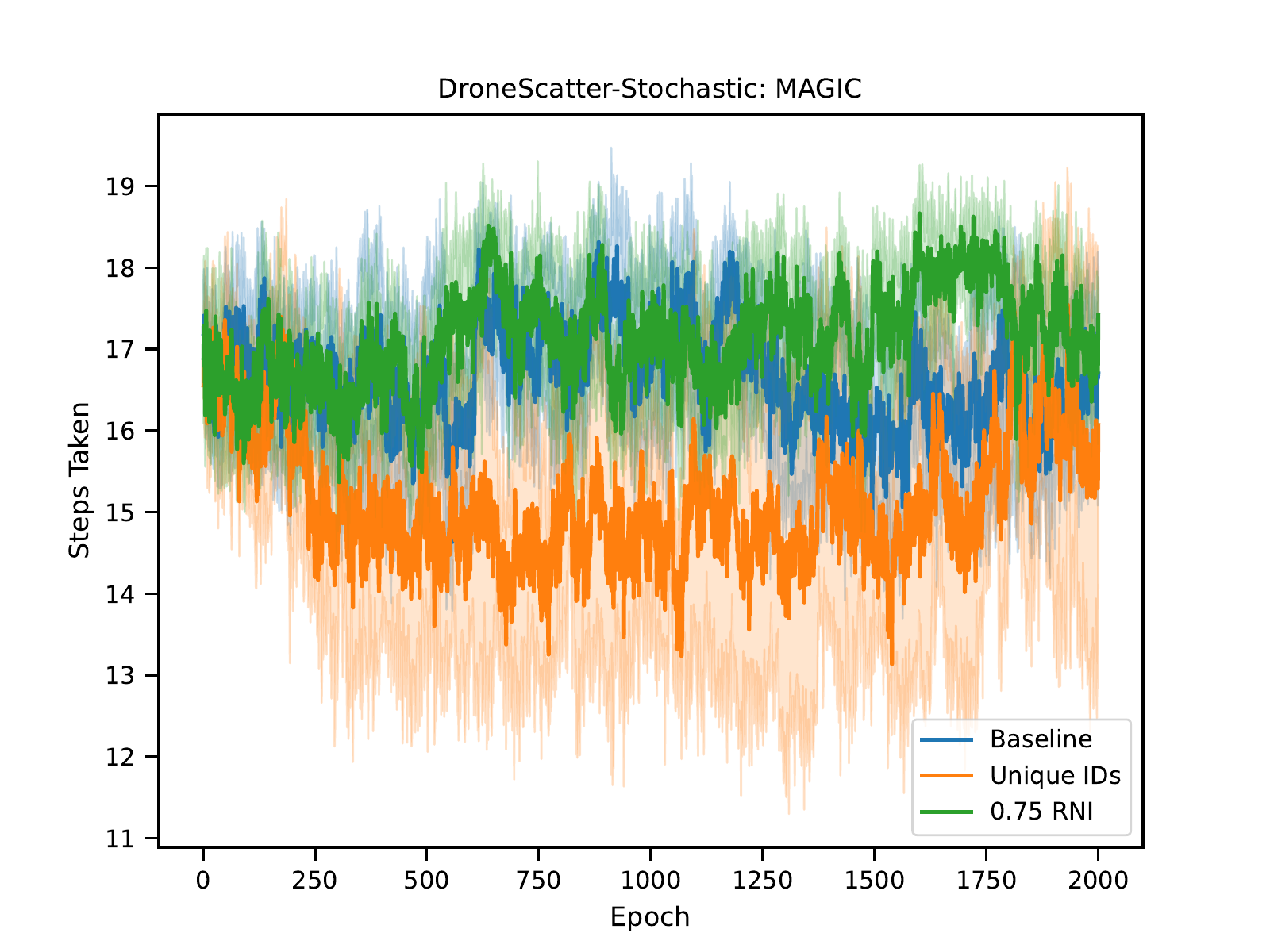}
\includegraphics[width=.32\linewidth]{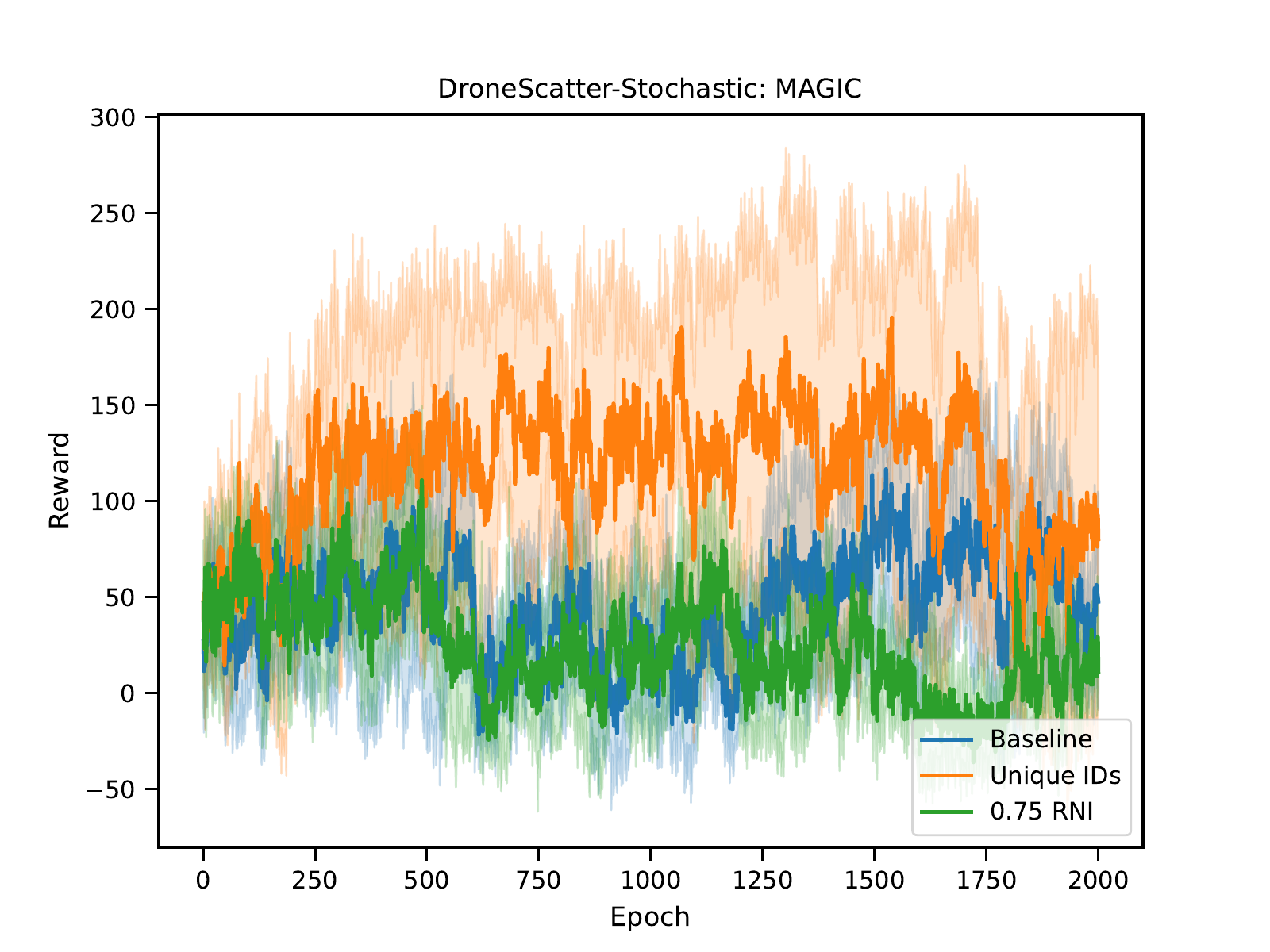}

\includegraphics[width=.32\linewidth]{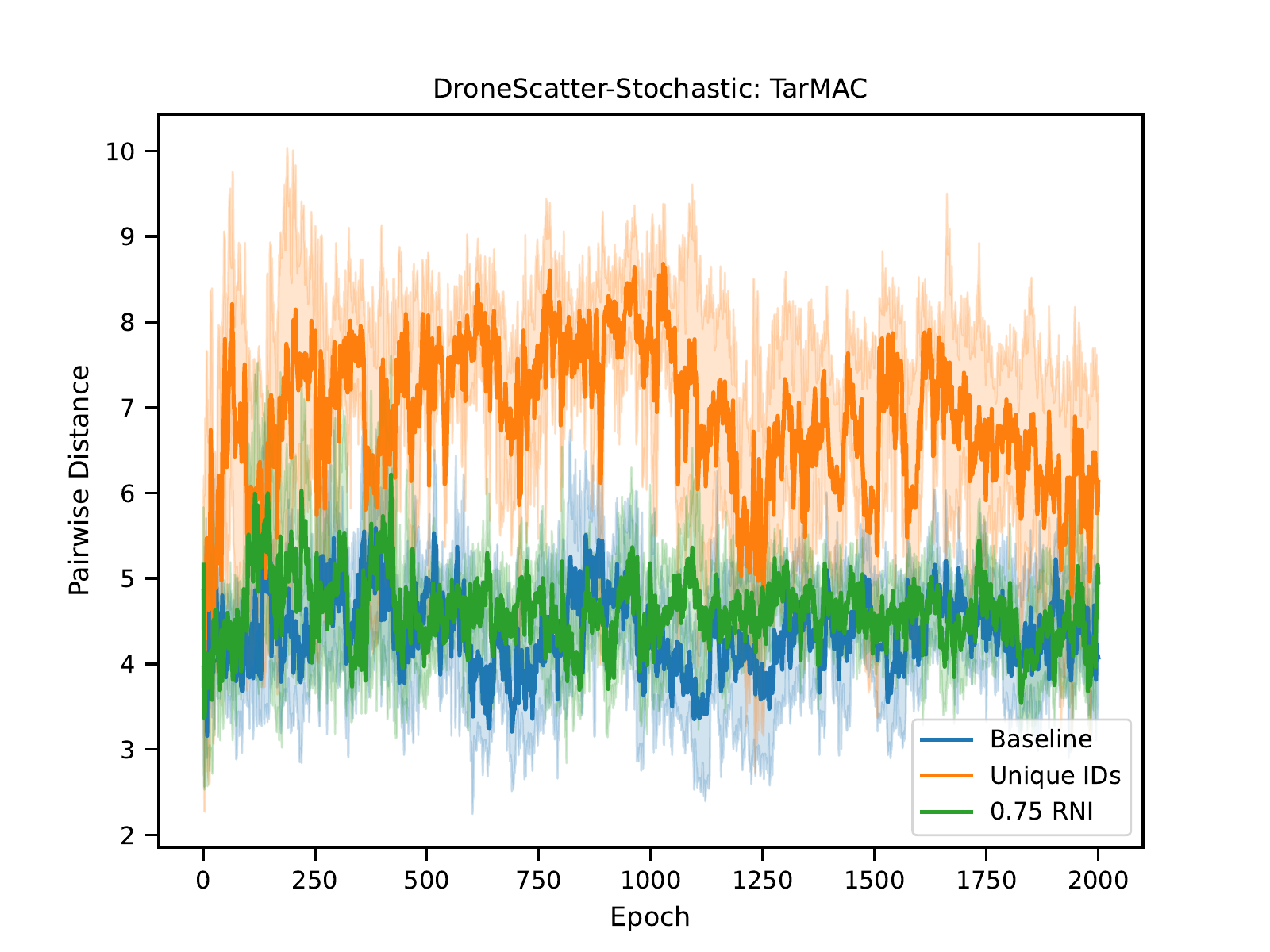}
\includegraphics[width=.32\linewidth]{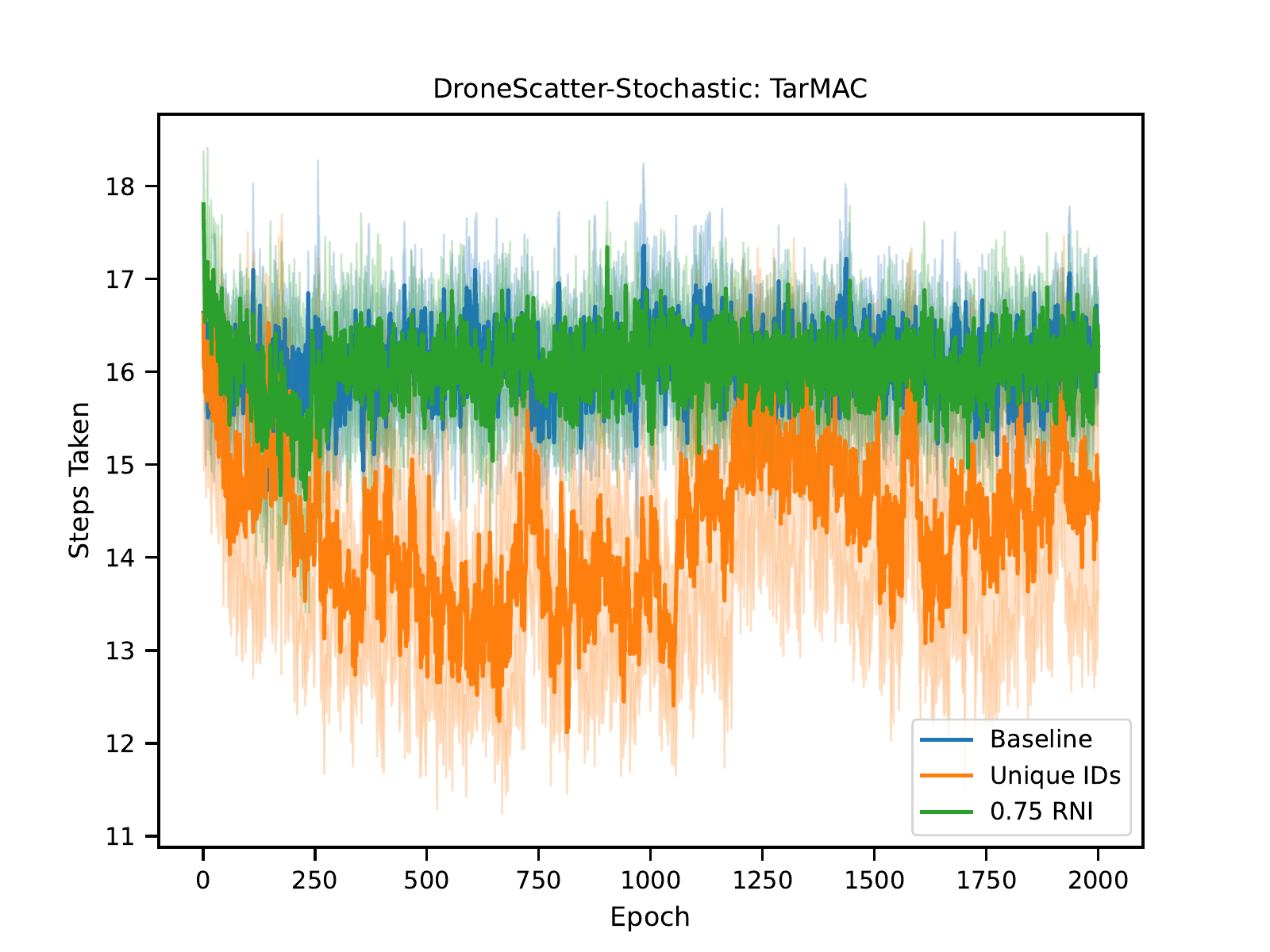}
\includegraphics[width=.32\linewidth]{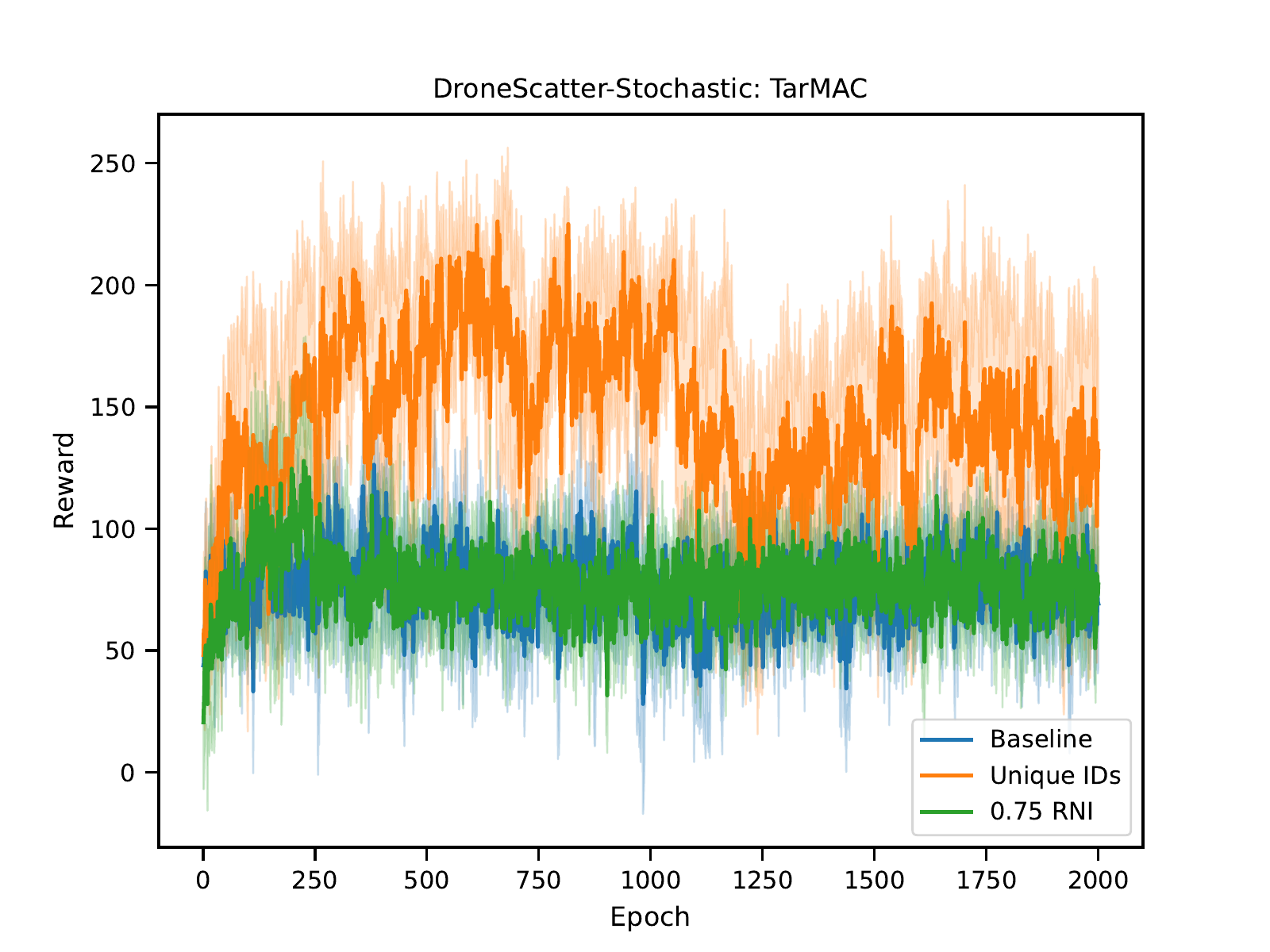}

\includegraphics[width=.32\linewidth]{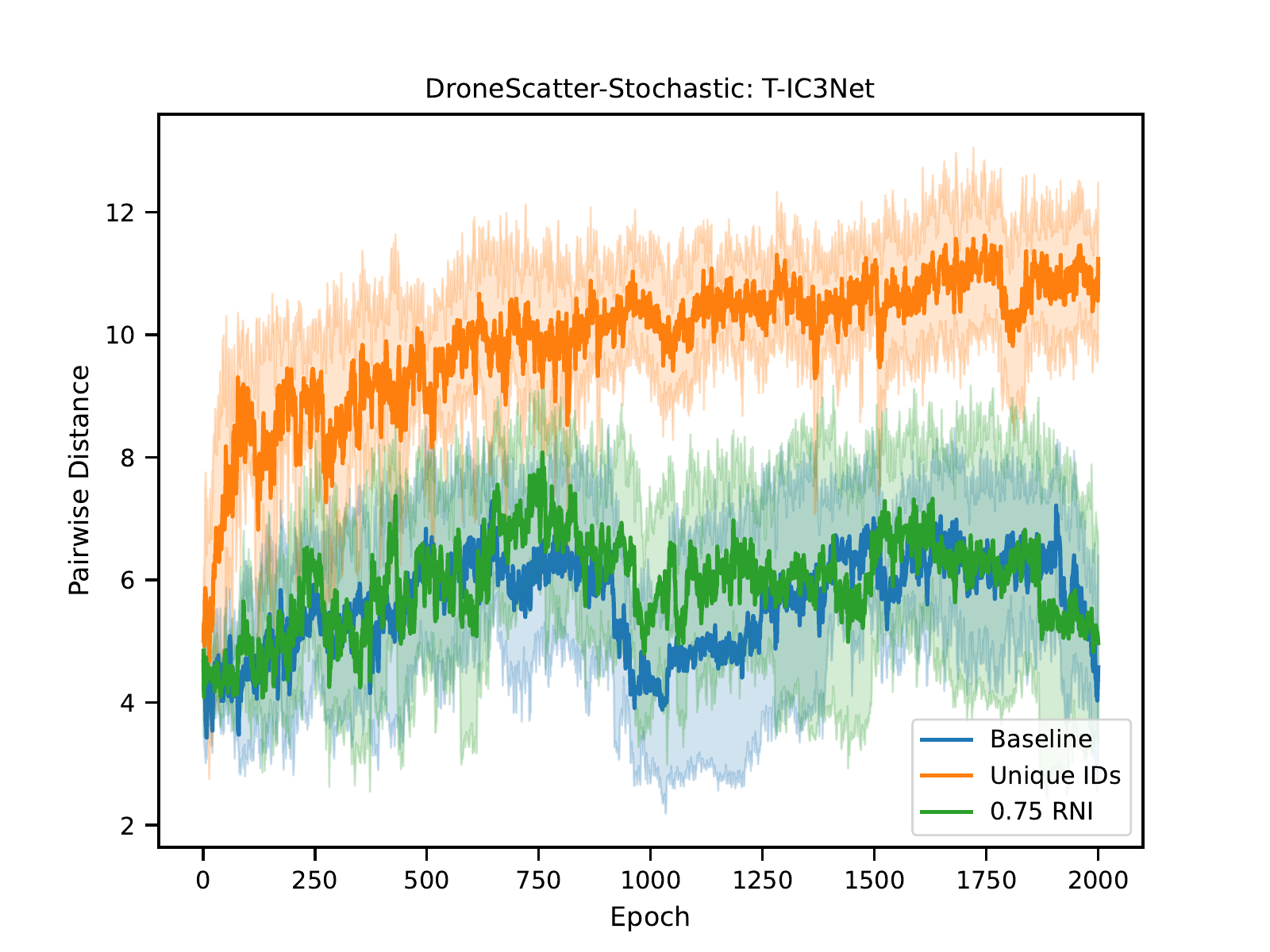}
\includegraphics[width=.32\linewidth]{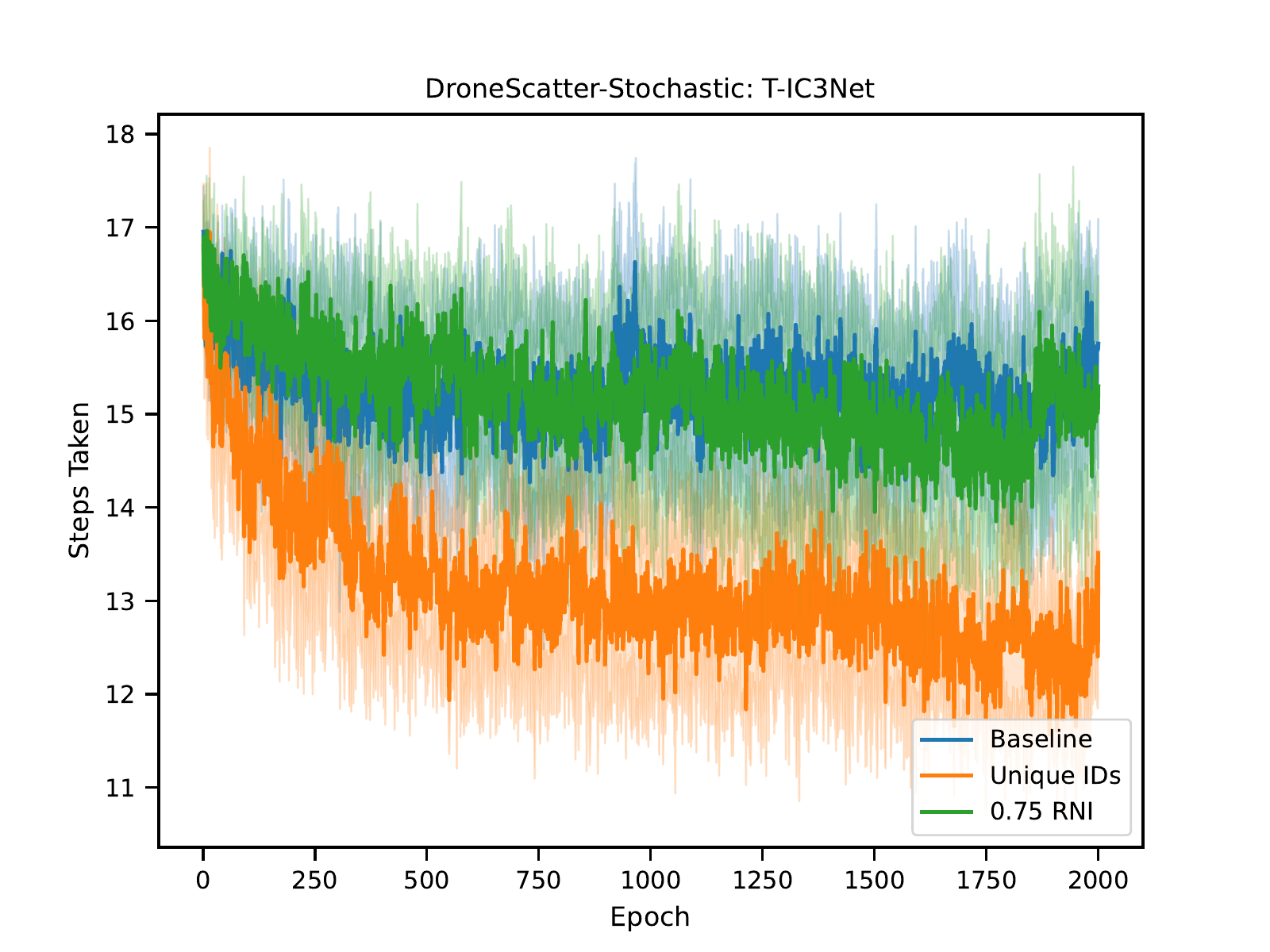}
\includegraphics[width=.32\linewidth]{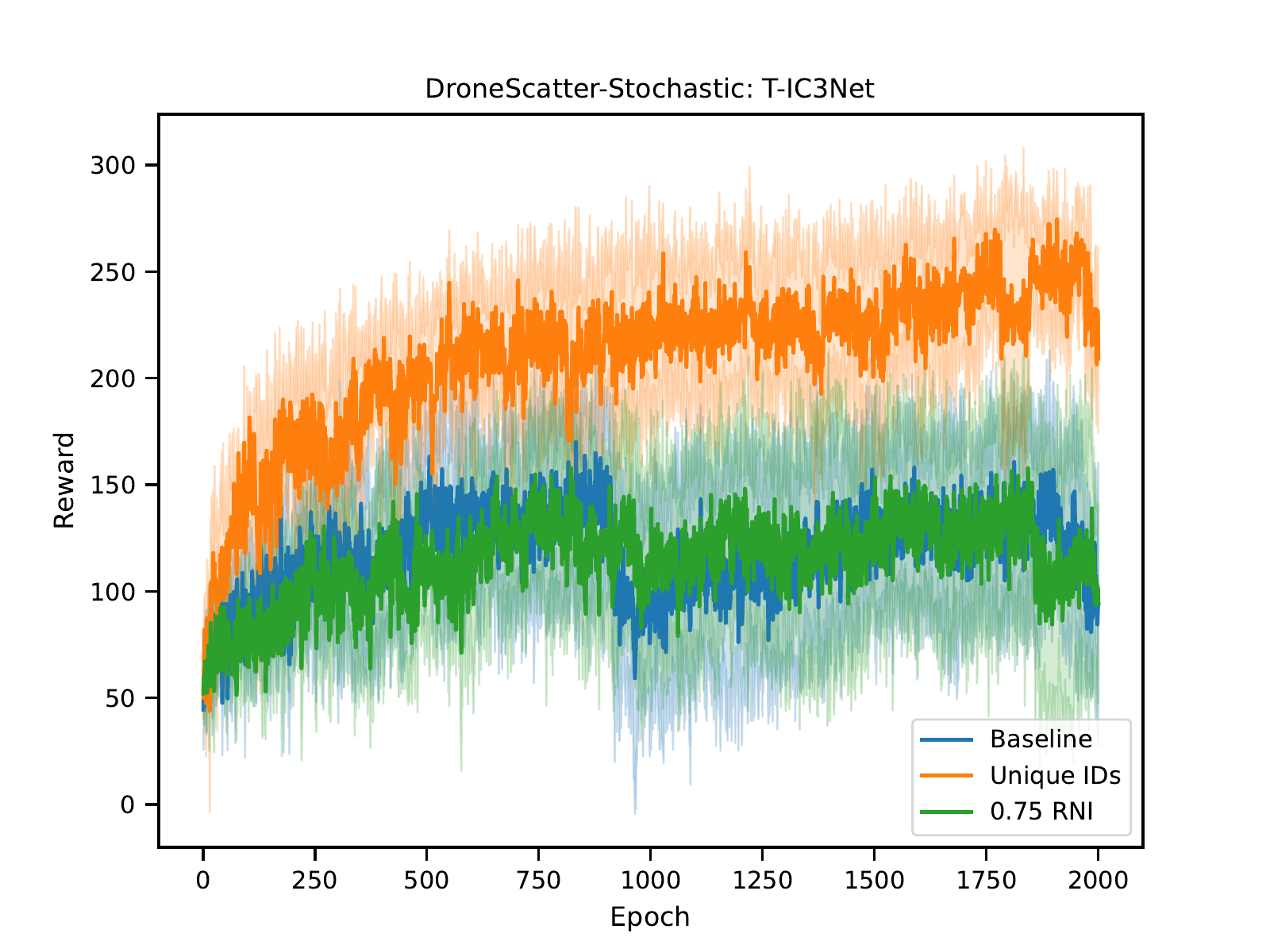}

\end{document}